\documentclass[11pt,american,british,english,intoc]{report}
\usepackage[T1]{fontenc}
\usepackage[latin9]{inputenc}
\usepackage[a4paper]{geometry}
\usepackage{color}
\usepackage{babel}
\usepackage{float}
\usepackage{rotfloat}
\usepackage{textcomp}
\usepackage{amsmath}
\usepackage{amsthm}
\usepackage{amssymb}
\usepackage{graphicx}
\usepackage{setspace}
\usepackage[authoryear]{natbib}
\usepackage{nomencl}
\providecommand{\printnomenclature}{\printglossary}
\providecommand{\makenomenclature}{\makeglossary}
\makenomenclature
\usepackage[unicode=true,pdfusetitle,
 bookmarks=true,bookmarksnumbered=true,bookmarksopen=false,
 breaklinks=true,pdfborder={0 0 1},backref=false,colorlinks=true]
 {hyperref}
\hypersetup{
 linkcolor=darkblue, citecolor=darkblue, urlcolor=darkblue, pdfstartview=XYZ, plainpages=false, pdfpagelabels}

\makeatletter

\providecommand{\tabularnewline}{\\}
\floatstyle{ruled}
\newfloat{algorithm}{tbp}{loa}[chapter]
\providecommand{\algorithmname}{Algorithm}
\floatname{algorithm}{\protect\algorithmname}

\numberwithin{equation}{section}
\numberwithin{figure}{section}
\numberwithin{table}{section}
\newcommand{\lyxrightaddress}[1]{
	\par {\raggedleft \begin{tabular}{l}\ignorespaces
	#1
	\end{tabular}
	\vspace{1.4em}
	\par}
}
\theoremstyle{plain}
\newtheorem{thm}{\protect\theoremname}[section]
\theoremstyle{plain}
\newtheorem{prop}[thm]{\protect\propositionname}
\theoremstyle{definition}
\newtheorem{defn}[thm]{\protect\definitionname}
\theoremstyle{plain}
\newtheorem{lem}[thm]{\protect\lemmaname}
\theoremstyle{remark}
\newtheorem{rem}[thm]{\protect\remarkname}
\theoremstyle{definition}
\newtheorem{example}[thm]{\protect\examplename}
\theoremstyle{plain}
\newtheorem{cor}[thm]{\protect\corollaryname}

\usepackage{graphicx}
\usepackage{color}
\usepackage{url}
\usepackage{hyperref}

\definecolor{darkblue}{rgb}{0,0,0.7}

\usepackage[titles]{tocloft}

\renewcommand{\@tocrmarg}{1.75em}

\usepackage{fancyhdr}
\newcommand{\headerfont}{\color[gray]{0.7}\fontfamily{phv}\fontsize{10}{11}\selectfont}
\newcommand{\footerfont}{\fontfamily{phv}\fontsize{11}{11}\selectfont}
\makeatother

\addto\captionsamerican{%
}
\addto\captionsamerican{\renewcommand{\algorithmname}{Algorithm}}
\addto\captionsamerican{\renewcommand{\corollaryname}{Corollary}}
\addto\captionsamerican{\renewcommand{\definitionname}{Definition}}
\addto\captionsamerican{\renewcommand{\examplename}{Example}}
\addto\captionsamerican{\renewcommand{\lemmaname}{Lemma}}
\addto\captionsamerican{\renewcommand{\propositionname}{Proposition}}
\addto\captionsamerican{\renewcommand{\remarkname}{Remark}}
\addto\captionsamerican{\renewcommand{\theoremname}{Theorem}}
\addto\captionsbritish{%
}
\addto\captionsbritish{\renewcommand{\algorithmname}{Algorithm}}
\addto\captionsbritish{\renewcommand{\corollaryname}{Corollary}}
\addto\captionsbritish{\renewcommand{\definitionname}{Definition}}
\addto\captionsbritish{\renewcommand{\examplename}{Example}}
\addto\captionsbritish{\renewcommand{\lemmaname}{Lemma}}
\addto\captionsbritish{\renewcommand{\propositionname}{Proposition}}
\addto\captionsbritish{\renewcommand{\remarkname}{Remark}}
\addto\captionsbritish{\renewcommand{\theoremname}{Theorem}}
\addto\captionsenglish{%
}
\addto\captionsenglish{\renewcommand{\corollaryname}{Corollary}}
\addto\captionsenglish{\renewcommand{\definitionname}{Definition}}
\addto\captionsenglish{\renewcommand{\examplename}{Example}}
\addto\captionsenglish{\renewcommand{\lemmaname}{Lemma}}
\addto\captionsenglish{\renewcommand{\propositionname}{Proposition}}
\addto\captionsenglish{\renewcommand{\remarkname}{Remark}}
\addto\captionsenglish{\renewcommand{\theoremname}{Theorem}}

\providecommand{\corollaryname}{Corollary}
\providecommand{\definitionname}{Definition}
\providecommand{\examplename}{Example}
\providecommand{\lemmaname}{Lemma}
\providecommand{\propositionname}{Proposition}
\providecommand{\remarkname}{Remark}
\providecommand{\theoremname}{Theorem}

\begin{document}
\title{Variational Bayes Inference in Digital Receivers}
\author{Viet Hung Tran}
\date{\includegraphics[width=0.1\textwidth]{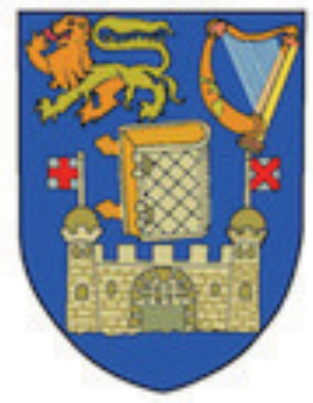}\\
A thesis submitted to Trinity College Dublin \\
for the degree of Doctor of Philosophy \\
(June 2014) \\
Supervisor: Assoc. Prof. Anthony Quinn }
\maketitle
\begin{abstract}
The digital telecommunications receiver is an important context for
inference methodology, the key objective being to minimize the expected
loss function in recovering the transmitted information. For that
criterion, the optimal decision is the Bayesian minimum-risk estimator.
However, the computational load of the Bayesian estimator is often
prohibitive and, hence, efficient computational schemes are required.
The design of novel schemes---striking new balances between accuracy
and computational load---is the primary concern of this thesis. 

Because Bayesian methodology seeks to construct the joint distribution
of all uncertain parameters in a hierarchial manner, its computational
complexity is often prohibitive. A solution for efficient computation
is to re-factorize this joint model into an appropriate conditionally
independent (CI) structure, whose factors are Markov models of appropriate
order. By tuning the order from maximum to minimum, this Markov factorization
is applicable to all parametric models. The associated computational
complexity ranges from prohibitive to minimal. For efficient Bayesian
computation, two popular techniques, one exact and one approximate,
will be studied in this thesis, as described next. 

The exact scheme is a recursive one, namely the generalized distributive
law (GDL), whose purpose is to distribute all operators across the
CI factors of the joint model, so as to reduce the total number of
operators required. In a novel theorem derived in this thesis, GDL---if
applicable---will be shown to guarantee such a reduction in all cases.
An associated lemma also quantifies this reduction. For practical
use, two novel algorithms, namely the no-longer-needed (NLN) algorithm
and the generalized form of the Forward-Backward (FB) algorithm, recursively
factorizes and computes the CI factors of an arbitrary model, respectively. 

The approximate scheme is an iterative one, namely the Variational
Bayes (VB) approximation, whose purpose is to find the independent
(i.e. zero-order Markov) model closest to the true joint model in
the minimum Kullback-Leibler divergence (KLD) sense. Despite being
computationally efficient, this naive mean field approximation confers
only modest performance for highly correlated models. A novel approximation,
namely Transformed Variational Bayes (TVB), will be designed in the
thesis in order to relax the zero-order constraint in the VB approximation,
further reducing the KLD of the optimal approximation. 

Together, the GDL and VB schemes are able to provide a range of trade-offs
between accuracy and speed in digital receivers. Two demodulation
problems in digital receivers will be considered in this thesis, the
first being a Markov-based symbol detector, and the second being a
frequency estimator for synchronization. The first problem will be
solved using a novel accelerated scheme for VB inference of a hidden
Markov chain (HMC). When applied to weakly correlated M-state HMCs
with n samples, this accelerated scheme reduces the computational
load from $O(nM^{2})$ in the state-of-the-art Viterbi algorithm to
$O(nM)$, with comparable accuracy. The second problem is addressed
via the TVB approximation. Although its performance is only modest
in simulation, it nevertheless opens up new opportunities for approximate
Bayesian inference to address high Quality-of-Service (QoS) tasks
in 4G mobile networks. 
\end{abstract}
\clearpage{}

\section*{Acknowledgements}

Looking back, it seems to me that writing this thesis is a natural
consequence of my life:\\

Firstly, I would like to thank my supervisor, Prof.$\ $Anthony Quinn,
for all the training and encouragement during my time at Trinity.
In fact, he is the best supervisor I could ever hope for. Without
his guidance and support, this thesis would never be complete.\\

I wish to thank Prof.$\ $Jean-Pierre Barbot and Prof.$\ $Pascal
Larzabal for encouraging me to pursue Bayesian methodology when I
finished my masters in ENS Cachan, Paris. \\

I also wish to acknowledge the academic support of Ho Chi Minh City
University of Technology, of which I was an undergraduate student
and currently am a lecturer. \\

Regarding my family, I am especially grateful to my elder brother,
Dr.$\ $Viet Hong Tran, whom I have followed since the time I was
born - from preliminary, middle, high school to the same college.
Without his passion for physics and electronics, I would never have
pursued them either. It is rather interesting to note that our sole
difference in our pathways is the time after entering College. While
he studied telecommunicaitons as a undergraduate and finished his
PhD thesis on automatics, I did exactly the opposite, i.e. I studied
automatics as a undergraduate and finished this PhD thesis on telecommunications.\\

Last, but definitely not least, I wish to dedicate this thesis to
my parents, whom I love much more than myself. Actually, writing this
thesis is the best way I know to make them feel proud of me.\\

\lyxrightaddress{\textbf{Viet Hung Tran }\\
University of Dublin, Trinity College \\
June 2014}

\clearpage{}

\section*{Summary }

This thesis is primarily concerned with the trade-off between computational
complexity and accuracy in digital receivers. Furthermore, because
OFDM modulation is key in 4G systems, there is an interest in better
demodulation schemes for the fading channel, which is the environment
that all mobile receivers must confront. The demodulation challenge
may then be divided into two themes, i.e. digital detection and synchronization,
both of which are inference tasks. 

A range of state-of-the-art estimation techniques for DSP are reviewed
in this thesis, focussing particularly on the Bayesian minimum-risk
(MR) estimator. The latter takes account of all uncertainties in the
receiver before returning the optimal estimate minimizing average
error. A key drawback is that exact Bayesian inference in the digital
detection context is intractable, since the number of possible states
grows exponentially with incoming data. 

In an attempt to design efficient algorithms for these probabilistic
telecommunications problems, we propose a core principle for computational
reduction: Markovianity. In order to generalize this principle to
arbitrary objective functions, we design a novel topology on the variable
indices, namely the conditionally independent (CI) structure. We achieve
this by a new algorithm, called the no-longer-needed (NLN) algorithm,
which returns a bi-directional CI structure for an arbitrary objective
function. Owing to the generalized distributive law (GDL) in ring
theory, any generic ring-sum operator can then be distributed across
this CI structure, and all the ring-product factors involving NLN
variables can then be computed via a novel Forward-Backward (FB) recursion
(not to be confused with the FB algorithm from conventional digital
detection). The reduction in the number of operators, when GDL can
validly be applied, is guaranteed to be strictly positive, because
of this CI structure, a fact established by a novel theorem on GDL
presented in this thesis. Note that, since the number of operations
falls exponentially with the number of NLN variables, the FB recursion
is expected to be attractive in practical telecommunications context.
The GDL principle is useful when designing and evaluating exact efficient
recursive computational flows. Furthermore, the application of GDL
to approximate iterative schemes---as opposed to recursive schemes---is
another focus of this thesis. 

When applied to a probability model, the topological CI structure
that NLN returns is shown to be equivalent to one of the CI factorizations,
returned by an appropriate chain rule, for the joint distribution.
In particular, the FB recursion is shown in this thesis to specialize
to both the state-of-the-art FB algorithm and to the Viterbi algorithm
(VA)---depending on which inference task (i.e. operators) we define---in
the case of the M -state hidden Markov chain (HMC), which is a key
model for digital receivers. Owing to the exponential fall in the
computational complexity of the FB recursion, as explained above,
FB and VA can return exact MR estimates, i.e. the sequence of marginal
MAP estimates and the joint MAP trajectory, respectively, with a complexity
$O(nM^{2})$ in both cases, i.e. growing linearly with the number,
$n$, of data. 

To achieve a trade-off between computational complexity and accuracy
in HMC trajectory estimation, the exact strategies are then relaxed
via deterministic distributional approximation of the posterior distribution,
via the Variational Bayes (VB) approximation. A novel accelerated
scheme is designed for the iterative VB (IVB) algorithm, which leaves
out the converged VB-marginal distributions in the next IVB cycle,
and hence, reduces the effective number of IVB cycles to about one
on average. This accelerated scheme is then carried over to the functionally
constrained VB (FCVB) algorithm, which is shown, for the first time,
to be equivalent to the famous Iterated Conditional Modes (ICM) algorithm,
returning a local joint MAP estimate. This new interpretation casts
fresh light on the properties of the ICM algorithm. When applied to
the digital detection problem for the quantized Rayleigh fading channel,
the accelerated ICM/FCVB algorithm yields attractive results. When
correlation in the Rayleigh process is not too high, i.e. the fading
process is not too slow, the simulation results show that this accelerated
ICM scheme achieves almost the same accuracy as FB and VA, but with
much lower computational load, i.e. $O(nM)$ instead of $O(nM^{2})$.
These properties follow from the newly-discovered VB interpretation
of ICM. 

In an attempt to deal with Bayesian intractability in more general
contexts, VB seeks the inde-pendence-constrained approximation that
minimizes the Kullback-Leibler divergence (KLD) to the true but intractable
posterior. The novel transformed Variational Bayes (TVB) approximation
is proposed as a way of reducing this KLD further, thereby improving
the accuracy of the deterministic approximation. The parameters are
transformed into a metric in which coupling between the transformed
parameters is weakened. VB is applied in this transformed metric and
the transformation is then inverted, yielding an approximation with
reduced KLD. 

As an application in telecommunications, the synchronization problem
in demodulation is then specialized to the frequency-offset estimation
problem, whose accuracy is critical for OFDM systems. When the frequency
offset of the basic single-tone sinusoidal model is off-bin, the accuracy
of the DFT-based maximum likelihood (ML) estimate is shown to be far
worse than that of the Bayesian MR estimate, the latter being the
continuous-valued posterior mean estimate. As stated above, the Bayesian
MR estimate is often not available in more general contexts, e.g.
joint synchronization and channel decoding, because the posterior
distribution is not in closed-form in these cases. When applied to
the single frequency estimation problem, TVB achieves an accuracy
far greater than that of VB, slightly better than that of ML, and
comparable to that of the marginal MAP estimate, as shown in simulation.
This experience encourages the exploration of the TVB approximation
in the general contexts above. 

\clearpage{}

\pagenumbering{roman}

\tableofcontents{}

\printnomenclature[3cm]{}

%auto-ignore
%%%%%%%%%%%%%%%%%% abbreviations

%---

\nomenclature[abb]{HMC}{Hidden Markov chain}

\nomenclature[abb]{FB}{Forward-Backward}

\nomenclature[abb]{VA}{Viterbi algorithm}

\nomenclature[abb]{ICM}{Iterated conditional modes}

\nomenclature[abb]{GDL}{Generalized distributive law}

\nomenclature[abb]{CI}{Conditionally independent}

%---

\nomenclature[abb]{VB}{Variational Bayes}

\nomenclature[abb]{IVB}{Iterative Variational Bayes}

\nomenclature[abb]{TVB}{Transformed Variational Bayes}

\nomenclature[abb]{FCVB}{Functionally constrained Variational Bayes}

\nomenclature[abb]{KLD}{Kullback-Leibler divergence}

\nomenclature[abb]{CE}{Certainty equivalent}

\nomenclature[abb]{EF}{Exponential family}

\nomenclature[abb]{CEF}{Conjugate exponential family}

\nomenclature[abb]{SEP}{Separable in parameter}

%---

\nomenclature[abb]{iid}{Independently and identically distributed}

\nomenclature[abb]{id}{Independently distributed}

\nomenclature[abb]{a.s.}{Almost surely}

\nomenclature[abb]{cdf}{Cumulative distribution function}

\nomenclature[abb]{pdf}{Probability density function}

\nomenclature[abb]{pmf}{Probability mass function}

\nomenclature[abb]{LOTUS}{Law of the unconscious statistician}

\nomenclature[abb]{CLT}{Central limit theorem}

\nomenclature[abb]{MCMC}{Markov chain Monte Carlo}

\nomenclature[abb]{AWGN}{Additive white Gaussian noise}

\nomenclature[abb]{DAG}{Directed acyclic graph}

%---

\nomenclature[abb]{ML}{Maximum likelihood}

\nomenclature[abb]{MAP}{Maximum a posteriori}

\nomenclature[abb]{MLSE}{Maximum likelihood sequence estimation}

%---

\nomenclature[abb]{MR}{Minimum risk}

\nomenclature[abb]{MSE}{Mean squared error}

\nomenclature[abb]{MMSE}{Minimum mean squared error}

\nomenclature[abb]{RMS}{Root mean square}

\nomenclature[abb]{LS}{Least squares}

\nomenclature[abb]{BER}{Bit error rate}

\nomenclature[abb]{WSS}{Wide-sense stationary}

\nomenclature[abb]{ARMA}{Autoregressive moving average}

%---

\nomenclature[abb]{DSP}{Digital signal processing}

\nomenclature[abb]{DTFT}{Discrete time Fourier transform}

\nomenclature[abb]{DFT}{Discrete Fourier Transform}

\nomenclature[abb]{FFT}{Fast Fourier transform}

\nomenclature[abb]{PSD}{Power spectral density}

%---

\nomenclature[abb]{OFDM}{Orthogonal frequency-division multiplexing}

\nomenclature[abb]{CDMA}{Code division multiple access}

\nomenclature[abb]{QoS}{Quality of service}

\nomenclature[abb]{A/D}{Analogue-to-digital}

\nomenclature[abb]{D/A}{Digital-to-analogue}

%---

\nomenclature[abb]{NLN}{No-longer-needed}

\nomenclature[abb]{FA}{First-appearance}

\nomenclature[abb]{NOF}{Non-overflowed}

%auto-ignore
\pagestyle{fancy}

\chapter{Introduction \label{=00005BChapter 1=00005D}	}

\pagenumbering{arabic} 
\setcounter{page}{1}

In the early years of this decade, 4G mobile systems have been widely
deployed around the world, in response to the complete dominance of
smartphones over traditional telephone. Then, in order to maintain
the timescale of ten years between each mobile generation, the 5G
system standard awaits a comprehensive specification in the next year
or two. 5G systems are currently expected to be about ten times faster
than 4G systems, much more energy-efficient, and moving towards massive
machine-to-machine communication {[}\citet{ch1:5G:IEEE_part1:2014,ch1:5G:IEEE_part2:2014}{]}.
This urgent challenge in mobile systems reflects the rapid development
of information technology, which will be looking for better methodologies
and faster computing algorithms from digital signal processing (DSP)
in coming years.

In order to propose new methods, we need a deeper understanding of
available techniques. This is the philosophy we will adopt in this
thesis.

\section{Motivation for the thesis}

Unlike fixed-line communication, the major challenge in the mobile
receiver is to maintain high Quality of Service (QoS) in the face
of challenging and rapidly changing physical environment. For the
same QoS, the mobile receiver requires more computational load than
a fixed-line one. Yet, the energy resource from a mobile's battery
is highly constrained resource. The trade-off between accuracy and
computational load \foreignlanguage{british}{favours} the reduction
in computational load. This motivates our research into efficient
inference scheme in mobile receivers. In this thesis, we seek new
trade-off possibility for digital receiver algorithm be on those provided
by conventional solution. 

The formal proof of central limit theorem (CLT) in the early twentieth
century encourages the focus on probability modeling and random processes.
Particularly, the point estimation via Maximum Likelihood (ML), after
the Fisher's work in the early 1920s. ML has become the state-of-the-art
estimator in DSP systems, owing to good accuracy. In the late 1960s,
the Viterbi algorithm (VA) was designed as a computationally efficient
recursive technique for ML sequence estimation (MLSE) for digital
sequence. While not achieving the highest accuracy for digital detection,
VA is still the state-of-the-art algorithm, owing to its computational
efficiency.

For a long time, Bayesian inference was not focused on the delivery
of practical systems, despite its consistency and ability to exploit
known prior structure. Being a probabilistic framework, the normalizing
constant is required for evaluating any posterior distribution, as
well as associated moments and interval probability. This normalizing
constant is usually intractable because it must account for all states
whose number increases exponentially with the number of data in digital
detection, i.e. curse of dimensionality. 

The Bayesian techniques were revived in the 1980s, owing to tractable
Markov Chain Monte Carlo (MCMC) simulation and other stochastic approximations
for posterior distributions. Because this stochastic approach is not
\foreignlanguage{british}{favoured} in energy- and space-constrained
mobile devices, the main impact of Bayesian results is mostly in offline
contexts. Particle filtering is making an impact in online processing,
but its various implementations are computationally expensive. Therefore,
their impact in mobile receiver design has been slight up to date.
More recently, deterministic distributional approximation methods,
e.g. Variational Bayes (VB), have shown great promise in providing
principled Bayesian iterative designs that are accurate/robust, while
also incurring far smaller computational load. Indeed, it is timely
to investigate how deterministic approximations in Bayesian inference
can furnish principled designs for iterative receivers.

Note that, the above historical review highlights the interesting
role of recursion and iteration techniques in signal processing in
telecommunications. In particular, we focus on exact recursive schemes
like VA and approximate iterative techniques like VB. Hence, on one
hand, the technical aim of this thesis is to design computationally
efficient iterative schemes, which are applicable to 4G mobile receivers.
On the other hand, the theoretical aim is to synthesize new exact
recursive computational flows, which have the potential to be used
in 5G mobile receivers. The effective combination of these two techniques,
i.e. recursion within iteration and vice versa, will also be considered.

\section{Scope of the thesis}

The thesis falls into the area of statistical signal processing for
telecommunications. In common with other areas of mathematical engineering,
we seek trade-off between accuracy and computational load in the devices
and algorithms. Then, from motivation above, the natural questions
are (i) whether there is a general principle guaranteeing faster computation
in the \textit{exact} case and (ii) whether we can find attractive
trade-off between accuracy and speed in \textit{approximate} computation. 

This thesis will resolve these two questions via two approaches, one
in computational management and one in Bayesian methodology. In turn,
theses are applied to two tasks of interest in telecommunications,
firstly inference for Hidden Markov Chain (HMC), and, secondly, iterative
receiver design. We will now summarize these two questions and these
two applications.

\subsection{Computational management for objective function \label{subsec:=00005BI=00005D:GDL}}

Regarding the first question (i) above, a reasonable answer is to
exploit conditionally independent (CI) structure. The trade-off can
be seen intuitively as follows: If the objective function involves
factors exhibiting full dependence on variables, then we expect the
exact valuation of the objective function has a maximum computational
complexity. Instead, it may be possible to factorize the objective
function so that the factor exhibits various degree of independence
from variables, in which case we should expect the computational load
to be reduced. The minimum complexity should occurs when no variables
are shared between factors. 

The task we set ourselves is to verify this intuition via a mathematical
tool, namely the generalized distributive law (GDL) in ring theory.
In computer science, the GDL has recently been applied to computation
on graphical models of arbitrary order and, also, a similar trade-off
was expressed using a graphical language. However, a theoretical result
guaranteeing that the GDL always reduces the computational load has
not yet been derived. Furthermore, we would like to derive such a
result from the perspective of set theory (i.e. set of variable indices
consistent with DSP culture) rather than the graph-theoretic culture
of machine learning. Addressing this problem is the principle task
of the thesis.

\subsection{Bayesian methodology \label{subsec:=00005BI=00005D:TVB}}

Regarding the second question (ii) above, we will confine ourselves
to the area of probabilistic inference. As we know, the optimal point
estimate is obtained by relaxing from minimum bit-error-rate (BER)
criterion to minimizing the average BER, corresponding to Bayesian
minimum risk (MR) estimate. Although the performance of this Bayesian
estimate is only optimal in the average sense, it is nevertheless
the most robust solution because it incorporates all the uncertainties
actually present in the system. 

The posterior distribution is often not tractable. Its stochastic
approximation via MCMC is typically slow, as mentioned previously.
In order to address this computational intractability, the zero-order
Markovian model (i.e. independent field) can be adopted, not as an
approximating model, but as a deterministic approximation of the posterior
distribution. It is important to recognize that the original model
is unchanged in this case: only the inference technique is changed
from exact computation to an independent approximation (the so-called
naive mean field approximation). The most important technique in this
context is the iterative Variational Bayes (VB) approximation, which
guarantees convergence to a local minimum of the Kullback-Leibler
divergence (KLD) from the approximate to the exact posterior. The
complexity of this converged iterative scheme is usually lower than
that of stochastic sampling methods. The accuracy of VB is, intuitively,
dependent on how small the KLD minimum is, and, in turn, how close
the original posterior distribution is to an independent field. In
this thesis, this inspires a new VB variant---which we call transformed
VB (TVB)---in which we transform the original model into one closer
to an independent structure, reducing KLD in this transformed metric.
This implies that KLD is also reduced in the original metric. This
is the second task of this thesis.

\subsection{Application I - Hidden Markov Chain \label{subsec:=00005BI=00005D:HMC}}

In theory, the most popular model of Markov model in DSP is the first-order
hidden Markov chain (HMC). The challenge is to compute the Bayesian
maximum a posteriori (MAP) estimate efficiently. Currently, there
are three well-known algorithms for the HMC, namely the Forward-Backward
(FB) algorithm, Viterbi algorithm (VA) and Iterated Conditional Modes
(ICM) algorithm, which computes exactly the sequence of maximum marginal
likelihood, the (joint) ML estimate, and a local (joint) ML estimate,
respectively. Using the GDL, we would like to explain why these three
estimation strategies achieve a computational load that is linear
in the number of samples. Also, we want to adopt the Bayesian perspective,
and verify that they return estimate based-on posterior distribution.
Using the VB approximation, we would also like to verify whether ICM
is a special case of VB, and, if so, to understand why the accuracy
of ICM is inferior to that of VA. 

Finally, from an understanding of GDL and VB, the challenge in computation
is to design a novel accelerated algorithm, not for recursion within
one iterative VB (IVB) cycle, but for iteration between IVB cycles.
The third task of this thesis is, therefore, to achieve a better trade-off
between accuracy and speed using accelerated VB scheme and to determine
if this trade-off is better than that of the state-of-the-art VA.

\subsection{Application II - Digital receiver \label{subsec:=00005BI=00005D:DigitalReceivers}}

The main application interest of this thesis is the telecommunications
system, particularly the mobile digital receivers, where the emphasis
is on computational reduction rather than on improving accuracy. Because
the digital demodulator is the critical inference stage in the receiver.
It will be our main application focus in this thesis. 

For a digital demodulator, there are three cases of inference problem
to be considered: unknown carrier\footnote{In Chapter \ref{=00005BChapter 3=00005D} we will take care to distinguish
between carrier and the channel.} but known data (pilot symbols), known (synchronized) carrier with
unknown data and both unknown. For each case, we examine a specific
demodulator problem, as follows: unsynchronized carrier frequency
estimation, synchronized symbol detection, and symbol detection for
the Rayleigh fading channel, respectively. Despite their ideality,
these problems address key challenges in current 4G mobile systems.
We will provide simulation evidence demonstrating the enhanced trade-off
for these demodulator problems using the techniques in this thesis.
This is the fourth task of this thesis.

\section{Structure of the thesis}

The inner chapters (\ref{=00005BChapter 2=00005D}-\ref{=00005BChapter 8=00005D})
of thesis will be divided into three main parts. In Chapter \ref{=00005BChapter 2=00005D},
we seek to map the current landscape of DSP for telecommunications,
motivating the aim of the thesis (Section \ref{sec:chap2:The-roadmap}).
In Chapters \ref{=00005BChapter 3=00005D}-\ref{=00005BChapter 5=00005D},
which are three methodological chapters of the thesis, we address
the computational management issue, raised in Section \ref{subsec:=00005BI=00005D:GDL}
above. The third main part of the thesis, consisting of Chapters \ref{=00005BChapter 6=00005D}-\ref{=00005BChapter 8=00005D},
will apply these methods to the three tasks, described in sections
\ref{subsec:=00005BI=00005D:TVB}-\ref{subsec:=00005BI=00005D:DigitalReceivers}.

The summary of each of the forthcoming chapters now follows.
\begin{itemize}
\item \textbf{Chapter }\ref{=00005BChapter 2=00005D}\textbf{ - Literature
review:} This chapter is divided into three sections in order to review,
briefly but thoroughly, the history and challenges of telecommunications
systems, state-of-the-art inference techniques in DSP, and applications
of these techniques in current telecommunications system. Because
digital demodulation is the main practical application of this thesis,
it is specifically addressed in the last section of this chapter.
Another aim of this chapter is to clarify and show evidence that the
Markov principle is ubiquitous in telecommunications systems.
\item \textbf{Chapter }\ref{=00005BChapter 3=00005D}\textbf{ - Observation
models for the digital receiver: }There are three purposes in this
chapter. Firstly, this chapter can be regarded as a technical review
of demodulation, focussing particularly on conventional techniques,
such as the matched filter and frequency-offset estimation. Secondly,
it establishes three practical digital receiver models for later considerations
and simulations in the thesis. Thirdly, this chapter aims to present
the brief, but insightful derivation of the Rayleigh model for the
fading channel. 
\item \textbf{Chapter }\ref{=00005BChapter 4=00005D}\textbf{ - Bayesian
parametric }\foreignlanguage{british}{\textbf{modelling}}\textbf{:}
The purpose of this chapter is two-fold. On one hand, this chapter
reviews the technical foundation of Bayesian methodology. On the other
hand, we emphasize the important role of the loss function in designing
optimal Bayesian point estimates, particularly minimum average BER
estimator. The VB approximation and its variant, FCVB, will also be
introduced in this chapter.
\item \textbf{Chapter }\ref{=00005BChapter 5=00005D}\textbf{ - Generalized
distributive law (GDL) for conditionally independent (CI) structure:}
The aim of this chapter is to solve the first task of the thesis (see
Section \ref{subsec:=00005BI=00005D:GDL}). A new theorem will be
derived, that guarantees the reduction in computational load in evaluating
the objective function via GDL, in those cases where GDL is applicable.
An algorithm, namely the no-longer-needed (NLN) algorithm, for applying
GDL to general ring-products of objective functions is then established.
Also, a generalized FB recursion for computing that objective function
via GDL is designed. The application of GDL to computational flow
for Bayesian estimation in Chapter \ref{=00005BChapter 4=00005D}
will also be provided. Lastly, the technique for optimal computational
reduction via GDL will be considered. 
\item \textbf{Chapter }\ref{=00005BChapter 6=00005D}\textbf{ }-\textbf{
Variational Bayes variants of the Viterbi algorithm:} The aim of this
chapter is to solve the third task of the thesis (see Section \ref{subsec:=00005BI=00005D:HMC}),
by applying the GDL's computational flow to an inference of HMC. The
insight of computational reduction in state-of-the-art FB and VA is
clarified by showing that, in this chapter, they are special cases
of FB recursion. Furthermore, the FB is shown to return an inhomogeneous
HMC, which is the posterior distribution of a homogeneous HMC. The
VA is then re-interpreted as a certainty equivalent (CE) approximation
of the inhomogeneous HMC. This re-interpretation motivates the design
of VB approximation for HMC, together with an accelerated scheme for
VB in this case. By specializing VB to FCVB, the FCVB is shown to
be equivalent to ICM algorithm, and hence, Accelerated FCVB is a faster
version of ICM, while maintaining exactly the same output, i.e. local
joint MAP estimate.
\item \textbf{Chapter }\ref{=00005BChapter 7=00005D}\textbf{ - The transformed
Variational Bayes (TVB) approximation:} The aim of this chapter is
to solve the second task of the thesis (see Section \ref{subsec:=00005BI=00005D:TVB}),
by improving the accuracy of the naive mean field approximation, produced
by VB (Chapter \ref{=00005BChapter 4=00005D}), via TVB. As a theoretical
application, the TVB algorithm is applied to the spherical distribution
family. As a practical application, TVB is then applied to the frequency-offset
synchronization problem, defined in Chapter \ref{=00005BChapter 3=00005D}.
\item \textbf{Chapter }\ref{=00005BChapter 8=00005D}\textbf{ - Performance
evaluation of VB variants for digital detection:} The aim of this
chapter is to resolve the fourth task of the thesis (see Section \ref{subsec:=00005BI=00005D:DigitalReceivers}),
by applying the results in Chapter \ref{=00005BChapter 6=00005D}
to Markovian digital detectors, established in Chapter \ref{=00005BChapter 3=00005D}.
Firstly, a homogenous Markov source transmitted over AWGN channel
is studied. The simulations will show the superiority of Accelerated
ICM/FCVB to VA. The possibility that Accelerated ICM/FCVB can run
faster than the currently-supposed fastest ML algorithm are also illustrated
and discussed in this case. Secondly, an augmented finite state Markov
model, constructed by Markov source and quantized Rayleigh fading
process, are considered. The simulations will illustrate three regimes
that Accelerated ICM/FCVB is superior, compatible and inferior to
VA, corresponding to low, middle and high correlation between samples
of Rayleigh process. The KLD is also plotted in this case, in order
to explain those three regimes via VB approximation perspective.
\item \textbf{Chapter }\ref{=00005BChapter 9=00005D}\textbf{ - Contributions
of the thesis and future works:} The contributions, proposal for future
works, and overall conclusion are provided in this chapter.
\end{itemize}

%auto-ignore
%auto-ignore
%%%% Common

\global\long\def\calQ{\mathcal{Q}}%

\global\long\def\calX{\mathcal{X}}%

\global\long\def\calH{\mathcal{H}}%

\global\long\def\calL{\mathcal{L}}%

\global\long\def\vtheta{\theta}%

\global\long\def\vphi{\phi}%

\global\long\def\vxi{\xi}%

\global\long\def\vOmega{\Omega}%

\global\long\def\xbold{\mathbf{x}}%

\global\long\def\btheta{\boldsymbol{\theta}}%

\global\long\def\htheta{\widehat{\theta}}%

\global\long\def\hxi{\widehat{\xi}}%

\global\long\def\fung#1{g\left(#1\right)}%

\global\long\def\funh#1{h\left(#1\right)}%

\global\long\def\setd#1#2{\{#1{}_{1},#1{}_{2},\ldots,#1_{#2}\}}%

\global\long\def\TRIANGLEQ{\triangleq}%

%%%% Chapter 2

\global\long\def\ndata{n}%

\global\long\def\norder{p}%

\global\long\def\matA{A}%

\global\long\def\fDoppler{f_{D}}%

\global\long\def\Tsample{T_{s}}%

\chapter{Literature review \label{=00005BChapter 2=00005D}}

The facts used for thesis' motivation in Chapter \ref{=00005BChapter 1=00005D}
will be verified in this chapter via a brief literature review, which
focuses on three themes - the telecommunications systems, the available
inference techniques, and the application of those techniques in telecommunications
- corresponding to three sections \ref{sec:chap2:The-roadmap}, \ref{sec:cahp2:Inference-methodology}
and \ref{sec:chap2:Review-of-digital} below.

\section{The roadmap of telecommunications \label{sec:chap2:The-roadmap}}

The ultimate aim of a telecommunications system is reliably to transfer
information over a noisy physical channel. These transmission systems
can be categorized into two domains: analogue and digital, although
the latter completely dominates the former in telecommunications nowadays
{[}\citet{ch2:bk:Tri_T_Ha}{]}. 

In order to motivate the research on digital receivers in this thesis,
some historical milestones and evolution of telecommunications will
be briefly reviewed in this section. 

\subsection{Analogue communication systems}

The origin of telecommunications is perhaps the discovery of the existence
of electromagnetic waves, as theoretically proved and experimentally
demonstrated firstly by Maxwell in 1873 {[}\citet{ch2:origin:Maxwell:1873}{]}
and Hertz in 1887 {[}\citet{ch2:origin:Hertz:1887}{]}, respectively.

Following the discoveries in physics, an analogue system was experimented
for radio transmission around mid-1870s {[}\citet{ch2:art:AM_history}{]}.
In early history, the most popular methods were Amplitude Modulation
(AM) and Frequency Modulation (FM), firstly appeared in {[}\citet{ch2:origin:AM_Mayer}{]}
and {[}\citet{ch2:patent:FM}{]}, respectively. Some of their breakthrough
applications were radio and television transmission (via AM), mobile
telephone and satellite communication (via FM), firstly experimented
by Pittsburgh's radio station in 1920, Zworykin in 1929, American
public service in 1946 and project SCORE in 1958, respectively {[}\citet{ch2:bk:4G:RF10}{]}. 

The first analogue cellular mobile system was also introduced by AT\&T
Laboratories in 1970 {[}\citet{ch2:origin:mobile_cellular:1970}{]}.
Based on radio transmission techniques, the analogue telephone systems
in 1980s could only offer speech and related services. The first international
mobile communications at the time were NMT in Nordic countries, AMPS
in USA, TARCS in Europe and J-TACS in Japan {[}\citet{ch2:bk:4G:LTE11}{]}.
The mobile system in this era is often called \textbf{\textit{``the
first-generation (1G) - Analogue transmission''}} in the literature. 

In general, the key task of analogue receiver is to reconstruct the
original waveform from noisily modulated signal {[}\citet{ch2:bk:Tri_T_Ha}{]}.
However, this analogue system only produces a modest performance,
compared with later invented digital system, in which the information
is extracted directly without the need of reconstructing carrier waveform.
Hence, different from analogue system, where roaming is not possible
and frequency spectrum of channel cannot be used efficiently {[}\citet{ch2:bk:4G:cellular04}{]},
the digital system is capable of providing flexibly multiplexing and
computable bit stream, which efficiently exploits the channel capacity.

\subsection{Digital communication systems}

The earliest digital form of telecommunications is perhaps the Morse
code, developed by Samuel Morse in 1837 for telegraphy {[}\citet{ch2:bk:Proakis:Comm01}{]}.
However, the modern digital communication only became practical in
1924 when Nyquist sampling-rate, i.e. a sufficient condition for fully
reconstructing continuous signal from its digital samples, was firstly
introduced in {[}\citet{ch2:origin:ShannonNyquist:Nyquist_rate}{]}.
Following Nyquist's work, Harley also studied the issue of maximal
data-rate that can be transmitted reliably over a band-limited channel
in {[}\citet{ch2:origin:ShannonHartley:Hartley}{]}. Finally, in 1948,
Shannon synthesized both Nyquist's and Harley's works and provided
existence proof for reliable transmission scheme, i.e the Shannon's
limit theorems, which serve as mathematical foundation for information
theory. 

\subsubsection{Generational evolution of digital communication systems \label{subsec:chap2:Generational-evolution-of-mobile}}
\begin{itemize}
\item \textbf{2G - Digital transmission:}
\end{itemize}
In the 1990s, although the analogue voice-centric system was still
dominant, the digital packet system gradually became popular. Internet
evolved from a low rate of 9.6 kbits/s with very few online people,
to a fixed-line dial-up modem of 56 kbits/s with graphical webpages
{[}\citet{ch2:bk:4G:beyond13}{]}. The concept of Internet Protocol
(IP) and Domain name servers (DNS) for digital data transmission were
also introduced {[}\citet{ch2:bk:4G:cellular04}{]}. 

In digital mobile system, \textbf{\textit{the second-generation (2G)}}
was also developed in this decade. The circuit-switched data connection
enabled text-based communication like Short Messages Service (SMS)
and emails at the rate 9.6 kbits/s {[}\citet{ch2:bk:4G:LTE11}{]}.
At the time, two well-known systems achieving that speed by assigning
multiple slots to users were GSM project of Europe, which exploited
Time-Division Multiple Access (TDMA), and IS-95 of Qualcomm in USA,
which exploited Code-Division Multiple Access (CDMA) {[}\citet{ch2:bk:4G:LTE11}{]}. 

By incorporating both analogue voice band and digital data packet
into single air interface, the GSM and IS-95 became the well-known
GPRS and IS-95B systems (also referred to as 2.5G systems), respectively
{[}\citet{ch2:bk:4G:LTE12}{]}. 
\begin{itemize}
\item \textbf{3G - Multimedia communication:}
\end{itemize}
In 2000s, the major breakthrough was broadband Digital Subscriber
Lines (DSL) and TV cable modem, which increased the Internet speed
from 56 kbits/s in dial-up modem to 1 Mbits/s and higher (e.g. 15
Mbits/s with ADSL 2+) {[}\citet{ch2:bk:4G:beyond13}{]}. The Internet
users were not only passive receivers but suddenly became creators
on the so-called Web 2.0 version. Since 2005, the effective Voice
over Internet Protocol (VoIP) has also become a high trend, while
the traditional fixed-line network telephone has seen a steady decline
in number of customers {[}\citet{ch2:bk:4G:beyond13}{]}. 

In mobile system, the UTMS and CDMA2000 systems have evolved from
GSM and IS-95 in Europe and USA, owing to the Third Generation Partnership
Project (3GPP) and 3GPP2 in International Telecommunications Union
(ITU), respectively {[}\citet{ch2:bk:4G:LTE12}{]}. Although the core
network of \textbf{\textit{the 3G system}} is almost the same as 2G,
except the variant air-interfaces of CDMA like Wideband CDMA (WCDMA)
{[}\citet{ch2:bk:4G:LTE12}{]}, the standard data rates has reached
1 Mbits/s and higher {[}\citet{ch2:bk:4G:RF10}{]}, owing to optimizing
operational process. Other 3G air-interfaces can also be designed
via microwave links like WiMAX and Mobile WiMAX, developed on the
basis IEEE 802.16 and 802.16e, respectively. Owing to high transfer
speed, both digital video and online multimedia streaming became widely
available. Hence, the 3G was also called the multimedia communication
era {[}\citet{ch2:bk:4G:cellular04}{]}. 

The digital broadcasting system also dominated the analogue communication
gradually. As of 2009, ten countries had shutdown analogue TV broadcast
{[}\citet{ch2:bk:4G:RF10}{]}. Based on the state-of-the-art H.264/MJPEG4
compression codec, the DVB-S2 and DVB-T2 (Digital Video Broadcasting
- Satellite and Terrestrial Second Generation, respectively) were
standardized in 2007 and 2009 respectively {[}\citet{ch2:bk:4G:RF10}{]}. 

Another application of satellite communication is USA Global Positioning
System (GPS) service, which provides relatively accurate user position.
By using spread-spectrum tracking code circuitry and triangulation
principle, mobile devices can track a propagation delay between transmitted
and received signal to four GPS satellites from any position on the
earth {[}\citet{ch2:bk:4G:RF10}{]}. 
\begin{itemize}
\item \textbf{4G - All-IP networks: }
\end{itemize}
In order to keep mobile system competitive in timescale of ten years,
3GPP organized a workshop to study the long term evolution (LTE) of
UTMS in 2004 {[}\citet{ch2:bk:4G:LTE12}{]} and then released a technical
report {[}\citet{ch2:tech:3GPP:release7}{]}. Afterward, the standardization
of \textbf{\textit{the fourth generation (4G-LTE) system}} was an
overlapped and iterative process {[}\citet{ch2:bk:4G:LTE11}{]}, which
took a lot of consideration on available technology, testing and verification. 

Since air interface is the interface that mobile subscriber is exposed
to, its frequency spectrum usage is crucial for mobile network success
{[}\citet{ch2:bk:4G:cellular04}{]}. Hence, although the core shared-channel
transmission scheme of 4G is still the same as that of previous generation,
i.e. dynamic time-frequency resource should be shared between users
{[}\citet{ch2:bk:4G:LTE11}{]}, 4G system employed the Orthogonal
Frequency-Division Multiple Access (OFDMA) air interface and other
variants, in place of WCDMA in 3G. Owing to small latency in OFDMA,
the data packet switching in 4G are smooth enough for continuous data
connection (e.g. speech communication and video chat), which could
not work seamlessly via busty data transmission of previous generations
{[}\citet{ch2:bk:4G:RF10}{]}. 

For that reason, 4G is also known as All-IP generation {[}\citet{ch2:bk:4G:cellular04}{]},
in which both voice and data transmission can be divided and re-merged
via individual packet routing (e.g. VoIP). The voice calls, although
enjoying the same Quality of Service (QoS), will be processed via
packet-switching circuit on mobile receivers, which is completely
different from voice-switching circuit requiring continuously physical
connection during the call {[}\citet{ch2:bk:4G:beyond13}{]} in previous
generations. 

In 2008, ITU published requirement sets for 4G system under the name
International Mobile Telecommunications - Advanced (IMT-Advanced)
{[}\citet{ch2:bk:4G:LTE12}{]}, which targets peak data rates of 100
Mbits/s for highly mobility access (i.e. with speeds of up to 250
km/h) and 1 Gbit/s for low mobility access (pedestrian speed or fixed
position) {[}\citet{ch2:bk:4G:RF10}{]}, together with other requirements
on spectral efficiency, user latency, etc. With that target, the High
definition (HD) TV programs is expected to be delivered soon on 4G
networks {[}\citet{ch2:bk:4G:video09}{]}. In 2010, both LTE-Advanced
and WiMAX 2.0 (IEEE 802.16m) systems were announced to meet IMT-Advanced
requirements {[}\citet{ch2:bk:4G:LTE12}{]}. The deployment of 4G
is also expected to be around 2015 {[}\citet{ch2:bk:4G:RF10}{]}. 
\begin{itemize}
\item \textbf{5G (undefined):}
\end{itemize}
Currently, the 4G standard was properly set up. Hence the current
trend is to define and set up \textbf{\textit{the 5G standard}}, just
like ten years ago. In 2012, the UK's University of Surrey secured
£35 million for new 5G research centre {[}\citet{ch2:link:5G:UK12}{]}.
In 2013, European Commission announced €50 million research grants
for developing 5G technology in 2020 {[}\citet{ch2:link:5G:EU13}{]}.
Although there is not any standard definition for 5G yet, a call for
submission on this topic has been circulated in digital signal processing
(DSP) society {[}\citet{ch2:link:5G:IEEE14}{]}. 

\subsubsection{Challenges in mobile systems \label{subsec:ch2:Challenges-in-mobile}}

For very long time, the mobile system had been dominated by voice
communication. Together with 4G launching, however, mobile data traffic
dramatically increased by a factor of over 100 and completely dominated
voice calls around 2010 {[}\citet{ch2:tech:4G:Ericsson11,ch2:bk:4G:LTE12}{]}.
In the same trend, about half of mobile phones sold in Germany in
2012 was actually smart-phones {[}\citet{ch2:bk:4G:beyond13}{]}.
The increase of network capacity is now critically demanded by the
growing use of smart-phones and IP-based service. Nevertheless, the
channel capacity in mobile system is theoretically bounded by Shannon's
channel capacity theorem (also known as Shannon--Hartley theorem),
which can be written in the simplest form as follows {[}\citet{ch2:bk:4G:LTE12}{]}:

\begin{equation}
C=B\log_{2}(1+SINR)\label{eq:ch2:Shannon:SINR}
\end{equation}
where $C$ is the channel capacity (bit/s) representing the maximum
data rate of all mobiles that one station can control, \textbf{$B$}
is the bandwidth of communication system in Hz and $SINR$ is the
\textit{signal to interference plus noise ratio}, i.e. the power of
receiver's desired signal divided by the power of noise and network
interference. Based on Shannon's channel capacity theorem (\ref{eq:ch2:Shannon:SINR}),
there are three main ways to increase the data transmission rate in
practice {[}\citet{ch2:bk:4G:LTE12}{]}, as explained below.

The first and natural way is to increase $SINR$. By constructing
more base stations, we can increase the maximum data rate that mobile
system can handle. However, this way is not always efficient because
of energy and economical cost. 

The second and fairly good way is to increase the bandwidth $B$.
Nevertheless, this method is rather limited since there is only finite
amount of radio spectrum, which is allocated and managed by ITU. 

The third and current way is to approach closer to channel capacity
$C$, determined by $B$ and $SINR$ (\ref{eq:ch2:Shannon:SINR}),
via communication technology. Overall, there are three phases in mobile
system that digital technology can assist to improve traffic performance: 
\begin{itemize}
\item The first phase is the transmitter: By applying multiplexing techniques
and/or by inserting reference header and error control packets, the
bandwidth can be efficiently exploited via user-sharing scheme. The
header normally consists of network information and Automatic-repeat
request (ARQ), which helps reducing the noise and interference effect
{[}\citet{ch2:bk:4G:RF10}{]}. For example, the overhead in 4G-LTE
is about 10\% of transmitted data {[}\citet{ch2:bk:4G:LTE12}{]}.
Nevertheless, too high overhead will cause latency and slow down the
overall data rate in mobile system. The challenge is to keep a low
overhead ratio while maintaining the overall QoS. 
\item The second phase is physical channel: a dynamic wireless channel is
more challenging than stationary guided channel or optical channel
{[}\citet{ch2:bk:4G:RF10}{]}. A typical phenomenon is the so-called
fading channel, in which the received signal is disturbed by Doppler
effect. Such an effect might happen because of receivers' mobility
or of environment reflection. For example, a challenge in users location
is to maintain the quality of GPS, which is recognized to be less
accurate in rural region and in building area {[}\citet{ch2:art:GPS:challenges05}{]}.
In 4G system, the required peak data rate for high-speed receiver
is also much less than that of stationary receiver, as shown above.
\item The third phase is receiver's performance: Owing to current popularity
of smart-phones, the computational capability of mobile devices is
improved significantly. By incorporating more complex processors (e.g.
VLSI), mobile receiver nowadays can compute more complicated operators.
Hence the most notable challenge is to optimize decoding algorithm
such that the number of operators can be reduced significantly. 
\end{itemize}
From the history of mobile system above, we can recognize a common
trend: the maximum data rate of mobile generation (e.g 9.6 kbits/s
in 2G, 56 kbits/s in 3G and 1 Mbits/s in 4G) was set almost the same
as that of previous fixed-line generation (e.g 9.6 kbits/s in early
internet, 56 kbits/s in dial-up modem and 1 Mbits/s in DSL). Therefore,
the challenges in mobile system are more about efficient operation
in different environment, rather than breaking the record of possible
maximum data rate of fixed-line communication. In other words, optimizing
the latency and computational load is a more serious issue in mobile
system than increasing the limit of decoding performance.

\subsubsection{The layer structure of telecommunications}

In practice, the design of telecommunications system is separated
into hierarchical abstraction levels. Each level hides unnecessary
details to higher levels and focuses on essential tasks driven by
features of lower levels. In general, parts of a system can be categorized
into two structures: hardware and software. 

In a hardware system, the typical levels are: physical level for physical
laws in semiconductor; circuit level for basic components like resistors
and transistors; element level for gates and logical ports; module
level for complex entities like CPUs and logic units; etc. {[}\citet{ch2:bk:SEP:sync02}{]}.

In a software system, communication protocols can be considered as
software module. The most popular model is the ITU's Open System Interconnection
(OSI) reference protocol model, which consists of seven stacked abstraction
layers {[}\citet{ch2:bk:4G:cellular04}{]}. From the lowest to highest
level, those seven layers are: 

- The Physical Layer represents interface connections (e.g. optical
cable, radio, satellite transmission, etc.), which are responsible
for actual transmission of data; 

- The Data Link Layer implements data packaging, error correction
and protocol testing; 

- The Network Layer provides network routing services; 

- The Transport Layer provides flow control, error detection and multiplexing
for transporting services through a network; 

- The Session Layer enables application identification; 

- The Presentation Layer prepares the data (e.g compression or de-compression); 

- The Application Layer acts as an interface of services provided
to the end users. 

The inference algorithms for digital receivers in this thesis (Chapters
\ref{=00005BChapter 6=00005D}-\ref{=00005BChapter 8=00005D}) can
be regarded belonging to Physical Layer of software system, although
some aspects on running-time in Physical Level of hardware system
are also taken into account (e.g. Section \ref{ch6:sub:bubble-sort-like}).
Nevertheless, as discussed in Chapter \ref{=00005BChapter 9=00005D},
those algorithms can be feasibly extended and applied to problems
in higher layers, e.g. decoding in Data Link Layer and network transmission
in Network and Transport Layer.

\section{Inference methodology \label{sec:cahp2:Inference-methodology}}

From the brief review in previous section, it is clear that communication
technology must rely on mathematical solutions in order to increase
both transmission speed and accuracy, particularly in the current
digital era. Because the ultimate aim is reliably to transmit a message
over a noisy channel, as mentioned before, the original transmitted
message is considered as unknown, as far as the receiver is concerned.
Hence, a methodology for inferring unknown quantities is obviously
critical in communication. In this section, state-of-the-art inference
techniques in digital signal processing will be briefly reviewed,
while their application in communication system will be presented
in next section. 

\subsection{A brief history of inference techniques \label{subsec:chap2:inference-history}}

In history, the Least Squares (LS) method was firstly presented in
print by Legendre in 1805 and quickly became standard tool for astronomy
in the early nineteenth century {[}\citet{ch2:BK:statistics:History86}{]}.
Because LS relies on inner product concept, which is considered to
underly most of applied mathematics and statistics {[}\citet{ch2:bk:Silverman:RKHS05}{]},
LS and its variant minimum mean square error (MMSE), proposed firstly
by Gauss {[}\citet{ch2:origin:LMS:Gauss1821}{]}, have been the most
popular criterions for inference technique since then. 

Earlier in 1713, the Bernoulli's book {[}\citet{ch2:origin:Bernoulli:Prob1713}{]},
which introduced the first law of large number (LLN), is widely regarded
as the beginning of mathematical probability theory {[}\citet{ch2:BK:statistics:History86}{]}.
From Bernoulli's results, De Moivre presented the first form of central
limit theorem (CLT) in 1738 {[}\citet{ch2:origin:CLT:Moivre1738}{]}
via Stirling's approximation {[}\citet{ch2:origin:CLT:Stirling}{]}.
Following De Moivre, the first attempts on dealing with inference
problem were presented separately in {[}\citet{ch2:origin:Bayes:Simpson1755}{]}
and {[}\citet{ch2:origin:Bayes:Bayes1763}{]}, via the concept of
inverse probability at the time {[}\citet{ch2:BK:statistics:History86}{]}.
The latter work was later called Bayes' theorem, firstly generalized
by Laplace in {[}\citet{ch2:origin:Bayes:Laplace1774,ch2:origin:Bayes:Laplace1781}{]}.
Those memoirs of Laplace were the most influential work of inference
probability in the eighteenth century {[}\citet{ch2:BK:statistics:History86}{]}. 

Nevertheless, probability theory only became widely recognized in
twentieth century, owing to the formal proof of CLT in {[}\citet{ch2:origin:CLT:Lyapunov:00}{]}.
The maximum likelihood, which is perhaps the most influential inference
technique in frequentist probability {[}\citet{ch2:origin:Fisher_ML:1912_1922}{]},
was introduced by Fisher in {[}\citet{ch2:origin:Fisher_info:Fisher22}{]}.
However, Fisher strongly rejected Bayesian inference techniques {[}\citet{ch2:origin:Fisher_ML:1912_1922}{]},
which he treated as the same as inverse probability concept. The Bayesian
theory has only revived and become popular since 1980s {[}\citet{ch2:origin:Bayes:modern1980}{]},
owing to the famous Markov Chain Monte Carlo (MCMC) algorithm invented
in physical statistics {[}\citet{ch2:origin:MCMC:Metropolis}{]}. 

\subsection{Inference formalism}

Given observed data, $\xbold\in\calX$, the aim of mathematical estimation
is to deduce some information, under a form of function $\widehat{\vxi}\TRIANGLEQ\funh{\xbold}$,
about unknown quantity $\xi\in\calQ$. A typical inference method
can be implemented via the following stages:
\begin{description}
\item [{(i)}] The very first stage is to impose models on $\xbold$ and
$\xi$. Those models are called either parametric or non-parametric,
if they only depend on either a set of parameters $\vtheta\in\vOmega$
or the whole spaces $\{\calX,\calQ\}$, respectively. Hence, loosely
speaking, a parametric model is designed specifically for $\xbold$
and $\xi$ (via $\vtheta$), while a non-parametric model is defined
specifically for the spaces $\calX$ and $\calQ$ (without any $\vtheta$).
\item [{(ii)}] The second stage is to choose a criterion in order to design
the function $\funh{\cdot}\in\calH$. The most common criterion is
to pick the optimal function $\hxi=h(\xbold)$ minimizing loss function
$\calL(\vxi,\xbold)$ (also known as error function). Note that, for
deterministic parametric model $\xbold=\fung{\vtheta}$, the loss
$\calL(\vxi,\vtheta)$ can be used instead. In some cases, such a
function $\funh{\cdot}$ is fixed and imposed by physical system.
Then, the remaining option is to study the behavior of function $\funh{\xbold}$.
Such a study is still useful, since we might be able to transmit the
$\xbold$ that minimizes the loss. 
\item [{(iii)}] The third stage, which is optional but mostly preferred,
is to impose a probability model dependent on both $\xbold$ and $\xi$.
Hence, the value of loss function $\calL$ is a random variable, whose
moments can be extracted. Because the computation of statistical moments
is often more feasible in practice, the optimized criterion in second
stage can be relaxed and loss function $\calL$ is required to be
minimized on average. 
\item [{(iv)}] The last stage, which is again optional but often applied
in practice, is to design good approximation for difficult computations
in above stages. The approximation techniques are vast and varied
from numerical computation, distributional approximation to model
approximation. In this thesis, however, distributional approximation
is of interest the most.
\end{description}
Based on the above procedure, some concrete inference methods will
be reviewed subsequently in the following, from the method involving
the least number of stages to the one with most of stages.

\subsection{Optimization techniques for inference}

In practice, when we know nothing about the model of $\xbold$, a
reasonable choice is to consider non-parametric approach. For a fast
algorithm, however, there are two choices: either artificially assuming
a parametric model for $\xbold$ or imposing an estimation model (either
linear or non-linear) for $\vxi$. The latter case will be considered
in this subsection. Note that, the optimization techniques here only
involve the first two stages (i-ii), because there is no probabilistic
model assumption at the moment.

\subsubsection{Estimation via linear models}

Regarding optimization's criterion, although the total variation (i.e.
$L_{1}$-norm) has gained popularity recently (e.g. in compressed
sensing {[}\citet{ch2:Art:ComSensing:Magazine}{]}), only Euclidean
distance (i.e. $L_{2}$-norm) for the loss $\calL(\vxi,\xbold)$ will
be reviewed here. The reason is that, the latter is still the dominant
criterion in DSP, owing to the Least Square (LS) method and its variants
{[}\citet{ch2:BK:Kay:Estimation98,ch2:bk:Proakis:DSP06}{]}. 

In the simplest linear form, the unknown quantity can be written in
vector calculus $\vxi=\matA\vtheta$, where matrix $\matA$ is assumed
known. The output of LS method is, therefore, the optimal value of
parameter $\htheta$ that minimizes the square error function $\left\Vert \xbold-\vxi\right\Vert _{2}^{2}$.
Note that, in this case, the loss has taken into account both model
design error for $\vxi$ and unknown noise embedded in $\xbold$.
Owing to linear property, the minimum point of loss function can be
found feasibly by setting derivative equal to zero, which yields the
set of linear normal equations {[}\citet{ch2:BK:Kay:Estimation98}{]}.
Such a technique is also called linear regression. In more general
form, where the matrix $A$ can be replaced by impulse response of
a linear filter, the LS method is also called adaptive filter method
in DSP {[}\citet{ch2:bk:Hayes:DSP96}{]}. 

The linear form also yields recursion form for LS in two cases {[}\citet{ch2:BK:Kay:Estimation98}{]}: 

- In spatial domain, if $\btheta_{\norder}\TRIANGLEQ\setd{\vtheta}{\norder}$,
where $\norder$ is the order of parameter model, the order-recursive
least square (Order-RLS) method returns the optimal $\widehat{\btheta_{\norder}}$
recursively from the LS optimal $\widehat{\btheta_{\norder-1}}$. 

- In temporal domain, if $\xbold_{\ndata}=\setd x{\ndata}$, where
$\ndata$ is the number of received data, the sequential LS (SLS)
method can return the optimal $\htheta$ for $\xbold_{\ndata}$ recursively
from the one for $\xbold_{\ndata-1}$. Owing to important online property,
the SLS has several variants, such as the weighted LS method {[}\citet{ch2:BK:Kay:Estimation98}{]}
or Recursive LS (RLS) methods {[}\citet{ch2:bk:Hayes:DSP96}{]}. The
latter cases are special cases of the former, in which the weights
are designed in order to either decrease the dependence of $\htheta$
on past values $\xbold_{i\in(-\infty,n)}$ exponentially down to zero
from the present time $\ndata$ (exponential weighted RLS method),
or truncate that dependence by a window (sliding window RLS method)
{[}\citet{ch2:bk:Hayes:DSP96}{]}.

The LS method can also be extended to decision problem under constraints.
In Constrained LS method, the parameter $\vtheta$ is subject to some
linear constraints, which can be solved feasibly via Lagrange multiplier
technique {[}\citet{ch2:BK:Kay:Estimation98}{]}. In Penalized LS
method, the square error function is added by a smoothly penalized
function dependent on $\vtheta$ {[}\citet{ch2:bk:Silverman:Penalized94}{]}.

\subsubsection{Estimation via non-linear models}

The LS criterion in linear case can also be applied to non-linear
model $\vxi=\fung{\vtheta}$, which is also called non-linear regression
{[}\citet{ch2:bk:nonlinear:Bard74}{]}. Because the minimization of
square error is often difficult in this case, a common solution is
to convert the non-linear problem back to a linear problem. There
are three popular techniques for that purpose. 

The first technique is transformation of parameters $\vphi=q(\vtheta)$,
such that $\vxi=\matA\vphi$ is a linear model. Although this method
can be applied successfully to sinusoidal parameter estimation via
trigonometric formula {[}\citet{ch2:BK:Kay:Estimation98}{]}, only
few non-linear cases can be solved by this way. 

The second technique is numerical approximation. A numerical grid
search on non-linear function can be implemented via Newton-Raphson
iteration, which returns a local minimum for loss function. Another
approximation is to linearize the loss function at a specific parameter
value of $\vtheta$ at each iteration. Such a technique is called
Gauss-Newton method, which omits the second derivatives from Newton-Raphson
iteration {[}\citet{ch2:BK:Kay:Estimation98}{]}.

The third technique is to solve the non-linear loss function via linear
regression in augmented space, namely Reproducing Kernel Hilbert Space
(RKHS). By Riesz representation theorem, a non-linear function can
be represented as an inner product between designed kernels in RKHS
{[}\citet{ch2:bk:Silverman:RKHS05}{]}, although the kernel form is
not always feasible to design.

\subsection{Probabilistic inference}

As a relaxation, we can consider $\xbold$ and $\vxi$ as realization
of unknown quantities. Based on Axioms of Probability, firstly formalized
in {[}\citet{ch2:origin:Kolmogorov:Axiom33}{]}, the unknown quantity
can be regarded either as random variable, which is a function mapping
a realization event $\omega\in\vOmega$ in a probability space of
triples $(\vOmega,\mathfrak{F},P)$ to (possibly vector) real value
{[}\citet{ch2:BK:Bernardo:Bayes94}{]}, or more generally as random
element, which maps that $\omega$ to measurable space $(E,{\cal E})$,
firstly defined in {[}\citet{ch2:origin:random:element48}{]}. By
this way, probabilistic model can be applied to $\xbold$ and $\vxi$,
instead of deterministic model. 

\subsubsection{Estimation techniques for stationary processes \label{subsec:chap2:Stationary-process}}

Firstly, let us regard a sequence of observed data $\xbold_{\ndata}=\setd x{\ndata}$
as a stochastic process of $\ndata$ random quantities. Although the
joint probabilistic model will not be specified, such a stochastic
process will be confined to be either strict- or wide-sense stationary
in this subsection. By definition, the strict-sense stationary process
requires that the joint distribution of any two data only depends
on the difference between their time points, while the wide-sense
relaxes the joint distribution constrain with the first two orders
of moments only. 

Because the covariance function of wide-sense stationary (WSS) signal
only depends on the lagged time, which, in turn, can be represented
as a power spectral density (PSD) in frequency domain, the computation
in that linear parametric model greatly facilitates the inference
task. Hence, the WSS property is widely assumed in DSP methods. Similarly,
the additive white Gaussian noise (AWGN) is the most popular noise
assumption in the literature, because a WSS Gaussian process, solely
characterized by the first two orders moment, is also a strict-sense
stationary process {[}\citet{ch3:BK:DigiComm:Madhow08}{]}. 

In theory, the famous Wold's representation theorem, firstly presented
in his thesis {[}\citet{ch2:origin:Wold:thesis54}{]}, guarantees
that any WSS process can be written as a weighted linear combination
of a lagged innovation sequence, which is a realization of white noise
process. In other words, given innovation sequence as the input, any
WSS discrete signal can be expressed either as the output of a causal
and stable innovation filter (i.e. an infinite impulse response (IIR)
filter) in frequency domain, or as Moving Average (MA) model with
infinite order in time domain {[}\citet{ch2:bk:Proakis:DSP06}{]}.
The latter is also called Wold decomposition theorem, which decomposes
the current value of any stationary time series into two different
parts: the first (deterministic) part is a linear combination of its
own past and the second (indeterministic) part is a MA component of
infinite order {[}\citet{ch2:bk:Wold:Bierens2004,ch2:art:Wold:Bierens2012}{]}. 

In practice, because the MA order can only be set finite, another
linear model with finite order, namely Auto-Regressive Moving-Average
(ARMA), is wildly applied to $\vxi$ as an approximation for Wold's
representation of WSS signal $\xbold$. The popular criterion in this
case is the Least Mean Square (LMS) error, in which the parameters
$\vtheta$ of the ARMA model of $\vxi$ has to be designed such that
the mean square error (MSE) function $E_{f(\xbold_{n})}(\left\Vert \vxi-\xbold_{n}\right\Vert _{2}^{2})$
is minimized. Owing to the similar form of square error, LMS criterion
can be solved efficiently via LS optimization techniques. Note that,
although the distribution form $f(\xbold_{n})$ is undefined, the
WSS assumption for $\xbold_{n}$ has greatly facilitated the computation
of minimum MSE (MMSE) criterion, which only requires the first and
second order moments of $f(\xbold_{n})$ {[}\citet{ch2:BK:Kay:Estimation98}{]}.
\begin{itemize}
\item \textbf{Wiener Filter:}
\end{itemize}
The MMSE estimator in this linear model is the well-known Wiener filter,
proposed in {[}\citet{ch2:origin:Wiener49}{]}. The engineering term
'filter' is used because it often refers to a process taking a mixture
of separate elements from input and returning manipulated separate
elements at the output {[}\citet{ch2:bk:AdaptiveFilter99}{]}. Such
elements might be frequency components or temporal sampling data. 

Wiener filter can be applied in three scenarios: filtering, smoothing
and prediction: 

- In filtering scenario, the underlying process value $\vxi_{n}$
at current time is estimated from $\xbold_{n}$ by solving the set
of linear normal equations, which is called Wiener-Hopf filtering
equations because the normal matrix in this case is the Toeplitz autocovariance
matrix {[}\citet{ch2:BK:Kay:Estimation98}{]}. In frequency domain,
such a correlation-based estimator can be considered as a time-varying
finite impulse response (FIR) filter. When the past data is considered
as infinite, the FIR filter becomes an IIR Wiener filter {[}\citet{ch2:bk:Proakis:DSP06}{]}. 

- In smoothing scenario, the underlying value $\vxi_{i}$ at any time
point $i$ is estimated from a theoretically infinite length signal
$\xbold$. Owing to the infinite length assumption, Fourier transform
is applicable and can be used to return the spectrum of estimator,
which is called infinite Wiener smoother {[}\citet{ch2:BK:Kay:Estimation98}{]}
in this case. 

- In prediction scenario, the unknown future data is estimated from
the current batch of data $\xbold_{n}$. In other words, the unknown
quantity in this case is $\vxi=x_{l>n}$ rather than underlying process
values. The normal equations in this case are called Wiener-Hopf equations
for $l$-step prediction {[}\citet{ch2:BK:Kay:Estimation98}{]}. If
$l=1$, those normal equations of linear prediction are identical
to Yule-Walker equations {[}\citet{ch2:origin:YuleWalker:Yule27,ch2:origin:YuleWalker:Walker31}{]},
which is used for finding Auto-Regressive (AR) parameters in AR process
{[}\citet{ch2:BK:Kay:Estimation98}{]}. 

Because the normal matrix has an extra Toeplitz property in this case,
many efficient algorithms were proposed to solve those normal equations.
Among them, Levinson-Durbin {[}\citet{ch2:origin:LevinsonDurbin:Levinson47,ch2:origin:LevinsonDurbin:Durbin60}{]}
and Schur algorithms {[}\citet{ch2:origin:Schur_algo:Shcur17,ch2:origin:Schur_algo:Gohberg86}{]},
which exploit recursive lattice filter structure, are the most well-known
{[}\citet{ch2:bk:Proakis:DSP06}{]}. In linear prediction, that two-stage
forward-backward lattice filter is also applied in forward and backward
linear prediction for the right-next future and right-previous past
data {[}\citet{ch2:bk:Proakis:DSP06}{]}, respectively. 

Note that the Wiener filter requires the true value of first and second
moments. i.e. the parameter of WSS $f(\xbold_{n})$, in order to compute
the estimators $\htheta$ for linear parameter $\vtheta$ of $\vxi$.
If those two moments are unknown a priori, they also need to be estimated.
For that purpose, a trivial method is to use empirical statistics,
extracted from available data, as their estimators. This method relies
on assumption of ergodic process, in which the moments of data at
arbitrary time point are equal to temporal statistics of one realization
of the process {[}\citet{ch2:bk:Proakis:DSP06}{]}. Nevertheless,
a good empirical approximation for statistical moments requires a
lot of observed data, which might cause latency and energy consuming
in practice. 
\begin{itemize}
\item \textbf{Adaptive filters:}
\end{itemize}
In cases where the block of data is too short or the first two moments
of WSS $\xbold_{n}$ are not known a priori, a popular approach is
to consider those two moments as unknown nuisance parameters. In DSP
literature, this approach is implemented via variants of Wiener filter,
namely adaptive filters, where finite blocks of observed data are
treated sequentially and adaptively. 

Instead of using the Levinson-Durbin algorithm for solving the normal
equations in Wiener filters, adaptive filters exploit variants of
the recursive LMS algorithms. Owing to the quadratic form of MSE,
the LMS algorithms always converge faster to the unique minimum of
MSE {[}\citet{ch2:bk:Proakis:DSP06}{]}. The standard LMS algorithm,
proposed in {[}\citet{ch2:origin:RLS:standard60}{]}, is a stochastic-gradient-decent
algorithm. Its complexity can be reduced via other gradient-based
LMS methods, such as averaging LMS or normalized LMS algorithms {[}\citet{ch2:bk:Proakis:DSP06}{]}.
For faster convergence, adaptive filters exploit the class of variant
Recursive Least Square (RLS) algorithm. The three major RLS algorithms
are standard RLS {[}\citet{ch2:origin:RLS:standard60}{]}, square-root
RLS {[}\citet{ch2:origin:RLS:squareroot77,ch2:origin:RLS:squareroot82}{]}
and Fast RLS {[}\citet{ch2:origin:RLS:fast78,ch2:origin:RLS:fast83}{]},
which exploit the eigenvalues of covariance matrix, the matrix inversion
via matrix decomposition and lattice-ladder filters via Kalman gain,
respectively {[}\citet{ch2:bk:AdaptiveFilter99,ch2:bk:Proakis:DSP06}{]}. 

Note that, the adaptive filters are also applicable to non-stationary
process. In that case, adaptive filters are merely parametric estimators,
which are artificially imposed on non-parametric model of data process
$\xbold_{\ndata}$.
\begin{itemize}
\item \textbf{Power Spectral Density} \textbf{(PSD) estimation:}
\end{itemize}
The autocovariance function can also be estimated via its PSD in the
frequency domain. In the literature, the three major PSD estimations
are non-parametric approach via the periodogram, a parametric approach
via ARMA modelling and a frequency-detection approach via filter banks
{[}\citet{ch2:bk:Proakis:DSP06}{]}: 

- By definition, the periodogram is the discrete-time Fourier transform
(DTFT) of the autocorrelation sequence (ACS) of sampled data. Because
of frequency leakage in windowing approaches, the periodogram does
not converge to the true PSD, although the sample ACS does converge
to the true ACS in the time domain {[}\citet{ch2:bk:Proakis:DSP06}{]}.
For this non-parametric approach, the proposed solution is to apply
averaging and smoothing operations upon the periodogram in order to
achieve a consistent estimator of the PSD. Such operations decrease
frequency resolution, and hence, reduce the variance of the spectral
estimate. The three well-known methods are Barllett {[}\citet{ch2:origin:PSD:Barlett48}{]},
Barlett-Tukey {[}\citet{ch2:origin:PSD:BlackmanTukey58}{]} and Welch
{[}\citet{ch2:origin:PSD:Welch67}{]}.

- For the parametric approach, the solution is to estimate the parameters
of an ARMA model representing the WSS process. Those parameters can
be estimated via linear prediction methods like Yule-Walker (for AR
model) or via order-RLS algorithms above. In the latter case, the
maximum order can be pre-defined via some asymptotic criterion like
Akaike information criterion (AIC) {[}\citet{ch2:origin:AIC:74}{]}.
In special cases, where the underlying signal is a linear combination
of sinusoidal components, the parameters can be detected via subspace
techniques like MUSIC {[}\citet{ch2:origin:MUSIC:86}{]} or rotational-invariance
technique like ESPRIT {[}\citet{ch2:origin:ESPRIT:86}{]}.

- In the filter bank method, as proposed in {[}\citet{ch2:origin:PSD:filter_bank69}{]},
the main idea is that the temporal signal can be processed in parallel
by a sequence of FIR filters, which serve as a spatial windows truncating
the spectrum in the frequency domain.

\subsubsection{Frequentist estimation}

In above random process, a parametric model is defined for $\vxi$,
whose purpose is to approximate data $\xbold$. In this subsection,
let us consider the other way around: a probabilistic model $f(\xbold|\vtheta)$
will be defined for $\xbold$, whose parameter $\vtheta$ can be estimated
via $\vxi$. 

In the frequentist viewpoint, probability relates to the frequencies
of possible outcome in an infinite number of realization of random
variable. The repeatability is, obviously, the basic requirement for
random variable in this philosophy. The frequentist literature often
replaces the notation $f(\xbold|\vtheta)$ of conditional distribution
with notation $f(\xbold;\vtheta)$ of likelihood if the unknown parameter
$\vtheta$ is regarded as fixed and/or unrepeatable value {[}\citet{ch2:BK:Kay:Estimation98}{]}. 
\begin{itemize}
\item \textbf{Consistent estimator:}
\end{itemize}
A popular criterion for frequentist's estimator $\vxi$ of parameter
$\vtheta$ is consistency condition, which states that $\vxi=h(\xbold_{\ndata})$
converges to $\vtheta$ in probability as $n\rightarrow\infty$. In
this asymptotic approach, the Maximum Likelihood estimator (MLE),
which maximizes the likelihood, can be shown to be consistent. Owing
to feasibility and constructive definition, MLE is perhaps the most
popular estimator in frequentist approach. 
\begin{itemize}
\item \textbf{Unbiased estimator:}
\end{itemize}
Another popular frequentist's criterion is unbiased condition, $b(\vtheta)=0$,
where $b(\vtheta)=E(\vxi)-\vtheta$ and the conditional mean is taken
via likelihood $f(\xbold|\vtheta)$. If the loss function $L(\vxi,\vtheta)$
is chosen as Euclidean distance (i.e squared error), the motivation
of unbiased condition is rooted from Mean Square Error (MSE) $mse(\vxi|\vtheta)$,
which is the sum of variance $var(\vxi|\vtheta)$ and squared bias
$b(\vtheta)^{2}$ {[}\citealp{ch2:BK:Bernardo:Bayes94}{]}. 

For minimum MSE (MMSE), the desired estimator in frequentist literature
is Minimum Variance Unbiased (MVU), in which unbiased condition $b(\vtheta)=0$
is assumed first, and Minimum Variance (MV) condition for $var(\vxi|\vtheta)$
is sought afterward. The important result for this MVU approach is
the Cramer-Rao bound (CRB) {[}\citealp{ch2:origin:CramerRao:Cramer,ch2:origin:CramerRao:Rao}{]},
which provides the bound for MVU estimator under regularity conditions. 

Nevertheless, the unbiased estimators might not exist in practice,
and hence, the applicability of CRB estimator is very limited. Moreover,
in term of MMSE, this unbiased approach is too constrained. A direct
computation of MMSE estimator, regardless of biased or not, should
be the ultimate aim after all. 

\subsubsection{Bayesian inference}

In Bayesian viewpoint, the probability is regarded as quantification
of belief, while Axioms of Probability are mathematical foundation
for calculating and manipulating that belief's quantification. In
this sense, Bayesian inference must involve two steps: 

- Firstly, the joint probability model $f(\xbold,\vtheta)$ must be
imposed, via e.g. empirical evidence in the past, uncertainty model
for unrepeatable physical system, our belief on frequencies of repeatable
outcome in future, or quantification of ignorance, etc. 

- Secondly, the posterior distribution $f(\vtheta|\xbold)$, which
quantifies our belief on parameter $\vtheta$ given observed data
$\xbold$, has to be derived from $f(\xbold,\vtheta)$ via probability
chain rule. This second step is also called Bayes' rule if $f(\xbold,\vtheta)$
is factored further into observation $f(\xbold|\vtheta)$ and prior
$f(\vtheta)$ distributions. In the past, the form $f(\vtheta|\xbold)$
was also called inverse probability, because conditional order between
parameters and data is reverse to that of likelihood $f(\xbold|\vtheta)$.

In practice, different from Frequentist approach, the aim of Bayesian
point estimator is to minimize expected value of loss function $E(L(\vxi,\vtheta))$,
but with respect to posterior $f(\vtheta|\xbold)$ instead of likelihood
$f(\xbold|\vtheta)$. Nevertheless, the main difficulty of Bayesian
techniques is that posterior distribution in practice is often intractable,
in the sense that the regular normalizing constant is not available
in closed-form. In that case, the distributional approximation for
posterior can be applied. In fact, as mentioned above, the availability
of multi-dimensional distributional approximations like MCMC is the
main reason for reviving Bayesian techniques in 1980s {[}\citet{ch2:origin:Bayes:modern1980}{]}. 

For convention, the pdf $f(\cdot)$ in this thesis is used for both
pdf and pmf distribution. The pmf is simply regarded a special case
of pdf and represented by probability weights of Dirac-delta functions
$\delta(\vtheta-\vtheta_{i})$ located at corresponding singular points
$\vtheta_{i}$. Note that, in this case, $\delta(\vtheta-\vtheta_{i})$
has to be regarded as a Radon--Nikodym probability measure, $\delta_{\vtheta}({\cal A})$,
for arbitrarily small $\sigma$-algebra set ${\cal A}$ in the sample
space $\Omega$, such that $\vtheta_{i}\in{\cal A}$, and the integral
involving $\delta(\cdot)$ needs to be understood as a Lebesgue integral. 

In an attempt to derive equivalence between Bayesian and Frequentist
techniques, the following two models for prior distribution are often
considered: 
\begin{itemize}
\item \textit{Uniform prior:} If the prior $f(\vtheta)$ is uniform over
sample space $\Omega$, the posterior distribution for $\vtheta$
is proportional to the likelihood. The Bayesian and Frequentist computational
results for MAP and ML estimates are then the same, although their
philosophy remains different. However, when the measure of sample
space $\Omega$ for $\vtheta$ is infinite, such a uniform prior will
become an improper prior.
\item \textit{Singular (Dirac-delta) function for prior:} If the prior is
assigned as $\delta(\vtheta-\vtheta_{0})$ at a singular value $\vtheta_{0}$,
the likelihood becomes $f(\xbold|\vtheta_{0})$ owing to sifting property
of Dirac delta function and, hence, justifies the philosophy of notation
$f(\xbold;\vtheta_{0})$ in frequentist. This prior is, however, not
a model of choice for Bayesian technique, because the posterior for
a Dirac-delta prior is exactly the same as that prior, by the sifting
property. In other words, once the prior belief on $\vtheta$ is fixed
at $\vtheta_{0}$, regardless of $\vtheta_{0}$ being known or unknown,
there is no observation or evidence that can alter that belief \textit{a
posteriori}. Hence, this singular function is not a good prior model
because it ignores any contrary evidence under Bayesian learning.
In application, the Dirac-delta function is mostly used in Certainty
Equivalent (CE) estimation, i.e. the plug-in method, for a nuisance
parameter subset of $\vtheta$ or in sampling distribution, as explained
in Section \ref{subsec:chap4:CE-approx}. 
\end{itemize}
Hence, care must be taken when interpreting Frequentist result as
special case of Bayesian result. For technical details of Bayesian
inference and its comparison with Frequentist, please see Chapter
\ref{=00005BChapter 4=00005D} of this thesis. 

\section{Review of digital communication systems \label{sec:chap2:Review-of-digital}}

In 1948, Shannon published his foundational paper {[}\citet{ch2:origin:Shannon:Limit}{]},
which guaranteed the existence of reliable transmission in digital
systems. By quantizing the original messages into a bit stream, the
digital system can feasibly manipulate the bit sequence, e.g. extracting
or adding redundant bits. The result is an encoded bit stream, which
is ready to be modulated into a robust analogue waveform transmitted
over noisy channel. The key advantage of digital receiver is that
it only has to extract the original bit stream from noisy modulated
signal, without the need of reconstructing carrier or baseband waveform
{[}\citet{ch2:bk:Tri_T_Ha}{]}. Hence, the aim of digital receiver
can be regarded as relaxation of that of analogue receiver. 

In its simplest form, a typical digital system can be divided into
several main blocks, as illustrated in Fig. \ref{fig:CHANNEL}. Note
that, owing to advances in methodology and technology nowadays, the
interface between those blocks becomes more and more blur. This unification
process is a steady trend in recent researches, as noted below. In
following subsections, both historical origin and state-of-the-art
inference techniques for telecom system will be briefly reviewed. 

\begin{figure}
\begin{centering}
\includegraphics[width=1\columnwidth]{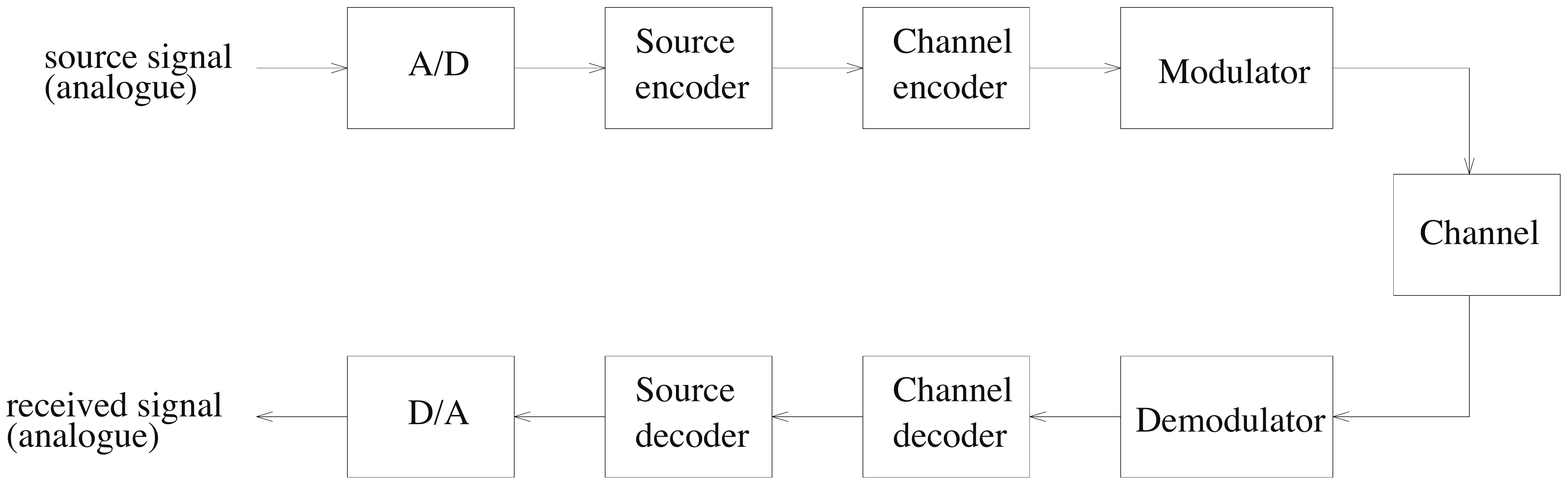}
\par\end{centering}
\caption{\label{fig:CHANNEL}Conventional stages of a digital communication
system}
\end{figure}

\subsection{A/D and D/A converters}

At the input, an analogue message will be converted into a digital
form as binary digits (or bits). Such a conversion is implemented
by the so-called Analogue-to-Digital (A/D) converter. At the output,
the Digital-to-Analogue (D/A) converter is in charge of reverse process,
which converts digital signal back to continuous form. In practice,
the A/D and D/A are used in both sampling and quantizing methods upon
temporal axis and spatial axis, respectively. Those two methods will
be briefly reviewed in this subsection.

\subsubsection{Temporal sampling}

The most important criterion in sampling process is invertible mapping,
which guarantees perfect reconstruction of original signal at the
output. A sufficient condition for successfully reconstructing signal
from its samples is that the sampling frequency at A/D is kept higher
than Nyquist rate (i.e. twice the highest frequency) of analogue message,
as firstly introduced in {[}\citet{ch2:origin:ShannonNyquist:Nyquist}{]}
and later proved in {[}\citet{ch2:origin:ShannonNyquist:Shannon}{]}.
Note that, the Shannon-Nyquist sampling theorem is, however, not a
necessary condition {[}\citet{ch2:origin:subNyquist:Landau67}{]}.
Because non-aliased sampled signal in frequency domain is a sequence
of shifted replicas of original signal, the reconstruction at D/A
is simply an ideal low-pass filter, which crops original signal out
of replicas. In time domain, such an ideal filter is called sync filter,
which is often replaced by pole-zero low-pass (smoothing) filters
like Butterworth or Chebyshev filters {[}\citet{ch2:bk:SEP:Haykin06}{]}. 

Since 1990s, the compressed sensing (CS) technique (also called compressive
sampling) has been proposed for sampling sparse signal {[}\citet{ch2:Art:ComSensing:Magazine}{]}.
Exploiting the sparsity, compressed sensing projects message signal
from original space into much smaller subspace spanned by general
waveforms (instead of sinusoidal waveforms in classical technique),
while exact recovery is still guaranteed under some conditions {[}\citet{ch2:art:ComSensing:Donoho}{]}.
Nevertheless, a drawback of CS is that the reconstruction has to rely
on global convex optimization via linear programming (LP) {[}\citet{ch2:Art:ComSensing:Magazine}{]},
instead of deterministic solution like traditional filters. In practice,
this technique has been applied to sub-Nyquist rate of multi-band
analogue signal {[}\citet{ch2:art:ComSensing:analog,ch2:art:ComSensing:subNyquist}{]}.

\subsubsection{Spatial quantization}

In typical quantization, there are three issues to be considered:
vertices of quantized cells, the quantized level within each cell
and the binary codeword associated with each level. The first two
issues, which are relevant to quantization's performance, will be
reviewed here, while the third issue, which is relevant to practical
compression rate, will be mentioned in next subsection on source encoder.

The simplest technique in quantization is to truncate and round analogue
value to the nearest boundary in a (either uniform or non-uniform)
grid of cells of amplitude axis {[}\citet{ch2:bk:Proakis:DSP06}{]}.
Each quantized level (i.e. the nearest boundary value in this case)
will then be assigned by a specific block of bits, which is often
called a codeword or symbol. For multi-dimensional case, each dimension
of analogue signal can be quantized separately. Such a technique is
commonly called scalar quantization in the literature (e.g. {[}\citet{ch2:bk:VQ_and_signalcompress}{]}).
In some cases, a transform coding, in which a linear transformation
is applied to message signal before implementing scalar quantization,
as firstly introduced in {[}\citet{ch2:origin:TransformCoding}{]},
might yield better performance than direct scalar quantization, e.g.
{[}\citet{ch2:art:VQ:TransformedGauss}{]}. 

If we regard data as a vector in multi-dimensional space, the vector
quantization (VQ) can be used as generalization of scalar quantization
(see e.g. {[}\citet{ch2:art:VQ:vec_vs_scalar}{]} for their comparison).
Instead of using parallel cells in scalar version, VQ divides data
space into multiple polytopes pointed to by boundary vectors. Then,
VQ maps each message vector within polytope cell into a quantized
vector (often being the nearest boundary vector) within that polytope.
The virtue of VQ is that any (either linear or non-linear) quantization
mapping can be represented equivalently as a specific VQ mapping {[}\citet{ch2:bk:VQ_and_signalcompress}{]}.
Hence, VQ is definitely among the best quantization mappings that
we can design. In history, original idea of VQ was scattered in the
literature. For example, VQ was firstly studied for asymptotic \foreignlanguage{british}{behaviour}
of random signal in {[}\citet{ch2:origin:VQ:Zador63}{]}, although
a version of VQ was used earlier in speech coding {[}\citet{ch2:origin:VQ:Dudley50}{]}.
In computer science, VQ is also known as k-means method, which is
named after {[}\citet{ch2:origin:k_means_MacQueen67}{]} and regarded
as cluster classification or pattern recognition method {[}\citet{ch2:art:VQ:pattern2011}{]}.

Different from sampling, quantization is an irreversible mapping.
Hence, the state-of-the-art reconstruction in D/A converter is simply
sample-and-holding (S/H) or higher-ordered interpolation operator
{[}\citet{ch2:bk:Proakis:DSP06}{]}.

\subsection{Source encoder and decoder}

The purpose of source encoder/decoder is to provide a compromise between
compression rate (i.e. number of representative bits per signal symbol)
and distortion measure (i.e. error quantity between reconstructed
and original signals). Given one of them, the ideal criterion is to
minimize the other. 

In history, the purely theoretical concept for compression rate was
Kolmogorov complexity {[}\citet{ch2:origin:Kolmogorov:Solomonoff64,ch2:origin:Kolmogorov:Kolmogorov65,ch2:origin:Kolmogorov:Chaitin69}{]},
which can be regarded as the smallest number of bits representing
the whole data. Because of analysis difficulty, Kolmogorov complexity
was subsequently replaced by minimum length description (MDL) principle
{[}\citet{ch2:origin:MDL:Rissanen78}{]}, in which the criterion was
shifted from finding shortest representative bit-block length to finding
an approximate model such that: the total number of bits representing
both approximate model and the original signal described by that model
is minimal. However, the computation for MDL is still complicated
and a subject for current researches {[}\citet{ch2:bk:MDL:Advances}{]}.

The historical breakthrough was a relaxation form of MDL in asymptotic
sense, which is the asymptotic bound of rate-distortion function,
firstly introduced and proved in foundational papers {[}\citet{ch2:origin:Shannon:Limit}{]}
and {[}\citet{ch2:origin:Shannon:RateDistortProof}{]}, respectively.
On one extreme of this bound, where desired rate is as small as null,
we achieve the best compression, but distortion would be high. On
other extreme where desired distortion is null, we achieve the so-called
lossless data compression scenario, but minimized compression rate
is still modest. Compromising those two cases, the lossy data compression
scenario, whose purpose is to reduce compression rate significantly
within tolerated small loss of distortion, has been widely studied
in the literature, as briefly reviewed below.

\subsubsection{Lossy data compression}

This scenario is sometimes called fixed-length code in the literature.
Given a fixed compression rate, its popular criterion is to minimize
the distortion measure (often chosen as mean square error (MSE)).
The research on lossy compression is vast, but currently it can be
roughly divided into three domains: A/D converter, transform coding
and wavelet-based compression. Note that their separation border might
be vague, e.g. Compressed sensing (CS) can be regarded as a hybrid
method of the first two methodologies. 

In information theory, the A/D converter can be regarded as special
case of lossy data compression (e.g. {[}\citet{ch2:Art:ComSensing:Magazine,ch2:bk:SourceChannel:Supelec}{]}): 
\begin{itemize}
\item For a fast sampling scenario, CS has been applied to speed-up medical
imaging acquisition {[}\citet{ch2:art:ComSensing_app:MRI}{]} or compressive
sensor network {[}\citet{ch2:art:ComSensing_app:sensor}{]} with acceptable
loss from sparsity recovery. Hence, CS is very useful in the context
of expensive sampling, although its compression performance is quite
modest {[}\citet{ch2:Art:ComSensing:Magazine}{]}. 
\item In quantization, the Lloyd-Max algorithm {[}\citet{ch2:origin:VQ:Lloyd_Max:60,ch2:origin:VQ:Lloyd_Max:82}{]}
and its generalization Linde-Buzo-Gray algorithm {[}\citet{ch2:origin:VQ:LBG_algo:80}{]}
(also known as k-means algorithm) are well-known algorithms for scalar
and vector quantization, respectively. Given initial quantized vectors
(or quantized levels in scalar case), the k-means algorithm iteratively
computes and assigns quantized vector as centroid of probability density
function (pdf) of the source signal within quantized cells {[}\citet{ch2:bk:compression:KhalidSayood}{]}.
At the convergence, the algorithm returns both boundary values of
the cells and quantized levels such that MSE is locally minimized
{[}\citet{ch2:bk:compression:KhalidSayood}{]}. In the case of discrete
source, if probability mass function (pmf) of the source is unknown,
it can be approximated by clusters of input source signal vectors
(i.e. similarly to histogram in scalar case) in offline basis, or
by adaptive algorithms {[}\citet{ch2:art:VQ:AdaptiveVQ}{]} in online
basis.
\item Another quantization technique resembling Vector Quantization (VQ)
is Trellis Coded Quantization (TCQ), introduced in {[}\citet{ch2:origin:TCQ:Marcellin_Fisher}{]}.
The main purpose of TCQ is to minimize the MSE of entire sequence
of quantized source samples, instead of individual MSE of each quantized
source sample. In order to avoid the exponential cardinality of trajectory
of quantized levels for entire source sequence, TCQ imposes Markovianity---local
dependence structure---on those trajectories (i.e. the quantized
levels of current source sample will depend on quantized values of
previous source sample) {[}\citet{ch2:bk:compression:KhalidSayood}{]}.
Then, the trajectory that minimizes MSE can be found via recursive
Viterbi algorithm for trellis diagram {[}\citet{ch2:origin:VA:Forney73}{]}.
In practice, this TCQ scheme was used in standard image compression
JPEG 2000 part II {[}\citet{ch2:art:TCQ:JPEG2000}{]}.
\end{itemize}
Before applying A/D converter, a pre-processing step involving transform
coding might be preferred in order to exploit both temporal and spatial
correlation in source signal. The main idea of transform coding is
to project source signal vector onto the basis capturing the most
important features {[}\citet{ch2:bk:SourceChannel:Supelec}{]}. 
\begin{itemize}
\item The earliest transform coding is a de-correlation technique, namely
Principal Components Analysis (PCA), firstly proposed for discrete
and continuous signal in {[}\citet{ch2:origin:PCA:Hotelling33}{]}
and {[}\citet{ch2:origin:PCA:Karhunen47,ch2:origin:PCA:Loeve}{]},
respectively. The main idea of PCA is to minimize the geometric mean
of variance of transformed components {[}\citet{ch2:bk:compression:KhalidSayood}{]}
by using eigenvectors of autocorrelation matrix as a transform matrix.
Then, only components with largest variances can be retained as important
features of source signal. 
\item The continuous version of PCA, which is also called Karhunen-Loeve
Transform (KLT), is the optimal transformation under MSE criterion,
yet its computational load is pretty high {[}\citet{ch2:origin:DCT:Ahmed91}{]}. 
\item Discrete Cosine Transform (DCT), firstly proposed in {[}\citet{ch2:origin:DCT:Ahmed74}{]},
is another transformation method, with similar performance to KLT
but much faster in operation {[}\citet{ch2:origin:DCT:Ahmed91}{]}.
In data compression, the DCT also yields much better performance than
Discrete Fourier Transform (DFT), particularly for correlated sources
like Markov source, since DCT mirrors the windowed signal and avoids
the distorted high frequency effect of sharp discontinuities at the
edges {[}\citet{ch2:bk:compression:KhalidSayood}{]}. In practice,
DCT is widely used in current standard image compressions, e.g. JPEG,
and video compressions, e.g MPEG {[}\citet{ch2:bk:compression:KhalidSayood}{]}.
\end{itemize}
Another pre-processing step is the so-called subband coding, which
can be generalized to be wavelet-based compression. The main idea
of subband coding is to separate source signal into different bands
of frequencies via digital filters, before pushing the outputs through
downsampling, quantization and encoding steps {[}\citet{ch2:bk:compression:KhalidSayood}{]}.
A drawback of traditional subband method is that the Fourier transform
is only perfectly local in frequency domain and none in time, i.e.
we cannot indicate when a specific frequency occurs. A trivial method
to work around this issue is the short temp Fourier transform (STFT),
which divides signal into separate block before applying Fourier transform
to each block. However, by uncertainty principle, a fixed window in
STFT cannot provide the localization in both time and frequency domains.
The wavelet method addresses this problem by re-scaling the window
such that low frequencies (longer time window) has higher frequency
resolution and high frequencies (shorter time window) has higher time
resolution. The wavelet kernel can also be designed with various orthogonal
waveforms, instead of strictly sinusoidal waveform in Fourier transform.
In practice, wavelet-based method is still not a standard compression
technique at the moment {[}\citet{ch2:bk:SourceChannel:Supelec}{]},
although there were several applications in image compression, e.g.
in JPEG2000 standard. 

For reconstruction, since both transform coding and wavelet technique
are reversible representation of data, the reconstruction in source
decoder is straightforward. Note that, in those cases, the lossy term
comes from the fact that some information is truncated by nullifying
a subset of coefficients, which does not affect the inverse mapping
process. For A/D converter, the reconstruction is similar to D/A converter
above, although the reconstructed data is merely an approximation
of original data in this case.

\subsubsection{Lossless data compression}

In many cases, lossless compression is required; e.g. in text compression,
a wrong character in the reconstructed message might lead to a completely
wrong meaning for the whole sentence. Since distortion is assumed
null, the ultimate aim of lossless compression is to minimize the
total number of representative bits of the source message. Hence,
in the data compression process, lossless compression is often applied
as the last step in order to reduce further the code length. In that
case, the input of lossless compression is discrete source and its
p.m.f will be used as an approximation of the p.d.f. of the original
source.

In practice, a relaxed criterion, which only requires that minimum
message length in an average sense be achieved, is widely accepted.
For the sake of simple computation in that average criterion, Shannon's
coding theorem further imposes a fairly strict assumption, where the
message source is iid. It shows that no lossless code mapping can
produce shorter average compression length than its entropy, which
is a function of the p.m.f of that iid source {[}\citet{ch2:origin:Shannon:Limit}{]}.
Following that result, the current three state-of-the-art algorithms,
namely Huffman-code, arithmetic code and Lempel-Ziv code, were designed
for three practical relaxations of Shannon's assumption, respectively:
iid source with known p.m.f; iid source with unknown p.m.f and correlated
source. These are now reviewed next:
\begin{itemize}
\item \textit{Huffman code} {[}\citet{ch2:origin:HuffmanCode}{]}: Given
the p.m.f of a discrete iid source, the famous Huffman code is the
optimal (and practical) code mapping, which yields the absolute minimal
average code length for the given iid source. Huffman's minimal average
length differs from Shannon's entropy by one compression bit per source
sample, since that is the difference between the integer value of
minimal length and the continuous values of entropy {[}\citet{ch2:bk:compression:KhalidSayood}{]}.
For fast compression, the Huffman code is designed as a prefix code,
in which no code word is a prefix to another codeword {[}\citet{ch2:bk:compression:KhalidSayood}{]}.
A prefix code, in turn, can be constructed for Huffman code as a binary
tree and facilitate the computation. The Huffman decoder is, owing
to this binary tree, fast and feasible since it can traverse through
the tree in the same manner as Huffman encoder {[}\citet{ch2:bk:compression:KhalidSayood}{]}.
In practice, variants of Huffman code have been used in standard image
compression, e.g. JPEG {[}\citet{ch2:origin:JPEG:Chen84,ch2:bk:compression:KhalidSayood}{]}.
\item \textit{Arithmetic code} {[}\citet{ch2:origin:ArithmeticCode:79}{]}:
The absolute minimal rate of Huffman code can be further reduced by
applying Huffman code to joint p.m.f of multiple block iid source
symbols, instead of p.m.f of a single symbol {[}\citet{ch2:BK:CoverAndThomas}{]}.
However, despite the fast reciprocal rate reduction, the computation
of the joint p.m.f grows exponentially with that number of symbols.
In order to avoid that computation, the arithmetic code makes use
of two key ideas: Firstly, it quantizes joint cumulative distribution
function (c.d.f), instead of joint p.m.f. Secondly, those quantized
intervals are refined recursively in an online Markovian fashion (i.e
the ``child'' c.m.f sub-interval of current trajectory divides its
``parent'' c.m.f interval of previous trajectory). Because the range
of c.m.f is the unit interval of real axis, this arithmetic code is
capable of representing infinite number of trajectories. Each of them
can be assigned as a rational number, with possibly long fractional
part (hence the name ``arithmetic''), within this unit interval.
Owing to Markovianity, the arithmetic code is able tractably to produce
a good code rate (within two bits of entropy {[}\citet{ch2:BK:CoverAndThomas}{]}),
compared to the minimal rate (within one bit of entropy) of Huffman
code above. For decoding, the arithmetic coded sequence in binary
base needs to be converted back to its original base value. This conversion
raises two issues: the decoding complexity and the rounding of converted
values. The former can be solved in a similar manner to the arithmetic
encoder, owing to the Markovianity. The latter can be solved in two
ways: either the length of original sequence is set a priori, or a
pilot symbol is included as the end-of-transmission {[}\citet{ch2:bk:compression:KhalidSayood}{]}
\item \textit{Lempel-Ziv code} {[}\citet{ch2:origin:LZ77,ch2:origin:LZ78}{]}:
For correlated source symbols from a set of alphabet, there are two
key issues to be solved: firstly, the design for codeword corresponding
to each alphabet and, secondly, the design for allocation indices
of appearance of that alphabet in a source trajectory. For the first
issue, a reasonable approach is to find the p.m.f of alphabet and
design a codebook based on that p.m.f. For the second issue, the allocation
indices need to be feasible to look up. The famous Lempel-Ziv (LZ)
codes solved both issues in an adaptive (i.e online) fashion: the
alphabet p.m.f is approximated by sequentially counting the frequency
of appearance of the alphabet, while allocation indices can be updated
by either sliding window approach (LZ77 algorithm {[}\citet{ch2:origin:LZ77}{]}),
or a tree-structure approach (LZ78 algorithm {[}\citet{ch2:origin:LZ78}{]}).
The LZ77 algorithm was proved to be asymptotically optimal in {[}\citet{ch2:art:LZ77:optimal}{]},
by showing that the compression rate of LZ77 converges to entropy
for ergodic source {[}\citet{ch2:BK:CoverAndThomas}{]}. Hence LZ77
has been used in many compression standards, e.g. ZIP, and in image
compressions, e.g. PNG {[}\citet{ch2:bk:compression:KhalidSayood}{]}.
The LZW algorithm, proposed in {[}\citet{ch2:origin:LZW}{]} as a
variant of LZ78, is widely used in many compression standards, e.g.
GIF {[}\citet{ch2:bk:compression:KhalidSayood}{]}, owing to its similar
performance to LZ78 and feasible implementation {[}\citet{ch2:bk:SourceChannel:Supelec}{]}.
The LZ codes also belong to the class of dictionary code, because
of its alphabet-frequency-index (hence the name dictionary) technique.
The decoding process for dictionary code is similar to table-lookup
process, where the table is the constructed dictionary and the look-up
process is implemented via allocation indices sent to the receiver.
\end{itemize}
In the literature, the types of lossless compression technique can
also be divided in several ways, such as: fixed-length-code (e.g Huffman
code) versus variable-length-code (e.g. arithmetic and LZ codes),
static code (i.e. offline) versus adaptive code (i.e. online), or
entropy code (i.e. for a known source p.m.f like Huffman code and
arithmetic codes) versus universal code (i.e. for an unknown source
p.m.f like arithmetic and LZ codes), etc.

The early history of data compression is interesting and involves
the generation of students in the era after Shannon: Shannon and Fano,
who were among the first pioneers of information theory, proposed
the theoretical Shannon-Fano coding in order to exploit c.d.f of the
source {[}\citet{ch2:BK:CoverAndThomas}{]}, although it had never
been used until arithmetic code was invented. In Fano's class of information
theory at MIT, his two students, Huffman and Peter Elias, also designed
two recursive lossless compression techniques, the Huffman and Shannon-Fano-Elias
coding, respectively, although the later was never published {[}\citet{ch2:bk:compression:KhalidSayood}{]}. 

Similarly to data compression, the early history of channel coding,
as briefly reviewed below, is just as interesting: Hamming was a colleague
of Shannon at Bell Labs when he invented the Hamming code {[}\citet{ch2:origin:HammingCode}{]},
which was also mentioned in {[}\citet{ch2:origin:Shannon:Limit}{]}.
Soon after, Peter Elias invented convolutional code {[}\citet{ch2:origin:ConvolutionCode}{]},
which, much later, subsequently led to the invention of the revolutionary
Turbo code {[}\citet{ch2:origin:Turbo:Berrou93}{]}. Gallager, a PhD
student of Elias, invented another revolutionary code, namely low-density-parity-check
(LDPC) code, in his doctoral thesis {[}\citet{ch2:origin:LDPC:Gallager}{]}. 

\subsection{Channel encoder and decoder}

When transmitted through a noisy channel, the bit stream can become
corrupted and unrecoverable. A reasonable solution is to strengthen
the transmitted message by adding in some extra bits, whose purpose
is to protect against the noise effect on message bits. Together,
the message and extra bits construct the so-called code bits, which
are transmitted through the channel. When the original message bits
are corrupted, those extra bits will be a valuable reference for the
channel decoder at a receiver to recover the message bits (hence the
name error-correcting-code in the literature). 

Nevertheless, too many extra bits requires too much redundant energy
in the transmitter and thereby increases the operational cost of communication
devices. The purpose of the channel encoder, therefore, is to maximize
the code rate (i.e. the ratio between number of message bits and number
of code bits), while maintaining the possibility of acceptable distortion
at the receiver. There exists, however, a limit for the code rate.
In the foundational paper of information theory {[}\citet{ch2:origin:Shannon:Limit}{]},
Shannon's channel capacity theorem introduced the asymptotic upper
bound of channel code rate, which is called channel capacity and is
solely dependent on the channel characteristics, provided that the
asymptotic average distortion is zero. Since then, a lot of effort
has been made to design the optimal channel codes, whose code rate
is close to that upper bound. Because of analysis difficulty in optimal
case, a relaxed criterion, where distortion is not zero but very small,
has been widely accepted in practice. 

In summary, a good channel code should satisfy three practical requirements:
high code rate, low computational load and sufficiently large Hamming
distance between any two codewords. The first and second ones represent
the requirement of low operational cost and speed of the communication
system, respectively. The third one is a consequence of the channel
characteristics (e.g. large Hamming or Euclidean distance between
codewords would reduce the uncertain error in binary symmetric channel
(BSC) and additive white Gaussian noise (AWGN) channel, respectively).
In order to facilitate the analysis of this requirement, all linear
codes currently make use of a max-min criterion: maximizing the minimum
codeword weight (i.e. its Hamming distance to the origin). Because
the sum of any two linear codewords in the finite field is also a
codeword, that criterion is equivalent to the task of maximizing the
minimum distance between any two codewords {[}\citet{ch2:bk:ToddMoon}{]}.
At the moment, non-linear channel codes have not been much investigated
or applied in practice {[}\citet{ch2:bk:ShuLin}{]}

The first codes satisfying all three requirements are two capacity-approaching
codes: LDPC and Turbo codes, which also reflect two main research
domains of channel code, namely block code and stream code, respectively.
We review this domain next.

\subsubsection{Block code}

A typical block code is a bijective linear mapping from the original
message space into a larger space, namely codeword space, over the
binary finite field. Owing to the tractability of linearity and the
availability of Galois field theory, research over channel codes has
been mostly focussed on this algebraic coding type in the early decades,
see e.g. {[}\citet{ch2:origin:AlgebraicCode:68,ch2:origin:AlgebraicCode:72}{]}.
Four historical milestones of block code in this period will be briefly
reviewed below.
\begin{itemize}
\item \textit{Hamming code} {[}\citet{ch2:origin:HammingCode}{]}: The first
practical channel code is Hamming code, whose minimum Hamming distance
is three. Hence, it is capable of correcting one error bit with a
hard-information decoder {[}\citet{ch2:bk:ToddMoon}{]} and it is
also called the class of single-error-correcting code {[}\citet{ch2:Art:ChannelCode:Review07}{]}.
However, the performance of Hamming is pretty modest in the AWGN channel,
even with soft-decoder.
\item \textit{Reed-Muller code} {[}\citet{ch2:origin:ReedMuller:Muller,ch2:origin:ReedMuller:Reed}{]}:
Reed-Muller codes were the first codes providing a mechanism to design
a code with desired minimum distance {[}\citet{ch2:bk:ToddMoon}{]}.
Another attractive property is speed of decoding, which is based on
fast Hadamard transform algorithms {[}\citet{ch2:bk:ToddMoon}{]}.
Although it was soon replaced by slightly better performance codes
(e.g. BCH code), it is still the best binary code for short-length
block codes and is currently being revisited in the literature, owing
to its good trade-off between performance and complexity via trellis
decoder {[}\citet{ch2:Art:ChannelCode:Review07}{]}.
\item \textit{Cyclic codes} {[}\citet{ch2:origin:CyclicCode:Prange}{]}:
In cyclic codes, any cyclic shift of a codeword yields another codeword.
Hence, its efficient encoding via cyclic shift-register implementation
is of advantage over other block codes. The major class of cyclic
codes with large minimum distance is the BCH codes, which was firstly
introduced in {[}\citet{ch2:origin:BCH:BoseChaudhuri,ch2:origin:BCH:Hocquengbem}{]}.
However, the asymptotic performance of BCH is not good (indeed, when
its code length at a fixed code rate becomes longer, the fraction
of errors that is possible to correct is close to zero {[}\citet{ch2:bk:ToddMoon}{]}).
Eventually BCH was dominated by its non-binary version, namely Reed-Solomon
code {[}\citet{ch2:origin:ReedSolomonCode}{]}. The ability of correcting
burst-error in RS makes it suitable for disk storage systems, and
hence, RS was widely used in compact disk (CD) writing system {[}\citet{ch2:Art:ChannelCode:Review07}{]}.
The important property of both BCH and RS is that they can be efficiently
decoded via finite field arithmetic {[}\citet{ch2:bk:ToddMoon,ch2:Art:ChannelCode:Review07}{]}.
\item \textit{LDPC code} {[}\citet{ch2:origin:LDPC:Gallager}{]}: Instead
of designing the encoder directly, LDPC code relies on the design
of sparsity in the parity-check matrix, which imposes sum-to-zero
constraint on multiple linear combinations of codewords. LDPC code
had been forgotten for over thirty years, until it was re-discovered
in {[}\citet{ch2:origin:LDPC:MacKay}{]}, which demonstrated the capacity-approaching
performance of LDPC code. Indeed, LDPC has excellent minimum distance,
which increases with the block code length (i.e. as the parity-check
matrix becomes more sparse) {[}\citet{ch2:bk:ToddMoon}{]}. As the
code length goes to infinity, LDPC codes have been shown to reach
channel capacity {[}\citet{ch2:art:LDPC:CapacityProof}{]}. Owing
to much lower error floor than Turbo code, LDPC code has been chosen
in many standard communication system, e.g. in DVB-S2 for digital
television or satellite communication system in NASA-Goddard Space
Flight Center {[}\citet{ch2:art:LDPC:+NASA}{]}. However, a drawback
of LDPC is the complicated encoding. Although the iterative decoding
complexity via message-passing algorithm only grows linearly with
block code length, owing to the sparsity, its dual encoder complexity
grows quadratically in block length. By some pre-processing steps,
the LDPC encoding complexity can be reduced significantly, and close
to linear in some cases, as proposed in {[}\citet{ch2:art:LDPC:Richardson01}{]}.
Applying finite geometry, a class of quasi-cyclic LDPC codes was recently
proposed in order to achieve linear encoding via cyclic shift-registers
{[}\citet{ch2:art:LDPC:+NASA,ch2:bk:ShuLin}{]}. Another drawback
of LDPC is that its efficient decoding is only an approximation, not
an exact arithmetic solution. 
\end{itemize}
Following the success of LDPC code, many researchers have focussed
on improving its decoding approximation, rather than designing a new
class of block encoder. Nevertheless, block codes have been received
more attention recently. Owing to the availability of algebraic methodologies
for finite field, the current trend is to design a good performance
code with shorter block length {[}\citet{ch2:Art:ChannelCode:Review07}{]},
since Shannon's channel limit is only valid for asymptotic case after
all.

\subsubsection{Stream Code \label{subsec:chap2:Stream-Code}}

By general definition, a stream code is a block code with infinite
length. In practice, the distinction between block code and stream
code is that the latter can continuously operate on a stream of source
bits (i.e. online case) and produce the same code bits as if it operates
on divided blocks of sources bits (i.e. offline case). The key advantage
of stream code is low complexity and delay of encoder. However, its
key drawback is the lack of a mathematical framework for evaluating
the encoder's performance. Currently, all good stream codes have to
be designed via trial-and-error simulation {[}\citet{ch2:bk:ToddMoon,ch2:Art:ChannelCode:Review07}{]}.
For a brief review, two state-of-the-art stream codes will be introduced
below:
\begin{itemize}
\item \textit{Convolutional code} {[}\citet{ch2:origin:ConvolutionCode}{]}:
The first practical stream code is Convolutional code, which maps
source bits to code bits via linear filtering operators (hence the
name ``convolution'') {[}\citet{ch2:bk:Richardson}{]}. In the encoder,
the filtering process is very fast, owing to the deployment of shift-register
memory circuits. In decoder, the changing of shift-register value
can be represented via the trellis diagram, whose transitions, in
turn, can be inferred jointly or sequentially via state-of-the-art
Viterbi {[}\citet{ch6:origin:VA:Viterbi67}{]} or Bahl-Cocke-Jelinek-Raviv
(BCJR) {[}\citet{ch6:origin:BCJR:Bahl74}{]} algorithms, respectively.
Owing to Markovianity, the transmitted sequence can be recursively
traced back via those inferred transitions, provided that the pilot
symbol at the beginning-of-transmission is known a priori. 
\item \textit{Turbo code} {[}\citet{ch2:origin:Turbo:Berrou93}{]}: The
first practical stream code approaching Shannon limit is Turbo code,
which was initially designed as a parallel concatenation of two Convolutional
codes connected by a permutation operator. That initial proposal for
Turbo code, although designed via trial-and-error process {[}\citet{ch2:bk:Turbo:Berrou2011}{]},
has revived the study of near-Shannon-limit channel codes and is still
the best performing method for the Turbo code class at the moment
{[}\citet{ch2:bk:ToddMoon}{]}. For Turbo decoding, the output of
the Convolutional decoder in each branch are iteratively fed as input
to the other brach until convergence. This iterative decoding (hence
the name 'Turbo' ) is, however, not guaranteed to converge, although
empirical studies show that it often converges in practical applications
{[}\citet{ch2:bk:ToddMoon,ch2:bk:Richardson}{]}. A drawback of Turbo
code is that its error-floor of bit-error-rate (BER) is rather high
at $10^{-5}$, owing to low minimum codeword distance. For high quality
transmission, whose BER requirement is much lower than $10^{-5}$,
the LDPC code with BER error-floor around $10^{-10}$ is much preferred
{[}\citet{ch2:bk:ShuLin}{]}. Currently, Turbo code is being applied
in many standard systems (e.g. CDMA2000, WiMAX IEEE 802.16, etc. {[}\citet{ch2:bk:Turbo:Berrou2011}{]}),
although telecommunications systems have gradually shifted from Turbo
code to LDPC code, owing to LDPC's good performance, as explained
above.
\end{itemize}

\subsection{Digital modulator \label{chap2:sub:Digital-modulator}}

The primary purpose of the digital modulator is to map a set of bits
to a set of sinusoidal bandpass signals for transmission. In the design
process of this mapping, there are three practical criteria to be
considered: power efficiency, spectral efficiency and bit detection
performance. In practice, if one of them is fixed, the other two are
reciprocally dependent on one another. 

Hence, there exists a natural trade-off in modulation design. Loosely
speaking, a modulation scheme is called more power efficient and/or
more spectral efficient if it needs less SNR-per-bit $E_{b}/N_{0}$
(also called bit energy per noise density) and/or it can transfer
more bit-rate $R$ as a ratio of the ideal channel's bandwidth $B$
(also called spectral-bit-rate or spectral efficiency $R/B$) , respectively,
to achieve the same bit-error-rate (BER) performance. Note that the
maximum bit-rate $R$ for reliable transmission is the Shannon channel
capacity $C$ while the maximum sampling rate for zero intersymbol
interference (ISI) is the Nyquist rate $2B$. The spectral efficiency
$R/B$, therefore, reflects the ratio between transmitter's bit rate
and receiver's sampling rate in a reliable and zero ISI transmission.
In the ideal scenario, where BER average is asymptotically zero, the
reciprocal dependence between power and spectral efficiency is given
by equation (\ref{eq:ch2:Shannon:SINR}), owing to Shannon channel
capacity theorem {[}\citet{ch2:bk:Tri_T_Ha}{]}. 

Despite the similarity in mathematical formulae and performance criteria,
the channel encoder and digital modulator are basically different
in the choice of the fixed term in three way trade-off above. Indeed,
the study of the channel encoder often neglects the spectral efficiency
issue, while digital modulation is mostly designed for reliable transmission
only (i.e. BER performance is kept fixed and small). Hence, the channel
encoder, whose main purpose is to achieve low BER at low SNR, can
be regarded as a pre-processing step for digital modulation, which
focuses on trade-off between power efficiency and spectral efficiency.
This is a reasonable division of tasks, because modulation involves
a D/A step and, hence, requires the study of signal spectrum, while
encoding only involves digital bits. Nevertheless, the interface between
them is, sometimes, not specific {[}\citet{ch2:bk:OFDM_CDMA}{]}.
This vague interface also happens between channel decoder and digital
demodulator, as we will see in Section \ref{subsec:chap2:Digital-detector}. 

As stated above, the purpose of digital modulation is to map information
bits to amplitude and/or frequency and/or phase of a sinusoidal carrier,
such that the trade-off between power efficiency and spectral efficiency
in a reliable transmission scheme is optimal. The solution for that
trade-off can be divided into two modulation schemes: either ``memoryless
mapping'' or ``mapping with memory'', i.e. either on a symbol-by-symbol
basis or via a Markov chain, respectively {[}\citet{ch2:bk:Proakis:Comm01}{]}.

\subsubsection{Memoryless modulation \label{subsec:chap2:Memoryless-modulation}}

When bandwidth is the premium resource, modulation on the carrier
amplitude and/or phase is preferred to carrier frequency. In the memoryless
scheme, each block of $k$ bits, called symbol, is separately mapped
to $M=2^{k}$ bandpass signals, whose characteristic is represented
by constellation points, as follows: 
\begin{itemize}
\item \textit{Binary (bit) modulation:} In this case, the carrier's amplitude,
phase or frequency can be modulated via Amplitude, Phase or Frequency
Shift Keying, i.e. ASK, PSK and FSK scheme, respectively. Regarding
ASK, also called on-off keying (OOK), it is mostly used in fiber optic
communication, because it can produce a bias threshold for turning
on-off the light-emitting diode (LED) {[}\citet{ch2:art:ASK:LED:2012}{]}.
Regarding PSK, which is perhaps the most popular binary modulation
{[}\citet{ch2:bk:Tri_T_Ha}{]}, it simply exploits the antipodal points
to represent the binary bits. That simple implementation, however,
might introduce phase ambiguity at the receiver. Regarding FSK, which
is a special and discrete form of Frequency Modulation (FM), the bit
is represented by a pair of orthogonal carriers, whose frequencies
are integer multiple of the bit rate $1/T_{b}$. That bit rate is
also the minimum frequency spacing required for preserving orthogonality
{[}\citet{ch2:bk:Tri_T_Ha}{]}. Despite its simplicity, binary modulation
has, however, low spectral efficiency.
\item \textit{$M$-ary modulation:} In this case, the spectral efficiency
can be greatly improved at the expense of power efficiency. The popular
criterion in this case is average Euclidean distance between constellation
points and origin, which reflects both error vulnerability and power
consumption of a transmitter. Hence, the $M$-ary ASK, which produces
$M$-ary pulse amplitude modulation (PAM), uses more power to achieve
the same BER performance as ASK, owing to the increase of average
Euclidean distance. Inheriting the high spectral efficiency of $M$-ary
ASK, the $M$-ary PSK provides higher power efficiency by increasing
the symbol's dimension from a line in $M$-ary ASK to a circle in
the Argand plane. \\
For large $M$, however, $M$-ary PSK is not power efficient, because
the Euclidean distance of adjacent points in the circle of $M$-ary
PSK becomes small with increasing $M$. The $M$-ary Quadrature Amplitude
Modulation (QAM) resolves this large $M$ issue by placing the constellation
points throughout the Argand plane, also known as in-phase and quadrature
(I-Q) plane, corresponding to real and imaginary axis in this case.
Furthermore, QAM often employs the Gray code, in which two \foreignlanguage{british}{neighbour}
points differ only by one bit, in order to decrease the uncertainty
error between \foreignlanguage{british}{neighbour} points. Owing to
those two steps, QAM is widely used in high spectral efficiency communication,
e.g. in current standard ITU-T G.hn of broadband power-line, wireless
IEEE 802.11a.g or WiMAX IEEE 802.16 {[}\citet{ch2:bk:mobile:system_2011}{]}.
\end{itemize}
For further increasing spectral efficiency with small loss of power
efficiency, a solution is to introduce orthogonality between carriers.
In this way, the correlator at the receiver can recover transmitted
information without interference, even when the set of possible carriers
is overlapped at the same time and/or frequency (hence the increase
in spectral efficiency). Because scaling does not affect the orthogonality
{[}\citet{ch2:bk:Tri_T_Ha}{]}, the amplitude of orthogonal carriers
can be normalized such that the transmitted power average is kept
unchanged (hence small loss in power efficiency). In practice, there
are two approaches for designing such orthogonality, either via digital
coded signals or via analogue multi-carriers, respectively, as reviewed
next:
\begin{itemize}
\item \textit{Code Shift Keying (CSK):} The simplest method of orthogonal
coded modulation is CSK, which bijectively maps $M$ information symbols
to $M$ orthogonal coded signals, notably Walsh functions {[}\citet{ch2:bk:Tri_T_Ha}{]}.
In practice, the number of orthogonal coded signal can be limited,
given a fixed symbol period. The orthogonal condition of coded streams
can be relaxed to a low cross-correlation condition of pseudo-random
or pseudo-noise (PN) code streams. Depending on the modulation scheme
(PSK or FSK) that the PN code is applied to, the result is called
direct sequence (DS) or frequency hopped (FH) spread spectrum signal,
respectively {[}\citet{ch2:bk:Proakis:Comm01}{]}. The name spread
spectrum comes from the fact that the rate is higher than unity, i.e.
the coded symbol period $T_{c}$ is smaller than the information symbol
period $T_{s}$. The signal bandwidth is, therefore, increased from
$1/T_{s}$ to $1/T_{c}$ {[}\citet{ch2:bk:OFDM_CDMA}{]}. \\
In multiple-access communication systems, where multiple users share
a single carrier, a similar orthogonal coding scheme is implemented
via Code Division Multiplexing (CDM), which multiplexes different
user bit streams orthogonally onto one carrier. In broader communication
networks, where users share the same channel bandwidth, CDM is generalized
to Code Division Multiple Access (CDMA), which is the standard spread
spectrum technique in 3G systems {[}\citet{ch2:bk:OFDM_CDMA}{]}.
The combination of CSK modulation with CDMA system is also a research
topic, aimed at increasing recovery performance and decreasing user
data interference {[}\citet{ch2:art:modem:CSK_2009}{]}.
\item \textit{Orthogonal Frequency Division Multiplexing (OFDM):} OFDM is
a key of 4G system, as reviewed in Section \ref{subsec:chap2:Generational-evolution-of-mobile}.
The key idea of OFDM is to employ orthogonal sub-carriers, before
multiplexing them into a single carrier for transmission. There are,
however, two different viewpoints on this concept of multi-carrier
modulation {[}\citet{ch2:bk:OFDM_CDMA}{]}: (i) in practical viewpoint,
each time slot  $T_{s}$ for the symbol period is fixed, while a filter
bank for $K$ bandpass filters, whose minimum frequency spacing is
symbol rate $1/T_{s}$ to preserve the orthogonality of the sub-carriers,
is employed for $K$ symbol pulse shapes of parallel data sub-streams;
(ii) In the textbook's viewpoint, the number $K$ of sub-carrier frequencies
is fixed, while modulation (e.g. QAM) is employed for sub-carriers
in time direction (hence, OFDM is also called Discrete Multi-Tone
(DMT) modulation if QAM is used {[}\citet{ch2:bk:Proakis:Comm01}{]}).
Note that, because a pulse cannot be strictly limited in both time
and frequency domains, (i) and (ii) are merely two viewpoints of the
same process: either truncating frequency-orthogonal time-limited
pulse in frequency domain, or truncating time-orthogonal band-limited
pulse in time domain, respectively. The method (i) is preferred in
practice, because, owing to the inverse Fast Fourier Transform (FFT),
the OFDM is a fast synthesis of Fourier coefficients modulated by
data sub-streams. Furthermore, before transmission, the D/A converter
for that synthesized digital carrier can be feasibly implemented via
oversampling (i.e. padding zero over unused DFT bins) in digital filters,
instead of via complicated analogue filters {[}\citet{ch2:bk:OFDM_CDMA}{]}.
\\
Historically, although the original idea of multi-carrier transmission
was first proposed in {[}\citet{ch2:origin:OFDM:Chang66,ch2:origin:OFDM:Saltzberg67}{]}
via a large number of oscillators, the implementation via digital
circuits for high-speed transmission was out of question for a long
time {[}\citet{ch2:bk:OFDM_CDMA}{]}. By including an extra guard
time-interval, which adds a cyclic prefix in the DFT block to itself,
OFDM was rendered suitable for mobile channels {[}\citet{ch2:origin:OFDM:Cimini85}{]}.
Indeed, the periodic nature of the DFT sequence in guarding-interval
makes the start of original symbols always continuous and greatly
facilitates the time synchronization issue in mobile systems {[}\citet{ch2:bk:Tri_T_Ha}{]}.
Since then, the main motivation of multi-carrier systems is to reduce
the effects of intersymbol interference (ISI), although the application
is two-fold. One one hand, longer symbol period and, hence, smaller
frequency spacing makes the system more robust against channel-time
dispersion and channel-spectrum variation within each frequency slot,
respectively {[}\citet{ch2:bk:Proakis:Comm01}{]}. On the other hand,
phase and frequency synchronization issues become more severe than
in traditional systems. Hence, in practice, OFDM has been used in
Wireless IEEE 802.11 and the terrestrial DVB-T standard, while the
cable DVB-C and satellite DVB-S systems still employ conventional
single carrier modulations {[}\citet{ch2:bk:OFDM_CDMA}{]}. The combination
of CDMA and OFDM, namely Multi-carrier CDMA, is also a promising method
for future mobile systems {[}\citet{ch2:bk:MC-CDMA:Hanzo03,ch2:bk:MC-CDMA:Fazel03}{]}. 
\end{itemize}

\subsubsection{Modulation with memory}

When power is the premium resource, modulation of carrier frequency
and of phase is preferred to carrier amplitude. Because the signal
is often continuously modulated in this case, it introduces a memory
(Markov) effect, in which the current $k$ symbol depends on the most
recent, say $L$, symbols. Hence, this modulation with memory can
be effectively represented via a Markov chain model, as follows:
\begin{itemize}
\item \textit{Differential encoding:} The simplest modulation scheme with
memory is differential encoding, in which the transition from one
level to another only occurs when a specific symbol is transmitted.
That level can be a Markov state of either phase or amplitude, corresponding
to differential $M$-ary PSK or Differential QAM {[}\citet{ch2:art:modem:differential_2006}{]},
respectively. 
\item \textit{$M$-ary FSK:} In this case, the frequency bandwidth can be
divided into $M$ frequency slots, whose minimum frequency spacing
is symbol rate $1/T_{s}$ to preserve the orthogonality of $M$ sinusoidal
carriers. Different from linear modulation scheme like QAM (i.e. the
sum of two QAM signals is another QAM signal), FSK is a non-linear
modulation scheme, which is more difficult to demodulate. Another
major drawback of $M$-ary FSK is the large spectral side lobes, owing
to abrupt switching between $M$ separate oscillators of assigned
frequencies {[}\citet{ch2:bk:Proakis:Comm01}{]}. The solution is
to consider a continuous-phase FSK (CP-FSK) signal, which, in turn,
modulates the single carrier's frequency continuously. In general,
CP-FSK is a special case of continuous phase modulation (CPM), where
the carrier's time-varying phase is the integral of pulse signals,
and represents the accumulation (memory) of all symbols up to the
modulated time {[}\citet{ch2:bk:Proakis:Comm01,ch2:bk:Tri_T_Ha}{]}.
In the literature, CPM is an important modulation scheme and widely
studied because of its efficient use of bandwidth {[}\citet{ch2:art:modem:CPM_1988,ch2:art:modem:CPM_2009}{]}. 
\end{itemize}
In order to increase power efficiency further with small loss of spectral
efficiency, a solution is to design an effective constellation mapping
in higher dimensional space. For example, the dimension of a trajectory
of $n$ $M$-state symbols is $M^{n}$, which increases exponentially
with $n$. Although the constellation design in the I-Q plane can
only be applied in two dimensions, its design principle can be applied
to a trajectory in $M^{n}$ dimensions. Then, the criterion for the
trajectory's constellation mapping can be chosen as a max-min problem,
which is to maximize the minimum Euclidean distance between any two
trajectories. Note that, the Gray code mapping for $n$ separate symbols
may fail to achieve that criterion and may yield only a small increase
in power efficiency {[}\citet{ch2:bk:Tri_T_Ha}{]}. 
\begin{itemize}
\item \textit{Trellis Coded Modulation (TCM)} {[}\citet{ch2:origin:TCM:Gottfried82}{]}:
The power efficiency, also known as coding gain, can be greatly improved
via TCM. The trajectory's constellation in TCM is designed via a principle
of mapping by set partitioning, i.e. the minimum Euclidean distance
between any two trajectories is increased with any partition of I-Q
constellation of each new symbol. Hence, the active point in current
I-Q constellation depends on both current symbol and the active point
in previous I-Q constellation. In other words, the current symbol
does not point to a point in I-Q constellation like in QAM, but to
a transition state between two consecutive I-Q constellation planes
(hence the name Trellis in TCM) {[}\citet{ch2:bk:Proakis:Comm01}{]}.
In the literature, the TCM is mostly designed to map the channel stream-code
bits (e.g. via Convolutional or Turbo code), instead of original bit
stream, in order to increase the joint decoder and demodulator performance
{[}\citet{ch2:bk:ToddMoon}{]}. There is, however, no mapping guaranteed
to achieve optimal Euclidean distance or BER performance {[}\citet{ch2:bk:TCM:Anderson_2003}{]}.
In current standard systems, although TCM is not selected in Wireless
IEEE 802.11a {[}\citet{ch2:bk:modem:TCM_2002}{]}, owing to difficulty
in design, it has been applied in Wireless IEEE 802.11b and IEEE 802.15.3
{[}\citet{ch2:bk:modem:TCM_2011}{]}
\end{itemize}

\subsection{The communication channel}

The communication channel is, by definition, the physical medium for
transmission of information. In practical channels, the common problem
is that, owing to unknown characteristics of the channel and transmission
devices,  random noise will be added to the transmitted signal. These
unknown quantities yields an uncertain signal, whose original form
can be inferred from the channel's probability model. Two major concerns
in the channel model are transmitted power and available channel bandwidth,
which represents the robustness against noise effect and the physical
limitations of the medium, respectively. Corresponding to those two
concerns, the characteristics of three major models in the literature,
namely AWGN, band-limited and fading channels, respectively, will
be briefly reviewed below. Those three channels represent the former
concern, the latter concern and both of them, respectively.

\subsubsection{The Additive White Gaussian Noise (AWGN) channel \label{subsec:chap2:AWGN-channel}}

The simplest model for the communication channel is the additive noise
process, which often arises from electronic components, amplifiers
and transmission interference. The noise model for the third cause
will be discussed in fading channel model below. For the first two
causes, the primary type of noise is thermal noise, which can be characterized
as Gaussian noise process {[}\citet{ch2:bk:Proakis:Comm01}{]} (hence
the name additive Gaussian channel). The noise is often assumed to
be white, that is, it has constant power spectral density (PSD), usually
denoted $N_{0}$ or $N_{0}/2$ for one-sided or two-sided PSD, respectively.
The Gaussian noise process is, therefore, wide-sense stationary and
zero mean. In the literature, the AWGN channel is perhaps the most
widely used model, owing to its mathematical simplicity. 

Nevertheless, the AWGN model is a mathematical fiction, because its
total power (i.e. the PSD integrated over all frequencies) is infinite
{[}\citet{ch2:bk:OFDM_CDMA}{]}. Its time sample has infinite average
power and, therefore, cannot be sampled directly without a filter.
In the literature, the ideal output of that filter at the receiver
is a noise model with finite power, namely discrete AWGN, which is
an iid Gaussian process with zero mean. 

\subsubsection{Band-limited channel \label{subsec:chap2:ISI-channel}}

In some communication channels, the transmitted signals are constrained
within a bandwidth limitation in order to avoid interference with
one another. If the channel bandwidth $B$ is smaller than the signal
bandwidth, the modulated pulse will be distorted in transmission.
In theory, the Nyquist criterion for zero intersymbol interference
(ISI) provides the necessary and sufficient condition for a pulse
shape to be transmitted over flat response channel without ISI. Such
a pulse with zero ISI is also called the Nyquist pulse in the literature.
The sampling rate $1/T_{s}$ must be greater than or equal to the
Nyquist rate $2B$, otherwise the Nyquist pulse does not exist {[}\citet{ch2:bk:Tri_T_Ha}{]}.
The Nyquist pulse with minimum bandwidth is the ideal sinc pulse.
However, because the sinc pulse is idealized, a raised-cosine filter
with small excess bandwidth is often used as Nyquist pulse in practice. 

The band-limited channel is also the simplest form of dispersive channel,
which responds differently with signal frequency. In practice, a noisy
dispersive channel is often modeled as linear filter channel with
additive noise. Such a noise can be white or \foreignlanguage{british}{colour}ed,
depending on whether the channel filter is put before or after the
noise. The additive colored noise with filtered PSD, which implies
correlation between samples, is more complicated and, therefore, is
typically transformed back to white noise model via a whitening filter
at the receiver.

\subsubsection{Fading channel \label{subsec:chap2:Fading-channel}}

In the mobile system, the transmitted signal always arrives at the
receivers as a superposition of multiple propagation paths, which
generally arise from signal reflection and scattering in the environment.
This type of fading channel actually appears in all forms of mobile
wireless communication {[}\citet{ch2:bk:Tri_T_Ha}{]}. The fading
process is characterized by two main factors, space varying (or frequency
selectivity) and/or time varying (or Doppler shift):
\begin{itemize}
\item \textbf{Space varying (or frequency selectivity) }
\end{itemize}
In the literature, this issue is characterized by a correlation frequency
$f_{corr}$ (also called coherence bandwidth), which is the inverse
of the delay spread $\Delta\tau$ arising from different traveling
time of multiple paths. The fading channel with or without ISI is
called frequency selective or flat (non-selective) fading channel,
respectively. In practice, flat fading can be approximately achieved
if the channel bandwidth $B$ satisfies $B\ll f_{corr}$. Since $B$
is of the order of $T_{s}^{-1}$ for a Nyquist basis, such a condition
is corresponding to $\Delta\tau\ll T_{s}$, i.e. the time delay is
much smaller than symbol period $T_{s}$. 

Note that, unlike ISI caused by channel filtering, the ISI in fading
channel is caused by random arrival time of multi-path signal copies
and, therefore, cannot be eliminated by pulse shaping in Nyquist criterion
for zero ISI {[}\citet{ch2:bk:Tri_T_Ha}{]}. For example, a symbol
period $T_{s}=10\ \mu s$, leading to an approximate data rate of
$200\ kbits/s$ for QPSK modulation, has the same order as $\Delta\tau=10\ \mu s$
of delay time corresponding to $3\ km$ difference of traveling paths
at light speed $c$ {[}\citet{ch2:bk:OFDM_CDMA}{]}. It means that
such a data transmission in practice cannot be free of ISI without
sophisticated techniques like equalizers, spread spectrum or multi-carrier
modulation. For example, the main motivation of OFDM is, intuitively,
to prolong the symbol period and, in turn, narrow the signal bandwidth.
In this way, OFDM avoids the use of a complex equalizer in demodulation,
although the Doppler spreading effect, as reviewed below, destroys
the orthogonality of OFDM sub-carriers and results in intercarrier
interference (ICI) {[}\citet{ch2:bk:Proakis:Comm01}{]}. 
\begin{itemize}
\item \textbf{Time varying (or Doppler shift)}
\end{itemize}
In the literature, this issue is characterized by correlation time
$t_{corr}$ (also called variation's timescale), which is the inverse
of maximum Doppler shift $\fDoppler=\frac{v}{c}f_{c}$ , given relative
speed $v$ between transmitter and receiver and carrier frequency
$f_{c}$. For example, the amplitude might be faded up to $-40\ dB$
at Doppler shift $f_{D}=50\ Hz$, corresponding to vehicle speed $v=60\ km/h$
and frequency carrier $f_{c}=900\ MHz$ {[}\citet{ch2:bk:OFDM_CDMA}{]}.
The fading channel is called slow or fast fading if signal envelope
fluctuates little or substantially within symbol period $T_{s}$,
respectively. The condition for slow fading is $T_{s}\ll t_{corr}$,
or equivalently $\fDoppler T_{s}\ll1$ (hence the name normalized
Doppler frequency $\fDoppler T_{s}$). In practice, the fluctuation
in carrier amplitude and phase is the superposition of multiple Doppler
shifts corresponding to multiple paths, which results in the so-called
Doppler spectrum instead of Doppler sharp spectral line at $f_{D}$. 

In a probability modelling context, each propagation path is considered
to contribute random delay and attenuation to the received signal,
whose envelope can be described as Rayleigh fading process (no guided
line-of-sight path), or Rician fading process (one strong line-of-sight
path), or the most general model, Nakagami fading process. Out of
the three, Clarke's Rayleigh model {[}\citet{ch2:origin:Fading:Clark68}{]}
is the most widely used model for wide-sense stationary fading process,
owing to its mathematical simplicity and tractability, compared to
the other two general models for real channel {[}\citet{ch2:bk:Proakis:Comm01}{]}.
Because the power spectral density (PSD) in flat fading with Clarke's
model is not a rational function, such a random process cannot be
represented by an auto-regressive (AR) model. This drawback leads
to difficulty in evaluation of channel detection {[}\citet{ch3:ART:FadingMarkov:tutorial08}{]}
and channel simulation {[}\citet{ch2:art:Fading:simulation06}{]}. 

For feasible evaluation of channel's detection, the finite-state Markov
channel (FSMC) was revived in {[}\citet{ch2:art:Fading:FSMC_95}{]}
for modeling the fading channel, owing to its computational simplicity.
Since then, the first-order FSMC model has been a model of choice
for fading channel, although it is more accurate for slow fading rates
than fast fading rates {[}\citet{ch3:ART:FadingMarkov:tutorial08}{]}. 

For feasible simulation, an approximate AR model, with sufficiently
large order $M$, for Clarke's model was proposed in {[}\citet{ch3:art:Fading:AR05}{]}.
Because such an AR model is essentially a Markov process with order
$M$, a high-order FSMC is also more accurate for fast fading channel
{[}\citet{ch2:art:Fading:Capacity_05}{]}. 

\subsection{Digital demodulator \label{subsec:chap2:Digital-demodulator}}

The purpose of the digital demodulator at receiver is to recover the
transmitted symbols carried on the modulated signal. The demodulation
is, therefore, an inference task based on the designed modulation
scheme and the noise model of the channel. In the practical system,
digital demodulation involves two steps, namely a signal processor
and digital detector {[}\citet{ch2:bk:Tri_T_Ha}{]}. The former, which
can be interpreted as the general form of A/D converter, is to convert
the noisy signal into a digital observation sequence, such that sufficient
statistics for transmitted symbols is preserved. The latter is to
infer transmitted symbols from this converted digital sequence. Those
two steps can be implemented sequentially or iteratively, as reviewed
below.

\subsubsection{Signal processor \label{subsec:chap2:Signal-processor}}

If the digital detector process is implemented directly in digital
software, high sampling rate around carrier frequency would be required.
To avoid such over-sampling, the solution is to obtain the lowpass
equivalent signal and produce a digital sequence from this lowpass
signal. For this purpose, the most important signal processor is the
matched filter {[}\citet{ch2:bk:Proakis:Comm01}{]}. Owing to equivalence
between convolution and correlation operator at sampling time point,
the matched filter is also equivalent to a correlator, whose purpose
is to extract the transmitted symbol via correlation between reference
and received signal. For a noiseless channel, this is obviously a
perfect extraction. For an additive noisy channel, the matched filter,
which is linear, can preserve the additivity of noise model in digital
sequence and, therefore, provide sufficient statistics for inferring
transmitted symbols. 

Nevertheless, there are two major difficulties for the matched filter.
Firstly, traditional matched filter requires a coherent demodulation,
which assumes that local reference carrier can be synchronized perfectly
with received signal in frequency and phase. Secondly, in the band-limited
channel, a cascade of a matched filter and an equalizer, which is
the compensator for reducing intersymbol interference (ISI), is no
longer a matched filter and might destroy the optimality of the original
matched filter {[}\citet{ch2:bk:Tri_T_Ha}{]}. Those two issues are
also major concern for signal processor in practical system. 
\begin{itemize}
\item \textbf{Synchronization}
\end{itemize}
At the receiver, three main parameters, which need to be synchronized
between the transmitted and received signal, are carrier phase, carrier
frequency and symbol timing. In practice, the state-of-the-art estimation
method for those three issues is Maximum Likelihood (ML), although
Bayesian techniques for synchronization have been proposed recently
in the literature, as reviewed below.

For phase synchronization, the ML phase estimator $\widehat{\phi}_{ML}$
is typically tracked by a Phase-Locked-Loop (PLL) circuit in practice.
Instead of computing $\widehat{\phi}_{ML}$ directly, which would
require a long observed time and cause delay in computation, the PLL
continuously tunes reference phase $\widehat{\phi}$ via a feedback
circuit until the value of the likelihood's derivative against carrier
phase becomes zero {[}\citet{ch2:bk:Proakis:Comm01}{]}. Owing to
its adaptability to the channel's variation and non-data-aided (NDA)
scheme, PLL research is extensive in the literature (see e.g. {[}\citet{ch2:bk:sync_phase:Lindsey72},\citet{ch2:bk:SEP:sync02}{]}).
Nevertheless, the data-aided (DA) (either pilot-based or non-pilot-based)
scheme for phase synchronization significantly increases the accuracy.
For pilot-based scheme, ML estimation $\widehat{\phi}_{ML}$ can be
derived feasibly {[}\citet{ch2:bk:sync_phase:Meyr97}{]}. For non-pilot-based
scheme, the estimation accuracy for low SNR is only acceptable via
channel code-aided (CA) scheme, although it is only possible to return
a local ML phase estimator via iterative EM algorithm in this case
{[}\citet{ch2:art:sync_Turbo:Herzet07}{]}. Recently, in {[}\citet{ch2:art:sync_phase:SEP_ICASSP}{]},
a Bayesian technique was applied to phase inference in a simple single-tone
carrier model, and showing that the Von Mises distribution is a suitable
conjugate prior for phase uncertainty in this case.

For frequency synchronization, three state-of-the-art methods in the
literature are the periodogram, phase increment and auto-correlation
methods {[}\citet{ch3:PhD:Freq:pilot09}{]}. The periodogram method
is equivalent to the ML estimation method {[}\citet{ch3:origin:Freq:ML74}{]},
while the other two are low-complexity sub-optimal methods, which
exploit the rotational invariance of displaced cisoidal signals {[}\citet{ch2:bk:Proakis:DSP06}{]}.
The technical details of these three methods will be presented in
Section \ref{sec:chap3.2:Freq}. In practice, the accuracy of frequency
estimation is high and, hence, frequency synchronization is much less
severe than phase synchronization. For example, in the AWGN channel,
the frequency-offset error in OFDM is typically around $1\%$ of the
sub-carrier spacing in high SNR {[}\citet{ch2:bk:Tri_T_Ha}{]}. Nevertheless,
frequency synchronization for OFDM is more severe in practical channels,
e.g. low SNR or fading channel {[}\citet{ch2:bk:OFDM_CDMA,ch2:bk:Proakis:Comm01}{]},
because the error in frequency estimation will destroy orthogonality
between sub-carriers. Recently, in {[}\citet{ch3:art:Freq:esprit13}{]},
the ESPRIT method was proposed for estimating OFDM's frequency offset
in low-cost direct-conversion receiver (DCR), which demodulates the
received signal in analogue domain and, hence, is prone to frequency-selective
I-Q imbalance. Regarding Bayesian techniques, the posterior distribution
for uncertain frequency is typically not in closed-form because sinusoidal
signal is a non-linear model of frequency. Hence, Bayesian techniques
for frequency have not been applied in practice, although, recently,
the MCMC method was proposed in the literature {[}\citet{ch2:art:sync_freq:MCMC_05}{]}
for approximating the frequency posterior distribution. 

For symbol timing, the typical scenario is data-aided (DA) synchronization
{[}\citet{ch2:art:sync_time:Bergman95}{]}. Such a scheme is reasonable,
since symbol timing involves the symbol data to begin with. For the
DA scheme, the ML delay estimator can be tracked via delay-locked
loop (DLL) circuit, which operates similarly to the PLL for phase.
Note that, a variant of DLL, namely Early-gate DLL, can also be applied
to non-data-aided (NDA) or low SNR scheme, by marginalizing out assumed
equi-probable symbols from observation model {[}\citet{ch2:bk:Proakis:Comm01}{]}.
Carrier phase and symbol timing can also be estimated jointly via
joint ML estimator in order to achieve higher accuracy {[}\citet{ch2:art:sync_time:time_phase_77}{]}.
Similarly, the delay time can be estimated along with phase offset
in above CA synchronization {[}\citet{ch2:art:sync_Turbo:Herzet07}{]}.
In OFDM, the symbol timing issue is less severe, owing to cyclic prefix
of symbols in guard time period {[}\citet{ch2:bk:OFDM_CDMA}{]}. 
\begin{itemize}
\item \textbf{Equalization}
\end{itemize}
If $H(f)$ is the product of the transmission filter and known channel
response, the matched filter $H^{*}(f)$ for zero ISI at receiver
can be designed such that the folded spectrum $X(f)$, i.e. the new
channel spectrum $H(f)H^{*}(f)$, satisfies the Nyquist criterion
for zero interference. Since the channel response is typically unknown,
a block of digital equalization filter for reducing ISI usually consists
of two parts. In the first part, a digital noise-whitening filter
$H_{w}(z)$ is designed such that the colored noise, i.e. the channel's
white noise filtered by matched filter, can be whitened and, hence,
uncorrelated. In the second part, a linear digital equalizer, $G(z)$,
can be chosen as one of two popular models, namely zero-forcing equalizer
$(H(z)H_{w}(z))^{-1}$ (for ideally noiseless channel) and MMSE equalizer
(for unknown noisy channel, but with WSS symbol sequence). The former
filter forces ISI to zero, but tends to increase noise power and,
hence, reduce SNR {[}\citet{ch2:bk:Tri_T_Ha}{]}. The latter filter,
whose coefficients can be estimated via the LMS or RLS algorithm,
minimizes the mean square error (MSE) between randomly WSS symbol
sequence and the output of equalizer {[}\citet{ch2:bk:Proakis:Comm01}{]}.
The performance of linear filter can be significantly increased via
the data-aided scheme, also known as the decision-feedback scheme,
when combined with digital detector. These decision-feedback equalizers
(DFE) are, however, non-linear filters {[}\citet{ch2:bk:Tri_T_Ha}{]}.
Note that, since the channel response is not known in practice, adaptive
algorithms need to be applied to these equalizers in order to track
channel's characteristics {[}\citet{ch2:bk:Proakis:Comm01}{]}.

\subsubsection{Digital detector \label{subsec:chap2:Digital-detector}}

As mentioned earlier in Section \ref{chap2:sub:Digital-modulator},
the interface between the channel decoder and digital detector is
blurred. In traditional definition, however, the output of digital
detector and, hence, of digital demodulator is hard-information of
transmitted symbol sequence, while the channel decoder increases the
detector's performance further, owing to channel encoding methods
{[}\citet{ch3:BK:DigiComm:Madhow08}{]}. In the literature, the term
``demodulation'' generally focuses on the output of the digital
detector, which relies on modulation scheme and channel characteristics,
while the term ``decoder'' implies the channel decoder, which relies
on encoding model (see e.g. {[}\citet{ch2:art:detector:eg_joint_03}{]}). 

The state-of-the-art techniques for digital detector are, therefore,
ML estimator and ML sequence estimator (MLSE), corresponding to two
detector schemes, namely symbol-by-symbol and joint symbol sequence,
respectively. Note that, for simplicity in study of demodulation,
the bit sequence at the input of modulator is often assumed as an
iid Bernoulli uniform sequence. Although ML and MAP estimators are
equivalent in this case, the term MAP estimators are often preserved
for more complicated input, e.g. Markov source or channel encoding
sequence.
\begin{itemize}
\item \textbf{Symbol-by-symbol detector}
\end{itemize}
In this scheme, each transmitted symbol is detected independently
from each other. This simple scheme is mostly applied to memoryless
modulation with AWGN channel. For the case of $M$-ary modulation,
the ML estimator simply returns the constellation point closest in
Euclidean distance to observed symbol in constellation plane. For
the case of orthogonal modulation, the user streams (in CDMA) and
sub-carriers (OFDM) are well separated via correlation in the signal
processor above {[}\citet{ch2:bk:OFDM_CDMA}{]}. Hence, the detector
for each user stream or sub-carrier is similar to $M$-ary case.
\begin{itemize}
\item \textbf{Symbol sequence detector}
\end{itemize}
In this scheme, transmitted symbol sequence is detected jointly via
Markovian model. Given transition values in trellis diagram, the MLSE
can be found efficiently via Viterbi algorithm (VA). This powerful
scheme has been applied to all kind of modulation and channel models. 

For memoryless modulation, the MLSE is applied when ISI occurs and/or
symbol timing is difficult. Firstly, $K$ observation samples per
overlapped time period will be collected into a sequence of length,
say $n$, of $M^{K}$-ary symbols. Secondly, assuming that two consecutive
$M^{K}$-ary symbols overlap in $L$ symbols, an $M^{m}\times M^{m}$
augmented transition matrix with $M^{2m}$ valid transitions can be
constructed, where $m=2K-L$. Thirdly, the soft-information for each
state in $M^{m}$ states are computed from the observation sequence.
Finally, the MLSE for those $n$ symbols are returned by VA. For high
$n$, such a MLSE is the closest point to the observation sequence
in exponentially $O(M^{nm})$-ary constellation plane, while VA's
computational load $O(nM^{2m})$ is always linear with $n$. 

Hence, the performance of MLSE is significantly better than that of
both symbol-by-symbol detector in AWGN channel and equalizer method
in band-limited channel, with modest increase in computational complexity
{[}\citet{ch2:bk:Proakis:Comm01}{]}. A similar result is also achieved
in the modulation with memory effect and FSMC-fading channel, since
both of them are Markovian model to begin with. 

\section{Summary}

In this chapter, we asserted that the challenge for mobile systems
is more about efficient computation than performance breakthroughs,
as we put this insight into the context of the generational evolution
of telecommunications systems. An interesting remark is that the standard
transmission speed of any mobile generation, from 1G to 5G, was always
set equal to that of fixed-line communication in the previous generation.
This insight in the mobile generations illustrates the need for efficient
computation and will be used in discussion of future work in Section
\ref{sec:chap9:Progress-achieved}.

The fading channel, i.e. the environment that all mobile phones must
confront, was also picked out as the main cause of performance degradation.
This review also raises the need for better trade-off schemes - between
accuracy and speed - for symbol detection in the receiver. For this
reason, the fading channel and the search for new trade-offs will
receive more attention in this thesis (e.g. sections \ref{sec:ch3:Fading-channel}
and \ref{subsec:chap8:Rayleigh-fading-channel}, respectively).

To motivate Bayesian inference in this thesis, non-Bayesian techniques
were first reviewed, along with their drawbacks. The Least Squares
(LS), optimal MMSE estimator, Wiener filter and ML estimator are central
techniques in non-probabilistic estimation. Some major limitations
in DSP for frequentist techniques, e.g. unbiased estimator, and Bayesian
theory, e.g. subjectivity of the prior model, were explained and clarified.
The Bayesian methodology introduced in this chapter will be presented
in more detail in Chapter \ref{=00005BChapter 4=00005D}. 

All major operational blocks in the telecommunications system were
briefly reviewed in order to emphasize the significant role of Markovianity.
Indeed, Markovianity appears in every block of telecommunications
system and, more importantly, in most computationally-efficient schemes
for these blocks. By Markovianity, we mean the invariant of the \foreignlanguage{british}{neighbourhood}
structure in each objective functional factor. The distributive law
for ring theory will be used in Chapter \ref{=00005BChapter 5=00005D}
as an attempt to exploit further the advantage of this Markovianity,
which is also one of the main themes for our future work (Section
\ref{sec:chap9:Progress-achieved}).

%auto-ignore
%auto-ignore
%%%% Common

\global\long\def\xbold{\mathbf{x}}%

\global\long\def\calA{\mathcal{A}}%

\global\long\def\calO{\mathcal{O}}%

\global\long\def\calI{\mathcal{I}}%

\global\long\def\calN{\mathcal{N}}%

\global\long\def\calCN{\mathcal{CN}}%

\global\long\def\PSD{N_{0}}%

\global\long\def\vtheta{\theta}%

\global\long\def\fung#1{g\left(#1\right)}%

\global\long\def\seti#1#2{#1\in\{1,2,\ldots,#2\}}%

\global\long\def\setd#1#2{\{#1{}_{1},#1{}_{2},\ldots,#1_{#2}\}}%

\global\long\def\TRIANGLEQ{\triangleq}%

%%%% Chapter 3

\global\long\def\ndata{n}%

\global\long\def\npath{K}%

\global\long\def\itime{i}%

\global\long\def\istate{k}%

\global\long\def\ipath{p}%

\global\long\def\Mstate{M}%

\global\long\def\acarrier{a}%

\global\long\def\fcarrier{f_{c}}%

\global\long\def\fset{f_{o}}%

\global\long\def\Oset{\Omega_{o}}%

\global\long\def\fsampling{f_{s}}%

\global\long\def\fdelta{\Delta f}%

\global\long\def\fDoppler{f_{D}}%

\global\long\def\carrier{\psi}%

\global\long\def\noise{z}%

\global\long\def\wgn{w}%

\global\long\def\phasenoise{\eta}%

\global\long\def\source{s}%

\global\long\def\sbold{\mathbf{\source}}%

\global\long\def\pulse{u}%

\global\long\def\Tperiod{T}%

\global\long\def\Tsample{T_{s}}%

\global\long\def\Tsampling{T_{s}}%

\global\long\def\tdelay{\Delta\tau}%

\global\long\def\phidelay{\phi}%

\global\long\def\phix{\phi}%

\global\long\def\phasepath{\varphi}%

\global\long\def\hdelay{h}%

\global\long\def\power{P_{0}}%

\global\long\def\ACF{R}%

\chapter{Observation models for the digital receiver \label{=00005BChapter 3=00005D}}

In this chapter, the demodulation task for three communication scenarios
will be briefly reviewed. For simplicity, the channel noise will be
assumed to be additive white Gaussian noise (AWGN). As explained in
Section \ref{subsec:chap2:AWGN-channel}, despite being idealistic,
AWGN model is a major type of corruption, which serves as basic assumption
in many practical channels {[}\citet{ch2:bk:Proakis:Comm01}{]}. In
this thesis, let us assume further that there is no intersymbol interference
(ISI) and channel time-delay. 

As explained in Section \ref{subsec:chap2:Stationary-process}, analogue
form of the received signal can be represented via the Wold decomposition
as follows:
\begin{equation}
x(t)=Re\left\{ \fung t\sum_{\itime=1}^{n}\source_{\itime}\pulse\left(t-\itime\Tsampling\right)\right\} +w(t)\label{eq:ch3:x(t)}
\end{equation}
where $w(t)$ is AWGN process with PSD $\sigma^{2}=N_{0}/2$ (W/Hz),
$\source_{\itime}\in\calA_{\Mstate}$ is the $\itime$th complex symbol
belonging to $\Mstate$-ary alphabet $\calA_{\Mstate}\TRIANGLEQ\left\{ \xi_{1},\ldots,\xi_{\Mstate}\right\} $,
the wave form $g(t)$ represents both carrier and channel characteristics,
$\pulse(t)$ is unit-energy Nyquist pulse shape of duration $\Tsampling$. 

Let us recall that the case of non-zero ISI can be arranged to be
close to zero via equalization, as explained in Section \ref{subsec:chap2:Signal-processor}.
Hence, for simplicity, the case of zero ISI will be assumed in this
thesis. Also, as reviewed in the same section, the symbol timing can
be synchronized jointly with carrier phase offset. Because the phase
synchronization issue will be left for future work, let us assume
that the symbol timing can be synchronized perfectly, and that the
duration $\Tsampling$ is the same for both symbol period and sampling
period in this thesis.

Note that, the parameterization $\theta$ of probabilistic observation
model $f(x(t)|\theta,\mathcal{I})$ depends on design of model $\mathcal{I}$,
which, in turn, depends on characteristics of both carrier signal
over channel $\fung t$ and the source $\sbold_{n}\TRIANGLEQ\setd{\source}{\ndata}$.
Then, four possible scenarios of model $\mathcal{I}$ are: 
\begin{itemize}
\item $\mathcal{I}_{1}$ - both known $\{\fung t,\sbold_{n}\}$
\item $\mathcal{I}_{2}$ - known $\fung t$, unknown $\sbold_{n}$
\item $\mathcal{I}_{3}$ - unknown $\fung t$, known $\sbold_{n}$
\item $\mathcal{I}_{4}$ - both unknown $\{\fung t,\sbold_{n}\}$
\end{itemize}
Leaving out the trivial case $\mathcal{I}_{1}$, three remaining scenarios
will be studied in this thesis. More specifically, three basic problems
in communication will be considered respectively: 
\begin{itemize}
\item $\mathcal{I}_{2}$: synchronized symbol detection in AWGN channel
(i.e. known carrier and channel characteristics, but unknown symbols)
\item $\mathcal{I}_{3}$: pilot-based frequency offset estimation in AWGN
channel (i.e. unknown carrier characteristic, but known channel characteristic
and symbols)
\item $\mathcal{I}_{4}$: synchronized symbol detection in quantized fading
channel (i.e. unknown channel characteristic and symbols, but known
carrier characteristic)
\end{itemize}
Those three problems, in respective order, will be presented in three
sections below. 

\section{Symbol detection in the AWGN channel}

As reviewed in Section \ref{chap2:sub:Digital-modulator}, the simplest
carrier wave form $\carrier(t)$ in this case is a complex sinusoid,
as follows:

\begin{equation}
\fung t=\psi(t)\triangleq ae^{j2\pi f_{c}t}\label{eq:ch3:g(t):case1}
\end{equation}
where carrier parameters $\vtheta=\{\acarrier,\fcarrier\}$, i.e.
amplitude $\acarrier$ and carrier frequency $\fcarrier$, are assumed
known in this section. For simplicity, the carrier phase is assumed
to be null. 

Since $\vtheta$ is given, the baseband signal can be retrieved and
sampled from $x(t)$ via matched filter, designed at carrier frequency.
As explained in Section \ref{subsec:chap2:Signal-processor}, the
matched filter is also equivalent to a correlator, owing to duality
between convolution and correlation operators. Note that, this equivalence
is only valid at sampling point $t=\itime\Tsampling$, $\seti{\itime}{\ndata}$
{[}\citet{ch2:bk:Tri_T_Ha}{]}, i.e. we have $x_{I}(\itime\Tsample)=\tilde{x}_{I}(\itime\Tsample)$
and $x_{Q}(\itime\Tsample)=\tilde{x}_{Q}(\itime\Tsample)$ in Fig.
\ref{fig:chap3:filter_correlator}. As explained in Section \ref{subsec:chap2:Digital-demodulator},
the correlator is also a general form of A/D converter. The key difference
is that the sample in the output of correlator is the result of a
projection, instead of sampling values in traditional A/D. Hence,
the correlator form will be presented below for the sake of clarity
and intuition.

\begin{figure}
\begin{centering}
\includegraphics[width=0.8\columnwidth]{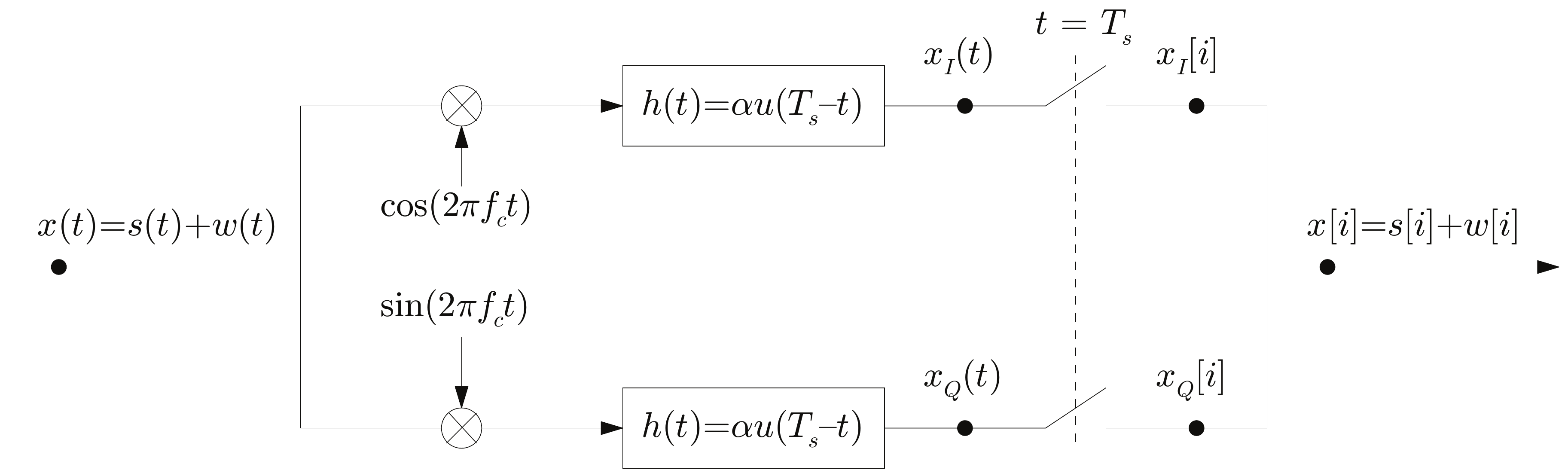}
\par\end{centering}
\begin{centering}
(a)
\par\end{centering}
\begin{centering}
$\ $
\par\end{centering}
\begin{centering}
\includegraphics[width=0.8\columnwidth]{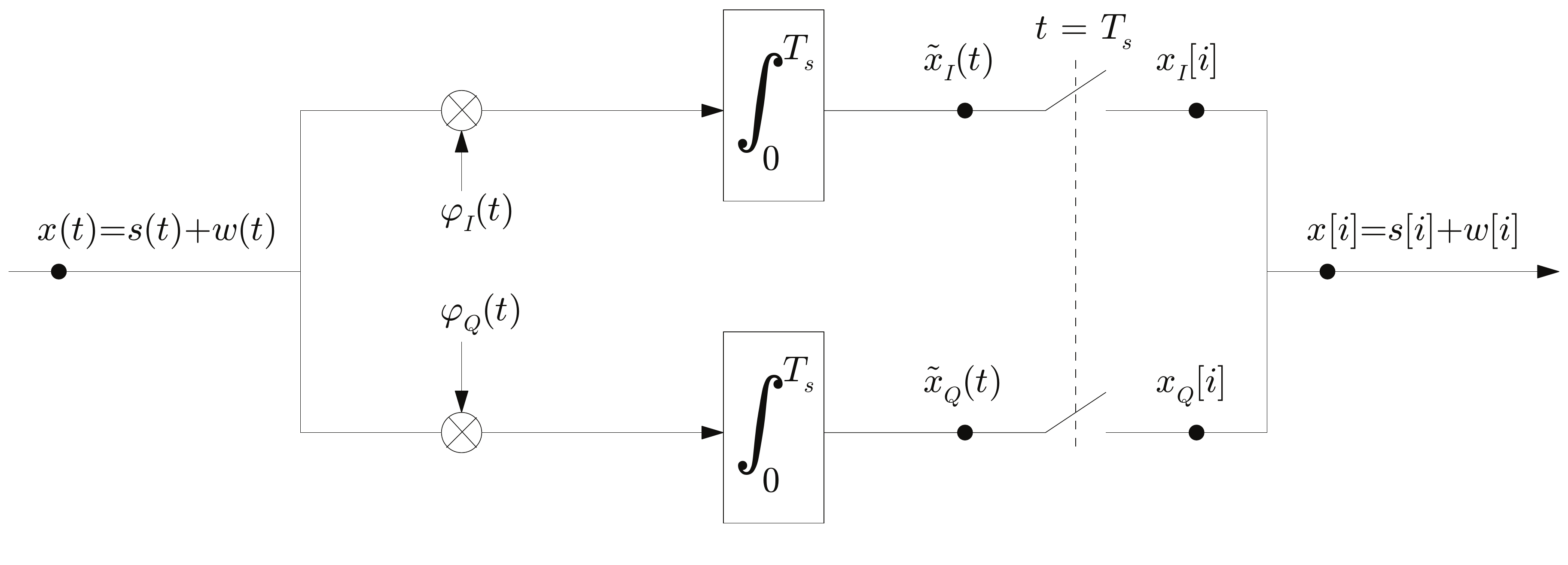}
\par\end{centering}
\begin{centering}
(b)
\par\end{centering}
\caption{\label{fig:chap3:filter_correlator}The equivalence between matched
filter (a) and correlator (b) in quadrature demodulator}
\end{figure}
In order to exploit the AWGN assumption, the key idea is to construct
an orthonormal basis spanning signal space, since the projection of
AWGN process in Hilbert space ${\cal H}$ onto that orthonormal basis
yields another AWGN process in sub-Hilbert space ${\cal S}$ with
the same PSD as the AWGN in ${\cal H}$ {[}\citet{ch3:BK:DigiComm:Madhow08}{]}.
For this purpose, let us consider a normalization constant $\alpha$
such that the in-phase $\varphi_{I}(t)=\alpha\pulse(t)\cos(2\pi f_{c}t)$
and quadrature $\varphi_{Q}(t)=\alpha\pulse(t)\sin(2\pi f_{c}t)$
components are orthonormal basis functions of signal space. Then,
the projections of $x(t)$ and $w(t)$ onto this signal space can
be collected as discrete complex variables, as follows:

\begin{equation}
\begin{cases}
x_{\itime} & =\left\langle x(t),\varphi_{I}(t)\right\rangle +j\left\langle x(t),\varphi_{Q}(t)\right\rangle \\
\wgn_{\itime} & =\left\langle w(t),\varphi_{I}(t)\right\rangle +j\left\langle w(t),\varphi_{Q}(t)\right\rangle 
\end{cases},\ \frac{t}{\Tsample}\in[\itime-1,\itime)\label{eq:ch3:w_k}
\end{equation}
where $\left\langle \cdot,\cdot\right\rangle $ is the inner product
in $\itime$th symbol interval, $\seti{\itime}{\ndata}$. Hence, the
output of quadrature demodulator in this case is the basic discrete
complex receiver in AWGN channel {[}\citet{ch3:art:AWGN:Forney98}{]},
as follows:

\begin{equation}
x_{\itime}=s_{\itime}+\wgn_{\itime},\ \itime\in\{1,\ldots,n\}\label{eq:ch3:x_k:case1}
\end{equation}
where the projection $\wgn_{\itime}$ of $w(t)$ via (\ref{eq:ch3:w_k})
is a discrete complex AWGN sequence, i.e. $\wgn_{\itime}\sim\calCN(0,\sigma^{2})$,
with the same variance $\sigma^{2}=N_{0}/2$ as $w(t)$ {[}\citet{ch3:BK:DigiComm:Madhow08}{]}.
Hence, the observation model can be written as follows:

\begin{equation}
f(x_{\itime}|s_{\itime})=\calCN_{x_{\itime}}(s_{\itime},\sigma^{2}),\ \itime\in\{1,\ldots,n\}\label{eq:ch3:f(x|s):case1}
\end{equation}

In classical estimation, the Maximum Likelihood (ML) estimator for
Gaussian noise can be found via the Least Squares (LS) method. In
the simplest case where $\sbold_{n}$ is a uniform iid sequence, each
symbol can be estimated separately, i.e. $\widehat{s_{\itime}}=\arg\min_{s_{\itime}\in\calA_{\Mstate}}\left\Vert x_{\itime}-s_{\itime}\right\Vert ^{2}$,
$\seti{\itime}{\ndata}$ . For more general case where $\sbold_{n}$
is a Markov source, Bayesian inference is needed. This Markovian case
will be studied in Chapter \ref{=00005BChapter 8=00005D}. 

\section{Frequency estimation in the AWGN channel \label{sec:chap3.2:Freq}}

When carrier parameters $\vtheta=\{\acarrier,\fcarrier\}$ are unknown,
e.g. owing to offset corruption in channel, a common solution for
these time-invariant parameters is to transmit a block of $\ndata$
known symbols (also called pilot symbols), in order to estimate $\vtheta$
before estimating true messages coming afterward {[}\citet{ch2:bk:Tri_T_Ha}{]}.
Hence, without loss of generality, let us assume that $s_{\itime}=1$
for all $\seti{\itime}{\ndata}.$ The carrier frequency $\fcarrier$
in this case becomes $\fcarrier+\fdelta$, where $\fdelta$ is the
unknown offset frequency. Denoting $\fsampling\TRIANGLEQ1/\Tsampling$
as symbol rate and $\fset\TRIANGLEQ\Tsampling\fdelta=\frac{\fdelta}{\fsampling}\in[-0.5,0.5)$
as normalized offset frequency, the channel wave form $\fung t$ in
(\ref{eq:ch3:g(t):case1}) now becomes: 

\begin{eqnarray}
\fung t & = & ae^{j2\pi(f_{c}+\fdelta)t}\label{eq:ch3:g(t):case2}\\
 & = & \psi(t)e^{j2\pi\fdelta t}\nonumber \\
 & \approx & \psi(t)e^{j2\pi\itime\fset},\ \frac{t}{\Tsample}\in[\itime-1,\itime)\nonumber 
\end{eqnarray}
in which the approximation in (\ref{eq:ch3:g(t):case2}) is valid
if $\fdelta$ is sufficiently smaller than symbol rate $\fsampling$,
or equivalently $\fset\ll1$. If this assumption is valid, receiver
can feasibly operate in steady-state condition, i.e. symbol timing
is synchronized first before $\fset$ is estimated {[}\citet{ch3:bk:Freq:Mengali97}{]}.
Applying the matched filter (i.e. correlator) at nominal carrier frequency
 $\fcarrier$ like above, the discrete complex data (\ref{eq:ch3:x_k:case1})
now becomes:

\begin{eqnarray}
x_{\itime} & = & ae^{j2\pi\itime\fset}+w_{\itime}\label{eq:ch3:x_k:case2}\\
 & = & ae^{j\Oset\itime}+w_{\itime},\ \itime\in\{1,\ldots,n\}\nonumber 
\end{eqnarray}
where $w_{\itime}$ is discrete AWGN sequence with variance $\sigma^{2}=N_{0}/2$
and $\Oset$ is digital offset frequency, with $\Oset\TRIANGLEQ2\pi\fset\in[-\pi,\pi)$.
The observation model can be written as follows:

\begin{equation}
f(x_{\itime}|a,\Oset)=\calCN_{x_{\itime}}(ae^{j\Oset\itime},\sigma^{2}),\ \itime\in\{1,\ldots,n\}\label{eq:ch3:f(x|s):case2}
\end{equation}

In the literature, the common concern is to estimate $\Oset$, while
amplitude $\acarrier$ is regarded as a nuisance parameter. Despite
being classical, single frequency estimation is still an interesting
issue in practice, as revised by {[}\citet{ch3:art:Freq:fast06,ch3:art:Freq:esprit13}{]}.
The important challenge is a trade-off between computational complexity
and estimation performance, particularly in the case of low signal-to-noise
ratio ($SNR$) {[}\citet{ch3:art:Freq:fast06}{]}. For this issue,
the periodogram, phase increment and auto-correlation are currently
three most common techniques in practice {[}\citet{ch3:PhD:Freq:pilot09}{]}
and will be briefly reviewed below.

\subsection{Single frequency estimation via Periodogram}

For a batch of data $\xbold_{\ndata}$, the classical Maximum Likelihood
(ML) estimator for frequency $\Oset$ is equivalent to the maximum
of periodogram {[}\citet{ch3:origin:Freq:ML74}{]}:

\begin{equation}
\widehat{\Oset}=\arg\max_{\Oset}\left|\sum_{\itime=1}^{\ndata}x_{\itime}e^{-j\Oset(\itime-1)}\right|^{2}\label{eq:ch3:Freq:Periodogram}
\end{equation}

In practice, the value of periodogram at DFT bins $\Oset=\frac{2\pi}{\ndata}m$
with $m\in\{0,\ldots,\ndata-1\}$ can be computed via Discrete Fourier
Transform (DFT) {[}\citet{ch3:PhD:Freq:pilot09}{]}. When $SNR=\acarrier^{2}/2\sigma^{2}$
is sufficiently high, the Mean Square Error (MSE) of ML estimator
$\widehat{\Oset}$ approaches the Cramér-Rao bound (CRB) for frequency
estimators {[}\citet{ch3:origin:Freq:ML74}{]}. However, when $SNR$
is below a certain threshold, the MSE of $\widehat{\Oset}$ rapidly
increases {[}\citet{ch3:art:Freq:iterative02}{]}. Another drawback
of periodogram is a high computational cost, even with sub-linear
complexity $\calO(\ndata\log\ndata)$ of FFT algorithm. {[}\citet{ch3:art:Freq:iterative02}{]}.
Hence, many sub-optimal estimators have been proposed to reduce the
computational complexity {[}\citet{ch3:art:Freq:fast06}{]}.

\subsection{Single frequency estimation via phase increment}

In order to avoid the high complexity in non-linear estimator in (\ref{eq:ch3:Freq:Periodogram}),
the noisy model for $x_{\itime}$ in (\ref{eq:ch3:x_k:case2}) can
be approximated by a noisy linear form when $SNR$ is sufficiently
high, as follows {[}\citet{ch3:origin:Freq:Tretter85}{]}:

\[
x_{\itime}\approx\acarrier e^{j\phix_{\itime}}=\acarrier e^{j(\Oset\itime+\eta_{\itime})},\ \itime\in\{1,\ldots,n\}
\]

Hence, the observed data in this case is: 
\begin{equation}
\phix_{\itime}=\Oset\itime+\eta_{\itime},\ \itime\in\{1,\ldots,n\}\label{eq:ch3:Freq:Tretter}
\end{equation}
where $\eta_{\itime}$ is discrete AWGN sequence with variance $\sigma^{2}/2\acarrier^{2}$,
i.e. $var(\eta_{\itime})=0.5/SNR$ {[}\citet{ch3:origin:Freq:Tretter85,ch3:origin:Freq:Kay89}{]}.
In order to avoid phase unwrapping in (\ref{eq:ch3:Freq:Tretter}),
Kay's method {[}\citet{ch3:origin:Freq:Kay89}{]} considered the differenced
phase data $\Delta\phix_{\itime}\TRIANGLEQ\phix_{\itime}-\phix_{\itime-1}=\angle x_{\itime}-\angle x_{\itime-1}$,
as follows:

\begin{eqnarray}
\Delta\phix_{\itime} & = & \arg\{x_{\itime}x_{\itime-1}^{*}\}\nonumber \\
 & = & \Oset+\Delta\phasenoise_{\itime},\ \itime\in\{2,\ldots,n\}\label{eq:ch3:Freq:Kay}
\end{eqnarray}
where $\Delta\phasenoise_{\itime}\TRIANGLEQ\phasenoise_{\itime}-\phasenoise_{\itime-1}$.
Hence, the phase increment method has replaced the original non-linear
model $f(x_{\itime}|a,\Oset)$ (\ref{eq:ch3:f(x|s):case2}) with its
approximated linear phase model $f(\Delta\phix_{\itime}|\Oset)$ in
(\ref{eq:ch3:Freq:Kay}). The ML estimator $\widehat{\Oset}$ for
$f(\Delta\phix_{\itime}|\Oset)$ is {[}\citet{ch3:origin:Freq:Kay89,ch3:bk:Freq:Mengali97}{]}:

\begin{equation}
\widehat{\Oset}=\sum_{\itime=2}^{\ndata}\lambda_{\itime}\arg\{x_{\itime}x_{\itime-1}^{*}\}\label{eq:ch3:Freq:Kay:ML}
\end{equation}
where $\lambda_{\itime}\TRIANGLEQ\frac{3}{2}\frac{\ndata}{\ndata^{2}-1}\left[1-\left(\frac{2\itime-\ndata}{\ndata}\right)^{2}\right]$
. Because the Kay's estimator (\ref{eq:ch3:Freq:Kay:ML}) is a weighted
average of phase increment, its complexity is low, with merely $\calO(\ndata)$
of complex multiplication. However, its main drawback is the moderate
performance, owing to linear approximation (\ref{eq:ch3:Freq:Tretter})
. Although Kay's estimator is unbiased and approaches Modified CRB
{[}\citet{ch3:bk:Freq:Mengali97}{]}, the $SNR$ value for low error
is high {[}\citet{ch3:PhD:Freq:pilot09}{]}. Many variants have been
proposed to improve the performance, while maintaining the low computational
cost (see for instance {[}\citet{ch3:art:Freq:iterative02,ch3:art:Freq:fast06}{]}). 

\subsection{Single frequency estimation via auto-correlation}

From (\ref{eq:ch3:Freq:Kay:ML}), we can see that the phase increment
has exploited the underlying rotational invariance of two temporally
displaced cisoidal signals {[}\citet{ch2:bk:Proakis:DSP06}{]}. Such
a property can also be exploited via discrete autocorrelation function
of $x_{\itime}$, defined as follows:

\begin{equation}
R[m]\TRIANGLEQ\frac{1}{\ndata-m}\sum_{\itime=m+1}^{\ndata}x_{\itime}x_{\itime-m}^{*},\ \seti m{\ndata-1}\label{eq:ch2:R=00005Bm=00005D:case2}
\end{equation}

Substituting $x_{\itime}$ (\ref{eq:ch3:x_k:case2}) into (\ref{eq:ch2:R=00005Bm=00005D:case2}),
we have {[}\citet{ch3:bk:Freq:Mengali97}{]}:

\begin{equation}
R[m]=e^{j\Oset m}+\vartheta_{m},\ \seti m{\ndata-1}\label{eq:ch3:Freq:R(m)}
\end{equation}
where $\vartheta_{m}$ is a zero-mean random noise. For efficiently
estimating $\Oset$ in (\ref{eq:ch3:Freq:R(m)}), because of lacking
noise model, Fitz method {[}\citet{ch3:origin:Freq:Fitz94}{]} considers
the time average error for the phase of $R[m]$, as follows:

\begin{equation}
\frac{1}{L}\sum_{m=1}^{L}\left(\arg\{R[m]\}-\arg\{e^{j\Oset m}\}\right)=\frac{1}{L}\sum_{m=1}^{L}\epsilon_{m}\approx0,\ 1<L<\ndata\label{eq:ch3:Freq:Fitz(0)}
\end{equation}
where the error $\epsilon_{m}$ is very small if $SNR$ is high and
the range $L<\frac{\pi}{\left|\Oset^{(\max)}\right|}$ can be properly
chosen via maximum uncertainty range $\pm\Oset^{(\max)}$ of $\Oset$
{[}\citet{ch3:bk:Freq:Mengali97}{]}. From linear equation (\ref{eq:ch3:Freq:Fitz(0)}),
we can compute the Fitz's estimator $\widehat{\Oset}$ feasibly:

\[
\widehat{\Oset}=\frac{2}{\ndata(\ndata-1)}\sum_{m=1}^{\ndata-1}\arg\{R[m]\}
\]

The Fitz's estimator is unbiased in the range $\pm\frac{\pi}{\ndata}$
of $\Oset$ and its MSE achieves the Modified CRB at $L=\frac{\ndata}{2}$.
When the range $L$ decreases, the computational load is lighter but
the accuracy also degrades {[}\citet{ch3:bk:Freq:Mengali97}{]}. Hence,
there is a trade-off between complexity and performance again. Some
improved versions of Fitz method can be found in {[}\citet{ch3:art:Freq:FitzVariant95,ch3:art:Freq:FitzVariant97}{]}.

\section{Symbol detection in the fading channel \label{sec:ch3:Fading-channel}}

As reviewed in Section \ref{subsec:chap2:Fading-channel}, in fading
channel, the transmitted signal is reflected from surroundings (e.g.
buildings, vehicles, etc.) and duplicated into multiple copies before
reaching mobile receiver. Because of multi-path environment, the received
signal is superposition of those copies, coming from different path
with various angles. In this thesis, the temporally fading effect,
which arises owing to motion of mobile receiver, of received signal
will be considered.

In flat fading channel (FFC), the multi-path delay spread $T_{d}$,
which is the maximum of difference $\tdelay$ in delay time of all
paths, is assumed to be small compared to symbol period $\Tsampling$.
Note that, because the coherence bandwidth $1/T_{d}$ in flat fading
is therefore larger than signal bandwidth $1/\Tsampling$, all frequency
components of signal will suffer the same magnitude of fading (hence
the name ``flat''). In this case, because the signal pulse is not
much affected by delay time on any $\ipath$th path, i.e. $\pulse\left(t-\itime\Tsampling-\tdelay_{\ipath}\right)\approx\pulse\left(t-\itime\Tsampling\right)$
{[}\citet{ch3:bk:Fading:cavers00}{]}, the channel waveform $\fung t$
of received signal in (\ref{eq:ch3:x(t)}) can be derived via superposition
principle, as follows:

\begin{eqnarray}
g(t) & = & \sum_{\ipath=1}^{\npath}ae^{j2\pi f_{c}(t-\tdelay_{\ipath})}\nonumber \\
 & = & \psi(t)\hdelay(t)\label{eq:ch3:g(t):case3}
\end{eqnarray}
in which the fading gain $\hdelay(t)$ of all $\npath$ arriving paths
is defined as:

\begin{eqnarray}
h(t) & = & \sum_{\ipath=1}^{\npath}e^{-j2\pi f_{c}\tdelay_{\ipath}}\nonumber \\
 & = & \sum_{\ipath=1}^{\npath}e^{j\phidelay_{\ipath}(t)}\label{eq:ch3:h(t)}
\end{eqnarray}

\begin{figure}
\begin{centering}
\includegraphics[width=0.5\columnwidth]{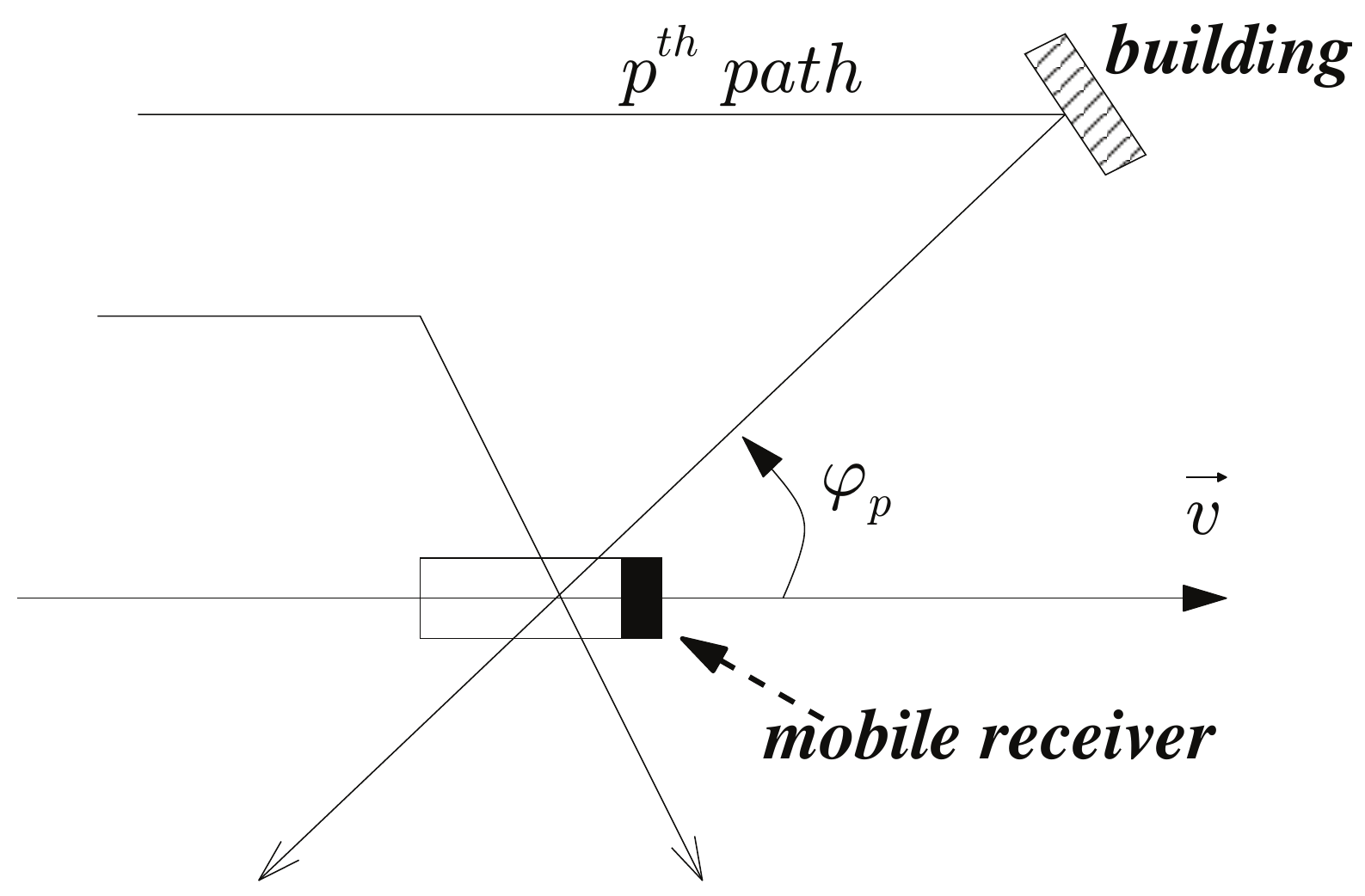}
\par\end{centering}
\caption{\label{fig:chap3:fading_path}Scattered transmission paths and receiver's
travelling velocity $\protect\overrightarrow{v}$}
\end{figure}
The computation for delay phase $\phidelay_{\ipath}(t)$ in (\ref{eq:ch3:h(t)})
can be carried out as follows: Denoting $v$ as velocity of mobile
receiver and $\phasepath_{\ipath}$ as arriving angle of $\ipath$th
path to receiver's moving direction {[}\citet{ch2:origin:Fading:Clark68}{]},
as illustrated in Fig. \ref{fig:chap3:fading_path}, the delay time
$\tdelay_{\ipath}$ caused by each path is $\tdelay_{\ipath}=\frac{\Delta x_{\ipath}}{c}=\frac{-(v\cos\phasepath_{\ipath})t}{c}$,
where $c$ is light speed and $\Delta x_{\ipath}=-(v\cos\phasepath_{\ipath})t$
is the path length change owing to receiver's motion {[}\citet{ch3:bk:Fading:cavers00}{]}.
Then, we have:

\begin{eqnarray}
\phidelay_{\ipath}(t) & = & -2\pi f_{c}\tdelay_{\ipath}=-2\pi f_{c}\frac{\Delta x_{\ipath}}{c}\nonumber \\
 & = & 2\pi f_{c}\frac{(v\cos\phasepath_{\ipath})t}{c}\nonumber \\
 & = & 2\pi(\fDoppler\mbox{\ensuremath{\cos\phasepath_{\ipath}})t},\ \seti{\ipath}{\npath}\label{eq:ch3:phi(t)}
\end{eqnarray}
in which the carrier Doppler shift for each path is $\fDoppler\mbox{\ensuremath{\cos\phasepath_{\ipath}}}=$$\frac{\fcarrier}{c}(v\cos\phasepath_{\ipath})$
and the frequency $\fDoppler$ is often called maximum Doppler shift,
since the maximum value $\cos\phasepath_{\ipath}=1$ is reached at
$\phasepath_{\ipath}=0$.

In practical communication system, since the extremely-fast FFC is
uncommon {[}\citet{ch3:ART:FadingMarkov:tutorial08}{]}, it is safe
to assume that complex fading gain will stay constant during any symbol
period, i.e. $\hdelay(t)\approx\hdelay_{\itime}\TRIANGLEQ\hdelay(\itime\Tsampling)$,
in which $\frac{t}{\Tsample}\in[\itime-1,\itime)$ and $\hdelay_{\itime}$
is the constant fading gain during $\itime$th symbol period. Then,
the channel waveform $\fung t$ in (\ref{eq:ch3:g(t):case3}) now
becomes: 
\begin{equation}
g(t)\approx\psi(t)\hdelay_{\itime},\ \frac{t}{\Tsample}\in[\itime-1,\itime)\label{eq:ch3:g=00005Bk=00005D:case3}
\end{equation}

Applying the matched filter at carrier frequency $\fcarrier$, the
discrete received data in this case are:

\begin{eqnarray*}
x_{\itime} & = & \hdelay_{\itime}s_{\itime}+\wgn_{\itime}\\
 & = & \left\Vert \hdelay_{\itime}\right\Vert e^{j\phidelay_{\itime}}s_{\itime}+\wgn_{\itime},\ \itime\in\{1,\ldots,n\}
\end{eqnarray*}

Because the fading phase $\phidelay_{\itime}$ can be separately estimated
along with channel synchronization, amplitude-only fading channel
is a popular choice for studying fading phenomenon in the literature
{[}\citet{ch3:ART:FadingMarkov:tutorial08}{]}. Then, for the amplitude-only
fading scenario, which is our interest in this thesis, we have:

\begin{equation}
x_{\itime}=\left\Vert \hdelay_{\itime}\right\Vert s_{\itime}+\wgn_{\itime},\ \itime\in\{1,\ldots,n\}\label{eq:ch3:receiver_case3}
\end{equation}

\subsection{Stationary Gaussian process in fading channel}

In practice, the sequence of arriving angle $\phasepath_{\ipath}$
can be modeled as uniformly iid random variable over the range $[-\pi,\pi)$
{[}\citet{ch3:ART:FadingMarkov:tutorial08}{]}. Therefore, the delay
phase $\phidelay_{\ipath}(t)$ in (\ref{eq:ch3:phi(t)}) is also distributed
uniformly over $[-\pi,\pi)$ at anytime $t$ {[}\citet{ch3:BK:Fading:Stuber11}{]}.
If the number $\npath$ of scattering paths is sufficiently large,
then owing to central limit theorem (CLT), the in-phase and quadrature
components of $h(t)=h_{I}(t)+jh_{Q}(t)$ in (\ref{eq:ch3:h(t)}) are,
respectively:

\begin{eqnarray}
h_{I}(t) & = & \sum_{\ipath=1}^{\npath}\cos\phidelay_{\ipath}(t)\label{eq:ch3:h_I}\\
h_{Q}(t) & = & \sum_{\ipath=1}^{\npath}\sin\phidelay_{\ipath}(t)\nonumber 
\end{eqnarray}
which can be approximated as independent Gaussian random variables
$\calN(0,\power/2)$, where $\power/2$ is called fading power per
dimension {[}\citet{ch3:BK:Fading:Stuber11}{]}. The average fading
gain $h(t)$ is a complex WSS Gaussian process, whose autocorrelation
function (ACF) is {[}\citet{ch3:BK:Fading:Stuber11}{]}:

\begin{eqnarray}
\ACF_{\hdelay}(\tau) & = & \ACF_{\hdelay_{I}}(\tau)+\ACF_{\hdelay_{Q}}(\tau)\nonumber \\
 & = & \power J_{0}(2\pi\fDoppler\tau)\label{eq:ch3:ACF(t)}
\end{eqnarray}
where $\ACF_{\hdelay_{I}}(\tau)=\ACF_{\hdelay_{I}}(\tau)=\frac{\power}{2}J_{0}(2\pi\fDoppler\tau)$
and $J_{0}(\cdot)$ is zero-order Bessel function of the first kind.
The result (\ref{eq:ch3:ACF(t)}) for Rayleigh fading channel is well
known in the literature {[}\citet{ch2:origin:Fading:Clark68,ch3:bk:Fading:cavers00,ch3:BK:Fading:Stuber11}{]}.
However, for a quick verification, let us derive (\ref{eq:ch3:ACF(t)})
briefly here. Substitute (\ref{eq:ch3:phi(t)}) to (\ref{eq:ch3:h_I}),
we have {[}\citet{ch3:BK:Fading:Stuber11}{]}:

\begin{eqnarray*}
\ACF_{\hdelay_{I}}(\tau) & = & E[h_{I}(t)h_{I}(t+\tau)]\\
 & = & \frac{\power}{2}E_{\phasepath_{\ipath}}[\cos\left(2\pi(\fDoppler\mbox{\ensuremath{\cos\phasepath_{\ipath}})}\tau\right)]
\end{eqnarray*}
owing to the central limit theorem (CLT). Then, by definition of zero-order
Bessel function of the first kind, $J_{0}(x)\TRIANGLEQ\frac{1}{\pi}\int_{0}^{\pi}\cos(x\sin(\phasepath))d\phasepath$,
we have:

\begin{eqnarray*}
\ACF_{\hdelay_{I}}(\tau) & = & \frac{\power}{2}\int_{-\pi}^{\pi}\cos\left(2\pi(\fDoppler\mbox{\ensuremath{\cos\phasepath_{\ipath}})}\tau\right)f(\phasepath_{\ipath})d\phasepath_{\ipath}\\
 & = & \frac{\power}{2}\frac{1}{\pi}\int_{0}^{\pi}\cos\left(2\pi(\fDoppler\mbox{\ensuremath{\sin\phasepath_{\ipath}})}\tau\right)d\phasepath_{\ipath}\\
 & = & \frac{\power}{2}J_{0}(2\pi\fDoppler\tau)
\end{eqnarray*}
where $f(\phasepath_{\ipath})=\frac{1}{2\pi}$ for $\phasepath_{\ipath}\in(-\pi,\pi]$
and zero orthewise. The computation for $\ACF_{\hdelay_{Q}}(\tau)$
is similar to $\ACF_{\hdelay_{I}}(\tau)$, therefore we have: 
\begin{equation}
\ACF_{\hdelay_{Q}}(\tau)=\ACF_{\hdelay_{I}}(\tau)=\frac{\power}{2}J_{0}(2\pi\fDoppler\tau)\label{eq:ch3:Rh_Q}
\end{equation}
The independence between $h_{I}(t)$ and $h_{Q}(t)$ can be verified
by first evaluating the cross-correlation function, as follows {[}\citet{ch3:BK:Fading:Stuber11}{]}:

\begin{eqnarray*}
\ACF_{\hdelay_{I}\hdelay_{Q}}(\tau) & = & E[h_{I}(t)h_{Q}(t+\tau)]\\
 & = & \frac{\power}{2}E_{\phasepath_{\ipath}}[\sin\left(2\pi(\fDoppler\mbox{\ensuremath{\cos\phasepath_{\ipath}})}\tau\right)]\\
 & = & 0
\end{eqnarray*}
owing to the central limit theorem (CLT). Since $h_{I}(t)$ and $h_{Q}(t)$
are jointly Gaussian processes, the uncorrelation (orthogonal) result
$\ACF_{\hdelay_{I}\hdelay_{Q}}(\tau)=0$ is sufficient for indicating
the independence of $h_{I}(t)$ and $h_{Q}(t)$, i.e.: 
\begin{equation}
\ACF_{\hdelay}(\tau)=\ACF_{\hdelay_{I}}(\tau)+\ACF_{\hdelay_{Q}}(\tau)\label{eq:ch3:Rh}
\end{equation}
which yields (\ref{eq:ch3:ACF(t)}). In the literature, this model
(\ref{eq:ch3:ACF(t)}) is often called Clarke's model for fading channel,
which has been widely applied in practice {[}\citet{ch3:ART:FadingMarkov:tutorial08}{]}
since its first appearance in {[}\citet{ch2:origin:Fading:Clark68}{]}. 

Because fading gain $h(t)$ is a wide-sense stationary (WSS) process,
it is possible to construct its Wold representation by infinite-order
AR process, as explained in Section \ref{subsec:chap2:Stationary-process}.
A finite-order AR model does not, however, fit the above Clarke's
model, since PSD of fading gain, as computed from (\ref{eq:ch3:ACF(t)}),
is not a rational function of frequency {[}\citet{ch3:ART:FadingMarkov:tutorial08,ch3:BK:Fading:Stuber11}{]}.
Recently, a simulation of approximate Clarke's model via high order
AR model (with memory up to 1000) was considered in {[}\citet{ch3:art:Fading:AR05}{]}.
Despite small simulation error, this AR approach is prohibitive in
the fading channel {[}\citet{ch3:ART:FadingMarkov:tutorial08}{]}

\subsection{Rayleigh process in fading channel \label{subsec:chap3:Rayleigh-process}}

For complex stationary Gaussian process in polar form, i.e. $h(t)=\left\Vert h(t)\right\Vert e^{j\phi(t)}$,
it is also well known that the squared magnitude $\gamma\TRIANGLEQ\left\Vert h(t)\right\Vert ^{2}$
is a random variable with $\chi^{2}$- distribution of two degree
of freedoms, as follows {[}\citet{ch3:bk:Fading:cavers00,ch3:BK:Fading:Stuber11}{]}:

\begin{equation}
f\left(\gamma\right)=\chi_{\gamma}^{2}(2)=\frac{1}{P_{0}}\exp\left(-\frac{\gamma}{P_{0}}\right)\label{eq:ch3:Chi-squared}
\end{equation}
in which the mean of $\gamma$ is $E[\gamma]=E\left[\left\Vert h(t)\right\Vert ^{2}\right]=P_{0}$,
as in (\ref{eq:ch3:ACF(t)}). Similarly, because the square-root of
$\chi^{2}-$random variable in this case is a Rayleigh random variable,
the distribution of magnitude $\left\Vert h(t)\right\Vert =\sqrt{\gamma}$
can be written as follows {[}\citet{ch3:bk:Fading:cavers00,ch3:BK:Fading:Stuber11}{]}:

\begin{equation}
f\left(\left\Vert h(t)\right\Vert \right)=\frac{\left\Vert h(t)\right\Vert }{\sigma_{0}^{2}}\exp\left(-\frac{\left\Vert h(t)\right\Vert ^{2}}{2\sigma_{0}^{2}}\right)\label{eq:ch3:Rayleigh}
\end{equation}
in which the Rayleigh squared-average is $E\left[\left\Vert h(t)\right\Vert ^{2}\right]=2\sigma_{0}^{2}$.
From (\ref{eq:ch3:Chi-squared}), we have $\sigma_{0}^{2}=E[\gamma]/2=P_{0}/2$.
Note that, because $J_{0}(0)=1$, the power average $E\left[\left\Vert h(t)\right\Vert ^{2}\right]$
can also be computed by setting $\tau=0$ in autocorrelation function
in (\ref{eq:ch3:ACF(t)}), i.e. $\sigma_{0}^{2}=\ACF_{\hdelay_{I}}(0)=\power/2$. 

Owing to standard form (\ref{eq:ch3:Rayleigh}), the dominant approach
in fading channel estimation is to quantize the Rayleigh distribution
$f\left(\left\Vert h(t)\right\Vert \right)$ in (\ref{eq:ch3:Rayleigh})
into finite cells and approximate $f\left(\left\Vert h(t)\right\Vert \right)$
via probability mass function {[}\citet{ch3:ART:FadingMarkov:tutorial08}{]}.
For representing the correlated Fading process, a finite-state Markov
chain (FSMC) is widely exploited, as reviewed in Section \ref{subsec:chap2:Fading-channel}.
The time-invariant transition matrix of that FSMC can be designed
via quantization of jointly Rayleigh variables, as presented in Appendix
\ref{App:chap:Quantization}. 

\section{Summary}

For \foreignlanguage{british}{modelling} digital receivers, three
fundamental system models were presented in this chapter. 

In the first case, the synchronized scheme, the matched filter was
shown to be an orthonormal correlator and, hence, preserves the sufficient
statistics in the data in the case of AWGN channel.

In the second case, the un-synchronized scheme, three state-of-the-art
(non-Bayesian) techniques for frequency-offset estimation were reviewed. 

In the third case, the synchronized fading scheme, the derivation
of the Rayleigh fading process for the amplitude of the received signal
- a derivation based solely on the original Gaussian process assumption
- was also presented briefly, providing us with the insight that it
is actually the square-root of Chi-square process. 

These models will be used, later in Chapters \ref{=00005BChapter 7=00005D}
and \ref{=00005BChapter 8=00005D}, in order to evaluate the performance
of novel inference methods in this thesis.

%auto-ignore
%auto-ignore
%%%% Common

\global\long\def\xbold{\mathbf{x}}%

\global\long\def\btheta{\mathbf{\boldsymbol{\theta}}}%

\global\long\def\xdata{x}%

\global\long\def\vtheta{\theta}%

\global\long\def\htheta{\widehat{\theta}}%

\global\long\def\vxi{\xi}%

\global\long\def\vpsi{\psi}%

\global\long\def\veta{\eta}%

\global\long\def\REAL{\mathbb{R}}%

\global\long\def\calTheta{\Theta}%

\global\long\def\spaceO{\Theta}%

\global\long\def\calX{\mathcal{X}}%

\global\long\def\calA{\mathcal{A}}%

\global\long\def\calE{\mathcal{E}}%

\global\long\def\calF{\mathcal{F}}%

\global\long\def\calN{\mathcal{N}}%

\global\long\def\calI{\mathcal{I}}%

\global\long\def\ndata{n}%

\global\long\def\nstate{m}%

\global\long\def\itime{i}%

\global\long\def\istate{k}%

\global\long\def\funh#1{h\left(#1\right)}%

\global\long\def\fung#1{g\left(#1\right)}%

\global\long\def\seti#1#2{#1\in\{1,2,\ldots,#2\}}%

\global\long\def\setd#1#2{\{#1{}_{1},#1{}_{2},\ldots,#1_{#2}\}}%

\global\long\def\TRIANGLEQ{\triangleq}%

%%%% Chapter 4

\global\long\def\ftilde{\tilde{f}}%

\global\long\def\fdelta{f_{\delta}}%

\global\long\def\KLD{KLD}%

\global\long\def\kldff{KLD_{\ftilde||f}}%

\global\long\def\Loss{L}%

\global\long\def\funL#1#2{\Loss\left(#1,#2\right)}%

\global\long\def\etheta#1{\htheta\left(#1\right)}%

\global\long\def\vL{\mathrm{L}}%

\global\long\def\vH{\mathrm{H}}%

\global\long\def\vutility{u}%

\global\long\def\vaction{a}%

\global\long\def\vmoment{m}%

\global\long\def\vphi{\phidelay}%

\global\long\def\vlambda{\lambda}%

\global\long\def\vy{y}%

\global\long\def\iVB{\nu}%

\global\long\def\sufficient{\mathcal{\tau}}%

\global\long\def\Fisher#1{I\left(#1\right)}%

\global\long\def\joint{f(\xbold,\vtheta)}%

\global\long\def\obs{f(\xbold|\vtheta)}%

\global\long\def\prior{f(\vtheta)}%

\global\long\def\posterior{f(\vtheta|\xbold)}%

\global\long\def\funftilde#1{\ftilde(#1|\xbold)}%

\global\long\def\zbinary{\vtheta_{\itime},\vtheta_{\backslash\itime}}%

\global\long\def\calL{\mathcal{L}}%

\global\long\def\ELoss{\bar{\calL}}%

\global\long\def\mrtheta{\widehat{\vtheta}_{MR}}%

\global\long\def\ztheta{\vtheta_{0}}%

\global\long\def\zEntropy{\calE_{\prior}}%

\chapter{Bayesian parametric modelling\label{=00005BChapter 4=00005D}}

The purpose of this chapter is to show that Bayesian inference method
is an effective tool for system modelling in telecommunication contexts
of interest in this thesis. Since efficient computation is a major
concern in mobile receivers, tractable Bayesian methods are primary
concern in this chapter. Moreover, for the sake of clarity, the general
form of Bayesian inference will be considered here, without any constraint
on model design, while specific models of interest will be studied
in later parts of this thesis.

\section{Bayes' rule}

The aim of parametric inference is to infer some information about
unknown quantity $\vtheta\in\spaceO$, given observed data $\xbold$.
For that purpose, the probabilistic solution is to impose a joint
distribution $f(\xbold,\vtheta)$ on both $\xbold$ and $\vtheta$.
By probabilistic chain rule, $f(\xbold,\vtheta)$ can be factorized
into two equivalent ways, as follows:

\begin{eqnarray}
f(\vtheta|\xbold)f(\xbold) & = & f(\xbold,\vtheta)\nonumber \\
 & = & f(\xbold|\vtheta)f(\vtheta)\label{eq:ch4:JOINT}
\end{eqnarray}

Hence, any information that data $\xbold$ can provide us about $\vtheta$
must be contained within the posterior distribution $f(\vtheta|\xbold)$.
Because the value $\xbold$ is known, the data inference $f(\xbold)=\int f(\xbold,\vtheta)d\vtheta$,
also known as predictive inference or occasionally as the evidence
{[}\citet{ch2:BK:Bernardo:Bayes94}{]}, is regarded as normalizing
constant for posterior distribution $f(\vtheta|\xbold)$, as follows:

\begin{eqnarray}
f(\vtheta|\xbold) & = & \frac{f(\xbold|\vtheta)f(\vtheta)}{f(\xbold)}\nonumber \\
 & \propto & f(\xbold|\vtheta)f(\vtheta)\label{eq:ch4:BAYES's rule}
\end{eqnarray}
where $\propto$ denotes normalizing operator, i.e. the right hand
side of $\propto$ is normalized to be sum-to-one over $\vtheta$.
The well-known formula (\ref{eq:ch4:BAYES's rule}) is called Bayes'
rule in the literature. Important texts on Bayes' theory and calculus
are found in {[}\citet{ch2:BK:Bernardo:Bayes94,ch4:bk:Bayes:bible:Jaynes,ch4:BK:Bible:Robert07}{]}.

\section{Subjectivity versus Objectivity}

Despite simplicity, Bayes' rule raises a philosophical issue on subjectivity
of probability, which is perhaps the most critical issue in Bayesian
inference {[}\citet{ch4:BK:Bible:Robert07}{]}. The most popular criticism
is to regard the prior and posterior distributions objectively, i.e.
as a measure of a repeatable or generable quantity, rather than ``quantification
of belief'' about $\vtheta$ in Bayesian philosophy {[}\citet{ch2:BK:Bernardo:Bayes94}{]}. 

Putting Bayesian interpretation aside, it can be seen that both prior
and posterior models are simply consequences of joint model $f(\xbold,\vtheta)$
design (\ref{eq:ch4:JOINT}), which is subjective \textit{per se}.
Box's famous comment that ``all models are wrong, but some are useful''
{[}\citet{ch4:origin:Models:Box79a}{]} was also stated as ``models
are never true, but it is only necessary that they are useful'' {[}\citet{ch4:origin:Models:Box79b}{]}.
The usefulness is, therefore, a necessary and subjective criterion
for model design. Hence, in order to avoid the philosophical ambiguity
between objectivity and subjectivity of probability (see e.g. {[}\citet{ch4:art:Bayes:Lindley00}{]}
and discussions therein), in this thesis, the imposed joint model
$f(\xbold,\vtheta)$ is considered as subjective belief, while the
prior and posteriors are considered as deductive consequences of the
subjective modelling. The justification is the following:

- Regarding subjectivity: In deterministic model, the purpose of parametric
inference is to return the optimal estimate $\htheta$ such that a
criterion like loss function $\calL\TRIANGLEQ\Loss(\xbold,\vtheta)$
can be minimized at $\htheta$. In probability context, the distribution
$f(\calL)$ is, therefore, a transformation of $f(\xbold,\vtheta)$.
Since joint model $\joint$, regardless of repeatability or unrepeatability,
is fundamentally imposed by our belief on the system, it is obviously
subjective and need to be useful under criterion $\Loss(\xbold,\vtheta)$. 

- Regarding objectivity: Note that chain rule of factorization in
(\ref{eq:ch4:JOINT}) follows the Axioms of Probability {[}\citet{ch2:BK:Bernardo:Bayes94}{]},
i.e. the computation of prior and posterior would never change the
originally set-up joint model $f(\xbold,\vtheta)$. Hence, the term
``Bayesian inference'' in this thesis simply refers to computation
of posterior distribution in (\ref{eq:ch4:BAYES's rule}), rather
than its Bayesian subjectivity meaning. Bayesian inference is, at
least in this thesis, a mathematical technique to solve the above
subjective problem. 

\section{Bayesian estimation as a Decision-Theoretic task}

Owing to Bayes' rule (\ref{eq:ch4:BAYES's rule}), the data information
about $\vtheta$ is all contained in posterior distribution $\posterior$.
In practice, however, the desired output is often a single point estimation
or decision $\vxi\TRIANGLEQ\etheta{\xbold}:\calX\rightarrow\spaceO$
of $\vtheta\in\spaceO$. In this section, Bayesian criteria for optimizing
value $\vxi$ will be reviewed.

\subsection{Utility and loss function \label{subsec:chap4:Utility-and-loss}}

The utility and loss definition was often expressed in different forms
in different fields. In this subsection, let us review some of these
forms, before applying the approach to Bayesian estimator in the consequence.

In decision theory, a chosen action $\vaction(\vtheta):\spaceO\rightarrow\calA$,
defined as a function of event $\vtheta$, is justified by a gain
or benefit of that action. The measure of that benefit is called utility,
denoted $\vutility(\vaction(\vtheta))$, whose axiomatization is firstly
presented in {[}\citet{ch4:origin:unility:vonNeumann44}{]}. 

In statistical decision, however, the action $\vaction$ is typically
chosen via a loss, i.e. negative utility, firstly formalized by Wald
{[}\citet{ch4:origin:loss:Wald49}{]}:

\begin{equation}
\Loss_{\vutility}(\vaction,\vtheta)=-\vutility(\vaction(\vtheta))\label{eq:ch4:loss-utility}
\end{equation}
where $\Loss_{\vutility}(\vaction,\vtheta):\calA\times\spaceO\rightarrow\REAL$.
The key role of definition (\ref{eq:ch4:loss-utility}) is that the
action $\vaction$ is now regarded as independent of parameter event
$\vtheta$ in $L_{\vutility}(\vaction,\vtheta)$. Moreover, in order
to guarantee a non-negative loss action, the loss definition in (\ref{eq:ch4:loss-utility})
is further constrained into a (regret) loss function, as follows {[}\citet{ch4:BK:DecisionTheory:09}{]}: 

\begin{equation}
\funL{\vaction}{\vtheta}\TRIANGLEQ\Loss_{\vutility}(\vaction,\vtheta)-\inf_{\vaction\in\calA}\Loss_{\vutility}(\vaction,\vtheta)\label{eq:ch4:regret_loss}
\end{equation}
where $\funL{\vaction}{\vtheta}:\calA\times\spaceO\rightarrow\mathbb{R}_{\geq0}$. 

In parametric inference, the action $a$ is to pick a value $\vxi$
in parameter space $\calTheta$, i.e. $a=\vxi\in\calA=\calTheta$.
By this way, the definition (\ref{eq:ch4:regret_loss}) becomes a
loss function $\calL\TRIANGLEQ\Loss(\vxi,\vtheta)$ for estimator
$\xi$ of $\vtheta$ such that $\Loss(\vxi,\vtheta):\Omega\times\Omega\rightarrow\mathbb{R}_{\geq0}$
and $\Loss(\vxi,\vtheta)=0\Leftrightarrow\vxi=\vtheta$. In deterministic
scenario, the standard criterion is to pick the value $\widehat{\vxi}$
of $\vxi$ such that the loss function is minimized, i.e. we have:
\[
\widehat{\vxi}=\arg\min_{\vxi\in\calTheta}\Loss(\vxi,\vtheta)
\]
. 

Lastly, in probability context, the estimator $\vxi$ is assigned
as a function of data $\xi\TRIANGLEQ\etheta{\xbold}:\calX\rightarrow\Omega$
in a functional space $\xi\in\Xi$ . In this context, the value $\calL=\Loss(\vxi,\vtheta)=\funL{\etheta{\xbold}}{\vtheta}$
is a deterministic function of two random variables $\{\xbold,\vtheta\}$.
Then, the aim of parametric inference in this case is to pick the
minimum risk (MR) function $\widehat{\vxi}=\mrtheta(\xbold)$ in functional
space $\Xi$ such that the expected loss function is minimized, i.e.
we have: 
\begin{equation}
\mrtheta(\xbold)=\arg\min_{\etheta{\cdot}\in\Xi}E_{f(\calL)}\funL{\etheta{\xbold}}{\vtheta}\label{eq:ch4:Minimum_risk}
\end{equation}
In general, a loss function $\calL$ can be designed via definition
of gain function (\ref{eq:ch4:loss-utility}), via definition of regret
function (\ref{eq:ch4:regret_loss}), by system requirement or simply
by imposing tractably mathematic form. Then, theoretically, the distribution
$f(\calL)$ in (\ref{eq:ch4:Minimum_risk}), $\calL\in\mathbb{R}_{\geq0}$,
can be derived via transformation of joint distribution $f(\xbold,\vtheta)$.
However, in practice, the exact form of $f(\calL)$ is often difficult
to compute. The most common solution is to confine ourself to expected
loss function $\ELoss\TRIANGLEQ E_{f(\calL)}(\calL)$. For this reason,
a theory of functional mean will be briefly reviewed first, before
we derive the computation of that mean value $\ELoss$.

\subsubsection{Expectation of a function}

In probability theory, law of the unconscious statistician (LOTUS)
is important {[}\citet{ch4:bk:LOTUS:proof:2005}{]}. Historically,
the term ``unconscious'' was used because some people forgot that
this law was \textit{not} a definition {[}\citet{ch4:origin:LOTUS:Ross70,ch4:bk:LOTUS:Schervish95}{]},
although some statisticians, e.g. {[}\citet{ch4:bk:LOTUS:not_amusing:2002}{]},
did not find that term amusing. The virtue of LOTUS is that we can
compute the expected value of a deterministic function $\vpsi=\fung{\vtheta}$
from original distribution $f(\vtheta)$, without the need of computing
transformed distribution $f(\psi)$, which might be difficult to carry
out in practice.
\begin{prop}
\label{prop:ch4:LOTUS} (Law of unconscious statistician (LOTUS))
(see e.g. {[}\citet{ch4:bk:LOTUS:proof:1985,ch4:bk:LOTUS:proof:2005}{]}
for rigorous proof)

If $\theta$ is a random variable with probability function $f(\vtheta)$
and $\vpsi=\fung{\vtheta}:\spaceO\rightarrow\REAL$ is a measurable
function, then $E_{f(\psi)}(\vpsi)=E_{f(\vtheta)}(\fung{\vtheta})=\int\fung{\vtheta}f(\vtheta)d\vtheta$.
\end{prop}

Somewhat relevant to LOTUS is the concept of certainty equivalent
(CE) in decision theory, defined as follows:
\begin{defn}
(Certainty Equivalent) {[}\citet{ch4:BK:DecisionTheory:09}{]} The
certainty equivalent $\htheta_{CE}\in\spaceO$, if existent, is a
special value of $\vtheta\in\spaceO$, such that:

\begin{equation}
\fung{\htheta_{CE}}=E_{f(\vtheta)}\fung{\vtheta}\label{eq:ch4:CE}
\end{equation}
\end{defn}

The equation (\ref{eq:ch4:CE}) means that expectation of functional
form can be evaluated by a single CE point $\htheta_{CE}$, if that
CE exists. Note that, the sufficient condition for existence of $\htheta_{CE}$
is that $\fung{\vtheta}$ is linear with $\vtheta$ (\ref{eq:ch4:CE}).

\subsection{Bayes risk \label{subsec:chap4:Bayes-risk}}

By Proposition \ref{prop:ch4:LOTUS}, the expected value $\ELoss$
in (\ref{eq:ch4:Minimum_risk}) can be found via the joint distribution
$f(\xbold,\vtheta)$, as follows:

\begin{equation}
\ELoss\TRIANGLEQ E_{f(\vxi,\vtheta)}\Loss(\vxi,\vtheta)=E_{f(\xbold,\vtheta)}\Loss(\etheta{\xbold},\vtheta)\label{eq:ch4:Bayesian_risk}
\end{equation}
In the literature, the expected loss in (\ref{eq:ch4:Bayesian_risk})
is also called Bayes risk {[}\citet{ch4:bk:Bayes:bible:Berger85}{]}.
In practice, because the computation of expectation (\ref{eq:ch4:Bayesian_risk})
via joint distribution $f(\xbold,\vtheta)$ is often not in closed
form, the Bayesian risk  (\ref{eq:ch4:Bayesian_risk}) can be estimated
via empirical (Monte Carlo) sampling of $f(\xbold,\vtheta)$. 

\subsubsection{Posterior expected loss \label{subsec:chap4:Posterior-expected-loss}}

The MR estimator $\mrtheta\TRIANGLEQ\mrtheta(\xbold)$ in (\ref{eq:ch4:Minimum_risk})
can also be found via posterior expected loss function, without the
need of computing the form $\joint$. By factorization $\joint=f(\xbold)\posterior$
in (\ref{eq:ch4:JOINT}), the Bayesian risk (\ref{eq:ch4:Bayesian_risk})
can also be computed by averaging posterior distribution, i.e.: $\ELoss=E_{f(\xbold)}\left\{ E_{f(\vtheta|\xbold)}\Loss(\etheta{\xbold},\vtheta)\right\} $
. Since $f(\xbold)\geq0$ for any $\xbold$, we have an equivalent
way to find the optimal estimator $\mrtheta$ in (\ref{eq:ch4:Minimum_risk}),
as follows:

\begin{equation}
\mrtheta=\arg\min_{\etheta{\cdot}}E_{f(\vtheta|\xbold)}\funL{\etheta{\xbold}}{\vtheta}\label{eq:ch4:Bayesian_Estimator}
\end{equation}

Hence, an advantage of Bayesian estimation method is that deriving
optimal estimator $\mrtheta$ via posterior distribution $\posterior$
is often much more feasible than via joint distribution $\joint$
{[}\citet{ch4:bk:Bayes:bible:Berger85}{]}. 

\subsubsection{Minimum risk estimators \label{subsec:chap4:Minimum-risk-estimators}}

From (\ref{eq:ch4:Bayesian_Estimator}), it is feasible to derive
the optimal estimators for several well-known loss functions. For
example, if $\Loss(\htheta,\vtheta)$ is quadratic loss $\left\Vert \htheta-\vtheta\right\Vert _{2}^{2}$,
zero-one loss $\delta(\htheta-\vtheta)$ or scalar absolute loss $\left\Vert \htheta-\vtheta\right\Vert _{1}$,
the minimum risk estimator $\mrtheta$ is the mean, mode or median
of posterior distribution $\posterior$, respectively {[}\citet{ch4:bk:Bayes:bible:Berger85}{]}.

In information theory, the Hamming distance is an important function.
Generally, a Hamming loss can be defined for continuous case, as follows:

\begin{equation}
\funL{\htheta}{\vtheta}\equiv Q(\htheta,\vtheta)=1-\frac{1}{\ndata}\sum_{\itime=1}^{n}\delta\left(\htheta_{\itime}-\vtheta_{\itime}\right)\label{eq:ch4:MR:Hamming}
\end{equation}
where $\widehat{\vtheta}=\setd{\htheta}{\ndata}$ is the set of estimates
and $\vtheta=\setd{\vtheta}{\ndata}$ is the set of parameters. The
Hamming loss in (\ref{eq:ch4:MR:Hamming}) can be minimized via the
following Lemma:
\begin{lem}
\label{lem:chap4:MR-Hamming} The minimum risk (MR) estimate $\widehat{\vtheta}_{MR}=\{\htheta_{1}(\xbold),\ldots,\htheta_{\ndata}(\xbold)\}$,
which minimizes $E_{f(\xbold,\vtheta)}Q(\htheta,\vtheta)$, is the
sequence of marginal MAP:
\end{lem}

\begin{equation}
\htheta_{\itime}(\xbold)=\arg\max_{\vtheta_{\itime}}f(\vtheta_{\itime}|\xbold),\ \seti{\itime}{\ndata}\label{eq:ch4:MR:marginalMAP}
\end{equation}

\begin{proof}
From (\ref{eq:ch4:MR:Hamming}), we can see that the MR estimate for
Hamming loss is:

\begin{eqnarray}
\min_{\widehat{\vtheta}(\cdot)}E_{f(\vtheta|\xbold_{n})}\funL{\widehat{\vtheta}(\xbold)}{\vtheta} & = & 1-\frac{1}{\ndata}\sum_{\itime=1}^{n}\max_{\widehat{\vtheta}(\cdot)}E_{f(\vtheta|\xbold_{n})}\delta\left(\htheta_{\itime}(\xbold)-\vtheta_{\itime}\right)\nonumber \\
 & = & 1-\frac{1}{\ndata}\sum_{\itime=1}^{\ndata}\max_{\vtheta_{\itime}}f(\vtheta_{\itime}|\xbold)\label{eq:ch4:MR:minHamming}
\end{eqnarray}
which yields (\ref{eq:ch4:MR:marginalMAP}).
\end{proof}
In the literature, a special case of Lemma \ref{lem:chap4:MR-Hamming},
in which $\vtheta$ is discrete and Dirac-$\delta(\cdot)$ is replaced
by Kronecker-$\delta[\cdot]$, is proved in {[}\citet{ch4:bk:lossHamming:image95,ch4:art:JuriLember11}{]}.
The above proof is provided in this thesis in order to cover the case
of continuous r.v. $\vtheta$. 

\section{Bayesian inference}

As explained in subsection \ref{subsec:chap4:Posterior-expected-loss},
although the ultimate aim of estimation task is to minimize the Bayesian
risk $\ELoss$ (\ref{eq:ch4:Bayesian_risk}) via joint model $f(\xbold,\vtheta)$,
the Minimum Risk (MR) estimated decision (\ref{eq:ch4:Bayesian_Estimator})
can be found equivalently via posterior distribution $\posterior$.
The tractable computation of posterior is, therefore, the main interest
in this thesis. 

Because the joint model $f(\xbold,\vtheta)$ depends on the design
of observation part $f(\xbold|\vtheta)$ and the prior part $f(\vtheta)$,
the technical issues with those two parts will be reviewed first in
this section. The role of posterior part $f(\vtheta|\xbold)$, which
is a mere consequence of the chain rule (\ref{eq:ch4:JOINT}), will
then be reviewed.

\subsection{Observation model}

In contrast to the deterministic approach, the probabilistic approach
considers the observed data $\xbold$ as \textit{one} realization
of observation distribution $f(\xbold|\vtheta)$, given a shaping
parameter $\vtheta$. In other words, the parametric model $f(\cdot|\vtheta)$
is a quantization model of the observer's/\foreignlanguage{british}{modeller's}
belief about the possible realization $\xbold$, of $\mathbf{X}$,
when observed. 

For clarifying potential confusion, let us emphasize again that, in
this thesis, the data $\xbold$ is regarded as one and only one realization,
drawn from $f(\xbold|\vtheta)$. Owing to this convention, it does
not matter whether the random quantity $\xbold$ is repeatable or
unrepeatable. In stochastic case where there are $\ndata$ observed
data, the notation $\xbold$ will be specialized to $\xbold_{\ndata}\TRIANGLEQ\setd{\xdata}{\ndata}$. 

In practice, the observation model $f(\xbold|\vtheta)$ is often imposed
by physical laws. In mathematical models, $f(\xbold|\vtheta)$ can
be flexibly parameterized by exploiting exchangeability, invariance
or sufficiency properties of data $\xbold$ {[}\citet{ch2:BK:Bernardo:Bayes94}{]}.
In that theoretical context, $f(\xbold|\vtheta)$ often belongs to
standard distributions, which are derived from experiments or defined
in probability textbooks {[}\citet{ch4:bk:Johnson_Kotz:DisMul97,ch4:bk:Johnson_Kotz:ContMul04,ch4:bk:Johnson_Kotz:ContUni04,ch4:bk:Johnson_Kotz:DisMul05}{]}.

For computational efficiency, the data sufficiency is the most important
property for us to exploit. If a statistic, i.e. a function of data,
extracts information on its parameter partially, a sufficient statistic
is much more efficient since it can extract that information fully.
Furthermore, sufficient statistics can represent the whole data in
a parametric model and, hence, might reduce the data complexity significantly.
For that reason, the parameterization technique via sufficient statistics
and its typical class, namely Exponential Family, is the key in this
thesis and will be briefly reviewed in this chapter.

\subsubsection{Sufficient statistics}

The sufficient statistics of an observation model $f(\xbold|\vtheta)$
can be identified via a well-known criterion, as follows:
\begin{prop}
\textbf{(Fisher-Neyman factorization criterion)} The statistics $\sufficient(\xbold)$
is called sufficient if and only if the observation distribution can
be factorized as {[}\citet{ch2:BK:Bernardo:Bayes94}{]}:

\begin{equation}
f(\xbold|\vtheta)=\funh{\sufficient(\xbold),\vtheta}\fung{\xbold}\label{eq:ch4:Neyman}
\end{equation}

for some functions $\funh{\cdot}\geq0$ and $\fung{\cdot}>0$. 
\end{prop}

Because the parameter $\vtheta$ only interacts with data $\xbold$
via function $\funh{\sufficient(\xbold),\vtheta}$ in (\ref{eq:ch4:Neyman}),
all the information of data $\xbold$ regarding $\vtheta$ is summarized
in $\sufficient(\xbold)$, hence its name sufficient statistic. 

In history, the notion of sufficient statistics was firstly defined
in {[}\citet{ch2:origin:Fisher_info:Fisher22}{]}, while the factorization
(\ref{eq:ch4:Neyman}) is explicitly established in {[}\citet{ch4:origin:NeymanFactor:Neyman35}{]}.
In classical inference, the sufficient statistics plays an important
role, mostly owing to Rao-Blackwell-Kolmogorov theorem {[}\citet{ch4:origin:Rao_Blackwell:Blackwell47,ch4:origin:Rao_Blackwell:Kolmogorov50,ch4:origin:Rao_Blackwell:Rao65}{]},
which establishes that unbiased estimators based on sufficient statistic
are the best estimators. In Bayesian inference, however, sufficient
statistics are simply regarded as a consequence of the Bayesian method
{[}\citet{ch2:BK:Bernardo:Bayes94}{]}. Owing to Bayes' rule (\ref{eq:ch4:BAYES's rule})
and Neyman factorization (\ref{eq:ch4:Neyman}), the posterior inference
normalizes out any data factor $\fung{\xbold}$ irrelevant to parameter
and, hence, always exploits the minimal sufficient statistics {[}\citet{ch4:bk:PointEst:Lehmann98}{]}. 

Nevertheless, the sufficiency principle plays a central role for data
simplification of two major observation classes: transformation family
(TF) and exponential family (EF). Indeed, most standard distributions,
see e.g. {[}\citet{ch4:bk:Johnson_Kotz:DisMul97,ch4:bk:Johnson_Kotz:ContMul04,ch4:bk:Johnson_Kotz:ContUni04,ch4:bk:Johnson_Kotz:DisMul05}{]},
belong to just these two classes {[}\citet{ch4:bk:PointEst:Lehmann98}{]}.
The former extracts sufficient statistics of interesting parameters
via a group of transformations, such as scaling or location shift,
while, on the other hand, preserving the original distribution's structure
{[}\citet{ch4:bk:SufficientStat:Cox06}{]}. The latter reduces data
complexity, regardless of sample size, to a fixed (usually small)
number of sufficient statistics without loss of information {[}\citet{ch4:bk:PointEst:Lehmann98}{]}.
Since EF class is widely exploited in signal processing for tractable
computation, it will be reviewed in this chapter. Some distributions
in TF class, namely spherical distributions, that also belong to EF
class will be presented as an application in Section \ref{sec:chap7:Spherical-family}.

\subsubsection{Exponential Family}

In the class of distributions parameterized via sufficient statistics,
the most important is the exponential family (EF), firstly pioneered
by {[}\citet{ch4:origin:EF:Darmois35,ch4:origin:EF:Koopman36,ch4:origin:EF:Pitman36}{]}.
The first motivation of the EF class is to exploit the fixed-dimension
property of sufficient statistics, as stated via Darmois-Koopman-Pitman
theorem {[}\citet{ch4:art:EF:Andersen70}{]}: ``Under regularity
conditions, a necessary and sufficient condition for the existence
of a sufficient statistic of fixed dimension is that the probability
density belongs to the exponential family''. Following up that result,
the sufficient data reduction in the EF class was then studied thoroughly
in {[}\citet{ch4:art:EF:Andersen70,ch4:art:EF:Brown86}{]}. The main
motivation for EF's usage nowadays is simply the computational tractability
engendered in the posterior distribution {[}\citet{ch4:BK:Bible:Robert07}{]}.
\begin{defn}
\textbf{(Exponential Family) }The observation model $f(\xbold|\vtheta)$
is a member of EF if and only if there is a separability between parameters
and data kernels, as follows {[}\citet{ch4:art:EF:Brown86}{]}: 

\begin{equation}
f(\xbold|\vtheta)=C(\vtheta)\funh{\xbold}\exp\left(R(\vtheta)\sufficient(\xbold)\right)\label{eq:ch4:EF:obs}
\end{equation}
In a more relaxed form, the scalar product $R(\vtheta)\sufficient(\xbold)$
in (\ref{eq:ch4:EF:obs}) might be replaced by a scalar product $\left\langle R(\vtheta),\sufficient(\xbold)\right\rangle $
that is linear in the second argument (such as the Euclidean inner
product in the case where $R(\vtheta)$ and $\sufficient(\xbold)$
are vector structures of equal dimension) {[}\citet{ch4:BK:AQUINN_06,ch4:art:EF:Nielsen09}{]}. 
\end{defn}

Comparing (\ref{eq:ch4:Neyman}) with (\ref{eq:ch4:EF:obs}), we recognize
that the data kernel $\sufficient(\xbold)$ is a sufficient statistic
in EF. Moreover, $\sufficient(\xbold)$ is also invariant with increasing
numbers of observation. For example, given an iid sequence $\xbold_{\ndata}=\setd{\xdata}{\ndata}$,
the EF observation (\ref{eq:ch4:EF:obs}) becomes $f(\xbold_{\ndata}|\vtheta)=\prod_{\itime=1}^{\ndata}f(\xdata_{\itime}|\vtheta)=C(\vtheta)^{\ndata}\funh{\xbold_{\ndata}}\exp\left(R(\vtheta)\sufficient(\xbold_{\ndata})\right)$,
where the sufficient statistic $\sufficient(\xbold_{\ndata})=\sum_{\itime=1}^{\ndata}\sufficient(\xdata_{\itime})$
preserves the dimension of initial statistics $\sufficient(\xdata_{1})$. 

Note that, if we regard the empirical mean $\frac{1}{\ndata}\sufficient(\xbold_{\ndata})$
as a moment constrain of iid sequence $\xbold_{\ndata}=\setd{\xdata}{\ndata}$,
the EF form (\ref{eq:ch4:EF:obs}) can also be found via maximum entropy
(MaxEnt) principle {[}\citet{ch4:art:Jordan:VB08}{]}.

\subsection{Prior distribution}

As explained above, the prior design $\prior$ cannot be separated
from the design of joint model $\joint$, which, in turn, depends
on the data characteristics. Hence, the goodness of estimation does
not only rely on inference techniques, but also on the quality of
model design. For any optimal decision, the first and foremost question
is whether we have considered all possible options, since too narrow
a set of options might lead us to a sub-optimal solution at best {[}\citet{ch4:origin:Models:Box79b}{]}.
The aim of prior design is, therefore, to embrace all possibilities
of parameter for the data set in the joint model.

In practice, there are three scenarios for prior design: 

- If we have no information on $\vtheta$ a priori, a non-informative
approach will be applied to prior design (see e.g {[}\citet{ch4:art:prior:Kass96}{]}
for full review and bibliography of this approach). The most well-known
priors in this case are uniform prior (also known as Laplace's prior)
{[}\citet{ch4:origin:prior:Laplace1814}{]}, Jeffreys' prior {[}\citet{ch4:origin:prior:Jeffreys46}{]}
and reference prior {[}\citet{ch4:origin:prior:Bernardo79}{]}, as
reviewed below.

- If we have all information on $\vtheta$, i.e. the prior distribution
is already given along with joint model, there is nothing for us to
do. However, if only the form of prior is defined, the prior design
becomes a tuning problem on shaping parameters of that form. An example
of this case is a conjugate prior, as explained below. Another example
is the multinomial distribution, which is the uniquely available form
for any discrete compact-support random variable. This multinomial
prior will be considered in Section \ref{sec:chap6:HMC}.

- If we have access to partial information on $\vtheta$, such as
moments or some constraints, a distributional optimizer can be sought.
For this case, some approximation methods, e.g. MaxEnt or moments
matching, can be applied {[}\citet{ch4:BK:Bible:Robert07}{]}. These
approximations will be explained in Section \ref{subsec:chap4:MaxEnt-Prior}.

In this subsection, the typical priors for those three cases will
be briefly reviewed. Laplace's prior, which is ignorant to data characteristics,
will be presented first in order to recall the ignorance principle
in prior design.

\subsubsection{Uniform prior}

The earliest principle, dated back to {[}\citet{ch4:origin:prior:Laplace1814}{]},
in prior design is the ``principle of indifference'' (also called
``principle of insufficient reason'' or Laplace's rule {[}\citet{ch4:art:prior:Kass96}{]}),
which imposes a uniform prior. This principle, however, receives some
serious criticism. Firstly, the uniform distribution is improper for
non-compact support $\vtheta$. Secondly and more fundamentally, the
principle of indifference ignores the re-parameterization issue in
observation model, which contains all the information of data about
parameter. Without taking that issue into account, the non-informative
prior $f(\vtheta)$ for \textit{a posteriori} estimation on $\vtheta$
might become an informative prior $f(\psi)$ for \textit{a posteriori}
estimation on $\vpsi$, where $\vpsi=\fung{\vtheta}$ is a one-to-one
mapping. 

\subsubsection{Jeffreys' prior}

In order to preserve the non-informative property in the re-parmeterization
issue, a criterion, namely ``invariance under re-parameterization'',
was originally required in Jeffreys' prior {[}\citet{ch4:origin:prior:Jeffreys46}{]}. 

Given observation model $f(\xbold|\vtheta)$, the Fisher information
is defined as follows: $\Fisher{\vtheta}\TRIANGLEQ E_{\obs}$$\left(\frac{\partial\log\obs}{\partial\vtheta}\right)^{2}$.
Owing to Jacobian transformation, the Fisher information is actually
invariant under re-parameterization, as follows: $\Fisher{\vtheta}=\Fisher{\vpsi}\left(\partial\vpsi/\partial\vtheta\right)^{2}$,
where $\vpsi=\fung{\vtheta}$ is a one-to-one mapping. Because square-root
of $\Fisher{\vtheta}$ yields a distributional transformation, the
Jeffreys prior is defined as $\prior\propto\sqrt{\Fisher{\vtheta}}$.
In the case of multi-dimensional parameter, the Jeffrey prior becomes
$\prior\propto\sqrt{\det\{\Fisher{\vtheta}\}}$, where $\Fisher{\vtheta}$
is the Fisher information matrix. 

Nevertheless, Jeffreys' prior is often improper, particularly in the
multi-dimensional case. For this reason, Jeffreys' prior is considered
as intuitive proposal, rather than a practical approach {[}\citet{ch2:BK:Bernardo:Bayes94}{]}.

\subsubsection{Reference prior}

A Bayesian, and somewhat objective, approach for prior design is to
consider the relationship between posterior and prior distributions,
given a fixed observation model. The criterion for optimal estimation
is, as explained above, to maximize the range of possibilities that
the prior can contribute to the posterior distribution. The reference
prior, firstly proposed in {[}\citet{ch4:origin:prior:Bernardo79}{]},
solved this problem via a variational approach, as follows: 
\[
\prior=\arg\max_{\prior}\left\{ E_{f(\xbold)}KLD\left(\posterior||\prior\right)\right\} 
\]
where KLD denotes Kullback-Leibler divergence {[}\citet{ch2:BK:CoverAndThomas}{]}:

\begin{equation}
KLD\left(\ftilde(\vtheta)||f(\vtheta)\right)\TRIANGLEQ E_{\ftilde(\vtheta)}\log\frac{\ftilde(\vtheta)}{f(\vtheta)}\geq0\label{eq:chap4:DEF=00003DKLD}
\end{equation}
If $\vtheta$ is scalar and continuous, the reference prior is identical
to Jeffreys' prior {[}\citet{ch2:BK:Bernardo:Bayes94}{]}, but this
is not true in the multi-dimensional case.

\subsubsection{Conjugate prior \label{subsec:chap4:Conjugate-prior}}

Another method for prior design is conjugate principle, which preserves
the prior and posterior within the same functional class, as follows:
\begin{defn}
(Conjugacy) {[}\citet{ch4:BK:Bible:Robert07}{]} A prior distribution
$f(\vtheta|\eta)\in\calF$ in distributional class $\calF$ is called
conjugate to an observation $\obs$ if its posterior distribution
also belongs to $\calF$, i.e. the distributional form is closed under
Bayes' rule (\ref{eq:ch4:BAYES's rule}). 
\end{defn}

Owing to conjugacy, the data update for posterior parameter can be
computed directly within data space itself, while the distributional
form stays unchanged. 

In particular, this invariance property under Bayes' rule plays a
central role in tractable and efficient computation for the EF class.
Indeed, because the EF observation model (\ref{eq:ch4:EF:obs}) preserves
data dimension, the dimension of its conjugate prior's parameters,
which are defined within the same data space, is also preserved \textit{a
posteriori}. 
\begin{defn}
(CEF class) The conjugate prior for EF observation model, which we
call the CEF distribution, can be defined as follows:
\end{defn}

\begin{equation}
f(\vtheta|\veta_{0})=C(\vtheta)^{\nu_{0}}\exp\left(R(\vtheta)v_{0}\right)\label{eq:ch4:EF:prior}
\end{equation}
where $\eta_{0}=\{\nu_{0},v_{0}\}$ is the shaping parameter. 

Obviously, the initialization of $\eta_{0}$ makes conjugate prior
somewhat informative, although special values of $\eta_{0}$ can make
conjugate EF prior identical to the non-informative Jeffrey prior
in some standard distributions.

\subsubsection{MaxEnt Prior \label{subsec:chap4:MaxEnt-Prior}}

Let us assume that, a priori, there is a set of mean constraints on
the parameter, as follows: 
\begin{equation}
f(\vtheta)\in\calF_{\vmoment}:\ E_{\prior}g_{\itime}(\vtheta)=\vmoment_{\itime}\label{eq:ch4:MaxEnt:mean_constraint}
\end{equation}
where all $\vmoment_{\itime}$ are known and $\calF_{\vmoment}$ is
the set of constrained distributions. The Maximum Entropy (MaxEnt)
principle, implied by {[}\citet{ch4:origin:MaxEnt:Jaynes80,ch4:origin:MaxEnt:Jaynes83}{]},
chooses the prior $\ftilde(\vtheta)\in\calF_{\vmoment}$ whose entropy
$\zEntropy$ is maximized, i.e. $\ftilde(\vtheta)=\arg\max_{\prior\in\calF_{\vmoment}}\zEntropy$
(once again, a variational principle). 

For defining entropy, however, there are two distinct cases. In the
discrete case, the entropy is traditionally defined as: $\zEntropy=-\sum_{\vtheta_{\istate}\in\spaceO}f(\vtheta_{\istate})\log f(\vtheta_{\istate})$.
In the continuous case, the relative Entropy is preferred, i.e. $\zEntropy=-\KLD(f(\vtheta)||f_{0}(\vtheta))$,
where $f_{0}(\vtheta)$ is a reference distribution. In practice,
$f_{0}(\vtheta)$ might be designed as  reference prior above. The
differential entropy, i.e. integration form of discrete Entropy $\zEntropy=-\int_{\vtheta}f(\vtheta)\log f(\vtheta)d\vtheta$,
is not always applicable in continuous case since it is sometimes
negative. 

The MaxEnt solution $\ftilde(\vtheta)$ for discrete and continuous
cases are $\ftilde(\vtheta)\propto\exp\left(-\sum_{\itime}g_{\itime}(\vtheta)\lambda_{\itime}\right)$
and $\ftilde(\vtheta)\propto\exp\left(-\sum_{\itime}g_{\itime}(\vtheta)\lambda_{\itime}\right)f_{0}(\vtheta)$,
respectively, where $\lambda_{\itime}$ are Lagrange multipliers of
the mean constraints (\ref{eq:ch4:MaxEnt:mean_constraint}). With
those forms, we can recognize that MaxEnt prior $\ftilde(\vtheta)$
is also a member of CEF class (\ref{eq:ch4:EF:prior}).

\subsection{Posterior distribution}

Given both observation and prior above, the joint model $f(\xbold,\vtheta)$
is already properly defined, and, hence, the computation of posterior
distribution in (\ref{eq:ch4:BAYES's rule}) is straight forward.
In this subsection, the main advantages of posterior distribution
as an inference object will be briefly reviewed and compared with
other inference techniques.

\subsubsection{Predictive inference}

Firstly, let us recall that, the predictive model $f(\xbold)$ on
observable $\xbold$ can be represented by marginalization over all
possible values of its parameter, as follows:

\[
f(\xbold)=\int\obs\prior d\vtheta
\]

Similarly, the posterior predictive model $f(\vy|\xbold)$ on observable
$\vy$, given data $\xbold$, can be represented via marginalization
over posterior distribution, as follows:

\begin{equation}
f(\vy|\xbold)=\int f(\vy|\xbold,\vtheta)\posterior d\vtheta\label{eq:ch4:predict:posterior}
\end{equation}

From (\ref{eq:ch4:predict:posterior}), we can see that posterior
$\posterior$ can be used as intermediate step to derive the inference
$f(y|\xbold)$ for unknown quantity $\vy$. This simple, yet elegant,
form of Bayesian prediction (\ref{eq:ch4:predict:posterior}) has
been used extensively in density estimation {[}\citet{ch4:bk:predict80}{]},
data classification {[}\citet{ch4:art:predict:classification92,ch4:art:predict:classification92b}{]},
model checking {[}\citet{ch4:art:predict:modelcheck96,ch4:bk:Gelman03}{]},
model averaging {[}\citet{ch4:art:predict:model95}{]}, etc. 

Note that, in the frequentist approach, because prior part $\prior$
is missing, the prediction (\ref{eq:ch4:predict:posterior}) has to
rely on plug-in approximation $\ftilde(\vtheta|\xbold)=\delta(\vtheta-\htheta)$,
sometimes referred to as a CE approximation: 
\begin{eqnarray}
\ftilde(\vy|\xbold) & = & \int f(\vy|\xbold,\vtheta)\delta(\vtheta-\htheta)d\vtheta\nonumber \\
 & = & f(\vy|\xbold,\htheta)\label{eq:ch4:predict:CE}
\end{eqnarray}

where $\htheta$ is often chosen as the Maximum Likelihood (ML) estimate,
i.e. $\htheta=\arg\max_{\vtheta\in\spaceO}\obs$ (see e.g. {[}\citet{ch4:bk:predict80}{]}
for the details). In the model-selection problem, such a substitution
is also popular, i.e. parameter model is often chosen first, before
the prediction step is carried out, although this approach is often
criticized for neglecting model uncertainty (see e.g. {[}\citet{ch4:art:predict:model95}{]}
and discussion therein).

\subsubsection{Hierarchical and nuisance parameters \label{subsec:ch4:Hierarchical-and-nuisance}}

In the case of binary partition $\vtheta=\{\zbinary\}$, where $\vtheta_{\itime}$
is the parameter of interest and $\vtheta_{\backslash\itime}$ is
the nuisance parameter, the Bayesian inference for $\vtheta_{\itime}$
can be readily derived via posterior $f(\vtheta|\xbold)$, as follows: 

\begin{eqnarray}
f(\vtheta_{\itime}|\xbold) & = & \int f(\vtheta|\xbold)d\vtheta_{\backslash\itime}\nonumber \\
 & \propto & \int f(\xbold|\vtheta)f(\vtheta)d\vtheta_{\backslash\itime}\label{eq:ch4:nuiss:theta_i}
\end{eqnarray}

By substituting the chain rule $f(\vtheta)=f(\vtheta_{\backslash\itime}|\vtheta_{\itime})f(\vtheta_{\itime})$
into (\ref{eq:ch4:nuiss:theta_i}), we can also derive a Bayes' rule
for $\vtheta_{\itime}$ directly, as follows:

\[
f(\vtheta_{\itime}|\xbold)\propto f(\xbold|\vtheta_{\itime})f(\vtheta_{\itime})
\]

where:

\begin{eqnarray}
f(\xbold|\vtheta_{\itime}) & = & \int f(\xbold|\vtheta)f(\vtheta_{\backslash\itime}|\vtheta_{\itime})d\vtheta_{\backslash\itime}\nonumber \\
 & = & \int f(\xbold,\vtheta_{\backslash\itime}|\vtheta_{\itime})d\vtheta_{\backslash\itime}\label{eq:ch4:nuiss:obs}
\end{eqnarray}

Note that this nuisance parameter issue is more difficult to solve
with frequentist method. Since prior $\prior$ is missing in this
case, the marginalization in (\ref{eq:ch4:nuiss:obs}) has to be approximated.
The common solution is to apply plug-in method, i.e. the nuisance
$\vtheta_{\backslash\itime}$ is replaced by its point estimation
$\widehat{\vtheta_{\backslash\itime}}$, which yields the so-called
profile likelihood {[}\citet{ch2:BK:Bernardo:Bayes94}{]}, as follows: 

\begin{eqnarray}
f_{p}(\xbold|\vtheta_{\itime}) & = & f(\xbold|\vtheta_{\itime},\widehat{\vtheta_{\backslash\itime}})\label{eq:ch4:nuiss:profile}
\end{eqnarray}

where again, $\widehat{\vtheta_{\backslash\itime}}$ is typically
chosen via ML principle. From (\ref{eq:ch4:nuiss:theta_i}) and (\ref{eq:ch4:nuiss:obs}),
we can see that the Bayesian inference for any subset of parameters
can be found by marginalizing out all nuisance parameters. However,
this approach often yields a complicated conditional distribution
which is often intractable {[}\citet{ch7:art:Transform:nuisance:Liseo06}{]}.
From (\ref{eq:ch4:nuiss:obs}), note that $f(\xbold|\vtheta_{\itime})$
is an infinite mixture of full observation model $f(\xbold|\vtheta)$,
with mixing density $f(\vtheta_{\backslash\itime}|\vtheta_{\itime})$.

Another approach is to produce an asymptotic integrated likelihood
via reference prior {[}\citet{ch7:art:Transform:ref_prior:Bernardo92,ch7:art:Transform:nuisance:Liseo93}{]}.
A major difficulty is that this approach depends strongly on an order
of parameters, which is relevant to the ordered grouping problem {[}\citet{ch2:BK:Bernardo:Bayes94}{]}.

In the $\ndata$-ary partition $\vtheta=\setd{\vtheta}{\ndata}$,
the direct computation of (\ref{eq:ch4:nuiss:theta_i}) is not feasible
in general {[}\citet{ch4:bk:Gelman03}{]}. This case is called hierarchical
parameter in the literature and will be studied in Section \ref{sec:chap6:HMC}.

\subsubsection{Sufficient statistics and shaping parameters}

From Fisher-Neyman factorization (\ref{eq:ch4:Neyman}), it is feasible
to recognize that, owing to normalizing operator in Bayes' rule (\ref{eq:ch4:BAYES's rule}),
the posterior inference for $\vtheta$ only depends on sufficient
statistic, rather than the whole data, as follows:

\[
f(\vtheta|\sufficient(\xbold))\propto f(\xbold|\vtheta)f(\vtheta)
\]

Hence, $\sufficient(\xbold)$ now becomes a shaping parameter for
the posterior distribution $\posterior$. In a slightly more general
case, where the prior $f(\vtheta|\veta)$ depends on known hyper-parameter
$\eta$ (also called shaping parameter in this thesis), the posterior
form can be written as follows: 

\[
f(\vtheta|\eta(\xbold))\propto\obs f(\vtheta|\veta)
\]
where $\eta(\xbold)=\{\sufficient(\xbold),\eta\}$ is called (data-updated)
shaping parameter for the posterior. 

The main challenges for posterior tractability are, therefore, to
identify the sufficient statistic and to design a prior such that
the computation of $\eta(\xbold)$ is feasible. Both of them can be
feasibly solved via definition of EF class. Multiplying the conjugate
prior (\ref{eq:ch4:EF:prior}) with EF observation (\ref{eq:ch4:EF:obs}),
the conjugate posterior can be feasibly derived, as follows:

\begin{equation}
f(\vtheta|\veta_{1}(\xbold))\propto C(\vtheta)^{\nu_{1}}\exp\left(R(\vtheta)v_{1}(\xbold)\right)\label{eq:ch4:EF:posterior}
\end{equation}
where $\nu_{1}=\nu_{0}+1$, $v_{1}(\xbold)=v_{0}+\sufficient(\xbold)$
and $\veta_{1}(\xbold)=\{\nu_{1},v_{1}(\xbold)\}$. Note that, because
the dimension of EF sufficient statistic $\sufficient(\xbold_{\ndata})=\sum_{\itime=1}^{\ndata}\sufficient(\xdata_{\itime})$
is preserved in iid case $\xbold_{\ndata}=\setd{\xdata}{\ndata}$,
the computation of posterior's shaping parameter $\veta_{1}(\xbold)$
is always tractable. Hence, the EF class plays an important role in
tractable Bayesian inference. In the online scheme, the computation
of $\sufficient(\xbold_{\ndata})$ can be carried out recursively
{[}\citet{ch4:BK:AQUINN_06}{]}.

\subsubsection{Asymptotic inference \label{subsec:chap4:Asymptotic-inference}}

Let us consider the negative logarithm of posterior distribution $\vL(\vtheta)=-\log f(\vtheta|\xbold_{\ndata})$
expanded up to second order of Taylor approximation, as follows:
\begin{equation}
\vL(\vtheta)=\vL(\ztheta)+(\vtheta-\ztheta)'\nabla\vL(\ztheta)-\frac{1}{2}(\phi-\phi_{0})'\vH(\ztheta)(\vtheta-\ztheta)+\ldots\label{eq:ch4:Taylor}
\end{equation}
 where $\nabla\vL(\ztheta)$ and $\vH(\ztheta)=-\nabla^{2}\vL(\ztheta)$
are the gradient vector and Hessian matrix evaluated at vector point
$\vtheta=\ztheta$, respectively, assuming regularity conditions on
$\vL(\vtheta)$ {[}\citet{ch2:BK:Bernardo:Bayes94}{]}

In the special case of the iid observation model, the first two orders
of the Taylor expansion yields an asymptotically converged form for
posterior distribution, as follows:
\begin{prop}
\label{prop:ch4:posteriorCLT} (Asymptotic posterior normality) {[}\citet{ch2:BK:Bernardo:Bayes94}{]}
Given iid observation model $f(\xbold_{\ndata}|\vtheta)=\prod_{\itime=1}^{\ndata}f(x_{i}|\vtheta)$
and maximum a posteriori (MAP) $\htheta$, i.e. $\nabla\vL(\htheta)=0$,
then under regularity conditions, the posterior distribution $f(\vtheta|\xbold_{\ndata})$
converges to the normal distribution $\calN_{\vtheta}\left(\htheta,\vH^{-1}(\htheta)\right)$,
when $\ndata\rightarrow\infty$. 
\end{prop}

The asymptotic posterior normality was firstly proposed and rigorously
proved in {[}\citet{ch4:origin:posteriorCLT:Laplace1810}{]} and {[}\citet{ch4:origin:posteriorCLT:LeCam53}{]},
respectively. Note that, given the iid observation model, we also
have a special converged form $\vH(\htheta_{ML})\rightarrow\ndata\calI(\htheta_{ML})$
as $\ndata\rightarrow\infty$, where $\htheta_{ML}=\arg\max_{\vtheta}f(\xbold_{\ndata}|\vtheta)$
{[}\citet{ch4:bk:Gelman03,ch2:BK:Bernardo:Bayes94}{]}. However, unless
the prior is uniform, the posterior $\posterior$ does not necessarily
converge to $\calN_{\vtheta}\left(\htheta_{ML},\vH^{-1}(\htheta_{ML})\right)$.

\section{Distributional approximation}

The computation of posterior distribution $\posterior$ is obviously
the main focus of Bayesian theory. However, the computation of the
posterior form via Bayes' rule (\ref{eq:ch4:BAYES's rule}) is often
intractable in practice. A common solution is, therefore, to use distributional
approximation $\funftilde{\vtheta}$ of $\posterior$, in which the
form $\funftilde{\vtheta}$ is tractable. The tractability means that
the computation can be carried out via a closed-form formula and/or
can be determined analytically in polynomial time.

The study of such approaches was the main reason for the revival of
Bayesian methodology in the 1980s. In this subsection, the most important
approximations will be briefly reviewed.

\subsection{Deterministic approximations}

\subsubsection{Certainty equivalence (CE) approximation \label{subsec:chap4:CE-approx}}

When moments of $\posterior$ are needed, but posterior form is hard
to derive, we can confine the posterior distribution into a single
point $\htheta(\xbold)$, as follows:

\begin{equation}
\funftilde{\vtheta}=\delta(\vtheta-\htheta(\xbold))\label{eq:ch4:approx:CE}
\end{equation}

Note that, by the sifting property, the functional moments $\widehat{\fung{\vtheta}}=E_{\posterior}\fung{\vtheta}$
can be approximated via substitution, as follows:

\begin{eqnarray}
\widehat{\fung{\vtheta}} & \approx & E_{\funftilde{\vtheta}}\fung{\vtheta}\label{eq:ch4:CE:moment}\\
 & = & E_{\delta(\vtheta-\htheta(\xbold))}\fung{\vtheta}\nonumber \\
 & = & \fung{\htheta(\xbold)}\label{eq:ch4:CE:substitution}
\end{eqnarray}

In the literature, this approximation (\ref{eq:ch4:CE:substitution})
is widely known as the plug-in substitution technique {[}\citet{ch4:BK:Bible:Robert07}{]}.
However, it will be called the Certainty Equivalent (CE) approximation
in this thesis, owing to its expectation form in (\ref{eq:ch4:CE:moment})
and the fact that it encodes all the uncertainty about $\vtheta$
\textit{a posteriori} by a single value $\htheta(\xbold)$ (\ref{eq:ch4:approx:CE}).
Although this CE approximation concept (\ref{eq:ch4:CE:moment}) is
different from the exact solution via CE principle (\ref{eq:ch4:CE}),
they will coincide if we can assign the CE, i.e. $\htheta(\xbold)=\htheta_{CE}(\xbold)$,
in (\ref{eq:ch4:approx:CE}) that satisfies (\ref{eq:ch4:CE}). The
name CE hence reflects the concept of distributional representation
via a representative value. In practice, $\htheta(\xbold)$ is often
chosen as the mean, mode, median of posterior distribution $\posterior$,
or as other minimum risk estimates (\ref{eq:ch4:Bayesian_Estimator}).

\subsubsection{Laplace approximation}

Owing to asymptotic posterior normality in Proposition \ref{prop:ch4:posteriorCLT},
the posterior $\posterior$ can be approximated via its asymptotic
form, i.e. $\funftilde{\vtheta}=\calN_{\vtheta}\left(\htheta,\vH^{-1}(\htheta)\right)$,
where $\htheta$ is the MAP estimate (mode) of $\posterior$. The
quality of this approximation obviously depends on the number of observations
and is typically poor in small samples. 

\subsubsection{MaxEnt approximation}

The posterior $\posterior$ can also be approximated via MaxEnt technique
in Section \ref{subsec:chap4:MaxEnt-Prior}, if distributional class
$\posterior\in\calF_{m}$ (\ref{eq:ch4:MaxEnt:mean_constraint}) is
already given. Similarly to derivation of MaxEnt prior, the MaxEnt
posterior $\funftilde{\vtheta}\in\calF_{m}$ is also a member of CEF
class, as follows:

\begin{equation}
\funftilde{\vtheta}\propto\exp(-\sum_{\itime}g_{\itime}(\vtheta)\lambda_{\itime}(\xbold))\label{eq:ch4:approx:MaxEnt:post}
\end{equation}
where Lagrangre multipliers $\lambda_{\itime}(\xbold)$ depend on
data $\xbold$ in this case. Note that, MaxEnt is a free-form variational
technique, since CEF form (\ref{eq:ch4:approx:MaxEnt:post}) of $\funftilde{\vtheta}$
is not fixed during approximation, but merely a solution of free-form
Entropy maximization process. 

\subsection{Variational Bayes (VB) approximation \label{subsec:chap4:Variational-Bayes-(VB)}}

In this thesis, the main deterministic approximation is the free-form
(variational) VB approximation. Let us provide a brief review of VB
and its variants in this subsection.

\subsubsection{Mean field theory}

The term ``variational'' originates from the term ``calculus of
variations'' {[}\citet{ch4:PhDThesis:VB:ICA:2002}{]}, in which the
(optimum) value of a definite integral (or a functional) deterministically
depends on the function in the argument of that integral {[}\citet{ch4:bk:VB:variational_calculus:1990}{]}.
The idea of the VB approximation has its roots in mean field theory
(MFT), which is originally a statistical quantum mechanics term, although
the definition of MFT was not specific in that early era {[}\citet{ch4:bk:VB:Mean_field:thermostatistics}{]}.
Loosely speaking, the MFT originally represents a technique for approximating
an interacting particle model by another non-interacting particle
model, such that the Helmholz free-energy is corrected up to the first
order {[}\citet{ch4:bk:VB:Mean_field:thermostatistics}{]}. The MFT
was later defined as a deterministic approximation for the expected
value of individual quantities in a generic statistical model, as
firstly introduced in neural networks in {[}\citet{ch4:origin:VB_in_neural_network:87}{]}. 

In Bayesian learning, the MFT was called ``ensemble learning'' {[}\citet{ch4:art:VB:MacKay:ensemble_learning}{]}
and re-defined as an approximate distribution, from a class of ``separable''
(i.e. independent) variables, to an arbitrary distribution, such that
the ``variational free-energy'' from the approximate distribution
to original distribution is minimized. The MFT was then, once again,
re-defined as the Variational Bayes method {[}\citet{ch4:art:VB:Attias:1999,ch4:art:VB:Jordan:2000}{]},
which minimizes the variational free-energy via the iterative Expectation-Maximization
(EM) and an iterative EM-like algorithm, called the VB EM algorithm
{[}\citet{ch4:PhdThesis:VB:VEM:2003}{]}. Note that, the methods based
on variational free-energy above mostly focus on point estimates and
neglect the optimization of the distributional form within the class
of approximate distributions of independent variables. 

Finally, the VB methodology was properly defined in {[}\citet{ch4:BK:AQUINN_06}{]}
as a free-form distributional approximation in the class of independent
variables, such that the Kullback-Leibler divergence (KLD) from the
approximate distribution to the original distribution is minimized.
An Iterative VB (IVB) algorithm was also proposed in {[}\citet{ch4:BK:AQUINN_06}{]}
in order to reach the local minimum of the KLD via an iterative gradient-based
method. Because this free-form definition of VB approximation is more
consistent with Bayesian methodology, this thesis will adopt this
VB approach.

\subsubsection{Iterative VB algorithm}

Let us consider a binary partition of parameters $\vtheta=\{\vtheta_{\itime},\theta_{\backslash\itime}\}$,
where $\theta_{\backslash i}$ denotes the complement of $\theta_{i}$
in $\vtheta$, i.e. the joint model $\joint$ has the following form:

\begin{equation}
\joint=f(\xbold,\vtheta_{\itime},\theta_{\backslash\itime})\label{eq:ch4:VB:joint}
\end{equation}
Then, the purpose of VB method is to seek an approximated distribution
$\funftilde{\vtheta}\in\calF_{c}$ in independent distribution class
$\calF_{c}:\ \breve{f}(\vtheta|\xbold)=\breve{f}(\vtheta_{\itime}|\xbold)\breve{f}(\vtheta_{\backslash\itime}|\xbold)$,
for which the Kullback-Leibler divergence $KLD_{\breve{f}||f}\TRIANGLEQ KLD(\breve{f}(\vtheta|\xbold)||\posterior)=E_{\breve{f}(\vtheta|\xbold)}\log\left(\frac{\breve{f}(\vtheta|\xbold)}{\posterior}\right)$
is minimized. 
\begin{thm}
(Iterative VB (IVB) algorithm) {[}\citet{ch4:BK:AQUINN_06}{]} \label{thm:chap4:Iterative-VB-(IVB)}Given
an arbitrary initial distribution $\ftilde^{[0]}(\vtheta|\xbold)=\ftilde^{[0]}(\theta_{i}|\xbold)\ftilde^{[0]}(\theta_{\backslash i}|\xbold)$,
the IVB algorithm updates VB-marginals in iterative cycle $\nu=1,2,\ldots$
until $\kldff$ is converged to a local minimum, as follows:

\begin{eqnarray}
\ftilde^{[\nu]}(\theta_{i}|\xbold) & \propto & \exp(E_{\ftilde^{[\nu-1]}(\theta_{\backslash i}|\xbold)}\log\joint)\label{eq:ch4:IVB}\\
\ftilde^{[\nu]}(\theta_{\backslash i}|\xbold) & \propto & \exp(E_{\ftilde^{[\nu]}(\theta_{\itime}|\xbold)}\log\joint)\nonumber 
\end{eqnarray}
\end{thm}

Note that, because $\ftilde^{[\iVB]}(\theta|\xbold)$ in IVB algorithm
(\ref{eq:ch4:IVB}) results from a gradient-based technique, convergence
to the global minimum is not guaranteed {[}\citet{ch4:BK:AQUINN_06}{]}.
Hence, $\ftilde^{[\infty]}(\vtheta|\xbold)$ is a local minimizer
of $KLD_{\breve{f}||f}$.

In practice, the computation of expectation in IVB algorithm (\ref{eq:ch4:IVB})
might be prohibitive or intractable. There are, however, some cases
in which this intractability can be avoided, as presented below.

\subsubsection{Separable-in-parameter (SEP) family}

From IVB cycles (\ref{eq:ch4:IVB}), it is feasible to recognize that
there exists a tractable class of joint distribution, such that the
IVB algorithm is tractable, as follows:
\begin{defn}
\label{DEF:(Separable-in-parameter-(SEP)}(Separable-in-parameter
(SEP) family) {[}\citet{ch4:BK:AQUINN_06}{]} The joint distribution
$\joint$ is said to belong to SEP family if its sub-parameters can
be split between separated kernels $\fung{\cdot}$and $\funh{\cdot}$,
as follows: 

\begin{equation}
\log\joint=\fung{\vtheta_{\itime},\xbold}\funh{\vtheta_{\backslash\itime},\xbold}\label{eq:ch4:VB:SEP_family}
\end{equation}
\end{defn}

Substituting the joint model (\ref{eq:ch4:VB:SEP_family}) back into
(\ref{eq:ch4:IVB}), the IVB scheme now becomes:

\begin{eqnarray}
\ftilde^{[\nu]}(\theta_{i}|\xbold) & \propto & \exp\left(\fung{\vtheta_{\itime},\xbold}\widehat{\funh{\vtheta_{\backslash\itime},\xbold}}^{[\iVB-1]}\right)\label{eq:ch4:IVB-SEP}\\
\ftilde^{[\nu]}(\theta_{\backslash i}|\xbold) & \propto & \exp\left(\widehat{\fung{\vtheta_{\itime},\xbold}}^{[\iVB]}\funh{\vtheta_{\backslash\itime},\xbold}\right)\nonumber 
\end{eqnarray}
where the iterative functional moments (called the VB moments) are
defined as: 
\begin{eqnarray*}
\widehat{\funh{\vtheta_{\backslash\itime},\xbold}}^{[\iVB-1]} & \TRIANGLEQ & E_{\ftilde^{[\nu-1]}(\theta_{\backslash i}|\xbold)}\funh{\vtheta_{\backslash\itime},\xbold}\\
\widehat{\fung{\vtheta_{\itime},\xbold}}^{[\iVB]} & \TRIANGLEQ & E_{\ftilde^{[\nu]}(\theta_{\itime}|\xbold)}\fung{\vtheta_{\itime},\xbold}
\end{eqnarray*}

From (\ref{eq:ch4:IVB-SEP}), we can see that VB marignals are available
if the VB moments can be computed. The advantage of the separability
constant is that these computations remain invariant (i.e. the integral
in each VB moment is not a function of $\iVB$). In this thesis, we
are only interested in SEP family for the IVB algorithm, owing to
this tractability property. 

Note that, similar to the Exponential Family (\ref{eq:ch4:EF:obs}),
the key motivation of the SEP family is to exploit the separability
between functional variables in order to achieve the tractability
in integral computation. In the former case, the computation in normalizing
constant is solved via separability between parameters and observed
data. In the latter case, the computation in IVB expectation is feasible
owing to separability between sub-parameters, given observed data.

\subsubsection{Functionally constrained VB (FCVB) approximation \label{subsec:chap4:Functionally-Constraint-VB}}

Another solution for tractably computing (\ref{eq:ch4:IVB}) is to
project one or all of the VB-marginals $\ftilde^{[\nu]}(\theta_{\backslash i}|\xbold)$,
$\ftilde^{[\nu]}(\theta_{\itime}|\xbold)$ into functionally constrained
classes $\fdelta^{[\nu]}(\theta_{\backslash i}|\xbold)$, $\fdelta^{[\nu]}(\theta_{i}|\xbold)$,
in particular the CE class (\ref{eq:ch4:approx:CE}), before they
are used in the expectation step of the IVB cycle (\ref{eq:ch4:IVB}).
This approximation scheme is called the FCVB approximation. 

In this way, the well-known Expectation-Maximization (EM) algorithm
can be recognized as a special case of FCVB, in which $\ftilde^{[\nu]}(\theta_{\backslash i}|\xbold)$
is projected to its local MAP point $\widehat{\theta_{\backslash i}}$,
i.e. $\fdelta^{[\nu]}(\theta_{\backslash i}|\xbold)=\delta(\theta_{\backslash i}-\widehat{\theta_{\backslash i}}^{[\nu]})$,
while $\ftilde^{[\nu]}(\theta_{i}|\xbold)$ is kept unchanged {[}\citet{ch4:BK:AQUINN_06}{]}.

Similarly, another form of FCVB is when both VB-marginals $\ftilde^{[\nu]}(\theta_{\backslash i}|\xbold)$,
$\ftilde^{[\nu]}(\theta_{\itime}|\xbold)$ are each projected into
their local MAP point $\fdelta^{[\nu]}(\theta_{\backslash i}|\xbold)=\delta(\theta_{\backslash i}-\widehat{\theta_{\backslash i}}^{[\nu]})$,
$\fdelta^{[\nu]}(\theta_{i}|\xbold)=\delta(\theta_{i}-\widehat{\theta_{i}}^{[\nu]})$,
respectively, as follows: 
\begin{lem}
\label{lem:(Iterative-FCVB-algorithm)}(Iterative FCVB algorithm)
Given an arbitrary initial value $\htheta^{[0]}=\{\widehat{\theta_{i}}^{[0]},$
$\widehat{\theta_{\backslash i}}^{[0]}\}$, the distribution $\fdelta^{[\iVB]}(\vtheta|\xbold)\in\calF_{\delta}:\ \breve{\fdelta}(\vtheta)=\delta(\theta_{i}-\widehat{\theta_{i}})\delta(\theta_{\backslash i}-\widehat{\theta_{\backslash i}})$
that locally minimizes $KLD(\fdelta(\vtheta|\xbold)||f(\vtheta|\xbold))$
can be found via Iterative FCVB algorithm at cycle $\nu=1,2,\ldots,$
as follows:

\textup{
\begin{eqnarray}
\widehat{\theta_{i}}^{[\nu]} & = & \arg\max_{\theta_{i}}(\exp(E_{\fdelta^{[\nu-1]}(\theta_{\backslash i}|\xbold)}\log\joint))\label{eq:ch4:FCVB:CE}\\
 & = & \arg\max_{\theta_{i}}f(\xbold,\theta_{i},\vtheta_{\backslash i}=\widehat{\theta_{\backslash i}}^{[\nu-1]})\nonumber \\
\widehat{\theta_{\backslash i}}^{[\nu]} & = & \arg\max_{\vtheta_{\backslash i}}(\exp(E_{\fdelta^{[\nu]}(\theta_{i}|\xbold)}\log\joint))\nonumber \\
 & = & \arg\max_{\vtheta_{\backslash i}}f(\xbold,\vtheta_{\itime}=\widehat{\theta_{i}}^{[\iVB]},\vtheta_{\backslash i})\nonumber 
\end{eqnarray}
}
\end{lem}

\begin{proof}
By definition, we have: 
\[
KLD(\fdelta(\vtheta|\xbold)||f(\vtheta|\xbold))=E_{\fdelta(\vtheta|\xbold)}\log\left(\frac{\fdelta(\vtheta|\xbold)}{\posterior}\right)=\log\left(\frac{1}{f(\vtheta=\htheta|\xbold)}\right)
\]
owing to the sifting property of $\fdelta(\vtheta|\xbold)=\delta(\vtheta-\htheta)$.
The local minimization can be seen intuitively by comparing (\ref{eq:ch4:FCVB:CE})
with (\ref{eq:ch4:VB:joint}). The value of posterior at any $\htheta$,
i.e. $f(\vtheta=\htheta|\xbold)=f(\xbold,\vtheta=\htheta)/f(\xbold)$,
does not decrease at any step in (\ref{eq:ch4:FCVB:CE}). Hence, the
value of $KLD(\fdelta(\vtheta|\xbold)||f(\vtheta|\xbold))$ is locally
minimized at a local MAP $\htheta$ at the convergence.
\end{proof}
Note that, the Iterative FCVB is not a double approximation, i.e.
it is not an approximation of VB approximation. Both IVB and Iterative
FCVB schemes are local minimizers of KLD distance from an independent
class, namely $\calF_{c}$ and $\calF_{\delta}$, respectively, to
the original distribution. However, because of similarity between
(\ref{eq:ch4:IVB}) and (\ref{eq:ch4:FCVB:CE}), the Iterative FCVB
approximation is considered as variant of IVB approximation in this
thesis.

It can be seen that the iterative CE updates (\ref{eq:ch4:FCVB:CE})
in FCVB algorithm is identical to Iterated Conditional Modes (ICM)
algorithm {[}\citet{ch4:art:ICM:Besag86}{]}, which is well known
to yield a local joint MAP estimate $\htheta=\{\widehat{\theta_{i}},\widehat{\theta_{\backslash i}}\}$
of original distribution $f(\vtheta|\xbold)$ at convergence {[}\citet{ch4:art:ICM:localMAP_2006}{]}.
Nevertheless, the name ``Iterative FCVB algorithm'' is preferred
to ``ICM algorithm'' in this thesis, since performance of this scheme
is easier to explain via independence property of distributional VB
approximation. 

\subsubsection{Non-iterative VB-related approximations}

In the literature, there are other non-iterative approximations that
can be considered as variants of the VB scheme. Although they will
not be used in this thesis, let us briefly review three typical cases
for the sake of completeness.

- If the purpose of the approximation is to minimize $\KLD_{MR}\TRIANGLEQ KLD(\posterior||\funftilde{\vtheta})$
instead of $\kldff$ in the VB scheme, the minimizer in this case
is $\funftilde{\vtheta}$ $=$ $f(\vtheta_{\itime}|\xbold)$ $f(\vtheta_{\backslash i}|\xbold)$,
i.e. the product of true posterior marginals. This scheme is widely
known as Minimum Risk (MR) approximation. 

- Obviously, the above MR approximation may not be interesting, since
those posterior marginals may be hard to compute in the first place.
However, if one posterior marginal, say $f(\vtheta_{\backslash i}|\xbold)$,
is given, the VB-marginal $\funftilde{\vtheta_{\itime}}$ can be found
by a single step of IVB algorithm (\ref{eq:ch4:IVB}), with $\funftilde{\vtheta_{\backslash i}}$
replaced by the true marginal $f(\vtheta_{\backslash i}|\xbold)$.
This non-iterative scheme is called the Quasi-Bayes approximation
in the literature. 

- If the true marginal $f(\vtheta_{\backslash i}|\xbold)$ in above
Quasi-Bayes scheme is not given, we can still replace $\funftilde{\vtheta_{\backslash i}}$
with a restricted form $\overline{f}(\vtheta_{\backslash i}|\xbold)$
, which can be imposed via some constraints on $\vtheta_{\backslash i}$.
Such a single-step VB scheme is called Restricted VB approximation. 

\subsection{Stochastic approximation}

A wide class of distributional approximation is the stochastic approximation,
in which the posterior $\posterior$ can be empirically approximated,
as follows:

\begin{equation}
\funftilde{\vtheta}=\frac{1}{\nstate}\sum_{\istate=1}^{\nstate}\delta(\vtheta-\vtheta^{(\istate)})\label{eq:ch4:sampling}
\end{equation}
where $\{\vtheta^{(1)},\vtheta^{(2)},\ldots,\vtheta^{(\nstate)}\}$
is an iid sample set (random sample), i.e. $\vtheta^{(\istate)}\sim\posterior$,
with $\seti{\istate}{\nstate}$. 

In low dimensions, the generation of $\vtheta^{(\istate)}$ can be
implemented via a wide range of sampling techniques, notably inversion,
rejection and importance sampling. In inversion sampling, the value
$F^{(\istate)}$ of cumulative function density (c.d.f) $F(\vtheta)\in[0,1]$
is generated first, while the scalar random variable $\vtheta$ can
be found via inverse function $F^{-1}$, i.e. $\vtheta^{(\istate)}=F^{-1}(F^{(\istate)})$.
In rejection sampling, the sample $\vtheta^{(\istate)}\sim f_{0}(\vtheta)$
is generated from the so-called dominated distribution $f_{0}(\vtheta)\geq f(\vtheta)$,
$\forall\vtheta\in\spaceO$ and, then, each sample $\vtheta^{(\istate)}$
is accepted with probability $p_{\istate}=f(\vtheta^{(\istate)})/f_{0}(\vtheta^{(\istate)})$
or else rejected. In importance sampling, we have $\ftilde(\vtheta)=\frac{1}{\nstate}\sum_{\istate=1}^{\nstate}w_{\istate}\delta(\vtheta-\vtheta^{(\istate)})$
, where each sample $\vtheta^{(\istate)}\sim\pi(\vtheta)$ is generated
from a reference distribution $\pi(\vtheta)$ and the weights are
computed as $w_{\istate}=f(\vtheta^{(\istate)})/\pi(\vtheta^{(\istate)})$. 

In high dimensions, the set $\{\vtheta^{(1)},\vtheta^{(2)},\ldots,\vtheta^{(\nstate)}\}$
can be generated dependently via a homogeneous Markov process $f(\vtheta_{\istate}|\vtheta_{\istate-1})$.
After a transition period $k_{c}$, the set $\{\vtheta^{(\istate)},\ldots,$
$\vtheta^{(\nstate)}\}$, $\istate>k_{c}$, is converged to iid sampling
set of $f_{s}(\vtheta)$, where $f_{s}(\vtheta)$ is, under mild regularity
conditions, the stationary distribution of homogeneous Markov process
$f(\vtheta_{\istate}|\vtheta_{\istate-1})$. This convergence in distribution
is independent of initialization $\vtheta^{(1)}$of this Markov process,
known as the forgetting property. By careful design of Markov transition
kernel $f(\vtheta_{\istate}|\vtheta_{\istate-1})$, we can achieve
the equality $f_{s}(\vtheta)=f(\vtheta)$. This is the basic principle
of Gibbs sampling in Markov Chain Monte Carlo (MCMC) method {[}\citet{ch2:origin:MCMC:Gibbs}{]}.
Another well known technique for designing transition kernel in MCMC
is Metropolis--Hastings algorithm {[}\citet{ch2:origin:MCMC:Metropolis}{]},
whose stationary distribution is the objective distribution. The key
advantage of MCMC is that an intractably high-dimensional distribution
can be tractably generated from multiple sampling steps of low-dimensional
conditional distributions. 

The most important condition of this technique is obviously the repeatability
of $\vtheta$. Then, owing to the sifting property of the Dirac-$\delta$
function, all functional moments of the empirical approximation, $\funftilde{\vtheta}$,
in (\ref{eq:ch4:sampling}) can be point-wise evaluated and, hence,
are always tractable. In prediction or online schemes, where the current
posterior becomes prior of next Bayesian inference step, the empirical
form (\ref{eq:ch4:sampling}) also satisfies the conjugacy principle,
and, hence, is always tractable. This idea is at the heart of particle
filtering for stochastic approximation of the nonlinear filtering
problem for time-variant parameters {[}\citet{ch6:art:VB:filtering:AQuinn08}{]}.
For reducing computational load, some variants of particle filtering,
which do not require a re-sampling step, were recently proposed in
the telecommunications context {[}\citet{ch4:art:particle_filtering:OFDM_2011,ch4:art:partical_filtering:ICASSP_2013}{]}.

\section{Summary}

A brief, but thorough, review for Bayesian techniques was given in
this chapter. It began with emphasis on the joint model, rather than
the posterior distribution, along with clarification on the issue
of subjectivity in belief quantification. 

Initially, the expectation of the loss function was taken with respect
to the joint model of parameters and data, via the axioms of decision
theory, and not with respect to the posterior distribution. The law
of the unconscious statistician (LOTUS) then showed that the mean
of loss function for the joint model, i.e. Bayesian risk, can be approximated
by Monte Carlo sampling and presented in simulation (e.g. as averaged
BER). Since Bayesian risk can be minimized equivalently via the minimum
risk (MR) estimator for posterior expected loss function, this motivated
the review of the posterior distribution, which, in turn, motivated
the reviews of prior and observation models, particularly the conjugacy
property in this chapter. 

Some distributional approximations, both deterministic and stochastic,
were surveyed, with a thorough review of VB and its special case,
functionally constrained VB (FCVB), for reducing Kullback-Leibler
divergence (KLD) associated with approximate distribution. These VB-based
methods will be used for designing novel algorithms in the context
of the hidden Markov chain for digital receivers, in Chapters \ref{=00005BChapter 6=00005D}.

%auto-ignore
%auto-ignore
%%%% Common

\global\long\def\REAL{\mathbb{R}}%

\global\long\def\DEAL{\mathbb{D}}%

\global\long\def\COMPLEX{\mathbb{C}}%

\global\long\def\RING{\mathcal{R}}%

\global\long\def\MSET{\mathcal{M}}%

\global\long\def\ASET{\mathcal{A}}%

\global\long\def\OMEGA{\Omega}%

\global\long\def\calO{{\cal O}}%

\global\long\def\calX{\mathcal{X}}%

\global\long\def\ndata{n}%

\global\long\def\nstate{m}%

\global\long\def\itime{i}%

\global\long\def\istate{k}%

\global\long\def\funh#1{h\left(#1\right)}%

\global\long\def\fung#1{g\left(#1\right)}%

\global\long\def\seti#1#2{#1=1,2,\ldots,#2}%

\global\long\def\setd#1#2{\{#1{}_{1},#1{}_{2},\ldots,#1_{#2}\}}%

\global\long\def\TRIANGLEQ{\triangleq}%

%%%% Chapter 5

\global\long\def\Ring{\mathcal{R}^{\calX_{\OMEGA}}}%

\global\long\def\spaceX{\mathcal{X}_{\OMEGA}}%

\global\long\def\noperator{\kappa}%

\global\long\def\ndim{m}%

\global\long\def\xBold{\mathbf{x}}%

\global\long\def\yBold{\mathbf{y}}%

\global\long\def\xHat{\widehat{\mathbf{x}}}%

\global\long\def\xMath{\mathrm{x}}%

\global\long\def\xbar{\bar{\mathrm{x}}}%

\global\long\def\xEta#1{\mathrm{x}_{\eta_{#1}}}%

\global\long\def\xpEta#1{\mathrm{x}_{\eta_{#1}^{+}}}%

\global\long\def\xbEta#1{\mathrm{x}_{\boldsymbol{\eta}_{#1}^{+}}}%

\global\long\def\xReverse#1{x_{[#1]}}%

\global\long\def\xbReverse#1{\mathbf{x}_{[#1]}}%

\global\long\def\Eta#1{\eta_{#1}^{+}}%

\global\long\def\xReverse#1{x_{[#1]}}%

\global\long\def\Etbold{\boldsymbol{\eta}}%

\global\long\def\bracket#1{[#1]}%

\global\long\def\bbracket#1{\boldsymbol{[#1]}}%

\global\long\def\bparent#1{\boldsymbol{(#1)}}%

\global\long\def\bcommon#1{\eta_{[#1]}}%

\global\long\def\kstate{k}%

\global\long\def\iton{i=1,\ldots,n}%

\global\long\def\iinn{i\in\{1,\ldots,n\}}%

\global\long\def\jtom{j=1,\ldots,m}%

\global\long\def\kinm{k\in\{1,\ldots,\ndim\}}%

\global\long\def\ktom{k=1,\ldots,m}%

\global\long\def\fhat{\widehat{\boldsymbol{f}}}%

\global\long\def\fBold{\mathbf{f}}%

\global\long\def\fbold{\boldsymbol{f}}%

\global\long\def\fprofile{f_{p}}%

\global\long\def\gBold{\mathbf{g}}%

\global\long\def\gbold{\boldsymbol{g}}%

\global\long\def\gHat{\widehat{\mathbf{g}}}%

\global\long\def\ghat{\widehat{g}}%

\global\long\def\gbar{\overline{g}}%

\global\long\def\oBold{\boldsymbol{\omega}}%

\global\long\def\COUNT#1#2{\calO_{#1}[#2]}%

\global\long\def\SUM{\boxplus}%

\global\long\def\SSUM#1{\underset{#1}{\boxplus}}%

\global\long\def\PROD{\odot}%

\global\long\def\SET{\mathcal{S}}%

\global\long\def\Sc{\mathcal{S}^{c}}%

\global\long\def\Ssubstract#1{\bar{\mathcal{S}}_{#1}}%

\global\long\def\Pibar{\overline{\pi}}%

\chapter{Generalized distributive law (GDL) for CI structure\label{=00005BChapter 5=00005D}}

A typical scenario in Bayesian inference is the marginalization over
hierarchical and nuisance parameters, as shown in Section \ref{subsec:ch4:Hierarchical-and-nuisance}.
Such a computation can be efficiently computed via the generalized
distributive law (GDL), as explained in this chapter. This GDL scheme
will also reveal efficient algorithms for Markovian model, i.e. a
special case of conditionally independent (CI) model, in Chapter \ref{=00005BChapter 6=00005D}. 

\section{Introduction}

When computing a sequence of operators upon a product of multivariate
functions (factors), the generalized distributive law (GDL) {[}\citet{ch2:origin:GDL:McEliece}{]}
has been proposed for reducing the computational load. Nevertheless,
the proposed computational flow of GDL has to be designed via a graph
representation of those factors. 

In this chapter, we will propose a novel topology representation,
namely conditional independent (CI) structure, for those factors.
This topology divides the factors into separate partitions, across
which the operators can be freely distributed via GDL. Because two
design stages, one for factors and one for operators, are isolated
in this scheme, the total number of arithmetic operators can be tuned
feasibly. This flexibility is useful for designing an optimal reduction
in computational complexity. 

\subsection{Objective functions \label{subsec:chap5:Objective-functions}}

In order to be consistent with the literature, the standard notation
$\{\cdot,\cdot,\ldots,\cdot\}$ for a set with separated partitions
will be applied throughout this chapter. Then, let $\xBold_{\OMEGA}$
be $m$-tuples variables within $\calX_{\OMEGA}$ space, i.e. $\xBold_{\OMEGA}=\{x_{1},\dots,x_{m}\}\in\calX_{\OMEGA}$,
where $\OMEGA$ is the total index set (universe), as follows: 
\begin{equation}
\OMEGA\TRIANGLEQ\{1,\ldots,k,\ldots,\nstate\}\label{eq:ch5:OMEGA}
\end{equation}
and dimension is equal to its cardinality. For simplicity, let us
assume here that all $x_{\kstate}$ belong to the same finite set
$\calX$, i.e. $x_{\kstate}\in\calX_{\kstate}=\calX$, $\ktom$, and
hence: 
\begin{equation}
\xBold_{\OMEGA}\in\calX_{\OMEGA}\TRIANGLEQ\underset{\nstate\ times}{\underbrace{\calX\times\ldots\times\calX}}=\calX^{\nstate}\label{eq:ch5:X:Omega}
\end{equation}
We also denote: 
\begin{eqnarray}
M & \TRIANGLEQ & |\calX|\label{eq:ch5:DEF:M}\\
\nstate & \TRIANGLEQ & |\OMEGA|\nonumber 
\end{eqnarray}
where $|\cdot|$ is the cardinal number of a finite set. Note that
all the results in this chapter can be generalized feasibly to the
case of different sets $\calX_{\kstate}$ and/or the case of continuous
variables {[}\citet{ch5:art:GDL:NewLook04}{]}.

Throughout this chapter, let us define $g:\calX^{|\omega|}\rightarrow\REAL$
as a generic (i.e. wildcard) function $g$ over index set $\omega\subseteq\OMEGA$
of variables $\xBold_{\omega}$. Then, as an imposed model, the main
function, $g(\xBold_{\OMEGA})$, is assumed to be a product of $\ndata$
known factors, as follows:

\begin{equation}
g(\xBold_{\OMEGA})=\prod_{\itime=1}^{\ndata}g(\xBold_{\omega_{i}})\label{eq:ch5:MODEL:x}
\end{equation}
where index sets are defined as $\omega_{i}\subseteq\OMEGA=\{1,\ldots,m\}$
and $\omega_{i}\neq\emptyset$, $\seti{\itime}{\ndata}$, such that:

\begin{equation}
\OMEGA=\omega_{1}\cup\cdots\cup\omega_{n}\label{eq:ch5:OMEGA:model}
\end{equation}

For shortening notation in (\ref{eq:ch5:MODEL:x}), let us also denote
$\gBold_{1:n}\TRIANGLEQ g(\xBold_{\OMEGA})$ and $g_{i}\TRIANGLEQ g(\xBold_{\omega_{i}})$,
which yield a neater form:

\begin{equation}
\gBold_{1:n}\equiv\prod_{i=1}^{n}g_{i}=g(\xBold_{\OMEGA})\label{eq:ch5:MODEL:i}
\end{equation}

Finally, if we define a generic operator ring-product $\PROD$ from
ring theory (see Definition \ref{DEF=00003DRing-product}) instead
of product $\prod$, the model (\ref{eq:ch5:MODEL:i}) can be treated
generally as follows:

\begin{equation}
\gBold_{1:n}\equiv\PROD_{i=1}^{n}g_{i}=g(\xBold_{\OMEGA})\label{eq:ch5:MODEL}
\end{equation}

For illustration, several examples of (\ref{eq:ch5:MODEL:i}) in this
thesis are (\ref{eq:ch4:VB:SEP_family}), (\ref{eq:Joint}) (see e.g.
{[}\citet{ch2:bk:ToddMoon}{]} for more examples in telecommunication
context). 

For later use, let us propose the following definition:
\begin{defn}
\label{DEF:chap5:variable and factor index}\textbf{ }\textit{(Index
of variable and index set of factor)} \\
The index $\kstate\in\OMEGA$ in (\ref{eq:ch5:OMEGA}) is called index
of variable (or variable index). In (\ref{eq:ch5:MODEL:x}-\ref{eq:ch5:OMEGA:model}),
$\omega_{i}$ is called the index set of the factor $g_{i}$ (or the
factor index set of $g_{i}$) in (\ref{eq:ch5:MODEL:i}-\ref{eq:ch5:MODEL}).
\end{defn}

\subsubsection{Single objective function}

In practice, it is often required to compute a sequence of operators
upon a sequence of factors $\gBold_{1:n}$ (\ref{eq:ch5:MODEL:i}).
For example, the objective function might be the output of summation:
\[
\sum_{\xBold_{S}}\gBold_{1:n}=\sum_{\xBold_{S}}g(\xBold_{\OMEGA})=g(\xBold_{\Sc})
\]
or maximization: 
\[
\max_{\xBold_{S}}\gBold_{1:n}=\max_{\xBold_{S}}g(\xBold_{\OMEGA})=g(\xBold_{\Sc})
\]
where $\xBold_{\OMEGA}=\{\xBold_{\Sc},\xBold_{S}\}$ and: 
\begin{equation}
\OMEGA=\{\Sc,\SET\}\label{eq:ch5:OMEGA:single}
\end{equation}

Nevertheless, a generic operator ring-sum $\SUM_{\xBold_{\SET}}$
from ring theory (Definition \ref{DEF=00003DRing-sum}) over $\xBold_{\SET}$
is not necessarily the sum or max. The objective function $g(\xBold_{\Sc})$
is then defined as the output of that operator upon $\gBold_{1:n}$
in (\ref{eq:ch5:MODEL}), as follows:

\begin{equation}
\SUM_{\xBold_{\SET}}\gBold_{1:n}=\SUM_{\xBold_{\SET}}g(\xBold_{\OMEGA})=g(\xBold_{\Sc})\label{eq:ch5:SUM:single}
\end{equation}
When dimension $m$ (\ref{eq:ch5:X:Omega}) is too high, the task
(\ref{eq:ch5:SUM:single}) leads to a heavy computational load in
practice {[}\citet{ch2:origin:GDL:McEliece,ch2:bk:ToddMoon}{]}. 

For later use, let us propose the following definition:
\begin{defn}
\label{DEF:chap5:operator and objective set}\textbf{ }\textit{(Index
set of operator $S$ and objective set $\Sc$) }\\
In (\ref{eq:ch5:OMEGA:single}), the set $S\subseteq\OMEGA$ is called
the index set of operators (or operator index set) and the compliment
$\Sc=\OMEGA\backslash\SET$ is called the objective set (of objective
function).
\end{defn}

\subsubsection{Sequential objective functions}

In practice, it is more often required to compute a sequence of objective
functions, rather than a single objective function. In this case,
the sequential objective functions can be defined as:

\begin{equation}
\SUM_{\xBold_{\SET_{j}}}\gBold_{1:n}=\SUM_{\xBold_{\SET_{j}}}g(\xBold_{\OMEGA})=g(\xBold_{\Sc_{j}}),\ \seti j{\noperator},\label{eq:ch5:SUM:sequential}
\end{equation}
in which:
\begin{equation}
\OMEGA=\{\Sc_{1},\SET_{1}\}=\ldots=\{\Sc_{\noperator},\SET_{\noperator}\}\label{eq:ch5:OMEGA:sequential}
\end{equation}

In a naive approach, we merely apply the computation (\ref{eq:ch5:SUM:single})
$\kappa$ times, each time with different operator index set $\SET_{j}$.
The computational load is, loosely speaking, about $\kappa$-fold
the computational cost in (\ref{eq:ch5:SUM:single}). Because we often
have $\noperator=n$ , this approach becomes impractical for high
$n$.

For efficient computation, we confine ourselves to a special topological
case, defined below (Definition \ref{(Non-overflowed-set)}) as the
non-overflowed condition for the objective sets $\setd{\Sc}{\noperator}$.
The results of the $\noperator$ formulae in (\ref{eq:ch5:SUM:sequential})
can be extracted, in a single sweep, from computational memory of
a single objective function (\ref{eq:ch5:SUM:single}) upon a union
of operator index sets, $\SET=\SET_{1}\cup\ldots\cup\SET_{\noperator}$
(hence the name ``sequential functions''), if that extraction does
not require any re-computation step and does not yield overflowed
memory (hence the name ``non-overflowed''). Note that, as shown
below, the case of scalar objective sets $\Sc_{j}=j$, which is widely
required in practice, always satisfy this condition. 

\subsection{GDL - the state-of-the-art }

If the two operators, $\SUM$ and $\PROD$, satisfy distributive law
of ring theory (Definition \ref{Pre-semiring}), the total number
of these operators in sequential objective functions (\ref{eq:ch5:SUM:sequential})
can be reduced significantly, as firstly formalized in {[}\citet{ch2:origin:GDL:McEliece}{]}
for the case of scalar objective sets $\Sc_{j}=j$. Such a proposal
motivated practical studies of the generalized distributive law (GDL)
in ring theory in order to design efficient algorithms {[}\citet{ch5:bk:ring:Glazek02}{]}.
Nevertheless, the use of GDL in the literature is still modest. The
main effort in the literature so far is to re-interpret known efficient
algorithms for sum and max operators into ring theory, which, in turn,
can generalize those algorithms for all operators $\SUM$ satisfying
GDL, as briefly reviewed below. 

In the early days, the main interest is to generalize some probabilistic
computational algorithms on the joint distribution $f(\xBold_{\OMEGA})$
into ring theory. For example, two well-known algorithms (in Chapter
\ref{=00005BChapter 6=00005D}) in Markov Chain decoder context -
the Viterbi algorithm (VA) for joint maximum-a-posteriori (MAP) and
the Forward-Backward (FB) algorithm (also known as BCJR algorithm
in channel decoder context) for sequence of marginal MAP {[}\citet{ch2:bk:ToddMoon}{]}
- were generalized in {[}\citet{ch5:art:GDL:ViterbiAlgo90}{]} and
{[}\citet{ch5:PhD:Wiberg96,ch5:art:BCJR:McEliece96}{]}, respectively.
In Bayesian networks, the generalized forms of belief propagation
and message-passing algorithms for a sequence of marginals were proposed
in {[}\citet{ch5:PhD:Nielsen01,ch5:art:MP:Kschischang01}{]}. 

Recently, the primary interest in GDL comes from graph theory. The
trend can be considered as having begun with the semi-tutorial paper
{[}\citet{ch2:origin:GDL:McEliece}{]}, which migrated the early results
into graphical learning language. GDL for graph was also applied in
other fields, like circuit design {[}\citet{ch5:art:circuit:Tong96}{]},
automatics {[}\citet{ch5:art:Automatics:MaxPlus10}{]} and entropy
computation in probability {[}\citet{ch5:art:MP:Entropy11}{]}. However,
a drawback for GDL development in graph theory has been the inconsistency
of the semiring concept, which is still under development in modern
algebra {[}\citet{ch5:bk:ring:Glazek02}{]}. Only recently, an attempt
at unifying the ring concepts for graphs was proposed in {[}\citet{ch5:bk:ring:Gondran:Minoux:08}{]}.

\subsection{The aims of this chapter}

From the above review, we can feasibly grasp the reason for emergence
of GDL in graph theory: the key point is that the graph provides a
structure representation of the original model (\ref{eq:ch5:MODEL}).
From that structure, the operators are then distributed directly into
factors $g_{i}$ via GDL. The computational flow, therefore, relied
on extra concepts and algorithms in graphical topology. For example,
application of GDL requires the notion of junction tree in {[}\citet{ch2:origin:GDL:McEliece}{]},
factor graph in {[}\citet{ch5:art:MP:Kschischang01}{]} or elimination
process {[}\citet{ch5:PhD:Nielsen01}{]} in Bayesian networks. This
indirect approach leads to ambiguity and misleading in counting number
of arithmetic operators and, eventually, in reduction of computational
load. 

For a more direct approach, there are three key steps in this chapter:

- A novel representation of (\ref{eq:ch5:MODEL}), namely the conditionally
independent (CI) structure, will be designed via the set algebra on
the factor index sets $\setd{\omega}{\ndata}$. All factors in (\ref{eq:ch5:MODEL:i}-\ref{eq:ch5:MODEL})
will be partitioned into ``bins'' beforehand, such that the variable
indices are conditionally separated (Fig. \ref{fig:Partitions-of-universe}).
Afterwards, the operator's set $\SET$ in (\ref{eq:ch5:OMEGA:single},\ref{eq:ch5:OMEGA:sequential})
will be divided and distributed into these ``bins'', via GDL. This
scheme does not only separate the stage of factor design (\ref{eq:ch5:OMEGA:model})
from the stage of operator design (\ref{eq:ch5:OMEGA:single},\ref{eq:ch5:OMEGA:sequential}),
but it also separates factors (a concept relating to model representation),
from GDL (a concept relating to operators). 

- For a better understanding of GDL, the basic abstract algebra in
ring theory will be provided from a computational perspective. A novel
theorem (Theorem \ref{thm:GDL}), which guarantees computational reduction
in GDL, will also be proved. Owing to this computational approach,
the insight of many inference tasks in probability and, their computational
load, will also be unified with respect to the computational flow
of GDL. 

- In the probability context, which is the main application of GDL
in this thesis and in current research, the equivalence between CI
structure and CI factorization of the joint distribution will be shown.
This equivalence also reveals a hidden consequence of GDL: it does
not only facilitate inference computations, but also implicitly re-factorizes
the original joint distribution. These development will be important
in explaining the tractability of Bayesian inference in HMC model
in Section \ref{sec:chap6:FB and VA via CI} of this thesis. 

\section{Conditionally independent (CI) topology \label{sec:Conditionally-separable}}

From (\ref{eq:ch5:OMEGA},\ref{eq:ch5:OMEGA:model}) and (\ref{eq:ch5:OMEGA:single},\ref{eq:ch5:OMEGA:sequential}),
we can see that there are three different ways to construct the universe
$\OMEGA$: variable indices, factor indices and operator indices,
respectively.

In this section, the topology of variable indices will be exploited
in two tasks: 

- For factor indices: an algorithm will be designed for exploiting
the conditionally independent (CI) structure embedded in the sequence
$\omega_{1},\ldots,\omega_{n}$ (\ref{eq:ch5:OMEGA:model}) and representing
the universe $\OMEGA$ in two sequence of CI partitions, which we
call no-longer-needed (NLN) and first-appearance (FA) variable indices.

- For operator indices: a sequence of distributed sets over those
two CI partitions will be defined. This task sets up recursive computation
of GDL in the sequel.

The topology results of these tasks are illustrated in Fig. \ref{fig:Partitions-of-universe}
and will be explained in step-by-step below.

\begin{figure}
\begin{centering}
\includegraphics[width=0.4\columnwidth]{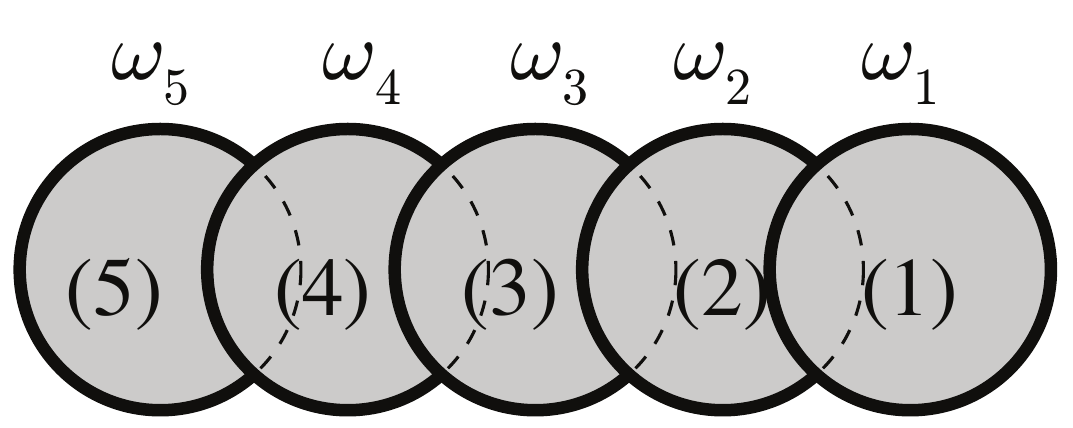}$\ \ \ $\includegraphics[width=0.4\columnwidth]{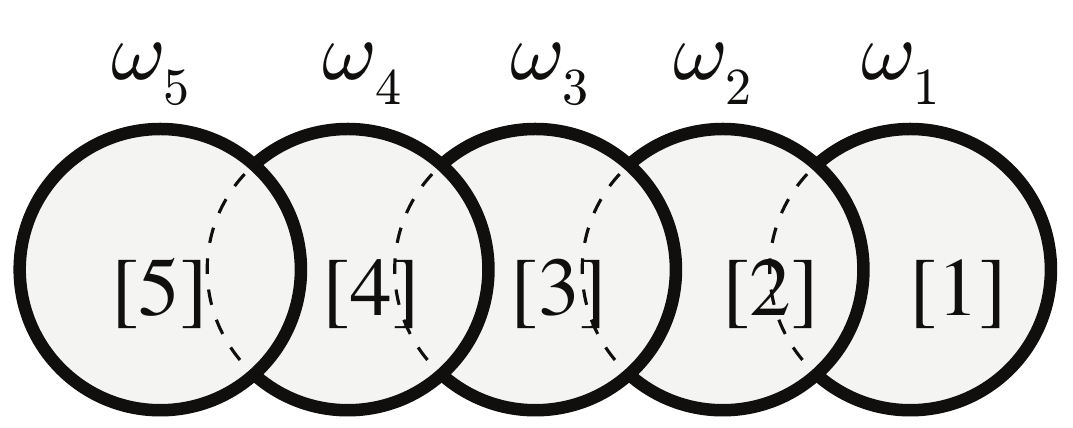}
\par\end{centering}
\begin{centering}
\includegraphics[width=0.4\columnwidth]{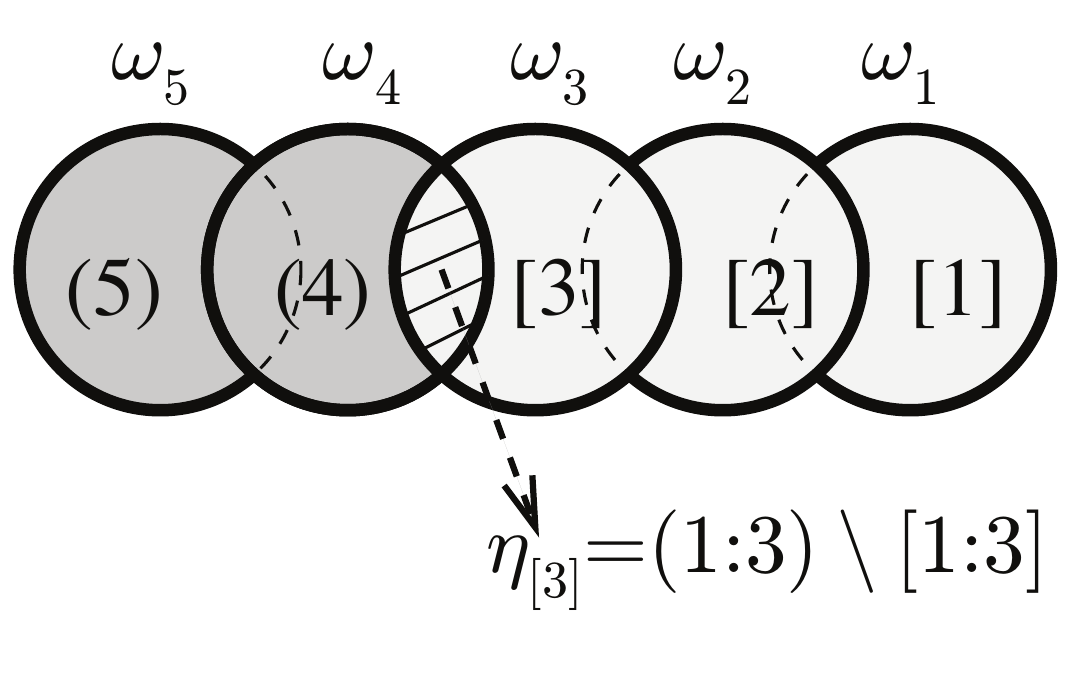}$\ \ \ $\includegraphics[width=0.4\columnwidth]{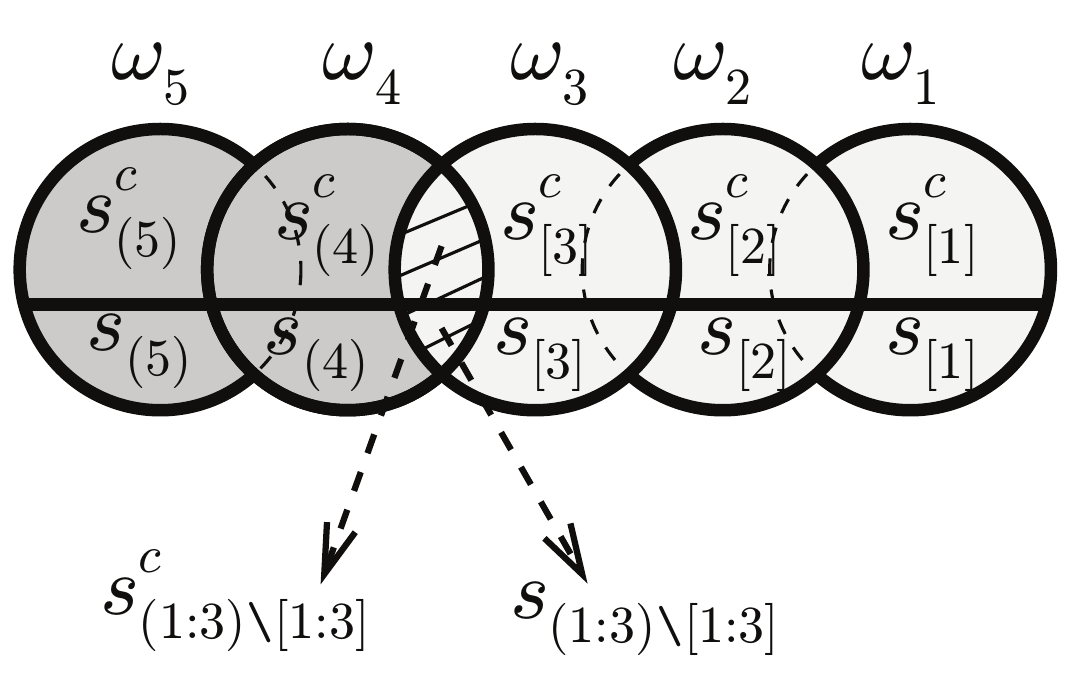}
\par\end{centering}
\caption{\label{fig:Partitions-of-universe}The topology of $n=5$ factor index
sets $\omega_{1},\ldots,\omega_{n}$ in (\ref{eq:ch5:OMEGA:model}),
together with the operator index set $\protect\SET\subseteq\protect\OMEGA$
in (\ref{eq:ch5:OMEGA:single})}
\end{figure}

\subsection{Separated indices of factors}

\subsubsection{No-longer-needed (NLN) algorithm}

In CI structure, the aim is to divide the universe $\OMEGA$ into
$n$ partitions (or ``bins'') of variables. From (\ref{eq:ch5:OMEGA:model}),
we can define two ways to partition $\OMEGA$, one backward and one
forward, as follows: 
\begin{eqnarray}
\OMEGA & = & \omega_{1}\cup\cdots\cup\omega_{n}\label{eq:ch5:(i)-=00005Bi=00005D}\\
 & = & \{\bparent 1,\ldots,\bparent n\}\nonumber \\
 & = & \{\bbracket 1,\ldots,\bbracket n\}\nonumber 
\end{eqnarray}
where the NLN index sets $\bbracket i$ (Definition \ref{(No-longer-needed-indices)})
and the FA index set $\bparent i$ (Definition \ref{(New-coming-indices)})
are (possibly empty) subsets of $\OMEGA$, $\iton$. 

Because $\omega_{i}$ might be overlapped with each other, let us
denote: 
\[
\oBold_{i:j}\TRIANGLEQ\omega_{i}\cup\cdots\cup\omega_{j},\ 1\leq i\leq j\leq n.
\]
Then the task is to extract $\bparent i$ and $\bbracket i$ from
$\{\omega_{1},\ldots,\omega_{n}\}$, as follows: 
\begin{defn}
\textit{(No-longer-needed (NLN) indices)} \label{(No-longer-needed-indices)}\\
Let us consider $n$ backward nested sets $\omega_{n}=\oBold{}_{n:n}\subseteq\oBold_{n-1:n}\subseteq\cdots\subseteq\oBold_{1:n}=\OMEGA$
and extract $n$ partitions of $\OMEGA$, as follows: $\bbracket i=\oBold_{i:n}\backslash\oBold_{i+1:n}$,
$\forall\iinn$ (upper-right schematic in Fig. \ref{fig:Partitions-of-universe}).
We call $\bbracket i$ a set of \textit{no-longer-needed }indices
for set $\omega_{\itime}$, $i=n,\ldots,1$, owing to the fact that
$\bbracket i\subseteq\omega_{i}$ and $\bbracket i\not\subset\oBold_{i+1:n}$. 
\end{defn}

$\ $
\begin{defn}
\textit{(First-appearance (FA) indices) }\label{(New-coming-indices)}\\
Let us consider $n$ forward nested sets $\omega_{1}=\oBold_{1:1}\subseteq\cdots\subseteq\oBold_{1:n-1}\subseteq\oBold_{1:n}=\OMEGA$,
we can also extract other $n$ partitions of $\OMEGA$, as follows:
$\bparent i=\oBold_{1:i+1}\backslash\oBold_{1:i}$, $\forall\iinn$
(upper-left schematic in Fig. \ref{fig:Partitions-of-universe}).
We call $\bparent i$ a set of \textit{first-appearance} indices at
point $\iton$, owing to the fact that $\bparent i\subseteq\omega_{i}$
and $\bparent i\not\subset\oBold_{1:i-1}$.
\end{defn}

For deriving a sequence of $\bbracket i$ and $\bparent i$, induced
by (\ref{eq:ch5:(i)-=00005Bi=00005D}), we can directly apply the
Definition \ref{(No-longer-needed-indices)}, \ref{(New-coming-indices)},
respectively. However, an extraction algorithm can also be designed,
as presented below. 

Firstly, let us design an $m\times n$ occupancy matrix $\OMEGA_{m\times n}$
for the universe $\OMEGA=\omega_{1}\cup\cdots\cup\omega_{n}$, as
follows:

\begin{equation}
\begin{array}{c}
\begin{array}{c}
(\OMEGA_{m\times n})\end{array}\\
\begin{array}{c}
m\\
\vdots\\
1
\end{array}
\end{array}\begin{array}{c}
\begin{array}{ccc}
\omega_{n}\ldots & \omega_{\itime}\ldots & \omega_{1}\end{array}\\
\left[\begin{array}{ccc}
\delta_{m,n} & \delta_{m,\itime} & \delta_{m,1}\\
\vdots & \vdots & \vdots\\
\delta_{1,n} & \delta_{1,\itime} & \delta_{1,1}
\end{array}\right]
\end{array}\label{eq:Adj-Matrix}
\end{equation}
where: 

- the whole matrix is, for convenience, labelled $(\OMEGA_{m\times n})$; 

- $n$ columns, from the first to the last column, are labelled by
$\omega_{\ndata},\ldots,\omega_{1}$, respectively (i.e in the reverse
order);

- $m$ rows, from the first to the last row, are labelled by variable
indices $k=m,\ldots,1$, respectively (i.e in the reverse order) in
the universe $\OMEGA$; 

- $\delta_{\istate,\itime}$ denotes binary indicator (i.e. a Kronecker
function): $\delta_{\istate,\itime}=1$ if $\istate\in\omega_{\itime}$,
and $\delta_{\istate,\itime}=0$ otherwise, $1\leq\istate\leq m$,
$1\leq\itime\leq n$. 

The reason for the reverse order in \foreignlanguage{british}{labelling}
is to preserve the band-diagonal property for the case of the Markov
chain, as illustrated in Example \ref{Occupancy Matrix-HMC}, what
follows shortly. 

From the occupancy matrix (\ref{eq:Adj-Matrix}), we can see that
the necessary and sufficient condition for $\istate\in\bbracket i=\oBold_{i:n}\backslash\oBold_{i+1:n}$,
is that - in row $k$ - value $1$ appears at column $\itime$ and
value $0$ appears at all columns $\{\itime+1,\ldots,\ndata\}$ ,
i.e. $\delta_{\istate,\itime}=1$ and $\delta_{k,j}=0$, with $j=\itime+1,\ldots,n$
. Hence, the set $\bbracket i$, $\iton$, consists of all those row
indices containing value $1$ at column $\omega_{\itime}$. 

Similarly, we can take the indices of rows for which the first-$1$
from the right occurs at column $\itime$, as constituting the $\itime$th
first-appearance (FA) set, $\bparent i$.

We can design an algorithm for construction of the $\ndata$ no-longer-needed
(NLN) sets, $\bbracket i$, as illustrated in Algorithm \ref{alg:chap7:(NLN)-algorithm}.

\begin{algorithm}
\textbf{Initialization:} 

- Constructing occupancy matrix $\OMEGA_{m\times n}$ (\ref{eq:Adj-Matrix})
for $\omega_{1},\ldots,\omega_{n}$. 

- Initialize $\bbracket i=\emptyset$, $\seti{\forall\itime}{\ndata}$

- Initialize $\ndata+1$ counter matrices, $C_{n+1}=\OMEGA_{m\times n}$
and $C_{\itime}=\emptyset$, $\seti{\forall\itime}{\ndata}$ 

\textbf{Recursion:} 

For each $\itime=n,\ldots,1$ (i.e. from the first to the last column),
do: \{ $C_{\itime}\leftarrow C_{\itime+1}$;

For each row $\istate$ at column $\omega_{\itime}$ in matrix $C_{\itime}$,
do: \{

if $\delta_{\istate,\itime}=1$, do: \{

- add value $\istate$ to the set $\bbracket i$ 

- delete row $\istate$ from $C_{\itime}$

\}\} Stop if $C_{\itime}=\emptyset$ \}

\textbf{Return:} the sets $\bbracket i$, $\seti{\itime}{\ndata}$

\textbf{Complexity:} $O(m)\leq O(m(\ndata-\itime_{c}+1))\leq O(mn)$
of Boolean comparison $\delta_{\istate,\itime}$ for $\nstate$ rows
and $(\ndata-\itime_{c}+1)$ columns, until stopping at column $\omega_{\itime_{c}}$,
$1\leq\itime_{c}\leq n$.

\caption{\label{alg:chap7:(NLN)-algorithm}No-longer-needed (NLN) algorithm}
\end{algorithm}
The NLN algorithm (Algorithm \ref{alg:chap7:(NLN)-algorithm}) for
sets $\bparent i$ are can be designed similarly, but taking the the
values $i=1,\ldots,\ndata$ (i.e. from the last to the first column).

Note that, the NLN algorithm (Algorithm \ref{alg:chap7:(NLN)-algorithm})
is merely a straightforward design from Definition \ref{(No-longer-needed-indices)},
without considering any sorting technique. Hence, the computational
complexity of the NLN algorithm is not optimized and ranges from lower
bound $O(\nstate)$ to upper bound $O(\nstate\ndata)$. If sorting
techniques are also applied, we can conjecture that the expectation
of NLN's complexity will be close to $m$. 

In the probability context, the output of NLN algorithm (Algorithm
\ref{alg:chap7:(NLN)-algorithm}) exactly corresponds to a valid probability
chain rule factorization for joint distribution, as explained in Section
\ref{sec:Re-interpretation-of-GDL}.
\begin{rem}
The two-directional topological NLN partition $\bbracket i$ and FA
partition $\bparent i$ of the universe index set $\OMEGA$, defined
in (\ref{eq:ch5:(i)-=00005Bi=00005D}), are novel. Likewise, the NLN
algorithm, whose purpose is to identify NLN and FA partitions of $\OMEGA$
induced by the factor index set algebra (\ref{eq:ch5:OMEGA:model}),
has not been proposed elsewhere in the literature. These partitions
will be important for our proposal to reduce computational load in
evaluating objective functions, later in the thesis.

A related algorithm to the NLN algorithm is the topological sorting
algorithm {[}\citet{ch5:BK:Bible:Algorithms}{]}, which returns a
topological ordering of a directed acyclic graph (DAG) (i.e. a valid
probability chain rule order). The key difference is that, the topological
sort permutes the order of factors $\omega_{\pi(1)},\ldots,\omega_{\pi(n)}$
until a valid chain rule order is achieved, while NLN algorithm maintains
the same order of factors $\omega_{1},\ldots,\omega_{n}$, but identifies
a new probability chain rule order via $\bbracket i,\ \seti{\itime}{\ndata}$,
extracted from $\omega_{1},\ldots,\omega_{n}$ (see Section \ref{sec:Re-interpretation-of-GDL})
\end{rem}

\begin{example}
\label{Occupancy Matrix} Let us assume that $m=5$, $n=4$ and $\omega_{4}=\{5,3,1\}$,
$\omega_{3}=\{4,3\}$, $\omega_{2}=\{3,2\}$, $\omega_{1}=\{2,1\}$.
Then, we can write down the occupancy matrix, as follows:

\begin{equation}
\begin{array}{c}
(\OMEGA_{5\times4})\\
5\\
4\\
3\\
2\\
1
\end{array}\left[\begin{array}{cccc}
\omega_{4} & \omega_{3} & \omega_{2} & \omega_{1}\\
([1]) & 0 & 0 & 0\\
0 & ([1]) & 0 & 0\\
{}[1] & 1 & (1) & 0\\
0 & 0 & [1] & (1)\\
{}[1] & 0 & 0 & (1)
\end{array}\right]\label{eq:Occupancy Matrix}
\end{equation}

where $[\cdot]$, $(\cdot)$ and $([\cdot])$ indicating elements
belonging to no-longer-needed (NLN) indices $\bbracket i$, first-appearance
(FA) indices $\bparent i$ and both of these sets, respectively. 

For illustration, the counter matrices, $C_{i}$, for NLN algorithm
sequentially are: $C_{5}=\OMEGA_{5\times4}$ and:

\[
\begin{array}{c}
(C_{4})\\
-\\
4\\
-\\
2\\
-
\end{array}\left[\begin{array}{cccc}
\omega_{4} & \omega_{3} & \omega_{2} & \omega_{1}\\
- & - & - & -\\
0 & [1] & 0 & 0\\
- & - & - & -\\
0 & 0 & 1 & 1\\
- & - & - & -
\end{array}\right],\ \begin{array}{c}
(C_{3})\\
-\\
-\\
-\\
2\\
-
\end{array}\left[\begin{array}{cccc}
\omega_{4} & \omega_{3} & \omega_{2} & \omega_{1}\\
- & - & - & -\\
- & - & - & -\\
- & - & - & -\\
0 & 0 & [1] & 1\\
- & - & - & -
\end{array}\right],\ \begin{array}{c}
(C_{2})\\
-\\
-\\
-\\
-\\
-
\end{array}\left[\begin{array}{cccc}
\omega_{4} & \omega_{3} & \omega_{2} & \omega_{1}\\
- & - & - & -\\
- & - & - & -\\
- & - & - & -\\
- & - & - & -\\
- & - & - & -
\end{array}\right]
\]

Hence, the NLN algorithm stops at $C_{2}=\emptyset$ and returns the
no-longer-needed (NLN) sets, as follows: $\bbracket 4=\omega_{4}=\{5,3,1\}$,
 $\bbracket 3=\{4\}$, $\bbracket 2=\{2\}$ and $\bbracket 1=\emptyset$.
\end{example}

$\ $
\begin{example}
\label{Occupancy Matrix-HMC} For later use, let us consider a canonical
(and simpler) example for a first-order Markov chain, with $m=5$,
$n=4$, i.e. we have $\omega_{4}=\{5,4\}$, $\omega_{3}=\{4,3\}$,
$\omega_{2}=\{3,2\}$, $\omega_{1}=\{2,1\}$. Then, similarly to Example
\ref{Occupancy Matrix}, we can write down the occupancy matrix, as
follows:

\begin{equation}
\begin{array}{c}
(\OMEGA_{5\times4})\\
5\\
4\\
3\\
2\\
1
\end{array}\left[\begin{array}{cccc}
\omega_{4} & \omega_{3} & \omega_{2} & \omega_{1}\\
([1]) & 0 & 0 & 0\\
{}[1] & (1) & 0 & 0\\
0 & [1] & (1) & 0\\
0 & 0 & [1] & (1)\\
0 & 0 & 0 & ([1])
\end{array}\right]\label{eq:OccupancyMatrix-HMC}
\end{equation}

where $[\cdot]$, $(\cdot)$ and $([\cdot])$ indicating elements
belonging to no-longer-needed (NLN) indices $\bbracket i$, first-appearance
(FA) indices $\bparent i$ and both of these sets, respectively. 

For illustration, the counter matrices, $C_{i}$, for NLN algorithm
sequentially are: $C_{5}=\OMEGA_{5\times4}$ and:

\[
\begin{array}{c}
(C_{4})\\
-\\
-\\
3\\
2\\
1
\end{array}\left[\begin{array}{cccc}
\omega_{4} & \omega_{3} & \omega_{2} & \omega_{1}\\
- & - & - & -\\
- & - & - & -\\
0 & [1] & 1 & 0\\
0 & 0 & 1 & 1\\
0 & 0 & 0 & 1
\end{array}\right],\ \begin{array}{c}
(C_{3})\\
-\\
-\\
-\\
2\\
1
\end{array}\left[\begin{array}{cccc}
\omega_{4} & \omega_{3} & \omega_{2} & \omega_{1}\\
- & - & - & -\\
- & - & - & -\\
- & - & - & -\\
0 & 0 & [1] & 1\\
0 & 0 & 0 & 1
\end{array}\right],\ \begin{array}{c}
(C_{2})\\
-\\
-\\
-\\
-\\
1
\end{array}\left[\begin{array}{cccc}
\omega_{4} & \omega_{3} & \omega_{2} & \omega_{1}\\
- & - & - & -\\
- & - & - & -\\
- & - & - & -\\
- & - & - & -\\
0 & 0 & 0 & [1]
\end{array}\right]
\]

Hence, the NLN algorithm stops at $C_{1}=\emptyset$ and returns the
no-longer-needed (NLN) sets, as follows: $\bbracket 4=\omega_{4}=\{5,4\}$,
$\bbracket 3=\{3\}$, $\bbracket 2=\{2\}$ and $\bbracket 1=\{1\}$.
\end{example}

\subsubsection{Ternary partition and in-process factors \label{subsec:chap5:Ternary-and-inprocess}}

.Let us denote $\bbracket{i:j}=\{\bbracket i,\ldots,\bbracket j\}$,
$\bparent{i:j}=\{\bparent i,\ldots,\bparent j\}$, $1\leq i\leq j\leq n$.
Then we have the following proposition:
\begin{prop}
(Ternary partition)\label{prop:(TernarySET)}\\
For any $\iton$, we can divide $\OMEGA=\oBold_{i+1:n}\cup\oBold_{1:i}$
into ternary partitions: 
\begin{equation}
\OMEGA=\{\bparent{i+1:n},\bcommon i,\bbracket{1:i}\}\label{eq:ch5:factor:Ternary}
\end{equation}
in which the (possibly empty) set $\bcommon i$ is called the set
of common indices (Fig. \ref{fig:Partitions-of-universe}): 
\begin{equation}
\bcommon i\TRIANGLEQ\oBold_{i+1:n}\cap\oBold_{1:i}=\bparent{1:i}\backslash\bbracket{1:i}\label{eq:chap5:eta-common-indices}
\end{equation}
Also, the left and right complement sets in (\ref{eq:ch5:factor:Ternary}),
respectively, are: 
\begin{eqnarray}
\bparent{i+1:n} & = & \oBold_{i+1:n}\backslash\bcommon i\label{eq:ch5:factor:left}\\
\bbracket{1:i} & = & \oBold_{1:i}\backslash\bcommon i\label{eq:ch5:factor:right}
\end{eqnarray}
\end{prop}

\begin{proof}
Firstly, let us recall the notation of the intersection part: $\bcommon i\TRIANGLEQ\oBold_{i+1:n}\cap\oBold_{i}$,
$\iton$. Then, the proof is carried out via following three steps:\textbf{\textit{}}\\
\textbf{\textit{Step 1:}} For proof of (\ref{eq:ch5:factor:right}):
by Definition \ref{(No-longer-needed-indices)} for NLN indices, we
have $\bbracket j\TRIANGLEQ\oBold_{j:n}\backslash\oBold_{j+1:n}$,
hence: 

\[
\bbracket{i+1:n}=\cup_{j=i+1}^{n}\bbracket j=\cup_{j=i+1}^{n}\{\oBold_{j:n}\backslash\oBold_{j+1:n}\}=\oBold_{i+1:n}
\]
where the last equality was derived via a sequence of merging, e.g.
$\{\oBold_{1:n}\backslash\oBold_{2:n}\}\cup\{\oBold_{2:n}\backslash\oBold_{3:n}\}=\oBold_{1:\ndata}$
. Then, since $\bbracket{i+1:n}=\oBold_{i+1:n}$, we also have:

\[
\bbracket{1:i}=\OMEGA\backslash\bbracket{i+1:n}=\OMEGA\backslash\oBold_{i+1:n}=\oBold_{1:i}\backslash\bcommon i
\]
where the last equality is a consequence of two basic set operators,
intersection and union, i.e.: $\bcommon i\TRIANGLEQ\oBold_{i+1:n}\cap\oBold_{1:i}$
and $\OMEGA=\oBold_{i+1:n}\cup\oBold_{1:i}$, respectively.\textbf{\textit{}}\\
\textbf{\textit{Step 2:}}\textbf{ }Similarly, for the proof of (\ref{eq:ch5:factor:left}):
by Definition \ref{(New-coming-indices)} of FA indices, we have: 

\begin{eqnarray*}
\bparent{1:i} & = & \cup_{j=1}^{i}\bparent j=\cup_{j=1}^{i}\{\oBold_{1:j+1}\backslash\oBold_{1:j}\}=\oBold_{1:i}\\
\Rightarrow\bparent{i+1:n} & = & \OMEGA\backslash\bparent{1:i}=\OMEGA\backslash\oBold_{1:i}=\oBold_{i+1:n}\backslash\bcommon i
\end{eqnarray*}
\textbf{\textit{Step 3:}} From above proofs of (\ref{eq:ch5:factor:left})
and (\ref{eq:ch5:factor:right}), the ternary partitions for $\OMEGA$
in (\ref{eq:ch5:factor:Ternary}) can be proved as follows: 

\begin{eqnarray*}
\OMEGA=\oBold_{i+1:n}\cup\oBold_{1:i} & = & \{\oBold_{i+1:n}\backslash\bcommon i,\bcommon i,\oBold_{1:i}\backslash\bcommon i\}\\
 & = & \{\bparent{i+1:n},\eta_{[i]},\bbracket{1:i}\},\ \iton
\end{eqnarray*}
\end{proof}
Proposition \ref{prop:(TernarySET)} also motivates a definition for
\textit{in-process} variable indices, as follows. These will be important
in recursive computation later:
\begin{prop}
(In-process indices)\\
In-process indices are defined as tri-partitioned index-sets, $\ASET_{i}$,
associated with the ternary partition $\OMEGA$ in (\ref{eq:ch5:factor:Ternary}),
as follows: 
\begin{eqnarray}
\ASET_{i} & \TRIANGLEQ & \{\bparent{i+1},\bcommon i,\bbracket i\}\label{eq:ch5:factor:in-process}\\
 & = & \mbox{\ensuremath{\omega_{i+1}\cup\omega_{i}}},\ \iton\nonumber 
\end{eqnarray}
Then, from (\ref{eq:ch5:factor:Ternary}) and (\ref{eq:ch5:factor:in-process}),
we have:

\begin{equation}
\OMEGA=\mathcal{A}_{1}\cup\cdots\cup\mathcal{A}_{n}\label{eq:ch5:OMEGA:unionA}
\end{equation}
\end{prop}

\begin{proof}
From Definition \ref{(No-longer-needed-indices)}, \ref{(New-coming-indices)}
and (\ref{eq:chap5:eta-common-indices}), we have: 

\begin{eqnarray*}
\ASET_{i} & \TRIANGLEQ & \{\bparent{i+1},\bcommon i,\bbracket i\}\\
 & = & \{\oBold_{1:i+1}\backslash\oBold_{1:i},\oBold_{i+1:n}\cap\oBold_{1:i},\oBold_{i:n}\backslash\oBold_{i+1:n}\}\\
 & = & \mbox{\ensuremath{\omega_{i+1}\cup\omega_{i}}},\ \iton
\end{eqnarray*}
as illustrated in lower-left schematic in Fig. \ref{fig:Partitions-of-universe}
for the case $i=3$.
\end{proof}
From (\ref{eq:ch5:factor:Ternary}), we can see that, in general,
the in-process index sets $\mathcal{A}_{i}$ in (\ref{eq:ch5:OMEGA:unionA})
are not separated partitions of $\OMEGA$. Note that, since $\bcommon i=\oBold_{i+1:n}\cap\oBold_{1:i}$,
it might be empty (e.g. when $\omega_{1},\ldots,\omega_{n}$ are separated
partitions of $\OMEGA=\{\omega_{1},\ldots,\omega_{n}\}$). In contrast,
the set $\ASET_{i}$, defined by (\ref{eq:ch5:factor:in-process}),
cannot be empty, since $\omega_{i}\neq\emptyset$, $\iton$, by definition
(Section \ref{subsec:chap5:Objective-functions}). Moreover, different
from union of tri-partitioned sets $\ASET_{i}$ in (\ref{eq:ch5:OMEGA:unionA}),
the union of common sets $\bcommon 1\cup\cdots\cup\bcommon n$ is
not necessary equal to $\OMEGA$. 

For these reasons, we prefer to compute the operators in (\ref{eq:ch5:SUM:sequential})
via tri-partitioned sets $\ASET_{i}$ instead of common sets $\bcommon i$.

\subsection{Separated indices of operators}

\subsubsection{Binary partition }

From (\ref{eq:ch5:OMEGA:single}) and (\ref{eq:ch5:(i)-=00005Bi=00005D}),
we can distribute the single operator index set $\SET$ and objective
index set $\SET^{c}$ across two $n$ alternative partitions, as follows:
\begin{eqnarray*}
\SET & = & \{s_{\bparent 1},\ldots,s_{\bparent n}\}\\
 & = & \{s_{\bbracket 1},\ldots,s_{\bbracket n}\}\\
\SET^{c} & = & \{s_{\bparent 1}^{c},\ldots,s_{\bparent n}^{c}\}\\
 & = & \{s_{\bbracket 1}^{c},\ldots,s_{\bbracket n}^{c}\}
\end{eqnarray*}
where, $\iton$, and: 
\begin{eqnarray}
s_{\bparent i} & \TRIANGLEQ & \SET\cap\bparent i\label{eq:ch5:s(i)_s=00005Bi=00005D}\\
s_{\bbracket i} & \TRIANGLEQ & \SET\cap\bbracket i\nonumber \\
s_{\bparent 1}^{c} & \TRIANGLEQ & \bparent i\backslash s_{\bparent i}\nonumber \\
s_{\bbracket 1}^{c} & \TRIANGLEQ & \bbracket i\backslash s_{\bbracket i}\nonumber 
\end{eqnarray}
as illustrated in lower-right schematic in Fig. \ref{fig:Partitions-of-universe}.

\subsubsection{Ternary partition and in-process operators}

Proceeding as in Section \ref{subsec:chap5:Ternary-and-inprocess},
then, from (\ref{eq:ch5:OMEGA:single}) and (\ref{eq:ch5:factor:Ternary}),
the ternary partitions of $\SET$ and $\SET^{c}$ can be defined,
respectively, as follows: 
\begin{eqnarray}
\SET & = & \{s_{\bparent{i+1:n}},s_{\bcommon i},s_{\bbracket i}\}\label{eq:ch5:SET:ternary}\\
\SET^{c} & = & \{s_{\bparent{i+1:n}}^{c},s_{\bcommon i}^{c},s_{\bbracket i}^{c}\},\nonumber 
\end{eqnarray}
where, $\iton$, and: 
\begin{eqnarray}
s_{\bcommon i} & \TRIANGLEQ & \SET\cap\bcommon i\label{eq:ch5:s_eta}\\
s_{\bcommon i}^{c} & \TRIANGLEQ & \bcommon i\backslash s_{\bcommon i}\nonumber 
\end{eqnarray}
as illustrated in lower-right schematic in Fig. \ref{fig:Partitions-of-universe}.

\subsubsection{Non-overflowed (NOF) condition \label{subsec:Sub-selection-of-set}}

For sequential objective functions (\ref{eq:ch5:OMEGA:sequential}),
let us consider a sequence of subsets $\SET_{j}$, such that $\SET=\SET_{1}\cup\cdots\cup\SET_{\noperator},$
and $\SET_{j}\subset\SET\subseteq\OMEGA$, $j=1,\ldots,\noperator$.
Denoting objective set $\Ssubstract j\neq\emptyset$ as complement
of $\SET_{j}$ in $\SET$, we will consider a special case of $\Ssubstract j$,
as follows:
\begin{defn}
(\textit{Non-overflowed (NOF) set and NOF condition}) \label{(Non-overflowed-set)}\\
A set $\Ssubstract j$, with $j\in\{1,\ldots,\noperator\}$, is called
non-overflowed (NOF) set if the following NOF condition is satisfied:
\begin{equation}
\exists i\in\{1,\ldots,n\}:\ \Ssubstract j\subseteq\SET_{\ASET_{i}}\TRIANGLEQ\{s_{\bparent{i+1}},s_{\bcommon i},s_{\bbracket i}\}\subseteq\SET,\ j\in\{1,\ldots,\noperator\}\label{eq:ch5:non-overflow:set}
\end{equation}
Otherwise, $\Ssubstract j$ is called an overflowed set. 
\end{defn}

The meaning of the name is as follows: given a tri-partitioned set
$\ASET_{i}$ (\ref{eq:ch5:factor:in-process}) being in process/memory
at point $\iinn$, the objective set $\Ssubstract j\subseteq\SET_{\ASET_{i}}$
can be extracted within current process/memory $\ASET_{i}$ and does
not cause extra/overflowed work. 

For later use, let us also define the complement of $\Ssubstract j$
in $\SET_{\ASET_{i}}$, as follows: 
\begin{eqnarray}
\SET_{\ASET_{i}}^{\backslash j} & \TRIANGLEQ & \SET_{\ASET_{i}}\backslash\Ssubstract j\label{eq:ch5:non-overflow:complement}\\
 & = & \SET_{\ASET_{i}}\cap S_{j}\nonumber \\
 & = & \{s_{\bparent{i+1}}^{\backslash j},s_{\bcommon i}^{\backslash j},s_{\bbracket i}^{\backslash j}\}\nonumber 
\end{eqnarray}
in which the intersection is carried out element-wise. Finally, let
us consider the most special case of \textit{non-overflowed }set,
which is scalar set, as follows:
\begin{prop}
\label{prop:Scalar-non-overflowed}(Objective scalar sets) \\
The singleton sets, $\Ssubstract j=\{j\}\in\SET=\OMEGA$, $j=1,\ldots,\noperator$,
where $\noperator=m$, are non-overflowed sets.
\end{prop}

\begin{proof}
Firstly, by definition of $\omega_{i}$, there always exists $i\in\{1,\ldots,n\}$
such that: 
\begin{equation}
\Ssubstract j=\{j\}\in\omega_{i+1}\cup\omega_{i}=\ASET_{i},\ \exists i\in\{1,\ldots,n\},\forall j\in\{1,\ldots,\noperator\}\label{eq:ch5:scalar_proof:a}
\end{equation}
where the last equality is owing to (\ref{eq:ch5:factor:in-process}). 

Secondly, owing to the assumption $\noperator=m$, we have $\SET=\cup_{j=1}^{m}\Ssubstract j=\cup_{j=1}^{m}\{j\}=\OMEGA$.

Thirdly, by definitions in (\ref{eq:ch5:factor:in-process},\ref{eq:ch5:non-overflow:set}),
we have:

\begin{equation}
\ASET_{i}=\{\bparent{i+1},\bcommon i,\bbracket i\}=\{s_{\bparent{i+1}},s_{\bcommon i},s_{\bbracket i}\}=\SET_{\ASET_{i}},\ \forall i\in\{1,\ldots,m\}\label{eq:ch5:scalar_proof:b}
\end{equation}
owing to $\SET=\OMEGA$ above and (\ref{eq:ch5:s(i)_s=00005Bi=00005D}-\ref{eq:ch5:s_eta}). 

Finally, from (\ref{eq:ch5:scalar_proof:a}) and (\ref{eq:ch5:scalar_proof:b}),
we have $\Ssubstract j\subseteq\SET_{\ASET_{i}}\subseteq\SET=\OMEGA,\ \exists i\in\{1,\ldots,n\},$
which satisfies the NOF condition (\ref{eq:ch5:non-overflow:set}),
$\forall j\in\{1,\ldots,\noperator\}$ and $\noperator=m$. 
\end{proof}
Note that, the assumption $\noperator=m$ is useful in Proposition
\ref{prop:Scalar-non-overflowed}, since it facilitate the verification
of (\ref{eq:ch5:scalar_proof:b}) in the proof. In Bayesian analysis,
the assumption $\noperator=m$ is often valid when a sequence of all
marginals needs to be computed (e.g for computing minimum risk estimator
computation (\ref{eq:ch4:MR:marginalMAP})). For the case $\noperator=\noperator_{0}<m$
but $\noperator_{0}$ is close enough to $m$, we may consider the
augmented case $\noperator=m$, i.e. $S=\OMEGA$,  like above, but
only return the results corresponding to the NOF sets $\Ssubstract j$,
$j\in\OMEGA\backslash\{\noperator_{0}+1,\ldots,m\}$. 

\section{Generalized distributive law (GDL)}

In this section, let us recall the abstract algebra and exploit the
distributive law associated with ring theory. For clarity, let us
firstly emphasize that ring theory can be applied feasibly to a set
of either variables or functions, i.e.:

- In abstract algebra {[}\citet{ch5:bk:ring:Handbook95}{]} and in
graph theory {[}\citet{ch2:origin:GDL:McEliece}{]}, the standard
definitions of a group and a ring begin with a standard set $\RING$,
associated with binary operators. 

- Because the point-wise value of any function $g:\calX\rightarrow\RING$
can be considered as a set $\RING$, for each $x\in\calX$, the above
standard definition of group and ring can be applied to a set of function
values $g(x)$ in the same way as standard set $\RING$ {[}\citet{ch5:bk:Ring:continuous60}{]}.
Hence, the GDL in ring theory can be applied directly to all functions
$g_{i}$ in (\ref{eq:ch5:MODEL}) {[}\citet{ch5:bk:ring:Gondran:Minoux:08}{]},
without the need of re-definition. 

Without loss of generality, let us avoid the former approach, which
is used in {[}\citet{ch2:origin:GDL:McEliece}{]}, and begin with
latter approach, i.e. abstract algebra for a set of functions. As
shown below, the latter is more general and actually an important
step in algorithm design, because the set of functions will clarify
the number of operators required in computing the objective functions
(\ref{eq:ch5:MODEL}).

\subsection{Abstract algebra \label{subsec:Abstract-algebra}}

Let us consider a set $\Ring$ of functions $g:\calX_{\OMEGA}\rightarrow\RING$.
Then we have following definitions:
\begin{defn}
(Commutative semigroup of functions) \label{Semigroup}\\
A commutative semigroup $(\Ring,\boxtimes)$ is a set $\Ring$ of
$\RING$-value functions on domain $\calX_{\OMEGA}$, associated with
closure binary operator $\boxtimes:\Ring\times\Ring\rightarrow\Ring:f(\xBold)\boxtimes g(\xBold)\rightarrow q(\xBold)\TRIANGLEQ(f\boxtimes g)(\xBold)$,
such that following two properties hold:\\
Associative Law: $f\boxtimes(g\boxtimes h)=(f\boxtimes g)\boxtimes h$\\
Commutative Law: $f\boxtimes g=g\boxtimes f$ \\
for any $\xBold\in\calX_{\OMEGA}$ and $f,g,h,q\in\Ring$. 
\end{defn}

$\ $
\begin{defn}
(Commutative pre-semiring of functions) \label{Pre-semiring}\\
A commutative pre-semiring $(\Ring,\boxplus,\odot)$ is a set $\Ring$
of $\RING$-value functions on domain $\spaceX$, associated with
two commutative semigroups $(\Ring,\boxplus)$ and $(\Ring,\odot)$,
such that following two properties hold:\\
Priority order of operation: $f\odot g\boxplus h=(f\odot g)\boxplus h\neq f\odot(g\boxplus h)$\\
Distributive law: $f\odot(g\boxplus h)=(f\odot g)\boxplus(f\odot h)$\\
for $f,g,h\in\Ring$.
\end{defn}

Note that, in abstract algebra, the above concept of pre-semiring
(Definition \ref{Pre-semiring}) and traditional semiring, which additionally
requires the existence of two identity elements (one for $\boxplus$
and one for $\odot$), are not equivalent  {[}\citet{ch5:bk:ring:Handbook95}{]}.
In graph theory, the definition of semiring together with the identity
elements is widely used in many papers {[}\citet{ch5:art:GDL:ViterbiAlgo90,ch2:origin:GDL:McEliece,ch5:art:MP:Entropy11}{]}.
Therefore, we use the name ``pre-semiring'', as proposed in {[}\citet{ch5:origin:PreSemiRing:Minoux01,ch5:bk:ring:Gondran:Minoux:08}{]},
and consider semiring as a special case of a pre-semiring. In contrast
to a semiring structure, the pre-semiring is more relaxed and does
not require the existence of the identity elements.

\subsection{Ring-sum and ring-product \label{subsec:chap5:Ring-sum-and-ring-product}}

Let us denote $\calX_{\omega}\subseteq\spaceX$ as a subspace of $\spaceX$,
with $\omega\subseteq\OMEGA$, as in Section \ref{subsec:chap5:Objective-functions}.
By generalizing the range from the set of real number $\REAL$ to
an arbitrary set $\RING$, the functions in our model (\ref{eq:ch5:MODEL})
become $g_{i}:\ \calX_{\omega_{i}}\rightarrow\RING$. Moreover, it
is feasible to recognize that $g_{i}\in\Ring$, since $\RING^{\calX_{\omega_{i}}}\subseteq\Ring$.
Consequently, we can apply the abstract algebra above in the pre-semiring
$(\Ring,\boxplus,\odot)$ to all functions, $g_{i}$. Note that, by
the above definitions from ring theory, the closure property of any
operators $\boxplus$, $\odot$ applied to $g_{\itime}$ is guaranteed
within the space $\Ring$, but not guaranteed within its sub-space
$\RING^{\calX_{\omega_{i}}}$.

For computation of our specific model (\ref{eq:ch5:MODEL}), we need,
however, to define some specific ring-sum and ring-product operators
properly. These definitions also clarify the number of required operators
in the sequel.
\begin{defn}
(Ring-sum) \label{DEF=00003DRing-sum}\\
Let us consider $g\in\RING^{\calX_{\omega}}$. Given $i\in\omega\subseteq\OMEGA$
and $\calX\TRIANGLEQ\{\bar{x}_{1},\ldots,\bar{x}_{M}\}$, then the
$\itime$th ring-sum is defined as:\\
$\underset{x_{i}}{\boxplus}g(\xBold_{\omega})=g_{\bar{x}_{1}}\boxplus g_{\bar{x}_{2}}\boxplus\cdots\boxplus g_{\bar{x}_{M}}$\\
where $g_{\bar{x}_{j}}\TRIANGLEQ g(x_{i}=\bar{x}_{j},\xBold_{\omega\backslash i})$
with $x_{i}\in\calX$, $j=1,\ldots,M$ and, as before, $M\TRIANGLEQ|\calX|$
(\ref{eq:ch5:DEF:M}) \\
In this way, we can define a projection function: $\underset{x_{S}}{\boxplus}:\ \RING^{\calX_{\omega}}\rightarrow\RING^{\calX_{\omega\backslash S}}$,
$\forall\SET=\{s_{1},\ldots,s_{|\SET|}\}\subseteq\omega$, as follows:

\begin{equation}
\underset{\xBold_{\SET}}{\boxplus}g(\xBold_{\omega})=\underset{x_{s_{|\SET|}}}{\boxplus}\cdots\underset{x_{s_{1}}}{\boxplus}g(\xBold_{\omega})\label{eq:ch5:DEF:Ring-sum}
\end{equation}
 
\end{defn}

$\ $
\begin{defn}
(Ring-product)\label{DEF=00003DRing-product}\\
Ring-product is an augmented function: $\odot_{k=i}^{j}:\ \RING^{\calX_{\omega_{i}}}\times\cdots\times\RING^{\calX_{\omega_{j}}}\rightarrow\RING^{\calX_{\omega}}$,
where $\omega=\omega_{i}\cup\cdots\cup\omega_{j}$, defined as follows:

\begin{equation}
\odot_{k=i}^{j}g_{k}=g_{j}\odot g_{j-1}\odot\cdots\odot g_{i}\label{eq:ch5:DEF:Ring-prod}
\end{equation}
\end{defn}

.

From the definitions above, we will apply the distributive law assumption
of pre-semiring in Definition \ref{Pre-semiring} to our ring-operators,
as follows:
\begin{defn}
(Generalized Distributive Law) \label{DEF=00003DGDL}\\
By definition of distributive law of pre-semiring in Definition \ref{Pre-semiring},
the generalized distributive law (GDL) \textit{ for elements} can
be defined as (from left to right):

\begin{equation}
g_{2}\odot(\underset{x_{S}}{\boxplus}g_{1})=\underset{x_{S}}{\boxplus}(g_{2}\odot g_{1}),\ \forall S\not\subset\omega_{2}\label{eq:ch5:GDL:element}
\end{equation}
Then, the reverse flow (from right to left) of (\ref{eq:ch5:GDL:element})
is called the generalized distributive law (GDL)\textit{ for operators},
i.e.:

\begin{equation}
\underset{x_{S}}{\boxplus}(g_{2}\odot g_{1})=g_{2}\odot(\underset{x_{S}}{\boxplus}g_{1}),\ \forall S\not\subset\omega_{2}\label{eq:ch5:GDL:op}
\end{equation}
\end{defn}

Notice the interesting difference in notation between (\ref{eq:ch5:GDL:element})
and (\ref{eq:ch5:GDL:op}). In ``GDL for elements'' (\ref{eq:ch5:GDL:element}),
it seems that the element $g_{2}$ is distributed across the operators
$\boxplus$, while in ``GDL for operators'' (\ref{eq:ch5:GDL:op}),
it looks like the operator $\underset{x_{\SET}}{\boxplus}$ is actually
the one being distributed. 

The ``GDL for elements'' (\ref{eq:ch5:GDL:element}) - for example,
$a_{3}(a_{1}+a_{2})=a_{3}a_{1}+a_{3}a_{2}$ - is more consistent with
the traditional distributive law in mathematics textbooks. However,
in practice, the ``GDL for operators'' (\ref{eq:ch5:GDL:op}) -
for example, $\sum_{i=1}^{2}(a_{3}a_{i})=a_{3}\left(\sum_{i=1}^{2}a_{i}\right)$
- is more popular and constitutes the definition of the GDL form,
e.g. in {[}\citet{ch2:origin:GDL:McEliece,ch5:PhD:Nielsen01}{]}).

In this thesis, the latter is preferred, i.e. ``GDL for operators''
(\ref{eq:ch5:GDL:op}) will be defined as GDL, although both terms
``GDL for elements'' and ``GDL for operators'' are obviously equivalent. 

\subsection{Computational reduction via the GDL \label{subsec:ch5:Computational-reduction-via-GDL}}

Let us investigate the computational load below by counting the number
of operators $\boxplus$ and $\odot$ involved in the ring-sum (\ref{eq:ch5:DEF:Ring-sum}),
ring-product (\ref{eq:ch5:DEF:Ring-prod}) and GDL (\ref{eq:ch5:GDL:op}). 

Intuitively, the computational reduction, from left hand side to right
hand side  in GDL (\ref{eq:ch5:GDL:op}), comes from the order of
the computation of projection (i.e. ring-sum), and augmentation (i.e.
ring-product). On the left hand side of (\ref{eq:ch5:GDL:op}), the
augmentation is implemented first, followed by projection, while the
order is reversed on the right hand side  of (\ref{eq:ch5:GDL:op}),
i.e projection is implemented first, followed by augmentation. Therefore,
the right hand side of (\ref{eq:ch5:GDL:op}) is more efficient because
it rules out the nuisance variables $x_{S}$ upfront (i.e. immediately
after recognizing that $x_{S}$ is no-longer-needed). In contrast,
the left hand side of (\ref{eq:ch5:GDL:op}) is less efficient because
the nuisance variables $x_{S}$ are dragged along in the augmentation
operator, $g_{2}\odot g_{1}$, and increases the number of computations
required for the final result. 

In practice, however, this reduction does not always imply an actual
reduction in computational load. For example, the quantities involved
in the left hand side, $\underset{x_{\SET}}{\boxplus}(g_{2}\odot g_{1})$
might be known beforehand and retrieved via table-lookup, while the
right hand side quantities might be unknown and hence requires computation.
Nevertheless, the general interest is to count the number of operators,
because that number is typically assumed to be proportional to computational
load in practice.

\subsubsection{Computational load in ring-sum and ring-product \label{subsec:chap5:cost-ring-sum}}

For convenience, let us denote $\calO[\cdot]$ as a counting procedure
returning the order of the number $O_{\boxtimes}$ of operators, $\boxtimes$,
applied on function in brackets, $[\cdot]$. 
\begin{lem}
\label{lem:COST} (Computational load in ring-sum and ring-product)\\
a) The computational load of a ring-sum only depends on dimension
$m$ of its function $g\in\RING^{\calX_{\omega}}$:

\[
\calO[\underset{\xBold_{S}}{\boxplus}g(x_{\omega})]=O_{\SUM}(M^{|\omega|}),\ \forall S\subseteq\omega
\]
b) The computational load of ring-product depends on the combined
dimension of its two functions:

\textup{
\[
\calO[g_{2}\PROD g_{1}]=O_{\PROD}(M^{|\omega_{1}\cup\omega_{2}|})
\]
}
\end{lem}

\begin{proof}
a) For ring-sum, $\underset{x_{i}}{\boxplus}$, in Definition \ref{DEF=00003DRing-sum},
we need $O_{\SUM}(M^{|\omega|})$ of operator $\SUM$. Hence, for
$\underset{\xBold_{\SET}}{\boxplus}$, we need $O_{\SUM}(\sum_{i=|\omega\backslash\SET|}^{|\omega|}M^{i})=O_{\SUM}(M^{|\omega|})$
of operator $\SUM$.

b) Because the result of binary ring-product $g_{1}(\xBold_{\omega_{1}})\PROD g_{2}(\xBold_{\omega_{2}})$
is a function with union domains $g(\xBold_{\omega_{1}\cup\omega_{2}})$,
it requires $O(M^{|\omega_{1}\cup\omega_{2}|})$ of operator $\PROD$
for computing the $M^{|\omega_{1}\cup\omega_{2}|}$ values of $\xBold_{\omega_{1}\cup\omega_{2}}$. 
\end{proof}

\subsubsection{Computational load in GDL}

From Lemma \ref{lem:COST}, we can prove the main theorem of this
chapter, as follows: 
\begin{thm}
\label{thm:GDL} The GDL for operators (\ref{eq:ch5:GDL:op}), if
applicable, always reduces the number of operators; i.e.:

\textup{
\begin{equation}
\calO[\underset{x_{S}}{\boxplus}(g_{2}\odot g_{1})]>\calO[g_{2}\odot(\underset{x_{S}}{\boxplus}g_{1})],\ \forall S\not\subset\omega_{2}\label{eq:ch5:GDL:op:big-O}
\end{equation}
}
\end{thm}

\begin{proof}
From (\ref{eq:ch5:GDL:op}) and Lemma \ref{lem:COST}, we have: 
\begin{eqnarray*}
\calO[\underset{x_{\SET}}{\boxplus}(g_{2}\odot g_{1})] & = & \calO[\underset{x_{\SET}}{\boxplus}\gBold_{1:2}]+\calO[g_{2}\odot g_{1}]\\
 & = & O_{\SUM}(M^{|\omega_{1}\cup\omega_{2}|})+O_{\PROD}(M^{|\omega_{1}\cup\omega_{2}|})
\end{eqnarray*}
where $\gBold_{1:2}=g_{2}\odot g_{1}$ is regarded as a point-wise
function $\gBold_{1:2}:\calX_{\omega_{1}\cup\omega_{2}}\rightarrow\RING$.
Likewise, we have:

\begin{eqnarray*}
\calO[g_{2}\odot(\underset{x_{\SET}}{\boxplus}g_{1})] & = & \calO[\underset{x_{\SET}}{\boxplus}g_{1}]+\calO[g_{2}\odot g(\xBold_{\omega_{1}\backslash\SET})]\\
 & = & O_{\SUM}(M^{|\omega_{1}|})+O_{\PROD}(M^{|\{\omega_{1}\backslash\SET\}\cup\omega_{2}|})
\end{eqnarray*}
Since $|\omega_{1}\cup\omega_{2}|\geq|\omega_{1}|$ and $|\omega_{1}\cup\omega_{2}|\geq|\{\omega_{1}\backslash\SET\}\cup\omega_{2}|$,
the following non-strict inequality follows:

\begin{equation}
\calO[\underset{x_{\SET}}{\boxplus}(g_{2}\odot g_{1})]\geq\calO[g_{2}\odot(\underset{x_{\SET}}{\boxplus}g_{1})]\label{eq:ch5:Proof=00003DGDL:op}
\end{equation}
Equality requires both $|\omega_{1}\cup\omega_{2}|=|\omega_{1}|$
and $|\omega_{1}\cup\omega_{2}|=|\{\omega_{1}\backslash\SET\}\cup\omega_{2}|$,
i.e. $\omega_{2}\subseteq\omega_{1}$ and $\SET\subseteq\omega_{1}\cap\omega_{2}$,
respectively. However, because the condition for GDL is $S\not\subset\omega_{2}$,
or, equivalently $\SET\subseteq\omega_{1}\backslash\{\omega_{1}\cap\omega_{2}\}$,
the equality in (\ref{eq:ch5:Proof=00003DGDL:op}) only happens if
the GDL (\ref{eq:ch5:GDL:op}) is invalid.
\end{proof}
In Theorem \ref{thm:GDL}, the GDL is recognized, for the first time,
as \textit{always} reducing the number of operators in any pre-semiring.
In the literature, there is no result guaranteeing such a reduction
in GDL for all cases in ring theoretic class. The GDL is only explicitly
recognized to reduce the number of operators in the trivial cases,
such as $a(b+c)=ab+ac$ in mathematical textbooks, and is instead
applied to specific models in order to evaluate the reduction on a
case-by-case basis {[}\citet{ch2:origin:GDL:McEliece,ch2:bk:ToddMoon,ch5:art:MP:EntropyRing08,ch5:art:MP:Entropy11}{]}. 

Moreover, Lemma \ref{lem:COST} provides an explicit formula for counting
the number of operators in GDL (\ref{eq:ch5:GDL:op:big-O}) via set
operators, rather than via a complicated graphical topology in {[}\citet{ch2:origin:GDL:McEliece,ch2:bk:ToddMoon}{]}. 

\section{GDL for objective functions \label{sec:chap5:GDL-for-objective}}

In this section, we will propose a novel recursive technique, called
forward-backward (FB) recursion, for computing the objective functions
in (\ref{eq:ch5:SUM:single}) and (\ref{eq:ch5:SUM:sequential}).
We call this technique FB, owing to its similarity with well known
FB algorithm for Hidden Markov Chain in the literature. This similarity
will be clarified in Section \ref{subsec:chap6:FB-HMC} of this thesis. 

\subsection{FB recursion for single objective function \label{sec:Computation-via-distributive}}

After setting up the pre-semiring, $(\RING^{X_{m}},\SUM,\PROD)$,
our aim is to compute the single objective function (\ref{eq:ch5:SUM:single}):

\begin{equation}
\underset{\xBold_{\SET}}{\boxplus}\gBold_{1:n}=\underset{\xBold_{\SET}}{\boxplus}(\odot_{i=1}^{n}g_{i})\label{eq:SUM-PRODUCT}
\end{equation}
with $\SET\subseteq\OMEGA$. From Lemma \ref{lem:COST}, we can see
that the cost of direct computation on the right hand side of (\ref{eq:SUM-PRODUCT})
is exponentially increasing with $m$, i.e. $O(M^{m})$, which is
impractical. Because the GDL (\ref{eq:ch5:GDL:op:big-O}) always reduces
the cost, we will exploit the conditionally independent (CI) structure
in (\ref{eq:SUM-PRODUCT}) and apply the GDL, as shown below.

\subsubsection{Binary tree factorization}

The joint ring-products (\ref{eq:ch5:MODEL}) can be written as:

\begin{equation}
\gBold_{1:n}=\gBold_{i+1:n}\PROD\gBold_{1:i},\ \ \ i\in\{1,\ldots,n-1\}\label{eq:BinaryTree=00003DFB}
\end{equation}
where $\gBold_{1:i}\TRIANGLEQ\odot_{j=1}^{i}g_{j}$ and $\gBold_{i:n}\TRIANGLEQ\odot_{j=i}^{n}g_{j}$,
together with forward and backward recursions:

\begin{eqnarray}
\gBold_{1:j} & = & g_{j}\PROD\gBold_{1:j-1},\ \ \ j=2,\ldots,i\label{eq:BinaryTree=00003DFW}\\
\gBold_{j:n} & = & g_{j}\PROD\gBold_{j+1:n},\ \ \ j=n-1,\ldots,i+1\label{eq:BinaryTree=00003DBW}
\end{eqnarray}
The domains of the functions in (\ref{eq:BinaryTree=00003DFB}-\ref{eq:BinaryTree=00003DBW})
are, respectively:

\begin{eqnarray}
dom(\gBold_{1:j}) & = & \{\xMath_{\bcommon j},\xBold_{\bbracket{1:j}}\}\label{eq:domain=00003DBinaryTree}\\
dom(\gBold_{j+1:n}) & = & \{\xBold_{\bparent{j+1:n}},\xMath_{\bcommon j}\}\nonumber \\
dom(\gBold_{1:n})=\xBold_{\OMEGA} & = & \{\xBold_{\bparent{j+1:n}},\xMath_{\bcommon j},\xBold_{\bbracket{1:j}}\}\nonumber 
\end{eqnarray}
using the notation defined in Proposition \ref{prop:(TernarySET)}.
Then, we will see below that operators, $\SUM$, upon the joint $\gBold_{1:n}$
can be distributed into the binary tree structure (\ref{eq:BinaryTree=00003DFB}),
owing to the GDL-for-operators form (\ref{eq:ch5:GDL:op}).

\subsubsection{Derivation of the FB recursion}

Let us denote $\widehat{\mathbf{g}}_{1:n}\TRIANGLEQ\underset{\xBold_{\SET}}{\boxplus}\gBold_{1:n}=g(\xBold_{S^{c}})$.
Owing to the separable domains in (\ref{eq:domain=00003DBinaryTree}),
we can apply the GDL (\ref{eq:ch5:GDL:op}) to computation of operator
$\underset{\xBold_{\SET}}{\boxplus}$ on (\ref{eq:BinaryTree=00003DFB}),
as follows:

\begin{equation}
\widehat{\mathbf{g}}_{1:n}=\underset{\xMath_{s_{\bcommon i}}}{\boxplus}(\gHat_{i+1:n}\odot\gHat_{1:i}),\ \ \ i\in\{1,\ldots,n-1\}\label{eq:Ternary=00003DGDL}
\end{equation}
where $\mbox{\ensuremath{\gHat_{1:i}}}\TRIANGLEQ\underset{\xBold_{s_{\bbracket{1:i}}}}{\boxplus}\gBold_{1:i}$
and $\gHat_{i:n}\TRIANGLEQ\underset{\xBold_{s_{\bparent{i:n}}}}{\boxplus}\gBold_{i:n}$
, together with forward and backward recursions:

\begin{eqnarray}
\gHat_{1:j} & =\underset{\xMath_{s_{\bbracket j}}}{\boxplus} & (g_{j}\odot\gHat_{1:j-1}),\ \ \ j=2,\ldots,i\label{eq:ch5:GDL:FW}\\
\gHat_{j:n} & =\underset{\xMath_{s_{\bparent j}}}{\boxplus} & (g_{j}\odot\gHat_{j+1:n}),\ \ \ j=n-1,\ldots,i+1\label{eq:ch5:GDL:BW}
\end{eqnarray}

The domain of functions in (\ref{eq:Ternary=00003DGDL}-\ref{eq:ch5:GDL:BW})
are:

\begin{eqnarray}
dom(\mbox{\ensuremath{\gHat_{1:j}}}) & = & \{\xMath_{\bcommon j},\xBold_{s_{\bbracket{1:j}}^{C}}\}\label{eq:domain=00003DFB}\\
dom(\gHat_{j+1:n}) & = & \{\xBold_{s_{\bparent{j+1:n}}^{C}},\xMath_{\bcommon j}\}\nonumber \\
dom(\widehat{\mathbf{g}}_{1:n})=\xBold_{\SET^{c}} & = & \{\xBold_{s_{\bparent{i+1:n}}^{C}},\xMath_{s_{\bcommon i}^{c}},\xBold_{s_{\bbracket{1:i}}^{C}}\}\nonumber 
\end{eqnarray}

Examining (\ref{eq:domain=00003DBinaryTree}) with (\ref{eq:domain=00003DFB}),
we can see that the variable indices in these domains can be feasibly
derived via the set algebras in (\ref{eq:ch5:factor:Ternary}) and
(\ref{eq:ch5:SET:ternary}), respectively. It is also important to
emphasize that, for computing a single objective function, $\widehat{\mathbf{g}}_{1:n}$
in (\ref{eq:Ternary=00003DGDL}), the value $i\in\{1,\ldots,n-1\}$
can be chosen arbitrarily, since each yields the same result $\widehat{\mathbf{g}}_{1:n}$.
These choices, however, may differ in computational load. 

\subsubsection{Intermediate steps}

Note that, when computing FB recursion (\ref{eq:Ternary=00003DGDL}-\ref{eq:ch5:GDL:BW}),
we have already computed the following intermediate results:

\begin{eqnarray}
\gamma_{i} & = & \gHat_{i+1:n}\odot\gHat_{i}\label{eq:GDL=00003DSecondOrder}\\
\gbar_{1:j} & = & g_{j}\odot\gHat_{1:j-1},\ \ \ \mbox{for}\ j=1,\ldots,i\nonumber \\
\gbar_{j:n} & = & g_{j}\odot\gHat_{j+1:n},\ \ \ \mbox{for}\ j=n,\ldots,i+1\nonumber 
\end{eqnarray}
with respect to the chosen index, $i\in\{1\ldots,n\}$. The domains
of the functions in (\ref{eq:GDL=00003DSecondOrder}) are, respectively:
\begin{eqnarray}
dom(\gamma_{i}) & = & \{\xBold_{s_{\bparent{i+1:n}}^{C}},\xMath_{\bcommon i},\xBold_{s_{\bbracket{1:i}}^{C}}\}\label{eq:domain=00003DGDL}\\
dom(\gbar_{1:j}) & = & \{\xMath_{\bcommon j},\xMath_{\bbracket j},\xBold_{s_{\bbracket{1:j-1}}^{C}}\}\nonumber \\
dom(\gbar_{j+1:n}) & = & \{\xBold_{s_{\bparent{j+2:n}}^{C}},\xMath_{\bparent{j+1}},\xMath_{\bcommon j}\}\nonumber 
\end{eqnarray}
which are set combinations of the two sets (\ref{eq:domain=00003DBinaryTree})
and (\ref{eq:domain=00003DFB}). 

These intermediate results (\ref{eq:GDL=00003DSecondOrder}) will
be useful for evaluating the computational load of FB recursion. They
are also valuable resources for efficient computation in sequential
objective functions, which we will consider in Section \ref{sec:Sequence-of-sets}.

\subsubsection{Computational load for a single objective function via FB}

From (\ref{eq:Ternary=00003DGDL}-\ref{eq:ch5:GDL:BW}), we can see
that the computational load for a specific $i\in\{1,\ldots,n\}$ depends
on three steps, one forward recursion (\ref{eq:ch5:GDL:FW}), one
backward recursion (\ref{eq:ch5:GDL:BW}) and one combination of these
(\ref{eq:Ternary=00003DGDL}). Then, the numbers of ring-sum $\SUM$
and ring-product $\PROD$ via the FB recursion (\ref{eq:Ternary=00003DGDL}-\ref{eq:ch5:GDL:BW})
are both equal to $\calO[\underset{\xBold_{S}}{\boxplus}\gBold_{1:n}]=\phi_{\SET}(i)$,
where $\phi_{\SET}(i)$ is noted explicitly to be a function of $i\in\{1,\ldots,n\}$,
as follows:

\begin{eqnarray}
\phi_{\SET}(i) & = & O(M^{W_{i}})+\sum_{j=i+1}^{n}O(M^{B_{j}})+\sum_{j=1}^{i}O(M^{F_{j}})\label{eq:GDL=00003Dtriple-O}
\end{eqnarray}
where $W_{i}$, $B_{j}$, $F_{j}$ denote the domain dimensions of
functions $\gamma_{i}$, $\gbar_{j:n}$, $\gbar_{1:j}$, in (\ref{eq:domain=00003DGDL}),
respectively, and, as before, $M\TRIANGLEQ|\calX|$ (\ref{eq:ch5:DEF:M}).
Note that, in practice, we often have $W,B_{j},F_{j}\ll m$, $j=1,\ldots,n$,
which yields a significant reduction in the total cost:

\[
\phi_{\SET}\ll O(M^{m})
\]

Although we only need to pick up one value $i$ from $\{1,\ldots,n\}$
in order to compute $\widehat{\mathbf{g}}_{1:n}$, this choice of
$i$ can greatly influence the computational load, $\phi_{S}(i)$,
as noted above. For this general case, a criterion for optimal choice
of $i$ has not yet been established in the literature. The discussion
on this optimization issue will be given in Section \ref{sec:chap5:optimizations}.

\subsection{FB recursion for sequential objective functions \label{sec:Sequence-of-sets}}

Let us now study an efficient scheme for computing sequential objective
functions in (\ref{eq:ch5:SUM:sequential}). Furthermore, for simplicity,
let us confine ourselves to the case of non-overflowed (NOF) sets
in Definition \ref{(Non-overflowed-set)}. The main advantage of NOF
condition is that the result of any NOF objective function $g(\Ssubstract j)=\underset{\xBold_{\SET_{j}}}{\boxplus}\gBold_{n}$
can be extracted from the FB recursion for evaluating the union objective
function $g(S^{c})=\underset{\xBold_{\SET}}{\boxplus}\gBold_{1:n}$,
where $\SET=\SET_{1}\cup\cdots\cup\SET_{\kappa}$ and $\OMEGA=\{S^{c},S\}$,
$j\in\{1,\ldots,\kappa\}$. 

Indeed, if NOF condition (Definition \ref{(Non-overflowed-set)})
is satisfied, $\xBold_{\Ssubstract j}$ must belong to union domain
of two recursive functions $\gbar_{1:i}$, $\gbar_{i:n}$ in (\ref{eq:domain=00003DGDL}),
for a specific $\iinn$, say $\itime=\itime_{j}$. Then, given that
specific value $i_{j}$, the result $\underset{\xBold_{S_{j}}}{\boxplus}\gBold_{n}$
for any $j=1,\ldots,\kappa$, can be extracted from a single FB recursion
$\underset{\xBold_{\SET}}{\boxplus}\gBold_{n}$ over $\SET$, without
the need to recompute FB recursion $\underset{\xBold_{S_{j}}}{\boxplus}\gBold_{n}$
for each $S_{j}$. We will divide the computation into two stages,
as presented below.

\subsubsection{FB recursion stage for union objective function}

In the first stage, defining the union set $\SET$, where $\SET=\SET_{1}\cup\cdots\cup\SET_{\kappa}$,
we compute one complete forward recursion (\ref{eq:ch5:GDL:FW}) and
one complete backward recursion (\ref{eq:ch5:GDL:BW}) for $\SET$,
i.e. until $i=n-1$ in forward recursion (\ref{eq:ch5:GDL:FW}) and
until $i=0$ for backward (\ref{eq:ch5:GDL:BW}). After finishing
this step, we have access to all forward and backward functions $\gHat_{1:i}$,
$\gHat_{i:n}$, $\gbar_{1:i}$, $\gbar_{i:n}$, $\forall\iton$, owing
to equations (\ref{eq:ch5:GDL:FW}-\ref{eq:GDL=00003DSecondOrder}).
This is the end of first stage.

Note that, because there is no need to evaluate the combination step
(\ref{eq:Ternary=00003DGDL}) in this stage, the values $\gamma_{i}$
in (\ref{eq:GDL=00003DSecondOrder}) are not computed either.

\subsubsection{FB extraction stage for sequential objective functions}

In the second stage, it is possible to extract two results, the union
objective function $g(S^{c})=\underset{\xBold_{\SET}}{\boxplus}\gBold_{1:n}$
and the sequence of NOF objective functions $g(\Ssubstract j)=\underset{\xBold_{\SET_{j}}}{\boxplus}\gBold_{n}$,
from memorized values $\gbar_{1:i},\gbar_{i:n},\forall\iton,$ of
the first stage. For clarity and completeness, let us consider the
extraction steps for the union and sequential cases, respectively.
The special case of non-overflowed (NOF) sets, namely scalar sets
(Proposition \ref{prop:Scalar-non-overflowed}), will also be discussed.\\
\\

\begin{itemize}
\item \textbf{For union set }$\SET$\textbf{:}
\end{itemize}
In order to compute the value $\gHat_{1:n}\TRIANGLEQ\underset{\xBold_{\SET}}{\boxplus}\gBold_{1:n}$,
we can extract values $\gbar_{1:i},\gbar_{i+1:n}$ at any $\iinn$
in current memory and compute:

\begin{eqnarray}
\gHat_{1:n}\TRIANGLEQ\underset{\xBold_{\SET}}{\boxplus}\gBold_{1:n} & = & \underset{\xBold_{\SET_{\ASET_{i}}}}{\boxplus}\gbar_{\ASET_{i}}\label{eq:GDL=00003Dgn}
\end{eqnarray}
where $\ASET_{i}$ is an in-processing tri-partitioned set, defined
in (\ref{eq:ch5:factor:in-process}), $\SET_{\ASET_{i}}$ are the
associated in-process operator index set (\ref{eq:ch5:non-overflow:set})
and $\gbar_{\ASET_{i}}$ is a new object, defined as follows:

\[
\gbar_{\ASET_{i}}\TRIANGLEQ\gbar_{i+1:n}\PROD\gbar_{1:i}
\]

Applying the GDL (\ref{eq:ch5:GDL:op}) and substituting the three
partitions of $\ASET_{i}$ in (\ref{eq:ch5:factor:in-process}) into
the right hand side of (\ref{eq:GDL=00003Dgn}), we retrieve equation
(\ref{eq:Ternary=00003DGDL}) for $\underset{\xBold_{S}}{\boxplus}\gBold_{1:n}.$
Hence, the step (\ref{eq:GDL=00003Dgn}) can be considered as extraction
step from FB recursions (\ref{eq:ch5:GDL:FW}-\ref{eq:ch5:GDL:BW}).
\begin{itemize}
\item \textbf{For non-overflowed (NOF) sets} $\Ssubstract j$\textbf{:}
\end{itemize}
Given non-overflowed (NOF) sets $\{\Ssubstract 1,\ldots,\Ssubstract{\noperator}\}$,
let us pick up the index $i$ satisfying condition $\Ssubstract j\subseteq\SET_{\ASET_{i}}$
in Definition \ref{(Non-overflowed-set)} and retrieve that in-processing
tri-partitioned set $\ASET_{i}$. For computing the sequence $\underset{\xBold_{\SET_{j}}}{\boxplus}\gBold_{1:n}=\underset{\xBold_{\SET\backslash\Ssubstract j}}{\boxplus}\gBold_{1:n}$,
the set $\SET_{\ASET_{i}}$ in (\ref{eq:GDL=00003Dgn}) can be replaced
with $\SET_{\ASET_{i}}^{\backslash j}$, defined in (\ref{eq:ch5:non-overflow:complement}),
as follows:

\begin{eqnarray}
\underset{\xBold_{\SET_{j}}}{\boxplus}\gBold_{1:n} & = & \underset{\xBold_{\SET_{\ASET_{i}}\backslash\Ssubstract j}}{\boxplus}\gbar_{\ASET_{i}}\label{eq:GDL=00003Dgk}\\
 & = & \underset{\xBold_{s_{\bcommon i}^{\backslash j}}}{\boxplus}(\ghat_{\ASET_{i}})\nonumber 
\end{eqnarray}

in which $j\in\{1,\ldots,\kappa\}$, and $\ghat_{\ASET_{i}}$ is a
new object, defined as follows:

\begin{equation}
\ghat_{\ASET_{i}}\TRIANGLEQ(\underset{\xBold_{s_{\bparent{i+1}}^{\backslash j}}}{\boxplus}\gbar_{i+1:n})\PROD(\underset{\xBold_{s_{\bbracket i}^{\backslash j}}}{\boxplus}\gbar_{1:i})\label{eq:ch5:GDL:gk:in-process}
\end{equation}
Hence, substituting the memorized values $\gbar_{1:i},\gbar_{i:n},\forall\iton$
of the first stage into (\ref{eq:ch5:GDL:gk:in-process}), it is straight-forward
to compute (\ref{eq:GDL=00003Dgk}) in one step (i.e. without recursion)
for any $j\in\{1,\ldots,\kappa\}$.
\begin{itemize}
\item \textbf{For a sequence of scalar sets} $\Ssubstract j=\{j\}\in S=\OMEGA$\textbf{:}
\end{itemize}
Let us recall, from Proposition \ref{prop:Scalar-non-overflowed},
the scalar sets $\Ssubstract j=\{j\}\in S=\OMEGA$ are indeed non-overflowed
(NOF) sets. This special case is of particular interest in practice.
For example, the task of computing all scalar functions $g(x_{j})=\underset{\xBold_{\backslash j}}{\boxplus}\gBold_{n}$,
$\jtom$, is a major concern in applications of GDL {[}\citet{ch2:origin:GDL:McEliece}{]}.
Because the number of variables $m$ is very high, $\noperator$ is
also very high, since $\kappa=m$ in this case. Hence, the above scheme
often significantly reduces the cost in this case of scalar NOF sets.

\subsubsection{Computational load for sequential objective functions}

The total computational load of FB recursion in this case can be evaluated
as in (\ref{eq:GDL=00003Dtriple-O}), with $\kappa$ extraction steps
(\ref{eq:GDL=00003Dgk}) in the second stage and a single step of
fully FB recursion in first stage, as follows: 

\begin{equation}
\phi_{\SET_{1:\kappa}}=\sum_{j=1}^{\kappa}O(M^{W_{j}})+\sum_{i=1}^{n}(O(M^{B_{i}})+O(M^{F_{i}}))\label{eq:GDL=00003Dtriple-O=00003Dk}
\end{equation}
where $W_{j}$ is the dimension of domain of the $\ghat_{\ASET_{i}}$
in (\ref{eq:GDL=00003Dgk}). Comparing (\ref{eq:GDL=00003Dgk},\ref{eq:GDL=00003Dtriple-O=00003Dk})
with (\ref{eq:Ternary=00003DGDL},\ref{eq:GDL=00003Dtriple-O}), respectively,
we can see that $W_{i}-|\Ssubstract j|\leq W_{j}\leq W_{i}$, with
$i,j$ satisfying the NOF condition  in (\ref{eq:ch5:non-overflow:set}).

In practice, the two-step FB recursive scheme of the first stage in
(\ref{eq:ch5:GDL:FW}-\ref{eq:ch5:GDL:BW}) , i.e. with one full forward
recursion and one full backward recursion, is roughly double the cost
of the one-step FB recursion in (\ref{eq:GDL=00003Dtriple-O}). Also,
the cost of the first stage in the two-step FB recursion often dominated
the cost of the second stage (extraction stage), i.e. $W_{i},W_{j}\ll B_{j},F_{j}$
with $i,j$ satisfying the NOF condition  in (\ref{eq:ch5:non-overflow:set}).
In that case, we can expect that $\phi_{\SET_{1:\kappa}}\approx2\phi_{S}$,
i.e. the total cost of computing the sequence $\underset{\xBold_{S_{j}}}{\boxplus}\gBold_{n}$,
is only double the cost of the single task $\underset{\xBold_{S}}{\boxplus}\gBold_{n}$.
Those costs are, therefore, of the same order.

\subsection{Computational bounds and optimization for FB recursion \label{sec:chap5:optimizations}}

By Theorem \ref{thm:GDL}, we know that GDL always reduces the total
number of operators, and, hence, computational load in FB recursion.
A well-posed question is to ask which choice of permutation order
of factors $g_{i}$ in $\gBold_{1:n}$ minimizes the number of operators
in FB recursion. This optimal choice is an open problem in the literature
{[}\citet{ch2:origin:GDL:McEliece}{]}. What we can achieve here is
specification of some computational bounds, and we will discuss some
difficulties of this optimization. Several practical approaches will
also be provided. 

\subsubsection{Bounds on computational complexity}

Firstly, let us consider the upper and lower bounds of the number
of operator in original model (\ref{eq:ch5:MODEL}). The derivation
of these bounds will illustrate the difficulties involved in finding
an optimal GDL-based computational reduction later.
\begin{prop}
\label{prop:ch5:lower-upper}The lower and upper bound of operator's
number in original model $\gBold_{1:n}=\odot_{i=1}^{n}g_{i}$, as
defined in (\ref{eq:ch5:MODEL}), is:

\begin{equation}
O(M^{m})\leq O_{\PROD}[\gBold_{1:n}]\leq O(nM^{m})\label{eq:lower-upper-bound}
\end{equation}
where $M\TRIANGLEQ|\calX|$, as defined in (\ref{eq:ch5:DEF:M}) and
$O_{\PROD}$ is defined as in Section \ref{subsec:chap5:cost-ring-sum}. 
\end{prop}

\begin{proof}
For a product $\gBold_{1:n}=\PROD_{i=1}^{n}g_{i}$, we have many ways
to compute $\gBold_{1:n}$ via parenthesization, owing to commutative
property. For example, the following two forms may yield different
costs:

\begin{equation}
\gBold_{1:n}=g_{\pi(1)}\PROD(g_{\pi(2)}\PROD(\cdots(g_{\pi(n)})))\label{eq:parenthesis=00003D1}
\end{equation}

\begin{equation}
\gBold_{1:n}=(g_{\pi(1)}\PROD g_{\pi(2)})\PROD\cdots\PROD(g_{\pi(n-1)}\PROD g_{\pi(n)})\label{eq:parenthesis=00003D2}
\end{equation}
where $\pi$, again, is a permutation of the set $\{1,\ldots,n\}$.
From Lemma \ref{lem:COST}, the computational load for (\ref{eq:parenthesis=00003D1})
is $O(\sum_{i=1}^{n}M^{|\oBold_{\pi(1):\pi(i)}|})$, with $\oBold_{\pi(i):\pi(j)}=\omega_{\pi(i)}\cup\cdots\cup\omega_{\pi(j)}$,
while the cost for (\ref{eq:parenthesis=00003D2}) is $O(\sum_{i=1}^{n-1}M^{|\omega_{\pi(i)}\cup\omega_{\pi(i+1)}|})$. 

In general, we can construct a binary tree corresponding to this task
of recursive parenthesization {[}\citet{ch5:art:GDL:NPcomplete97,ch5:BK:Bible:Algorithms}{]}.
Because we have $n$ leaf nodes for this binary tree, corresponding
to $n$ functions $g_{i}$, the total number of \foreignlanguage{british}{binarily}
combined nodes (internal nodes and root node) is $n$. At the root
of the binary tree is an arbitrary binary partition of $\gBold_{1:n}$,
i.e. $\phi_{Left}\TRIANGLEQ\gBold_{\pi(1):\pi(i)}$ and $\phi_{Right}\TRIANGLEQ\gBold_{\pi(i+1):\pi(n)}$,
where $\gBold_{\pi(i):\pi(j)}=g_{\pi(i)}\PROD\cdots\PROD g_{\pi(j)}$.
Since we always have $|\oBold_{\pi(1):\pi(i)}\cup\oBold_{\pi(i+1):\pi(n)}|=|\OMEGA|=m$,
$\forall i\in\{1,\ldots,n\}$, the computation cost for (\ref{eq:parenthesis=00003D2})
at the root is fixed:

\begin{equation}
\calO[\phi_{Left}\PROD\phi_{Right}]=O(M^{m})\label{eq:Proof=00003DRootNode}
\end{equation}

Because the computation of $\gBold_{\pi(1):\pi(i)}$ and $\gBold_{\pi(i+1):\pi(n)}$,
in turn, can be recursively parenthesized into other binary sub-partitions,
respectively, this scheme constructs a binary tree, in which each
binarily combined node represents the result of ring-product $\phi_{Left}\PROD\phi_{Right}$
of two multiplied functions coming up from the left and right of that
combined node. In a bottom-up manner, the cost of any internal node
is, therefore: 

\begin{equation}
0\leq\calO[\phi_{Left}\PROD\phi_{Right}]=O(M^{|\oBold_{Left:}\cup\oBold_{Right}|})\leq O(M^{m})\label{eq:Proof=00003DInternalNode}
\end{equation}
where $\emptyset\subseteq\oBold_{Left:},\oBold_{Right:}\subseteq\OMEGA$.
Because we have $n-1$ internal nodes with varied cost (\ref{eq:Proof=00003DInternalNode})
and one root node with fixed cost (\ref{eq:Proof=00003DRootNode}),
the number of ring-products $O_{\PROD}[\gBold_{1:n}]$ is therefore
bounded by (\ref{eq:lower-upper-bound}).
\end{proof}

\subsubsection{Minimizing computational load in FB recursion \label{subsec:chap5:Minimizing-compuational-cost}}

As shown in the proof of Proposition \ref{prop:ch5:lower-upper} above,
the total number of commutative ring-products in $\gBold_{1:n}=\PROD_{i=1}^{n}g_{i}$
depends on two issues: (i) the $n!$ permutations and (ii) the choice
of parenthesization between them. By investigating those two issues,
the optimization scheme for $\underset{x_{S}}{\boxplus}\gBold_{1:n}$
was proven to be an NP-complete problem {[}\citet{ch5:art:GDL:NPcomplete97,ch2:origin:GDL:McEliece}{]}.
Hence, a tractable technique for finding the optimum is not available.
Nevertheless, we still have some other options to consider, as follows:

- The FB recursion for GDL was originally inspired by two topological
sets, one forward (NLN) and one backward (FA), in Definition \ref{(No-longer-needed-indices)},\ref{(New-coming-indices)}.
For this reason, a reasonable solution is to extend the FB to multi-direction
approach, instead of two directions, NLN and FA. Such a multi-direction
scheme is actually a merit of graph theory, which facilitates the
visual representation of the model. Nevertheless, a topological CI
structure via set algebra in this chapter is still useful. Because
the operators in pre-semiring $(\Ring,\SUM,\PROD)$ are both binary,
all combinatorial of operators must consist of binary relationship.
Hence, in multi-direction scheme, the CI structure above may be still
applicable to a general Bayesian networks, which concerns about chain
rule order of factors in distribution.

- In the probability context, we can design the permutation $\pi$
in the joint distribution (\ref{eq:JOINT}) such that $f(\xBold_{\OMEGA})$
is factorized into a chain rule order. Then, the product of $f_{\pi(i)}$
can be computed in a reverse order to the chain rule. In this way,
the total number of product will be $O_{\PROD}(\sum_{i=1}^{n}M^{i})=O(M^{n})$,
where $\sum_{i=1}^{n}M^{i}=\frac{M^{n+1}-1}{M-1}-1$, i.e. it always
reaches the lower bound in (\ref{eq:lower-upper-bound}). Such a value
of $\pi$ can always be found in linear time via topological sorting
algorithm {[}\citet{ch5:BK:Bible:Algorithms}{]}.

- In Bayesian inference, the permutation $\pi$ also has an important
role. In Section \ref{sec:Re-interpretation-of-GDL}, we will see
that permutations of the full conditional distributions (\ref{eq:Posterior})
yield different factorization forms for the GDL. Hence, although the
form for re-factorized conditional distribution is not computed in
FB recursion, the computation cost of FB recursion will vary, based
on that implicit re-factorization form (\ref{eq:Ternary=00003DFactorization}).
In this sense, a re-factorization with the minimal number of \foreignlanguage{british}{neighbour}
variables might be preferred, in order to reduce the dimension of
intermediate functions in FB recursion via the GDL. This scheme is
consistent with the minimum message length problem in model representation.

\section{GDL in the probability context \label{sec:chap5:GDL-in-probability}}

Note that, the first step in computing $\underset{\xBold_{S}}{\boxplus}\gBold_{1:n}$
is to identify three elements in the pre-semiring $(\Ring,\boxplus,\odot)$,
i.e. the functional set $\Ring$ and two binary operators $\boxplus,\odot$,
which satisfy all properties in Definitions \ref{Semigroup}-\ref{Pre-semiring}.
This general framework is very flexible in practice. For example,
{[}\citet{ch2:origin:GDL:McEliece,ch2:bk:ToddMoon}{]} gives examples
of semirings that are useful in graph theory and decoding context. 

In this section, we present application of GDL in the probability
context. For this purpose, we will define and apply some practical
pre-semirings, which are summarized in Table \ref{tab:Pre-semirings}.

\begin{table*}
\begin{centering}
\begin{tabular}{|c|c|c|c|c|c|c|}
\hline 
$\RING$ & $g_{i}$ & $\boxplus$ & $\PROD$ & $\underset{\xBold_{\SET}}{\boxplus}(\PROD_{i=1}^{n}g_{i})$ & short name & Purpose \tabularnewline
\hline 
\hline 
$[0,1]$ & $f_{i}$ & + & $\times$ & $\sum_{\xBold_{\SET}}f(\xBold_{\OMEGA})$ & sum-product & Marginalization\tabularnewline
\hline 
$[0,1]$ & $f_{i}$ & $\max$ & $\times$ & $\max_{\xBold_{\OMEGA}}f(\xBold_{\OMEGA})$ & max-product & Joint mode\tabularnewline
\hline 
$(-\infty,1]$ & $\log f_{i}$ & $\max$ & + & $\max_{\xBold_{\OMEGA}}(\log f(\xBold_{\OMEGA}))$ & max-sum & Joint mode\tabularnewline
\hline 
$\DEAL$ & $f_{i}\angle\log q_{i}$ & + & $\times$ & $1\angle E_{f(\xBold_{\OMEGA})}\log q(\xBold_{\OMEGA})$ & Dual number & Entropy\tabularnewline
\hline 
\end{tabular}
\par\end{centering}
\caption{\label{tab:Pre-semirings}Some commutative pre-semirings $(\protect\Ring,\boxplus,\odot)$
for the probability context, with. $f_{i}\protect\TRIANGLEQ f(x_{i}|\protect\xMath_{\eta_{i}})$,
$q_{i}\protect\TRIANGLEQ q(x_{i}|x_{\nu_{i}})$ and $\protect\SET\subseteq\protect\OMEGA=\{1,\ldots n\}$ }
\end{table*}

\subsection{Joint distribution}

Without loss of generality, let us assume that $\nstate=\ndata$,
i.e. the number of factors and variables are the same in this section.
Then, let us consider a joint distribution, $f(\xBold_{\OMEGA})$,
of $n$ discrete random variables $\xBold_{\OMEGA}\in\calX_{\OMEGA}=\calX{}^{n}$,
$\OMEGA=\{1,\ldots,n\}$. By the chain rule of probability, the distribution
$f(\xBold_{\OMEGA})$ can be factorized into a chain rule order, as
follows: 
\begin{equation}
f(\xBold_{\OMEGA})=\prod_{i=1}^{n}f(x_{i}|x_{\eta_{i}})\label{eq:PRIOR}
\end{equation}
where $f(x_{i}|\xEta i)=f(x_{i}|\xBold_{i-1})$ is conditional distribution
of $x_{i}$ given its neighbor variables $\xEta i$, $\iton$. Similar
to our universal model (\ref{eq:ch5:MODEL:i}), we assume that the
value of functions $f(x_{i}|\xEta i)$ and \foreignlanguage{british}{neighbour}
sets $\eta_{i}$ in (\ref{eq:PRIOR}) are known. 

\subsection{GDL for probability calculus}

Given the joint distribution (\ref{eq:PRIOR}), we are interested
in three kinds of inference: (i) the sequence of scalar marginals,
(ii) joint mode and (iii) functional moments. We will explain briefly
how to deploy GDL in each of these contexts, next.

\subsubsection{Sequence of scalar marginals \label{subsec:chap5:Sequence-of-marginals}}

Let us specify pre-semiring, $(\Ring,\boxplus,\odot)$, as ,$([0,1]^{\calX^{\ndata}},+,\times)$,
where $\RING=[0,1]$ (the unit line segment in $\REAL$), and where
$+,\times$ are traditional scalar addition and multiplication for
real numbers. Because $g_{i}=f_{i}\TRIANGLEQ f(x_{i}|\xMath_{\eta_{i}})\in[0,1]^{\calX^{\ndata}}$
is a distribution, we can compute a sequence of marginals $f(x_{i})=\sum_{\xBold_{\backslash i}}f(\xBold_{\OMEGA})$,
where $\xBold_{\backslash i}$ is the complement of $x_{i}$ in $\xBold_{\OMEGA}$,
as follows:

\begin{eqnarray}
f(x_{i})=\sum_{\xBold_{\backslash i}}f(\xBold_{\OMEGA}) & = & \sum_{\xBold_{\backslash i}}(\prod_{i=1}^{n}g_{i}),\ \mbox{with}\ g_{i}=f_{i}\label{eq:Marginal}
\end{eqnarray}

for $\iton$. Then, the sequence of $n$ scalar sets in (\ref{eq:Marginal})
can be computed feasibly via FB recursion for sequential objective
functions, as in Section \ref{sec:Sequence-of-sets}. 

An application of (\ref{eq:Marginal}) for HMC model, namely FB algorithm,
will be presented in the Section \ref{subsec:chap6:FB-HMC}.

\subsubsection{Joint mode \label{subsec:chap5:Joint-mode}}

The elements in joint mode $\xHat_{\OMEGA}=\{\widehat{x}_{1},\ldots,\widehat{x}_{n}\}$,
defined as $\xHat_{\OMEGA}=\arg\max_{\xBold_{\OMEGA}}f(\xBold_{\OMEGA})$,
can be found via either one of two forms, as follows:

\begin{eqnarray}
\widehat{x}_{i} & = & \arg\max_{x_{i}}(\max_{\xBold_{\backslash i}}f(\xBold_{\OMEGA}))\label{eq:x^hat}\\
 & = & \arg\max_{x_{i}}(\max_{\xBold_{\backslash i}}(\log f(\xBold_{\OMEGA})))\nonumber 
\end{eqnarray}

for $\iton$. Corresponding to two ways of computing $\widehat{x}_{i}$
(\ref{eq:x^hat}), we have two ways to assign pre-semiring $(\Ring,\boxplus,\odot)$,
either with $([0,1]^{\calX^{\ndata}},\max,\times)$ or $([-\infty,0]^{\calX^{\ndata}},$
$\max,+)$, respectively, as follows:

\begin{eqnarray}
\max_{\xBold_{\backslash i}}f(\xBold_{\OMEGA}) & = & \max_{\xBold_{\backslash i}}(\prod_{i=1}^{n}g_{i}),\ \mbox{with}\ g_{i}=f_{i}\label{eq:max}\\
\max_{\xBold_{\backslash i}}(\log f(\xBold_{\OMEGA})) & = & \max_{\xBold_{\backslash i}}(\sum_{i=1}^{n}g_{i}),\ \mbox{with}\ g_{i}=\log f_{i}\nonumber 
\end{eqnarray}

for $\iton$. Once again, both sequences of $n$ scalar sets in (\ref{eq:max})
can be computed feasibly via FB recursion for sequential objective
functions in Section \ref{sec:Sequence-of-sets}. Substituting the
results (\ref{eq:max}) into (\ref{eq:x^hat}), we can retrieve the
joint mode. 

An application of (\ref{eq:x^hat}-\ref{eq:max}) for HMC model, namely
bi-directional Viterbi algorithm, will be presented in Section \ref{subsec:chap6:Bi-directional-VA}.

\subsubsection{Entropy}

Consider another joint \textit{reference} distribution of $\xBold_{\OMEGA}\in\calX{}^{n}$,
i.e. $q(\xBold_{\OMEGA})=\prod_{i=1}^{n}q(x_{i}|x_{\nu_{i}})$, where
$q_{i}\TRIANGLEQ q(x_{i}|x_{\nu_{i}})$ is the associated full conditional
distribution of $x_{i}$ given its \foreignlanguage{british}{neighbour}
variables $x_{\nu_{i}}$, $\iton$. Consider the following functional
moment:

\begin{equation}
E_{f(\xBold_{\OMEGA})}\log(q(\xBold_{\OMEGA}))=E_{f(\xBold_{\OMEGA})}(\sum_{i=1}^{n}\log q_{i})\label{eq:Entropy-Task}
\end{equation}

Comparing (\ref{eq:Entropy-Task}) with our GDL model (\ref{eq:SUM-PRODUCT})
, the task is to transform the sum $\sum_{i=1}^{n}$ in (\ref{eq:Entropy-Task})
into some form of real value products $\prod_{i=1}^{n}$, a special
case of ring-product $\PROD_{i=1}^{n}$, in order to achieve a computational
load reduction via GDL. Such a transformation can be effected via
the so-called \textit{dual number} in matrix form {[}\citet{ch5:art:DualNumber:matrix88}{]},
as follows:

\begin{equation}
g_{i}=f_{i}\angle\log q_{i}\TRIANGLEQ f_{i}\left[\begin{array}{cc}
1 & \log q_{i}\\
0 & 1
\end{array}\right]\label{eq:Dual=00003Dgi}
\end{equation}
for $\iton$. The dual number, originally proposed in {[}\citet{ch5:origin:DualNumber:Clifford1873}{]},
belongs to a ring $\DEAL$, not a field like complex numbers $\mathbb{C}$
{[}\citet{ch5:art:DualNumber:complex75}{]}. However, it shares with
complex numbers the property of magnitude $f_{i}$ and angle $\angle\log q_{i}$
(see Appendix \ref{App:chap:Dual-Number}). Therefore, $g_{i}$ belongs
to the ring $(\DEAL^{\calX^{\ndata}},+,\cdot)$, a special case of
pre-semiring $(\Ring,\boxplus,\odot)$, where $+$ and $\cdot$ are
the usual matrix summation and matrix multiplication. Note that, in
the above special matrix form (\ref{eq:Dual=00003Dgi}), it is feasible
to verify that matrix multiplication is commutative. 

Substituting (\ref{eq:Dual=00003Dgi}) into (\ref{eq:Entropy-Task}),
we have:

\begin{eqnarray}
\underset{\xBold_{\OMEGA}}{\sum}\gBold_{1:n} & = & \sum_{\xBold_{\OMEGA}}\prod_{i=1}^{n}g_{i}\label{eq:Entropy-like}\\
 & = & E_{f(\xBold_{\OMEGA})}(1\angle\sum_{i=1}^{n}\log q_{i})\nonumber \\
 & = & 1\angle(E_{f(\xBold_{\OMEGA})}\sum_{i=1}^{n}\log q_{i})\nonumber 
\end{eqnarray}

For the purpose of reducing the number of traditional sum and product
operators, we can apply the GDL to the right hand side of (\ref{eq:Entropy-like}).
Then, the value of the angle is extracted from the result of $\underset{\xBold_{\OMEGA}}{\sum}\gBold_{1:n}$
and reported as the value of $E_{f(\xBold_{\OMEGA})}\log(q(\xBold_{\OMEGA}))$
in (\ref{eq:Entropy-Task}). 

In the literature, the above task of computing $E_{f(\xBold_{\OMEGA})}\log(q(\xBold_{\OMEGA}))$
(\ref{eq:Entropy-Task}) via GDL was specialized to the computation
of entropy {[}\citet{ch5:art:MP:Entropy11}{]}, Kullback-Leibler divergence
(KLD) {[}\citet{ch5:art:MP:EntropyRing08}{]}, and the Expectation-Maximization
(EM) algorithm {[}\citet{ch5:art:EM:semiring09}{]}. In each case,
the derivation relied heavily on complicated operators in ring theory,
rather than on the simple matrix operator in (\ref{eq:Dual=00003Dgi}).
Also, a unified recursion for implementing GDL---such as is achieved
by the FB recursion above---is missing in those papers. 

Another potential application of (\ref{eq:Entropy-like}) is in the
Iterative VB (IVB) algorithm, as presented in Section \ref{subsec:chap4:Variational-Bayes-(VB)}.
Notice the similarity between the expectation for log functions in
(\ref{eq:Entropy-like}) above and in the IVB algorithm (\ref{eq:ch4:IVB})
.

\subsubsection{Bayesian computation}

Let the role of the joint distribution $f(\xBold_{\OMEGA})$ in (\ref{eq:PRIOR})
be a prior in Bayes' rule, as follows:

\begin{eqnarray}
f(\xBold_{\OMEGA}|\yBold_{\OMEGA}) & \propto & f(\yBold_{\OMEGA}|\xBold_{\OMEGA})f(\xBold_{\OMEGA})\label{eq:Posterior}\\
 & = & \prod_{i=1}^{n}f(y_{i}|\xBold_{\nu_{i}},y_{\xi_{i}})f(x_{i}|x_{\eta_{i}})\nonumber 
\end{eqnarray}
where $f(\yBold_{\OMEGA}|\xBold_{\OMEGA})=\prod_{i=1}^{n}f(y_{i}|\xBold_{\nu_{i}},y_{\xi_{i}})$
is a given observation distribution (model) with known (i.e. observed
or realized) discrete values $\yBold_{\OMEGA}\in{\cal Y}^{n}$. Because
the form (\ref{eq:Posterior}) also belongs to our generic model structure
(\ref{eq:ch5:MODEL:i}), we can implement the above inference schemes
for $2n$ factors (\ref{eq:Posterior}) in the same way as for $n$
factors in the general distribution (\ref{eq:PRIOR}). 

\subsection{GDL for re-factorization \label{sec:Re-interpretation-of-GDL}}

In this subsection, the role of the CI topology in Section \ref{sec:Conditionally-separable}
will be explained in the probability context. In this way, we will
appreciate that GDL is, in essence, a tool for exploiting that topology.

For this purpose, let us re-consider the joint distribution $f(\xBold_{\OMEGA})$
in (\ref{eq:PRIOR}). Owing to commutativity of product, we have $n!$
ways to permute those $n$ factors, as follows: 
\begin{equation}
f(\xBold_{\OMEGA})=\prod_{i=1}^{n}f(x_{\pi(i)}|\xEta{\pi(i)})\label{eq:JOINT}
\end{equation}
where $\pi$ is a permutation over the set $\OMEGA$. Let us define
the $n$ full conditionals $g_{i}=f_{\pi(i)}\TRIANGLEQ f(x_{\pi(i)}|\xEta{\pi(i)})$,
together with the sets $\omega_{i}=\{\pi(i),\eta_{\pi(i)}\}$, $\forall i\in\OMEGA=\{1,\ldots,n\}$.
Consider, further, the following binary parenthesization:

\begin{eqnarray}
f(\xBold_{\OMEGA}) & = & (\prod_{j=i+1}^{n}f_{\pi(j)})(\prod_{j=1}^{i}f_{\pi(j)})\label{eq:Binary}\\
 & = & \gBold_{i+1:n}\gBold_{1:i},\ \iinn\nonumber 
\end{eqnarray}

Note that, because the permutation $\pi$ is arbitrary, $f(\xBold_{\OMEGA})$
in (\ref{eq:JOINT}) - for a particular permutation - may not be in
probability chain rule order. Consequently, the forward $\gBold_{1:i}=\prod_{j=1}^{i}f_{\pi(j)}$
and backward $\gBold_{i:n}=\prod_{j=i}^{n}f_{\pi(j)})$ products,
with the same domains as in (\ref{eq:domain=00003DBinaryTree}), are
merely positive functions and may not be distributions in general.
In practice, the model (\ref{eq:JOINT}) happens very often, since
the chain rule order is very often not available (e.g in {[}\citet{ch2:origin:GDL:McEliece}{]}).

\subsubsection{Conditionally independent (CI) factorization}

Given the index set $\omega_{i}=\{\pi(i),\eta_{\pi(i)}\}$ in (\ref{eq:JOINT}),
we will see below that there exists a close relationship between topology
in Section \ref{sec:Conditionally-separable} and re-factorization
forms of (\ref{eq:JOINT}).
\begin{prop}
\label{prop:Reverse-ChainRuleOrder} The first-appearance (FA) $\bparent i$
and no-longer-needed (NLN) $\bbracket i$ sets yield two choices of
probabilistic chain rule order, one forward and one backward, respectively,
for $f(\xBold_{\OMEGA})$ in (\ref{eq:JOINT}), as follows:

\begin{equation}
f(\xBold_{\OMEGA})=\prod_{i=1}^{n}\fBold_{\bparent i}=\prod_{i=1}^{n}\fBold_{\bbracket i}\label{eq:Chain=00003D=00005B=00005D}
\end{equation}

where:

\begin{eqnarray}
\fBold_{\bparent i}\TRIANGLEQ f(x_{\bparent i}|x_{\bcommon{i-1}}) & = & f(x_{\bparent i}|\xBold_{\bparent{1:i-1}}),\ \ \ i=1,\ldots,n\label{eq:f(i)}\\
\fBold_{\bbracket i}\TRIANGLEQ f(x_{\bbracket i}|x_{\bcommon i}) & = & f(x_{\bbracket i}|x_{\bbracket{i+1:n}}),\ \ \ i=n,\ldots,1\nonumber 
\end{eqnarray}
\end{prop}

\begin{proof}
We only need to prove the case $\bparent i$, because the case $\bbracket i$
follows the same logic. Then, the key solution is to prove the following
relationship:

\begin{equation}
f(\xBold_{\bparent{i:n}}|\xBold_{\bparent{1:i-1}})=f(\xBold_{\bparent{i:n}}|x_{\bcommon{i-1}}),\ \forall i\in\{1,\ldots,n\}\label{eq:Proof=00003Df(i)}
\end{equation}
because if (\ref{eq:Proof=00003Df(i)}) is valid for all $i$, equations
(\ref{eq:f(i)}) is valid by induction.\\
For proving (\ref{eq:Proof=00003Df(i)}), let us exploit the chain
rule and the properties (\ref{eq:ch5:(i)-=00005Bi=00005D},\ref{eq:ch5:factor:left})
of FA sets $\bparent i$, as follows: 
\begin{eqnarray}
f(\xBold_{\OMEGA}) & = & f(\xBold_{\bparent{i:n}}|\xBold_{\bparent{1:i-1}})f(\xBold_{\bparent{1:i-1}})\label{eq:ch5:proof:CI_chain_rule}\\
f(\xBold_{\oBold_{i:n}}) & = & f(\xBold_{\bparent{i:n}}|x_{\bcommon{i-1}})f(x_{\bcommon{i-1}})\nonumber 
\end{eqnarray}
Then, from (\ref{eq:Binary}), we can compute the joint and marginal
distributions in (\ref{eq:ch5:proof:CI_chain_rule}), as follows:

\begin{eqnarray}
f(\xBold_{\OMEGA}) & = & \gBold_{i:n}\gBold_{1:i-1}\label{eq:joint marginal}\\
f(\xBold_{\bparent{1:i-1}})=\sum_{\xBold_{\bparent{i:n}}}f(\xBold_{\OMEGA}) & = & (\sum_{\xBold_{\bparent{i:n}}}\gBold_{i:n})\gBold_{1:i-1}\nonumber \\
f(\xBold_{\oBold_{i:n}})=\sum_{\xBold_{\bbracket{1:i-1}}}f(\xBold_{\OMEGA}) & = & \gBold_{i:n}(\sum_{\xBold_{\bbracket{1:i-1}}}\gBold_{1:i-1})\nonumber \\
f(x_{\bcommon{i-1}})=\sum_{\{\xBold_{\bparent{i:n}},\xBold_{\bbracket{1:i-1}}\}}f(\xBold_{\OMEGA}) & = & (\sum_{\xBold_{\bparent{i:n}}}\gBold_{i:n})(\sum_{\xBold_{\bbracket{1:i-1}}}\gBold_{1:i-1})\nonumber 
\end{eqnarray}
in which we have applied the distributive law in (\ref{eq:joint marginal}),
owing to separable domains of the functions $\gBold$ (see equations
(\ref{eq:domain=00003DBinaryTree})). Substitute (\ref{eq:joint marginal})
into (\ref{eq:ch5:proof:CI_chain_rule}), we can evaluate both conditional
distributions in (\ref{eq:ch5:proof:CI_chain_rule}):

\begin{eqnarray}
f(\xBold_{\bparent{i:n}}|\xBold_{\bparent{1:i-1}}) & = & \frac{f(\xBold_{\OMEGA})}{f(\xBold_{\bparent{1:i-1}})}=\frac{\gBold_{i:n}}{\sum_{\xBold_{\bparent{i:n}}}\gBold_{i:n}}\label{eq:chap5:conditional_g(i)}\\
f(\xBold_{\bparent{i:n}}|x_{\bcommon{i-1}}) & = & \frac{f(\xBold_{\oBold_{i:n}})}{f(x_{\bcommon{i-1}})}=\frac{\gBold_{i:n}}{\sum_{\xBold_{\bparent{i:n}}}\gBold_{i:n}}\label{eq:chap5:conditional_g(ii)}
\end{eqnarray}
which yield (\ref{eq:Proof=00003Df(i)}). 
\end{proof}
Note that, even though the function $f(\xBold_{\OMEGA})=\gBold_{1:n}=\prod_{\itime=1}^{\ndata}g_{i}$
is a valid probability distribution, the product $\gBold_{i:n}=\prod_{j=i}^{\ndata}g_{j}$,
for arbitrary $i$, is not necessarily a valid distribution, if the
functional factors $g_{i}$, $\iton$, do not follow a probability
chain rule order of $f(\xBold_{\OMEGA})$. In fact, the function $g_{i}$
may not be a valid probability distribution to begin with. 

Even so, owing to GDL and CI structure in (\ref{eq:joint marginal}),
the functions $\gBold_{i:n}$ and $\sum_{\xBold_{\bparent{i:n}}}\gBold_{i:n}$
in (\ref{eq:chap5:conditional_g(i)}-\ref{eq:chap5:conditional_g(ii)})
are not necessarily valid probability distributions in order for equality
(\ref{eq:Proof=00003Df(i)}) to be valid. In order words, the CI structure
has identified a chain rule order, and provided a practical method
for factorizing arbitrary distribution. 

This remark is important and interesting, particularly in probability
context. Given an arbitrary non-negative function $\gBold_{1:n}\in[0,\mathbb{R}^{+})$,
the issue of computing normalizing constant $\sum_{\xBold_{\OMEGA}}\gBold_{1:n}$
for $f(\xBold_{\OMEGA})=\frac{\gBold_{1:n}}{\sum_{\xBold_{\OMEGA}}\gBold_{1:n}}$,
where $f(\xBold_{\OMEGA})\in[0,1]$, is typically prohibitive, because
of the curse of dimensionality. In contrast, the recursive computation
for normalizing constant of conditional distributions in (\ref{eq:chap5:conditional_g(i)}-\ref{eq:chap5:conditional_g(ii)})
is typically efficient, since the number of operators falls exponentially
with number of NLN variables. In other words, it is possible to recursively
factorize $f(\xBold_{\OMEGA})$ and compute its conditional distributions
form in polynomial time, owing to application of GDL in (\ref{eq:chap5:conditional_g(i)}-\ref{eq:chap5:conditional_g(ii)}),
without the need of computing the prohibitive normalizing constant
$\sum_{\xBold_{\OMEGA}}\gBold_{1:n}$ over the whole set of $\xBold_{\OMEGA}$.
Similarly, any moments of $f(\xBold_{\OMEGA})$ can be computed recursively,
via that CI factorization form (\ref{eq:Chain=00003D=00005B=00005D})
of $f(\xBold_{\OMEGA})$, instead of being computed directly over
$f(\xBold_{\OMEGA})$, which essentially requires computing normalizing
constant $\sum_{\xBold_{\OMEGA}}\gBold_{1:n}$.

This interesting CI factorization form can be verified feasibly via
the following Proposition:
\begin{prop}
\label{prop:Ternary=00003DReverseChain}The ternary partition of $\OMEGA$
in Proposition \ref{prop:(TernarySET)} yields a ternary factorization
form for $f(\xBold_{\OMEGA})$, as follows:

\begin{equation}
f(\xBold_{\OMEGA})=f(\xBold_{\bparent{i+1:n}}|\xBold_{\bcommon i})f(\xBold_{\bcommon i})f(\xBold_{\bbracket{1:i}}|\xBold_{\bcommon i})\label{eq:Ternary=00003DFactorization}
\end{equation}

where $i\in\{1,\ldots,n\}$ and:

\begin{eqnarray}
f(\xBold_{\bparent{i+1:n}}|\xBold_{\bcommon i}) & = & \prod_{j=i+1}^{n}\fBold_{\bparent j}\label{eq:sequence=00003Df(i)f=00005Bi=00005D}\\
f(\xBold_{\bbracket{1:i}}|\xBold_{\bcommon i}) & = & \prod_{j=1}^{i}\fBold_{\bbracket j}\nonumber 
\end{eqnarray}
\end{prop}

\begin{proof}
Because the sequences $\mbox{\ensuremath{\fBold}}_{\bparent i},\mbox{\ensuremath{\fBold}}_{\bbracket i}$
are each in chain rule order, in consequence of Proposition \ref{prop:Reverse-ChainRuleOrder},
both equations in (\ref{eq:sequence=00003Df(i)f=00005Bi=00005D})
satisfy the chain rule. For (\ref{eq:Ternary=00003DFactorization}),
we can see that $\prod_{j=1}^{i}\mbox{\ensuremath{\fBold}}_{\bparent i}=f(\xBold_{\bparent{1:i}})=f(\xBold_{\bcommon i})f(\xBold_{\bbracket{1:i}}|\xBold_{\bcommon i})$
satisfies the chain rule, since $\bparent{1:i}=\{\bcommon i,\bbracket{1:i}\}$.
\end{proof}
The re-factorized form (\ref{eq:Ternary=00003DFactorization}) for
a special case, namely HMC model, will be illustrated in Fig. \ref{fig:FB=00003Db}
in Section \ref{sec:chap6:FB and VA via CI}.

\subsubsection{CI topology versus CI factorization}

Note that, the CI topological structure via FB recursion (\ref{eq:Ternary=00003DGDL}-\ref{eq:ch5:GDL:BW})
is a computational technique, while the CI factorization via chain
rule (\ref{eq:Ternary=00003DFactorization}) is a probabilistic methodology.
In order words, the former involves quantitative values and practical
implementation, while the latter gives us insights about model characteristic.
Nevertheless, both of them yields the same result under GDL, as shown
next.

For illustration , let us consider the pre-semiring $([0,1]^{\calX^{\ndata}},\SUM,\PROD)$.
Then, the inference tasks $\underset{\xBold_{\SET}}{\boxplus}f(\xBold_{\OMEGA})=\underset{\xBold_{\SET}}{\boxplus}\PROD_{i=1}^{n}f_{\pi(i)}$,
as summarized in Table \ref{tab:Pre-semirings}, can be computed via
two equivalent forms, the original form (\ref{eq:Binary}) and the
ternary factorization (\ref{eq:Ternary=00003DFactorization}). Applying
GDL to these two forms, respectively, we have:

\begin{eqnarray}
\underset{\xBold_{\SET}}{\boxplus}f(\xBold_{\OMEGA}) & = & \underset{\xBold_{s_{\bcommon i}}}{\boxplus}(\underset{\xBold_{s_{\bparent{i+1:n}}}}{\boxplus}\gBold_{i+1:n}\PROD\underset{\xBold_{s_{\bbracket{1:i}}}}{\boxplus}\gBold_{1:i})\label{eq:Binary-Form}
\end{eqnarray}
and:

\begin{equation}
\underset{\xBold_{\SET}}{\boxplus}f(\xBold_{\OMEGA})=\underset{\xBold_{s_{\bcommon i}}}{\boxplus}(f(\xHat_{\bparent{i+1:n}}|\xBold_{\bcommon i})\PROD f(\xBold_{\bcommon i})\PROD f(\xHat_{\bbracket{1:i}}|\xBold_{\bcommon i}))\label{eq:Ternary-Form}
\end{equation}
where $f(\xHat_{\bparent{i+1:n}}|\xBold_{\bcommon i})\TRIANGLEQ\mbox{\ensuremath{\underset{\xBold_{s_{\bparent{i+1:n}}}}{\boxplus}}}f(\xBold_{\bparent{i+1:n}}|\xBold_{\bcommon i})$
and $f(\xHat_{\bbracket{1:i}}|\xBold_{\bcommon i})\TRIANGLEQ\underset{\xBold_{s_{\bbracket{1:i}}}}{\boxplus}f(\xBold_{\bbracket{1:i}}|\xBold_{\bcommon i})$. 

Comparing (\ref{eq:Binary-Form}) with (\ref{eq:Ternary-Form}), we
can see that the result of GDL applied to the original form (\ref{eq:Binary})
is equivalent to the result of GDL applied to the re-factorization
form (\ref{eq:Ternary=00003DFactorization}), without the need to
compute that re-factorization form (\ref{eq:Ternary-Form}). This
equivalence is useful when computing sequential objective functions
$\underset{\xBold_{\SET_{j}}}{\boxplus}f(\xBold_{\OMEGA})$. For example,
the $n$ scalar marginals can be computed directly via the original
form (\ref{eq:Marginal}), without the need to derive the re-factorization
form (\ref{eq:Ternary=00003DFactorization}). 

\section{Summary}

In this chapter, the generalized distributive law (GDL) was revisited
and new insights were gained from a topological perspective. Let us
summarize here three main achievements of this new perspective:

Firstly, we have defined the GDL via an abstract algebra for functions,
rather than the approach using variables in the literature. Hence,
it was feasible to show that the GDL always reduces the total number
of operators, when applicable. 

Secondly, by separating the concept of operator indices from variable
indices, we have applied set algebra and set up a conditionally independent
(CI) structure for the original model. This topological CI structure
was also shown to be equivalent to CI factorization in the probability
context. Hence, the GDL is better understood as a tool exploiting
original CI structures, rather than being a cause of that CI structure.
Conversely, the design of CI structures, embedded in the original
model, can be guided by the amount of reduction achieved when applying
the GDL.

Thirdly, a new computational structure, namely FB recursion, for the
GDL was also designed. When applied to Bayesian inference, the FB
recursion is also a generalized form of well-known algorithms, such
as the Forward-Backward (FB) algorithm and Viterbi algorithm (VA),
both of which we will study in Chapter \ref{=00005BChapter 6=00005D}.
Furthermore, a new interpretation of entropy computation via the GDL
was also provided. This interpretation will be useful in understanding
the relationship between the Viterbi algorithm (VA) and Variational
Bayes (VB) approximation for the hidden Markov chain (HMC) in the
next chapter.

%auto-ignore
%auto-ignore
%%%% Common

\global\long\def\REAL{\mathbb{R}}%

\global\long\def\DEAL{\mathbb{D}}%

\global\long\def\COMPLEX{\mathbb{C}}%

\global\long\def\RING{\mathcal{R}}%

\global\long\def\MSET{\mathcal{M}}%

\global\long\def\ASET{\mathcal{A}}%

\global\long\def\OMEGA{\Omega}%

\global\long\def\calO{{\cal O}}%

\global\long\def\calX{\mathcal{X}}%

\global\long\def\ndata{n}%

\global\long\def\nstate{M}%

\global\long\def\itime{i}%

\global\long\def\istate{k}%

\global\long\def\funh#1{h\left(#1\right)}%

\global\long\def\fung#1{g\left(#1\right)}%

\global\long\def\seti#1#2{#1\in\{1,2,\ldots,#2\}}%

\global\long\def\setd#1#2{\{#1{}_{1},#1{}_{2},\ldots,#1_{#2}\}}%

\global\long\def\TRIANGLEQ{\triangleq}%

%%%% Chapter 6

\global\long\def\fBold{\boldsymbol{f}}%

\global\long\def\TBold{\mathbf{T}}%

\global\long\def\LBold{L}%

\global\long\def\fprofile{f_{p}}%

\global\long\def\fdelta{f_{\delta}}%

\global\long\def\ftdelta{\widetilde{f}_{\delta}}%

\global\long\def\fdbar{\overline{f_{\delta}}}%

\global\long\def\WBold{\mathbf{W}}%

\global\long\def\XBold{\mathbf{X}}%

\global\long\def\ftilde{\widetilde{f}}%

\global\long\def\fnutilde#1{\widetilde{f}^{[#1]}}%

\global\long\def\fctilde{\widetilde{f}_{c}}%

\global\long\def\fbar{\overline{f}}%

\global\long\def\fhat{\widehat{f}}%

\global\long\def\fpbar{\overline{\fprofile}}%

\global\long\def\xBold{\mathbf{x}}%

\global\long\def\sBold{\mathbf{s}}%

\global\long\def\oneBold{\boldsymbol{1}}%

\global\long\def\btheta{\vtheta}%

\global\long\def\ntheta{n}%

\global\long\def\iIVB{\nu}%

\global\long\def\nIVB{\nu_{c}}%

\global\long\def\eIVB{\nu_{e}}%

\global\long\def\KLDVB{KLD_{VB}}%

\global\long\def\KLDFCVB{KLD_{FCVB}}%

\global\long\def\Fc{\mathbb{F}_{c}}%

\global\long\def\Ffc{\mathbb{F}_{f.c}}%

\global\long\def\Normal{\mathcal{N}}%

\global\long\def\Oc{\mathcal{O}}%

\global\long\def\Rc{\mathbb{R}}%

\global\long\def\Psi{\psi}%

\global\long\def\LAMBDA{\boldsymbol{\Lambda}}%

\global\long\def\lAMBDA{\boldsymbol{\lambda}}%

\global\long\def\khat{\widehat{k}}%

\global\long\def\lhat{\widehat{l}}%

\global\long\def\lVA{\widehat{l}^{(VA)}}%

\global\long\def\Lhat{\widehat{L}}%

\global\long\def\LVA{\widehat{L}^{(VA)}}%

\global\long\def\tradVBML{tradVB_{(ML)}}%

\global\long\def\tradFCVBML{tradFCVB_{(ML)}}%

\global\long\def\VBML{VB_{(ML)}}%

\global\long\def\FCVBML{FCVB_{(ML)}}%

\global\long\def\Vst{VB_{(ML)}^{1st}}%

\global\long\def\Fst{FCVB_{(ML)}^{1st}}%

\global\long\def\Mu{Mu}%

\global\long\def\iton{i\in\{1,\ldots,n\}}%

\global\long\def\iinn{i\in\{1,\ldots,n\}}%

\global\long\def\gbar{\overline{g}}%

\global\long\def\fDoppler{f_{D}}%

\global\long\def\Tsample{T}%

\chapter{Variational Bayes variants of the Viterbi algorithm \label{=00005BChapter 6=00005D}}

\section{Introduction}

For state inference of a Hidden Markov Chain (HMC) with known parameters,
we will study, in this chapter, four well known algorithms in the
literature, corresponding to a trade-off between performance and computational
load: 

- Forward-Backward (FB) algorithm can compute the exact marginal posterior
distributions recursively, yet the cost is typically prohibitive in
practice. 

- By confining the inference problem to certainty equivalent (CE)
estimate (Section \ref{subsec:chap4:CE-approx}), the Viterbi algorithm
(VA) is able to compute recursively the exact joint Maximum-a-posteriori
(MAP) state trajectory estimate, with acceptable complexity. 

- Further restricting CE estimate to a local joint MAP, the Iterated
Conditional Modes (ICM) algorithm is even faster than VA, yet ICM's
reliability is undermined because of the dependence on initialization
and a lack of understanding of the methodology currently. 

- Maximum Likelihood (ML) is the fastest estimation method, but neglects
the Markov structure of the hidden field and consequently has the
worst performance.

In this chapter, we will re-interpret these methods within a fully
Bayesian perspective:

- FB will be shown to be a consequence of FB factorization of the
posterior distribution, which is an inhomogeneous HMC. 

- VA actually returns shaping parameters of another HMC approximation,
whose joint MAP estimate is equal to the exact joint MAP trajectory
of posterior. This novel Bayesian interpretation of VA not only reveals
the nature of VA, but also opens up an approximation framework for
HMC. 

- As a variant of VA, but further confined to the independent class
of hidden field posterior, Variational Bayes (VB) approximation is
a reasonable choice for the conditionally independent (CI) structure
of HMC posterior. Owing to this CI structure, a novel speed-up scheme
for iterative VB (IVB) algorithm in VB method will also be proposed
in the chapter. 

- Finally, ICM will be shown to be equivalent to the so-called functionally
constrained VB (FCVB) approximation. 

\subsection{A brief literature review of the Hidden Markov Chain (HMC)}

For many decades, the first-order Hidden Markov model (HMM) has been
widely used as a stochastic model for the dependent (dynamic) sequential
data. The fundamental problems of HMM are to infer both its parameters
and the latent variables. For general treatment of all kind of HMM,
we refer to the textbooks {[}\citet{ch6:bk:HMM:Cappe05,ch6:bk:HMM:Schnatter06}{]},
which have a thorough review of HMMs in the literature. 

Throughout the chapter, we focus on the simplest case of HMM, namely
finite state homogeneous HMC with known parameters. Despite simplicity,
this model has been used successfully in various application domains,
e.g. speech processing {[}\citet{ch6:Art:HMM:speech:Rabiner89}{]},
digital communication {[}\citet{ch6:origin:BCJR:Bahl74,ch2:bk:ToddMoon}{]}
and image analysis {[}\citet{ch6:art:HMM:image:Li_2000}{]}.

The label inference of Markov chain became recursively tractable owing
to Forward-Backward (FB) algorithm, firstly proposed by Baum et al
{[}\citet{ch6:origin:FB:Baum70}{]}. In their Baum-Welch algorithm
(currently known as the Expectation-Maximization (EM) algorithm for
HMC with unknown transition matrix {[}\citet{ch6:bk:HMM:Cappe05}{]}),
the FB algorithms is used as an Expectation step for label field.
FB algorithm was also discovered in other fields under different names,
such as BCJR algorithm {[}\citet{ch6:origin:BCJR:Bahl74}{]} in channel
decoders (as reviewed in Section \ref{subsec:chap2:Stream-Code}),
Kalman filtering and smoothing (two-filters formula) {[}\citet{ch6:art:Kalman:two_filters_69}{]}
in Gaussian linear state space model and the sum-product algorithm
{[}\citet{ch6:art:SumProduct:Pearl88}{]} in graphical learning. 

For HMC, the recursive marginalization in FB, however, are slow and
become a serious problem in applications requiring a fast estimation
method. Hence, the point-estimate-based Viterbi algorithm (VA), firstly
proposed in {[}\citet{ch6:origin:VA:Viterbi67}{]}, was designed to
recursively evaluate the true maximum-a-posteriori (MAP) of joint
trajectory. By replacing marginalization with maximization, VA can
be computed much more quickly than FB, which requires all marginal
inference of each label. Owing to efficient recursive computation,
the application of VA is vast (see for example the history of VA in
{[}\citet{ch6:ART:VA:history06}{]}). VA is often presented via the
so-called weighted length in trellis diagram, a concept in graphical
learning, as firstly formalized in {[}\citet{ch2:origin:VA:Forney73}{]}.
This approach does not explain its relationship with FB properly,
and also lack important insight of its approximated property.

The fast Iterated Conditional Modes (ICM) algorithm, firstly proposed
in {[}\citet{ch4:art:ICM:Besag86}{]}, is widely used in two scenarios.
The first one is Markov random fields {[}\citet{ch4:bk:lossHamming:image95,ch6:PhD:ICM:HMM:Dauwels05,ch4:art:ICM:localMAP_2006}{]},
in which ICM is applied to finding local joint MAP of the hidden label
field with low computational load {[}\citet{ch6:art:ICM:VQ_channel_10}{]}.
The second scenario is Expectation Conditional Maximization (ECM)
algorithm, in which ICM is used to replace the M-step in the Expectation
Maximization (EM) algorithm {[}\citet{ch6:art:ECM:mixture:08}{]}.
However, ECM has been deployed only for non-closed forms of M-step,
with ICM used instead as a closed-form approximation. To the best
of our knowledge, the material in this chapter is the first to study
ICM as a closed-form approximation for the HMC with known parameters,
and to characterize ICM as a VB variant. 

\subsection{The aims of this chapter}

In this chapter, we will provide a deterministic Bayesian approximation
framework for label inference in the HMC, and study the trade-off
between performance and computational load. 

Firstly, the FB algorithm will be presented as a factorization scheme
for an inhomogeneous HMC posterior. 

Then, we will show that VA is a sparse CE-based HMC approximation
of original HMC, in which their joint MAPs are undisturbed. Because
tracking the joint MAP in that sparse HMC is much faster than in the
original HMC, this Bayesian perspective does not only reveal the core-trick
of complexity reduction in VA, but also motivates another Bayesian
approximation, namely a Variational Bayes (VB) approximation from
mean field theory. 

Fundamentally, VB seeks an approximating distribution within the independent
functional class, such that its Kullback-Leibler divergence (KLD)
to the original distribution is minimized. In the literature, VB methodology
has been applied successfully to intractable inference of HMM with
unknown parameters {[}\citet{ch6:art:VB:filtering:AQuinn08,ch6:art:VB:HMM:Titterington09}{]}.
Although the Markov chain with known parameters in this paper is completely
tractable, we still use, for the first time, VB for HMC label inference
as an attempt to further reduce the computational load. 

Furthermore, a novel accelerated scheme will be proposed in order
to reduce computational load of iterative VB algorithm significantly.
In CI structures such as HMC, this accelerated scheme can reduce the
total number of IVB cycles to a factor of $\log(n)$, with $n$ denoting
the number of labels.

As a consequence, a functionally constrained VB (FCVB) approximation
will be developed to produce a local joint MAP estimate for hidden
label field, based on iterative CE propagation among all of VB marginal
distributions. This FCVB scheme will be shown to be equivalent to
ICM algorithm. The virtue of the FCVB optic will be helpful to understand
the property of ICM. Moreover, it will allow us to inherit a novel
accelerated scheme arising in the VB approximation for the HMC model.
The accelerated FCVB constitutes a faster version of ICM. 

In simulation, the performance of FCVB will be shown to be comparable
to that of VA when the transition probabilities in HMC are not too
correlated (i.e. when the correlation coefficient between any two
simulated transition probabilities for the HMC is not too close to
one in magnitude). Note that, FCVB is an iterative scheme, while VA
is not. Notwithstanding this, the independent structure of FCVB makes
its computation per iteration much lower than VA, yielding a much
reduced net computational load. 

Finally, we will briefly recall the Bayesian risk theory of Hamming
distance criterion, which was also reviewed in {[}\citet{ch4:bk:lossHamming:image95,ch4:art:JuriLember11}{]}.
This will allow us to explain the performance ranking we find under
simulations in Chapter \ref{=00005BChapter 8=00005D}. From best to
worst, they are FB, VA, VB and FCVB algorithms.

\section{The Hidden Markov Chain (HMC) \label{sec:chap6:HMC}}

Assume that we receive a sequence of data $\xBold_{n}\TRIANGLEQ[x_{1},\dots,x_{n}]'$
(the observed field), where $x_{i}\TRIANGLEQ x[i]\in\Rc$ are samples
at discrete time point $i\in\{1,\dots,n\}$. Let us consider two simple
stochastic models for $\xBold_{\ndata}$, as follows: 
\begin{itemize}
\item The simplest model for $\xBold_{n}$ is independent identical distributed
(iid) random variables. This model, however, is too strict and neglects
the dependent structure of interest in data.
\item The next relaxation for $\xBold_{n}$ is the non-identical one, i.e.
the independently distributed (id) case, in which we assume $x_{i}$
is sampled from one of $M$ classes of known observation distributions
$f_{k}(x_{i})$, $k\in\{1,\ldots,M\}$. At each time $\seti in$,
let us define a soft classification $M\times1$ vector (an $M$-dimensional
statistic), as follows: 
\begin{equation}
\Psi(x_{i})\TRIANGLEQ[f_{1}(x_{i}),\ldots,f_{M}(x_{i})]'\label{eq:ch6:DEF:psi}
\end{equation}
Furthermore, let us define the label vector $l_{i}\in\{\boldsymbol{\epsilon}(1),\ldots,\boldsymbol{\epsilon}(M)\}$
pointing to the state, $\istate$, of $x_{\itime}$, where $\boldsymbol{\epsilon}(k)$
is the $\istate$th elementary vector: 
\[
\boldsymbol{\epsilon}(k)=\left[\delta[k-1],\dots,\delta[k-M]\right]'
\]
and $\delta[\cdot]$ is Kronecker-$\delta$ function. Owing to the
id structure, the observation model, given the matrix $L_{n}\TRIANGLEQ[l_{1},\ldots,l_{n}]$,
is 
\begin{equation}
f(\xBold_{n}|L_{n})=\prod_{i=1}^{n}f(x_{i}|l_{i})=\exp(\sum_{i=1}^{n}l_{i}'\log\Psi(x_{i}))\label{eq:ch6:Observation}
\end{equation}
where, akin to Matlab convention, operators such as exp and $\log$
are taken element-wise. Note that, we adopt the vector form in right
hand side of (\ref{eq:ch6:Observation}) in order to emphasize its
Exponential Family (EF) structure, as defined in (\ref{eq:ch4:EF:obs}).
\end{itemize}
Then, for the id case, the task of inferring the class of $x_{i}$
is equivalent to inference task of $l_{i}$ in (\ref{eq:ch6:Observation}).
For this purpose, let us consider two simple prior models for label
sequence: 
\begin{itemize}
\item Once again, the simplest model for discrete $l_{i}$ is iid sampling
of the multinomial distribution with known probabilities, $p\TRIANGLEQ[p_{1},\ldots,p_{M}]'$
in the probability simplex (i.e. the sum-to-one simplex), as follows:
\[
f(l_{i}|p)=Mu_{l_{i}}(p)=l_{i}'p=\exp(l_{i}'\log p)
\]
where the vector form is adopted in order to emphasize the CEF form,
defined in (\ref{eq:ch4:EF:prior}). In this thesis, the notation
$Mu_{l_{i}}(p)\equiv Mu_{l_{i}}(1,p)$ denotes multinomial distribution
with one realization in total. Note that, in this case, the $\itime$th
observation model $f(x_{i})$ is a mixture of $M$ known components:
\[
f(x_{i})=\sum_{l_{i}}f(x_{i}|l_{i})f(l_{i}|p)=p'\Psi(x_{i})
\]
Owing to id structure of $\xBold_{n}$ and conjugacy in EF, the posterior
distribution of $l_{i}$ also belongs to id distribution class, as
follows: 
\[
f(l_{i}|\xBold_{n},p)=f(l_{i}|x_{i},p)=Mu_{l_{i}}(\gamma_{i})\propto f(x_{i}|l_{i})f(l_{i}|p)
\]
 in which the length-$\nstate$ vector $\gamma_{i}$ is the shaping
parameter and $\gamma_{i}\propto\Psi(x_{i})\circ p$, $\iton$. The
normalization constant is derived by noting that $\sum_{\istate=1}^{\nstate}\gamma_{\istate,i}=1$,
notation $\propto$ denotes a sum-to-one operator and $\circ$ is
the Hadamard product.
\item However, the above independent mixture model is too strict in practice,
since it neglects the temporal dependence. A widely adopted dependent
model for $l_{i}$ is the homogeneous HMC: 
\begin{equation}
f(\LBold_{n}|\TBold,p)=\prod_{i=2}^{n}f(l{}_{i}|l_{i-1},\TBold)f(l_{1}|p)\label{eq:ch6:PRIOR}
\end{equation}
in which the known parameters are $\nstate\times\nstate$ transition
matrix $\TBold$ with sum-to-one columns (i.e. a left stochastic matrix)
and initial probability vector $p$ in the probability simplex, as
illustrated by the Directed Acyclic Graph (DAG) in Fig. \ref{fig:DAG}.
By definition of the HMC, we have: 
\end{itemize}
\begin{eqnarray}
f(l_{i}|l_{i-1},\TBold) & = & Mu_{l_{i}}(\mbox{\ensuremath{\TBold}}l_{i-1}),\ i\in\{2,\ldots n\}\label{eq:HMC}\\
f(l_{1}|p) & = & Mu_{l_{1}}(p)\nonumber 
\end{eqnarray}

Although the posterior probability $f(\LBold_{n}|\xBold_{n})$ of
individual joint trajectories $L_{n}$ in above HMC can be computed
feasibly, the full inference of $L_{n}$ is intractable owing to the
exponential increase in the number $M^{n}$ of joint trajectories
with $\ndata$, a problem referred to as the curse of dimensionality
{[}\citet{ch6:bk:curse_dimension:Karny97}{]}. This remark will be
clarified in the sequel. The discovery of a tractable Bayesian methodology
for this problem will be central to this chapter. Note that, the key
idea behind computational reduction in this case is simply to confine
our inference task to special cases, and to avoid computing all $\nstate^{\ndata}$
joint posterior probabilities $f(\LBold_{n}|\xBold_{n})$. Those confined
inferences can be marginal distributions, point estimates or distributional
approximations.

Throughout the chapter, there will be some convenient conventions
for shortening notations: $\{\TBold,p\}$ will be omitted occasionally,
e.g. $f(l_{i}|l_{i-1})\equiv f(l_{i}|l_{i-1},\TBold),$ when the context
is clear. The $M\times1$ zero and unity vectors are defined as $\boldsymbol{0}_{M\times1}$
and $\boldsymbol{1}_{M\times1}$, respectively. The range of the index
will be denoted by subscripts, e.g.: $\xBold_{i+1:n}\equiv[x_{i+1},\ldots,x_{n}]'$,
$L_{i+1n}=[l_{i+1},\ldots,l_{n}]$. When $i\notin\{1,\dots n\}$,
empty set convention will be applied, e.g. $L_{n+1:n}=\emptyset$,
and $L_{0}=\emptyset$. The value of arbitrary distribution $f(\theta)$
at $\widehat{\theta}$ will be denoted as $f(\widehat{\theta})\equiv f(\theta=\widehat{\theta})$.
For avoiding ambiguity, point estimates such as $\widehat{l_{i}}$
or $\widehat{L_{n}}$ will be re-defined as marginal MAP, joint MAP,
etc. separately in each section. For computational load, let us denote
exponential, multiplication, addition and maximization operators as
$EXP$, $MUL$, $ADD$ and $MAX$, respectively. 

\begin{figure}
\begin{centering}
\includegraphics[width=0.6\columnwidth]{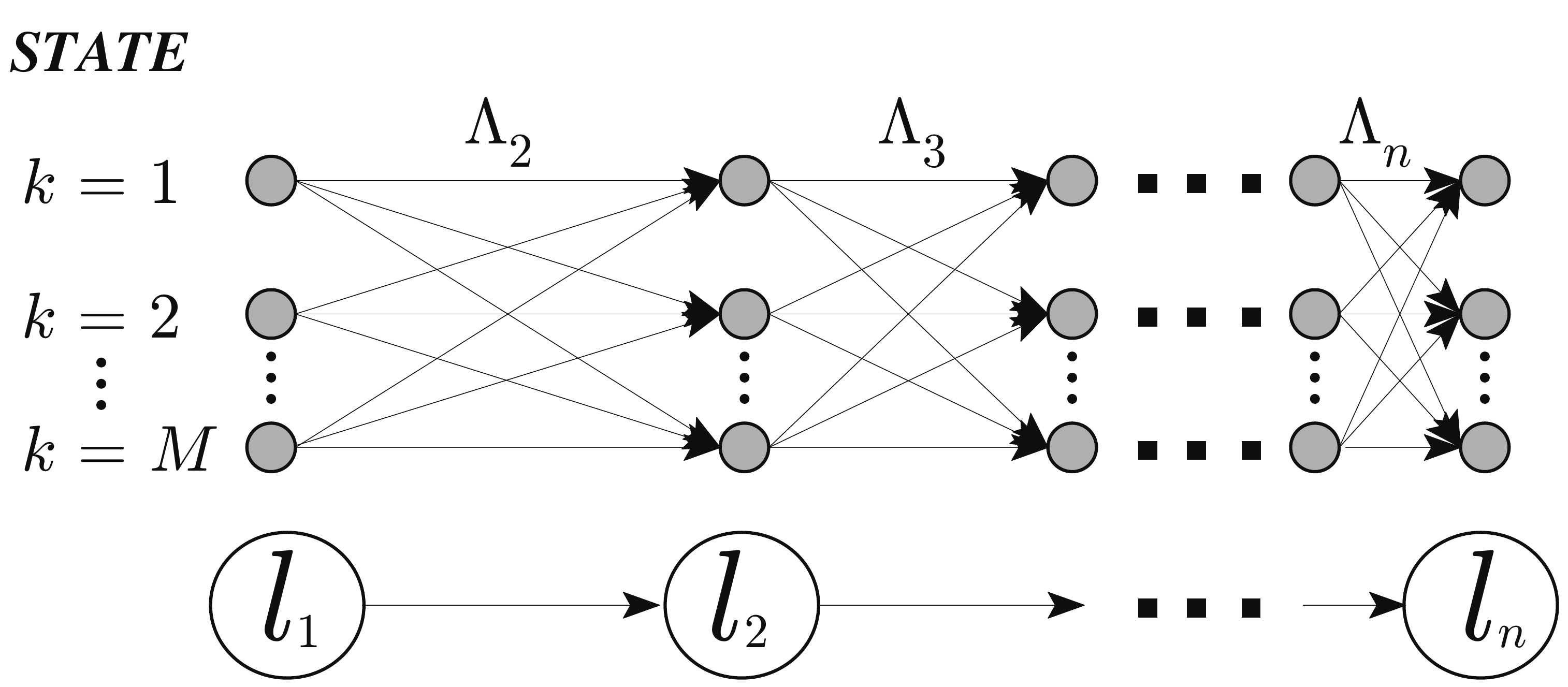}
\par\end{centering}
\begin{centering}
\includegraphics[width=0.6\columnwidth]{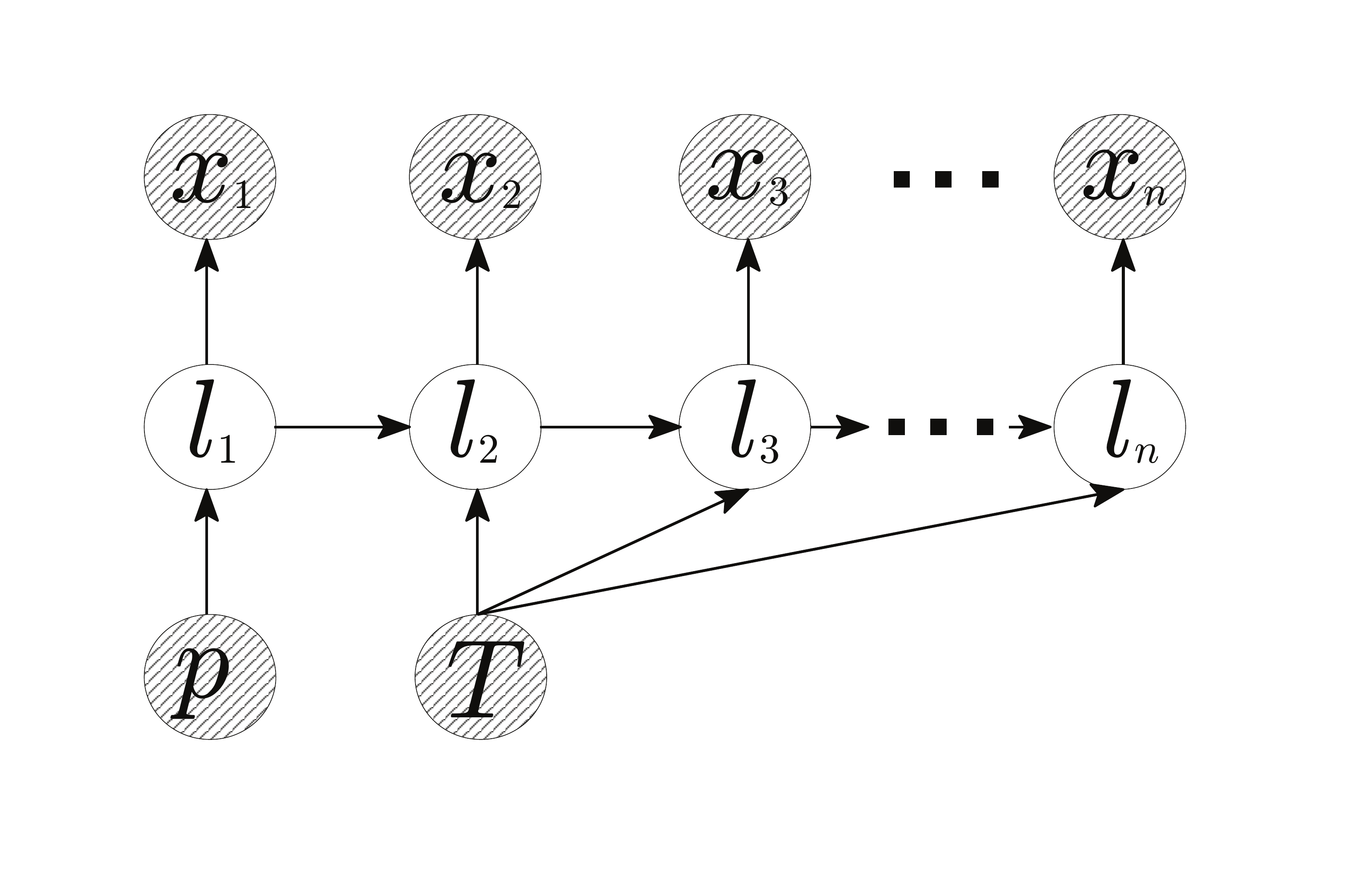}
\par\end{centering}
\caption{\label{fig:DAG} Trellis diagram (top) and DAG (bottom) for HMC. $\protect\LAMBDA_{i}$
are the transition metric lengths at time $\protect\itime$. }
\end{figure}

\subsection{Sequence of marginals for general label field}

The purpose of this sub-section is to study the computational load
when we confine the inference from the $\nstate^{\ndata}$-term joint
$f(\LBold_{n}|\xBold_{n})$ to just $n$ smoothing inference $f(l_{i}|\xBold_{n})=\sum_{L_{\backslash i}}f(L_{n}|\xBold_{n})$,
where $L_{\backslash i}$ is the complement of $l_{i}$ in $L_{n}$,
$i\in\{1,\ldots,n\}$. The general (not necessarily Markov-constrained)
model $f(\xBold_{n},L_{n})$ for label field $L_{\ndata}$ will be
studied in this sub-section, while the HMC model will be studied in
next sub-section. 

\subsubsection{Scalar factorization for label field}

Firstly, let us investigate why direct marginalization over the joint
$f(L_{n}|\xBold_{n})\propto f(\xBold_{n},L_{n})$ is intractable.
By the general chain rule, any general model $f(\xBold_{n},L_{n})$
can be factorized into scalar factors, as follows:

\begin{equation}
f(\xBold_{n},L_{n})=\prod_{i=2}^{n}f(x_{i},l_{i}|\xBold_{i-1},L_{i-1})f(x_{1},l_{1})\label{eq:FACTOR}
\end{equation}

Now, let us examine the cost of computing the sequence of posterior
marginals: because there are $n$ multiplication factors in (\ref{eq:FACTOR}),
the probability $f(\xBold_{n},L_{n})$ of each trajectory $L_{n}$,
given $\xBold_{\ndata}$, needs $O(n)$ of MUL operators. In order
to marginalize out $L_{\backslash i}$, for each $f(l_{i}|\xBold_{n})$,
we have to evaluate all probabilities of $M^{n-1}$ trajectories $L_{\backslash i}$,
i.e. $M^{n-1}\times O(n)\approx O(nM^{n})$ of MULs in total. Finally,
for $n$ smoothings $f(l_{i}|\xBold_{n})$, we would need $\ndata\times O(nM^{n-1})\approx O(n^{2}M^{n})$
of MULs, which increases exponentially with $n$. 

\subsubsection{Binary partitions for label field}

Apart from scalar factors, the general model (\ref{eq:FACTOR}) can
also be factorized into two forward and backward sub-trajectories,
$L_{i}$ and $L_{i+1:n}$, respectively:

\begin{equation}
f(\xBold_{n},L_{n})=f(\xBold_{i+1:n},L_{i+1:n}|\xBold_{i},L_{i})f(\xBold_{i},L_{i})\label{eq:FB=00003Dgeneral}
\end{equation}
where $i\in\{1,\ldots,n\}$. Those two sub-trajectories, in turn,
can be binarily factorized in the manner of a binary tree:

\begin{eqnarray}
f(\xBold_{i},L_{i}) & = & f(x_{i},l_{i}|\xBold_{i-1},L_{i-1})f(\xBold_{i-1},L_{i-1})\label{eq:FWBW=00003Dgeneral}\\
f(\xBold_{i:n},L_{i:n}|\xBold_{i-1},L_{i-1}) & = & f(x_{i},l_{i}|\xBold_{i-1},L_{i-1})f(\xBold_{i+1:n},L_{i+1:n}|\xBold_{i},L_{i})\nonumber 
\end{eqnarray}

Note that, the cost of computing the sequence of posterior marginals
via FB factorization (\ref{eq:FWBW=00003Dgeneral}) is at least the
same as that cost in scalar factorization in (\ref{eq:FACTOR}), because
there is no CI structure for $L_{\ndata}$ in original $\ndata$ factors
(\ref{eq:FACTOR}). Nevertheless, when the general model (\ref{eq:FB=00003Dgeneral})
is specialized to the HMC model, the FB factorization (\ref{eq:FWBW=00003Dgeneral})
will lead to the tractable FB algorithm, as explained below. 

\subsection{Sequence of marginals for the HMC \label{subsec:chap6:FB-HMC}}

By exploiting Markovianity in the HMC model, the FB algorithm, firstly
proposed in {[}\citet{ch6:origin:FB:Baum70}{]}, computes all smoothings
$f(l_{i}|\xBold_{n})$ in a recursive way, without the need to compute
the $M^{n}$ values of $f(L_{n}|\xBold_{n})$ explicitly. In this
sub-section, the traditional FB algorithm will be re-interpreted as
a recursive update of Bayesian sufficient statistics. This fully Bayesian
treatment will be helpful in understanding the underlying methodology,
which we will later present in Section \ref{sec:Point-estimation}
on point estimation.

The recursive FB algorithm, as shown below, can divide the trajectory
$L_{n}$ into sub-trajectories and reduce the complexity down from
exponential form $O(n^{2}M^{n})$ in (\ref{eq:FACTOR}) down to linear
form $\Oc(2nM^{2})$ of MUL. For this purpose, let us study the scalar
factorization in our HMC model first. 

\subsubsection{Markovianity}

Multiplying the id observation model (\ref{eq:ch6:Observation}) by
the HMC prior (\ref{eq:ch6:PRIOR}), the joint distribution for the
HMC context is:

\begin{eqnarray}
f(\xBold_{n},L_{n}|\TBold,p) & = & f(\xBold_{n}|\LBold_{n})f(\LBold_{n}|\TBold,p)\label{eq:Joint}\\
 & = & \prod_{i=2}^{n}f(x_{i},l_{i}|l_{i-1},\TBold)f(x_{1},l_{1}|p)\nonumber 
\end{eqnarray}
in which we have exploited Markovianity of $L_{\ndata}$. The augmented
form for $x_{i}$ and $l_{i}$, $i\in\{1,\ldots,n\}$, is: 

\begin{eqnarray}
f(x_{1},l_{1}|p) & = & f(x_{1}|l_{1})f(l_{1}|p)\label{eq:augment}\\
f(x_{i},l_{i}|l_{i-1},\TBold) & = & f(x_{i}|l_{i})f(l_{i}|l_{i-1},\TBold)\nonumber 
\end{eqnarray}

By comparing the general model (\ref{eq:FACTOR}) with the HMC model
(\ref{eq:Joint}), we can recognize the following Markov property,
which will be exploited throughout the chapter:

\begin{equation}
f(x_{i},l_{i}|l_{i-1})=f(x_{i},l_{i}|\xBold_{i-1},L_{i-1})\label{eq:MarkovProperty}
\end{equation}

\subsubsection{FB recursion}

Let us substitute Markov property (\ref{eq:MarkovProperty}) into
(\ref{eq:FWBW=00003Dgeneral}) and marginalize the result in (\ref{eq:FB=00003Dgeneral}).
In this way, we can evaluate the smoothing marginals $f(l_{i}|\xBold_{n})\propto\sum_{L_{\backslash i}}f(\xBold_{n},L_{n})$
via generalized distributive law (GDL) (Section \ref{sec:chap5:GDL-for-objective}),
as follows:

\begin{equation}
f(l_{i}|\xBold_{n})\propto f(\xBold_{i+1:n}|l_{i})f(l_{i}|\xBold_{i})\label{eq:FB=00003DHMC}
\end{equation}
in which the two marginalized sub-trajectories $f(l_{i}|\xBold_{i})\propto\sum_{L_{i-1}}f(\xBold_{i},L_{i})$
and $f(\xBold_{i+1:n}|l_{i})=\sum_{L_{i+1:n}}f(\xBold_{i+1:n},L_{i+1:n}|l_{i})$,
in turn, can be computed recursively:

\begin{eqnarray}
f(l_{i}|\xBold_{i}) & \propto & \sum_{l_{i-1}}f(x_{i},l_{i}|l_{i-1})f(l_{i-1}|\xBold_{i-1})\label{eq:FWBW=00003DHMC}\\
f(\xBold_{i:n}|l_{i-1}) & = & \sum_{l_{i}}f(x_{i},l_{i}|l_{i-1})f(\xBold_{i+1:n}|l_{i})\nonumber 
\end{eqnarray}
where $i\in\{1,\ldots,n\}$. By replacing direct marginalizations
of $L_{i-1}$ and $L_{i+1:n}$ in (\ref{eq:FWBW=00003Dgeneral}) with
the chain rule for marginalization (\ref{eq:FWBW=00003DHMC}), we
have greatly reduced the total computational load of $O(M^{n-i})$
and $O(M^{i-1})$, for each time $i$, down to $O(M^{2})$ and $O(M^{2})$,
respectively. Hence, the cost for all $n$ smoothings $f(l_{i}|\xBold_{n})$
is $O(2nM^{2})$. 
\begin{rem}
\label{Remark=00003DFB=00003DHMC} Note that, recognizing Markovianity
(\ref{eq:MarkovProperty}) is the vital step for this scheme. Otherwise,
the distributive law cannot be applied to general model (\ref{eq:FB=00003Dgeneral})
to yield (\ref{eq:FB=00003DHMC}). In this sense, the FB factorization
(\ref{eq:FB=00003Dgeneral}-\ref{eq:FWBW=00003Dgeneral}) is a natural
way to exploit the conditional independent (CI) structure in the HMC,
thereby reducing the complexity via the generalized distributive law
(GDL) (see Sections \ref{sec:chap5:GDL-for-objective}, \ref{subsec:chap5:Sequence-of-marginals}).
\end{rem}

\subsubsection{FB algorithm}

From (\ref{eq:HMC}), (\ref{eq:augment}) and (\ref{eq:FWBW=00003DHMC}),
the shaping parameters $\alpha_{i}$ for filtering marginals, $f(l_{i}|\xBold_{i})=Mu_{l_{i}}(\alpha_{i})$,
as well as the un-normalized length-$\nstate$ vector statistics $\beta_{i}\equiv\beta_{i}(\xBold_{i+1:n})$
governing the backward observations model $f(\xBold_{i+1:n}|l_{i})=l_{i}'\beta_{i}$,
$i=1,\ldots,n$, can be evaluated recursively and in parallel, as
follows:

\begin{eqnarray}
\alpha_{1} & \propto & \Psi_{1}\circ p\label{eq:alpha}\\
\alpha_{i} & \propto & \Psi_{i}\circ(\TBold\alpha_{i-1}),\ i=2,\ldots,n\nonumber 
\end{eqnarray}
where $\Psi_{i}$ are the soft-classification vectors (\ref{eq:ch6:DEF:psi}),
and:

\begin{eqnarray}
\beta_{n} & = & \boldsymbol{1}_{M\times1}\label{eq:beta}\\
\beta_{i} & = & \TBold'(\Psi_{i}\circ\beta_{i+1}),\ i=n-1,\dots,1\nonumber 
\end{eqnarray}

By substituting (\ref{eq:alpha}-\ref{eq:beta}) into (\ref{eq:FB=00003DHMC}),
the shaping parameters, $\gamma_{i}$, of the smoothing marginals,
$f(l_{i}|\xBold_{n})=Mu_{l_{i}}(\gamma_{i})$, can be readily evaluated:
\begin{equation}
\gamma_{i}\propto\beta_{i}\circ\alpha_{i},\ i=\{1,\ldots,n\}\label{eq:gamma}
\end{equation}
.

The FB algorithm, firstly proposed\textcolor{red}{{} }in {[}\citet{ch6:origin:FB:Baum70}{]},
consists of simultaneous forward (\ref{eq:alpha}) and backward (\ref{eq:beta})
recursions for computing $\alpha_{i}$ and $\beta_{i}$, respectively.
However, the $\beta_{\itime}$ evaluation (\ref{eq:beta}) typically
incurs a memory overflow. Stabilization is achieved via a normalization
step, which was firstly proposed in {[}\citet{ch6:Art:HMM:speech:Rabiner89}{]},
as shown in Algorithm \ref{alg:chap6:FB-algorithm}.

\begin{algorithm}
\textbf{Storage:} $2n$ length-$\nstate$ vectors $\alpha_{i}$, $\beta_{i}$
in (\ref{eq:alpha}-\ref{eq:beta}) 

\textbf{Recursion}: evaluate (\ref{eq:alpha}-\ref{eq:beta}), and
normalize $\beta_{i}:=\beta_{i}$ /$\sum_{k=1}^{M}\beta_{k,i}$ 

\textbf{Termination: }Evaluate $\gamma_{i}\propto\beta_{i}\circ\alpha_{i}$,
$i\in\{1,\dots,n\}$.

\textbf{Return} $\widehat{l}_{i}=\arg\max_{l_{\itime}}(\gamma_{i}'l_{\itime})$,
$i=1,\dots,n$. 

\caption{\label{alg:chap6:FB-algorithm} FB algorithm}
\end{algorithm}

\section{Point estimation for HMC \label{sec:Point-estimation}}

In practice, it is often desired to compute point estimate $\widehat{L_{n}}=[\widehat{l_{1}},\ldots,\widehat{l_{n}}]$
for the hidden label field. For those discrete labels, the mode is
a reasonable estimate. However, in general, the decision on what inference
should be used leads to a trade-off between performance and computational
load, as shown below.

\subsection{HMC estimation via Maximum likelihood (ML)\label{subsec:ML}}

If the HMC prior $f(\LBold_{n}|\TBold,p)$ is neglected, we can directly
evaluate ML estimate $\widehat{L_{n}}=\arg\max_{L_{n}}f(\xBold_{n}|\LBold_{n})=\arg\max_{L_{n}}\prod_{i=1}^{n}f(x_{i}|l_{i})$,
as follows: $\widehat{l_{i}}=\arg\max_{l_{i}}\left(l_{i}'\Psi_{i}\right)$,
where $i\in\{1,\ldots,n\}$. Because the maximization operator is
very fast and straightforward, the computational complexity of ML
estimator is very low and no memory is required.

\subsection{HMC estimation via the MAP of marginals}

The sequence of marginal MAP can be defined as $\widehat{l_{i}}=\arg\max_{l_{i}}f(l_{i}|\xBold_{n})$,
$i\in\{1,\ldots,n\}$. The smoothing marginals, $f(l_{i}|\xBold_{n})$,
are provided by the output of FB algorithm (\ref{eq:gamma}), i.e.
$\widehat{l_{i}}=\arg\max_{l_{i}}(l_{i}'\mbox{\ensuremath{\gamma_{i}})}$. 

Notice that the sequence of marginal MAP may be a zero-probability
trajectory {[}\citet{ch6:bk:HMM:Cappe05,ch6:bk:HMM:Fraser08}{]}.
Hence, in many cases, the joint MAP of the length-$\ndata$ trajectory,
$L_{\ndata}$, is preferred, since it is always a non-zero-probability
trajectory. We address this task next.

\subsection{HMC estimation via the MAP of trajectory \label{subsec:MAP-of-trajectory}}

By Bayes' rule, the MAP of trajectory is $\widehat{\LBold}_{n}=\underset{\LBold_{n}}{\arg\max}f(\LBold_{n}|\xBold_{n})=\underset{\LBold_{n}}{\arg\max}f(\xBold_{n},\LBold_{n})$.
Because maximizing $f(\LBold_{n}|\xBold_{n})$ directly over $M^{n}$
trajectories is prohibitive, we will compute $\widehat{\LBold}_{n}$
sequentially via $\widehat{l_{i}}\in\widehat{L_{n}}$, $\iton$. This
can be achieved via two approaches: parallel memory-extraction (bi-directional
VA) and recursive memory-extraction (VA). In order to understand the
underlying methodology of the latter, we will present the former first. 

\subsubsection{Bi-directional VA \label{subsec:chap6:Bi-directional-VA}}

The computation of the MAP element $\widehat{l_{i}}\in\widehat{L_{n}}$,
as defined above, will be tractable if we extract it, not from the
joint inference, $f(L_{\ndata}|\xBold_{\ndata})$, but from $n$ profile
smoothing inferences: 
\[
\fprofile(l_{i}|\xBold_{n})\TRIANGLEQ f(l_{i}|\widehat{\LBold}_{\backslash i},\xBold_{n})\propto\max_{\LBold_{\backslash i}}f(\xBold_{n},\LBold_{n})
\]
where $\seti in$. Then, we have:

\begin{equation}
\widehat{l_{i}}=\arg\max_{l_{i}}\fprofile(l_{i}|\xBold_{n}),\ \iton\label{eq:TwoVA=00003DCE}
\end{equation}

In the same binary-tree approach of the FB algorithm (\ref{eq:FB=00003DHMC}),
applying the Markov property (\ref{eq:MarkovProperty}) to the joint
model (\ref{eq:FB=00003Dgeneral}) and maximizing the result, we have:

\begin{equation}
\fprofile(l_{i}|\xBold_{n})\propto f_{p}(\xBold_{i+1:n}|l_{i})f_{p}(l_{i}|\xBold_{i})\label{eq:TwoVA=00003DHMC}
\end{equation}
in which the profile filtering inferences $\fprofile(l_{i}|\xBold_{i})\equiv f(l_{i}|\widehat{\LBold}_{i-1},\xBold_{\itime})\propto\max_{L_{i-1}}f(\xBold_{i},L_{i})$
and backward profile observations $f_{p}(\xBold_{i:n}|l_{i-1})\equiv f(\xBold_{i:n}|\widehat{L}_{i:n},l_{i-1})\propto\max_{L_{i:n}}f(\xBold_{i:n},L_{i:n}|l_{i-1})$,
can be maximized recursively:

\begin{eqnarray}
f_{p}(l_{i}|\xBold_{i}) & \propto & \max_{l_{i-1}}f(x_{i},l_{i}|l_{i-1})\fprofile(l_{i-1}|\xBold_{i-1})\label{eq:TwoVA=00003DFW}\\
f_{p}(\xBold_{i:n}|l_{i-1}) & \propto & \max_{l_{i}}f(x_{i},l_{i}|l_{i-1})f_{p}(\xBold_{i+1:n}|l_{i})\label{eq:TwoVA=00003DBW}
\end{eqnarray}
where $\iton$. By replacing the marginalizations (\ref{eq:FB=00003DHMC})
in the FB algorithm with maximizations, we can evaluate profile distributions
in (\ref{eq:TwoVA=00003DHMC}) in $\log$ domain, which reduces the
computational load from $O(2nM^{2})$ multiplications for FB, down
to $O(2nM^{2})$ additions, with the addition being much faster than
multiplication in practice. This variant scheme is well known as bi-directional
VA in the literature {[}\citet{ch6:art:VA:logMax:Viterbi98}{]}. The
bi-directional VA is also called soft-output variant of VA {[}\citet{ch2:bk:ToddMoon}{]},
because it produces both hard and soft information, i.e. both $\widehat{l_{i}}$
and $\fprofile(l_{i}|\xBold_{n})$, respectively.

\subsubsection{The Viterbi Algorithm (VA) \label{subsec:chap6:The-Viterbi-Algorithm}}

In the second approach, which is the traditional VA, we will, once
again, exploit Markovianity and further reduce computational load
by establishing time-variant relations $h_{i}$: $\widehat{l_{i-1}}=h_{i}(\widehat{l_{i}})$,
$i\in\{2,\ldots,n\}$, where $\widehat{l_{i}}\in\widehat{L_{n}}$,
the joint MAP of trajectory, as before. Then, in the HMC, these $\widehat{l_{i}}$
can be computed recursively, without the need to evaluate the profile
distributions $\fprofile(l_{i}|\xBold_{n})$, as in (\ref{eq:TwoVA=00003DCE}). 

For motivation, let us evaluate the pair $\{\widehat{l_{i-1}},\widehat{l_{i}}\}$
via second-order of profile smoothing distributions $\fprofile(l_{i-1},l_{i}|\xBold_{n})\TRIANGLEQ f(l_{i-1},l_{i}|\widehat{\LBold}_{\backslash\{i-1,i\}},\xBold_{n})$,
expanded from the first-order profiles (\ref{eq:TwoVA=00003DHMC}),
as follows:

\begin{equation}
\fprofile(l_{i-1},l_{i}|\xBold_{n})\propto f_{p}(\xBold_{i+1:n}|l_{i})f(x_{i},l_{i}|l_{i-1})\fprofile(l_{i-1}|\xBold_{i-1})\label{eq:2ndOrder=00003Dprofile}
\end{equation}
with $i\in\{2,\ldots,n\}$. Hence, the second way to find $\{\widehat{l_{i-1}},\widehat{l_{i}}\}$
is a stage wise maximization, in which we need to find one of them
first, $\widehat{l_{i}}=\arg\max_{l_{i}}(\max_{l_{i-1}}\fprofile(l_{i-1},l_{i}|\xBold_{n}))$,
and then substitute $l_{i}=\widehat{l_{i}}$ into (\ref{eq:2ndOrder=00003Dprofile}),
from which $\widehat{l_{i-1}}=h_{i}(l_{i}=\widehat{l_{i}})$, $i\in\{2,\ldots,n\}$,
may be computed as follows:

\begin{align}
\widehat{l_{i-1}} & \TRIANGLEQ\arg\max_{l_{i-1}}\fprofile(l_{i-1}|l_{i}=\widehat{l_{i}},\xBold_{n})\nonumber \\
 & =\arg\max_{l_{i-1}}(f(x_{i},l_{i}=\widehat{l_{i}}|l_{i-1})\fprofile(l_{i-1}|\xBold_{i-1}))\label{eq:condCE}\\
 & \TRIANGLEQ h_{i}(\widehat{l_{i}})\nonumber 
\end{align}
where (\ref{eq:condCE}) follows from (\ref{eq:2ndOrder=00003Dprofile})
and the factor $f_{p}(\xBold_{i+1:n}|l_{i}=\widehat{l_{i}})$ was
excluded in $\arg\max$ operator in (\ref{eq:condCE}). In this way,
$h_{i}(l_{i})$ can be recognized as a conditional certainty equivalence
(CE): $h_{i}(l_{i})=\widehat{l_{i-1}}(l_{i})\TRIANGLEQ\arg\max_{l_{i-1}}\fprofile(l_{i-1}|l_{i},\xBold_{n})$.
Moreover, comparing (\ref{eq:condCE}) with (\ref{eq:TwoVA=00003DFW}),
we can see that $h_{i}(l_{i})$ are consequences of the forward step
and, hence, its values can be stored in the memory in an online manner.
Given $\widehat{l_{n}}$ at the end of the forward step, i.e. via
(\ref{eq:TwoVA=00003DCE}), we can directly trace back all other labels
$\widehat{l_{n-1}},\ldots,\widehat{l_{1}}$ via the stored values
$\widehat{l_{i-1}}=h_{i}(\widehat{l_{i}})$, $i\in\{2,\ldots,n\}$,
without the need to evaluate the backward step (\ref{eq:TwoVA=00003DBW})
and profile smoothing distributions (\ref{eq:TwoVA=00003DCE}-\ref{eq:TwoVA=00003DHMC}).
Hence, VA halves the computational load of bi-directional VA by requiring
computation of the forward step (\ref{eq:TwoVA=00003DFW}) only. 

Now, we can formalize VA via two steps, as detailed next.

\subsubsection*{Viterbi Forward step}

By substituting (\ref{eq:HMC}) into (\ref{eq:augment}), we can express
$f(x_{i},l_{i}|l_{i-1})$ in (\ref{eq:TwoVA=00003DFW}) in exponential
form, as follows:

\begin{align}
f(x_{1},l_{1}|p) & =\exp(-l_{1}'\lAMBDA_{1})\label{eq:augment=00003DEF}\\
f(x_{i},l_{i}|l_{i-1},\TBold) & =\mbox{\ensuremath{\exp}}(-l_{i}'\LAMBDA_{i}l_{i-1}),\ i=2,\dots,n\nonumber 
\end{align}
in which the sufficient statistics $\lAMBDA_{i}$, collected into
a sequence of $n$ length-$M$ vectors, and the information measures
$\LAMBDA_{i}\equiv\LAMBDA_{i}(x_{i})$, being a sequence of $\ndata-1$
$M\times M$ matrices, can be defined as follows:

\begin{eqnarray}
\lAMBDA_{1} & = & -(\log(\Psi_{i})+\log p)\label{eq:lambda1}\\
\LAMBDA_{i} & = & -(\log(\Psi_{i}\mathbf{1}_{M\times1}')+\log\TBold),\ i=2,\dots,n\label{eq:metric length}
\end{eqnarray}
where $\Psi_{i}$ are the soft-classification vectors (\ref{eq:ch6:DEF:psi}).

Finally, substituting (\ref{eq:augment=00003DEF}) back into (\ref{eq:TwoVA=00003DFW}),
the shaping parameters $\overline{\alpha}_{i}$ of the profile filtering
distributions $f_{p}(l_{i}|\xBold_{i})=Mu_{l_{i}}(\overline{\alpha}_{i})$
in (\ref{eq:TwoVA=00003DFW}), $\iton$, can be evaluated in log-domain
as follows:

\begin{equation}
\overline{\alpha}_{i}\propto\exp(-\lAMBDA_{i})\label{eq:VA=00003Dprofile=00003Dalpha}
\end{equation}

\begin{eqnarray}
\lAMBDA_{j,i} & = & \min_{k}(\LAMBDA_{i}(j,k)+\lAMBDA_{k,i-1})\label{eq:VA=00003Dlambda=00003DFW}\\
\kappa_{j,i} & = & \arg\min_{k}(\LAMBDA_{i}(j,k)+\lAMBDA_{k,i-1})\label{eq:VA=00003DCE=00003DBW}
\end{eqnarray}
where $\lAMBDA_{j,i}$ and $\kappa_{j,i}$, $j\in\{1,\ldots,M\}$,
denote $j$th element of vectors $\lAMBDA_{i}$ and $\kappa_{i}$,
respectively, and $\LAMBDA_{i}(j,k)$ denotes element at $j$th row
and $k$th column of matrix $\LAMBDA_{i}$, $1\leq j,k\leq M$. Note
that the conditional CEs $\widehat{l_{i-1}}(l_{i})$ in (\ref{eq:condCE})
can be found feasibly via (\ref{eq:VA=00003DCE=00003DBW}): $h_{i}(l_{i})=\widehat{l_{i-1}}(l_{i})=\epsilon(l_{i}'\kappa_{i})$,
$i\in\{2,\ldots,n\}$, where $\epsilon(j)$ in general denotes the
$j$th elementary vector in $\mathbb{R}^{\Mstate}$.

In the literature, $\lAMBDA_{i}$ and $\LAMBDA_{i}$ are not given
this novel Bayesian perspective. Indeed, the term ``metric length''
is often assigned to the elements $\LAMBDA_{i}(j,k)$ {[}\citet{ch2:origin:VA:Forney73}{]},
owing to their positive value and representation as an edge in trellis
diagram (Fig. \ref{fig:DAG}). 

Moreover, the forward recursions of profile filterings $f_{p}(l_{i}|\xBold_{i})$
in (\ref{eq:TwoVA=00003DFW}) are often considered as maintaining
$M$ ``survival'' maximal trajectories, reduced from the original
$M^{i}$ trajectories in $f(L_{i}|\xBold_{i})$. The reason for this
language is the fact that one of the $M$ length-$\itime$ trajectories
in $f_{p}(l_{i}|\xBold_{i})$ is the prefix of the global MAP trajectory
$\Lhat_{\ndata}$. Hence, each element $\lAMBDA_{j,i}$ in (\ref{eq:VA=00003Dlambda=00003DFW})
is often called the ``weighted length'' of the $j$th survival trajectory
at time $\itime$, $j\in\{1,\ldots,M\}$ {[}\citet{ch2:origin:VA:Forney73}{]}.

\subsubsection*{Viterbi back-tracking step}

From (\ref{eq:condCE}), the joint MAP trajectory, $\widehat{\LBold}_{n}$,
can be evaluated via a fast backward recursion. The last label estimate
is found first, i.e. we have $\widehat{l_{n}}=\underset{l_{n}}{\arg\max}\fprofile(l_{n}|\xBold_{n})=\boldsymbol{\epsilon}(\widehat{k}_{n})$,
where:

\begin{equation}
\khat_{n}=\underset{j}{\arg\min}(\lAMBDA_{j,n})\label{eq:VA=00003Dk_n}
\end{equation}
Then, the previous labels $\widehat{l_{i-1}}=\widehat{l_{i-1}}(l_{i}=\widehat{l_{i}})=\boldsymbol{\epsilon}(\khat_{i-1})$,
$i=2,\dots,n,$ leading to $\widehat{l_{n}}$, can be recursively
traced back using the $\kappa_{i}$ vectors of Viterbi forward step
(\ref{eq:VA=00003DCE=00003DBW}), as follows:

\begin{equation}
\khat_{i-1}=\kappa_{j=\khat_{i},i}\label{eq:VA=00003Dk_i=00003DBW}
\end{equation}

\begin{rem}
\label{Remark=00003DVA=00003DHMC} Notice that, Markovianity is the
vital condition, exploited by the VA in computational reduction. Without
it, the $\max$ operators cannot be distributed recursively in (\ref{eq:TwoVA=00003DHMC}).
Computation, both the recursive addition in (\ref{eq:FWBW=00003DHMC})
and maximization in (\ref{eq:TwoVA=00003DHMC}) are special cases
of the generalized distributive law (GDL) (see Sections \ref{sec:chap5:GDL-for-objective},
\ref{subsec:chap5:Joint-mode}). 
\end{rem}

\subsubsection*{Viterbi algorithm}

The VA, firstly presented in {[}\citet{ch6:origin:VA:Viterbi67}{]}
for channel decoding context, was formalized via trellis diagram for
the HMC in {[}\citet{ch2:origin:VA:Forney73}{]} (Fig. \ref{fig:DAG}).
Note that the MAP of trajectory may change entirely based on the last
observation $x_{n}$, owing to (\ref{eq:VA=00003Dk_n}). The VA (Algorithm
\ref{alg:chap6:Viterbi-Algorithm}) is, therefore, an offline (batch-based)
algorithm.

\begin{algorithm}
\textbf{Storage:} Length-$\nstate$ vectors $\kappa_{i}$, $\iinn$,
as in (\ref{eq:VA=00003DCE=00003DBW})

\textbf{Initialization}: evaluate (\ref{eq:lambda1})

\textbf{Recursion:} For $\iinn$: evaluate (\ref{eq:lambda1}-\ref{eq:metric length})
and (\ref{eq:VA=00003Dlambda=00003DFW}-\ref{eq:VA=00003DCE=00003DBW})

\textbf{Termination: }evaluate (\ref{eq:VA=00003Dk_n}- \ref{eq:VA=00003Dk_i=00003DBW})

\textbf{Return} $\widehat{l}_{i}=\boldsymbol{\epsilon}(\widehat{k}_{i})$,
$i=1,\dots,n$. 

\caption{\label{alg:chap6:Viterbi-Algorithm}Viterbi Algorithm (with similar
convention to {[}\citet{ch2:origin:VA:Forney73}{]})}
\end{algorithm}

\section{Re-interpretation of FB and VA via CI factorization \label{sec:chap6:FB and VA via CI}}

In the Bayesian viewpoint, the Markov property (\ref{eq:MarkovProperty})
not only reduces the computational load for inference on the joint
model (\ref{eq:FB=00003Dgeneral}), but fundamentally, also yields
a CI factorization (\ref{eq:Joint}) for posterior distribution of
label field. In this section, we will re-interpret the outputs of
FB and VA, and show that they are merely a consequence of CI structure,
i.e. FB returns the shaping parameters of HMC re-factorization, while
VA returns the shaping parameters of a degenerated HMC.

\subsection{Inhomogeneous HMC}

In the literature, it has been shown that the posterior distribution
of hidden label field, given id observations (\ref{eq:ch6:Observation})
and prior homogeneous HMC (\ref{eq:ch6:PRIOR}), is an inhomogeneous
HMC {[}\citet{ch6:bk:HMM:Cappe05}{]}. In this sense, HMC prior (\ref{eq:ch6:PRIOR})
is a conjugate prior of id observation model (\ref{eq:ch6:Observation}),
in the sense defined in Section \ref{subsec:chap4:Conjugate-prior}.
We will clarify this result via the following proposition:
\begin{prop}
The posterior distribution of the homogeneous HMC (\ref{eq:ch6:PRIOR}),
given conditionally id observation (\ref{eq:ch6:Observation}), is
an inhomogeneous HMC:

\begin{eqnarray}
f(\LBold_{n}|\xBold_{n}) & = & f(L_{i+1:n}|l_{i},\xBold_{n})f(l_{i}|\xBold_{n})f(L_{i-1}|L_{i:n},\xBold_{n})\label{eq:TERNARY}
\end{eqnarray}
in which $\iinn$ and:

\begin{eqnarray}
f(L_{i-1}|L_{i:n},\xBold_{n})=f(L_{i-1}|l_{i},\xBold_{i-1}) & = & \prod_{j=1}^{i-1}f(l_{j}|l_{j+1},\xBold_{j})\label{eq:Ti=00003DFWBW}\\
f(L_{i+1:n}|L_{i},\xBold_{n})=f(L_{i+1:n}|l_{i},\xBold_{i+1:n}) & = & \prod_{j=i+1}^{n}f(l_{j}|l_{j-1},\xBold_{j:n})\nonumber 
\end{eqnarray}
\end{prop}

\begin{proof}
Note that (\ref{eq:TERNARY}) simply follows the probability chain
rule for any general distribution $f(\LBold_{n}|\xBold_{n})$. For
the proof of (\ref{eq:Ti=00003DFWBW}), substituting Markov property
(\ref{eq:MarkovProperty}) into the joint inference (\ref{eq:FB=00003Dgeneral}),
we have:

\begin{eqnarray*}
f(L_{i-1}|L_{i:n},\xBold_{n}) & = & \frac{f(\xBold_{n},L_{n})}{\sum_{L_{i-1}}f(\xBold_{n},L_{n})}=\frac{f(\xBold_{i-1},L_{i})}{\sum_{L_{i-1}}f(\xBold_{i-1},L_{i})}\\
 & = & f(L_{i-1}|l_{i},\xBold_{i-1}),\ \forall\iinn
\end{eqnarray*}

\begin{eqnarray*}
f(L_{i:n}|L_{i-1},\xBold_{n}) & = & \frac{f(\xBold_{n},L_{n})}{\sum_{L_{i:n}}f(\xBold_{n},L_{n})}=\frac{f(\xBold_{i:n},L_{i:n}|l_{i-1})}{\sum_{L_{i:n}}f(\xBold_{i:n},L_{i:n}|l_{i-1})}\\
 & = & f(L_{i:n}|l_{i-1},\xBold_{i:n}),\ \forall\iinn
\end{eqnarray*}

Because the above equations are valid for all $\iinn$, we can derive
the right hand side of (\ref{eq:Ti=00003DFWBW}) by induction.
\end{proof}
\begin{figure}
\begin{centering}
\includegraphics[width=0.8\columnwidth]{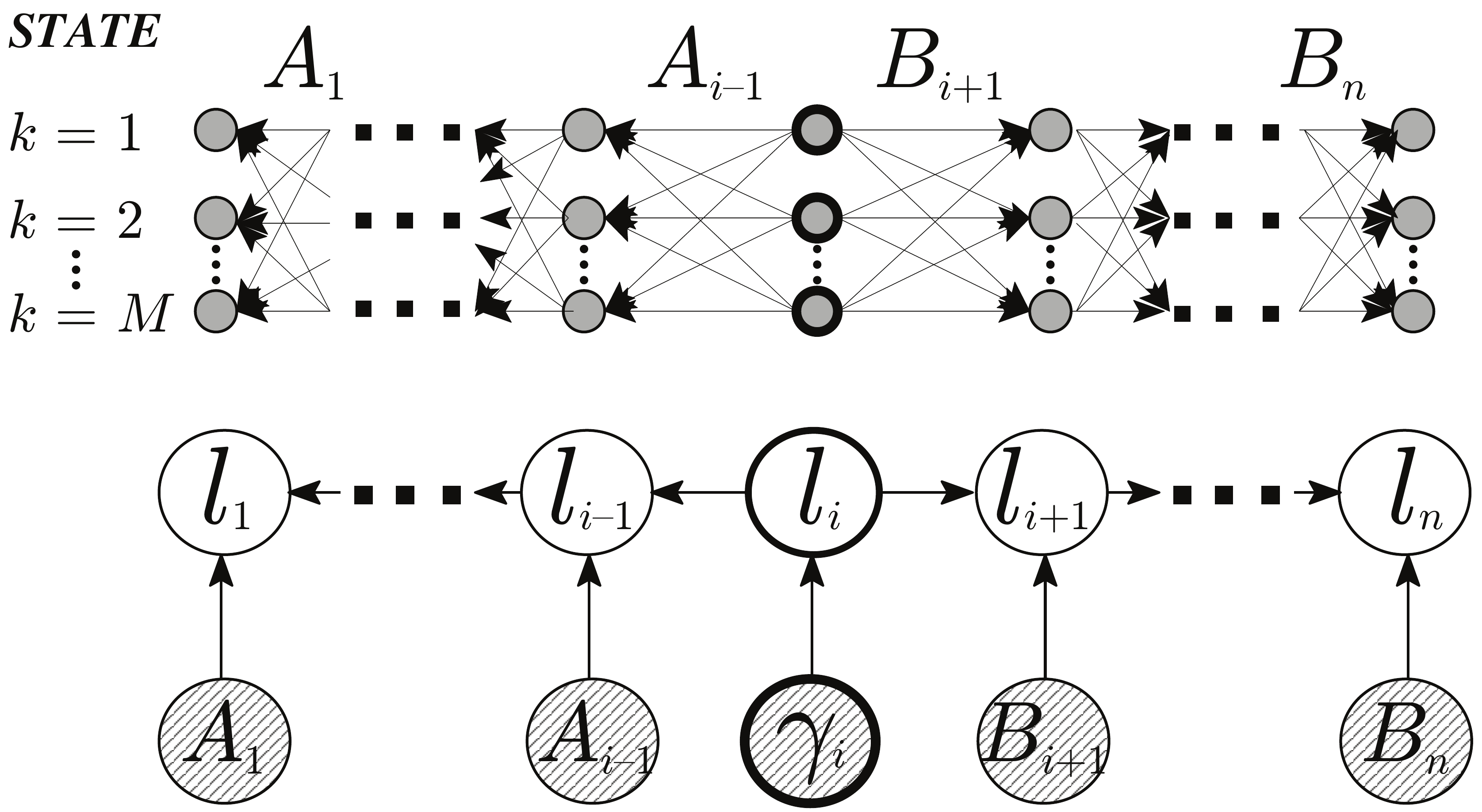}
\par\end{centering}
\caption{Posterior distribution of HMC in Fig. \ref{fig:DAG}, via FB algorithm:
trellis diagram (top) and DAG (bottom). Symbols are explained in the
text. \label{fig:FB=00003Db}}
\end{figure}
Note that, the FB algorithm (Algorithm \ref{alg:chap6:FB-algorithm})
is actually implied by the CI factorization in (\ref{eq:TERNARY}),
as shown in Propositions \ref{prop:Reverse-ChainRuleOrder},\ref{prop:Ternary=00003DReverseChain}.
Then, the backward transitions probabilities: 
\[
f(l_{\itime-1}|l_{\itime},\xBold_{\itime})=\frac{f(x_{\itime},l_{\itime},l_{\itime-1}|\xBold_{\itime-1})}{\sum_{l_{\itime}}f(x_{\itime},l_{\itime},l_{\itime-1}|\xBold_{\itime-1})}
\]
 and forward transitions probabilities: 
\[
f(l_{\itime}|l_{\itime-1},\xBold_{\itime:n})=\frac{f(\xBold_{\itime:n},l_{\itime}|l_{\itime-1})}{\sum_{l_{\itime}}f(\xBold_{\itime:n},l_{\itime}|l_{\itime-1})}
\]
in (\ref{eq:Ti=00003DFWBW}) can be computed via second-order filtering
marginals: 
\[
f(x_{\itime},l_{\itime},l_{\itime-1}|\xBold_{\itime-1})=f(x_{i},l_{i}|l_{i-1})f(l_{i-1}|\xBold_{i-1})
\]
and backward observation model: 
\[
f(\xBold_{\itime:n},l_{\itime:n}|l_{\itime-1})=f(x_{i},l_{i}|l_{i-1})f(\xBold_{i+1:n}|l_{i})
\]
extracted from forward and backward recursions in (\ref{eq:FWBW=00003DHMC}),
respectively. The following corollary will clarify this fact.
\begin{cor}
\label{cor:FB=00003DA,B}The posterior distributions in (\ref{eq:Ti=00003DFWBW})
can be evaluated via FB algorithm (Algorithm \ref{alg:chap6:FB-algorithm}),
as follows:

\begin{eqnarray}
f(l_{i}|\xBold_{n}) & = & Mu_{l_{i}}(\gamma_{\itime}),\ i\in\{1,\ldots,n\}\label{eq:proof=00003Dgamma}\\
f(l_{i}|l_{i+1},\xBold_{i}) & = & Mu_{l_{i}}(A_{i}'l_{i+1}),\ i\in\{2,\ldots,n\}\nonumber \\
f(l_{i+1}|l_{i},\xBold_{i+1:n}) & = & Mu_{l_{i+1}}(B_{i}l_{i}),\ i\in\{1,\ldots,n-1\}\nonumber 
\end{eqnarray}
where $A_{i}(k,:)\propto\TBold(k,:)\circ\alpha_{i}'$ and $B_{i}(:,k)\propto\beta_{i+1}\circ\LAMBDA_{i}(:,k)$
are right- and left-stochastic matrices of sufficient statistics,
with $(k,:)$ and $(:,k)$ denoting $k$th row and $k$th column of
a matrix, respectively, as illustrated in Fig \ref{fig:FB=00003Db}. 
\end{cor}

\begin{proof}
By the chain rule, we have:

\begin{eqnarray}
f(l_{i}|l_{i+1},\xBold_{i}) & = & \frac{f(l_{i+1}|l_{i})f(l_{i}|\xBold_{i})}{\sum_{l_{i}}f(l_{i+1}|l_{i})f(l_{i}|\xBold_{i})}\label{eq:proof=00003DA}
\end{eqnarray}

\begin{eqnarray}
f(l_{i+1}|l_{i},\xBold_{i+1:n}) & = & \frac{f(\xBold_{i+1:n}|l_{i+1})f(l_{i+1}|l_{i})}{\sum_{l_{i+1}}f(\xBold_{i+1:n}|l_{i+1})f(l_{i+1}|l_{i})}\label{eq:proof=00003DB}
\end{eqnarray}

Notice that all of the terms in right hand side (\ref{eq:proof=00003Dgamma}-\ref{eq:proof=00003DB})
have been derived in (\ref{eq:alpha}-\ref{eq:gamma}).
\end{proof}

\subsection{Profile-approximated HMC}

In order to find the joint MAP of trajectory, $\Lhat_{\ndata}$, the
bi-directional VA computed a sequence of profile distributions $\fprofile(l_{i}|\xBold_{n})$
via bi-directional maximization (\ref{eq:TwoVA=00003DHMC}-\ref{eq:TwoVA=00003DBW})
on joint model (\ref{eq:Joint}). Since joint model (\ref{eq:Joint})
can be re-factorized into CI structure (\ref{eq:TERNARY}), then,
applying the same bi-directional maximization to (\ref{eq:TERNARY}),
we can define $n$ approximated HMCs, corresponding to each choice
of $\iinn$ in (\ref{eq:TERNARY}), as follows:

\begin{equation}
\fbar_{i}(\LBold_{n}|\xBold_{n})=\left(\prod_{j=i+1}^{n}\fprofile(l_{j}|l_{j-1},\xBold_{j:n})\right)\fprofile(l_{i}|\xBold_{n})\left(\prod_{j=1}^{i-1}\fprofile(l_{j}|l_{j+1},\xBold_{j})\right)\label{eq:chap6:bi-VA--HMC}
\end{equation}
where $\iinn$, and:

\begin{eqnarray}
f(l_{i}|\xBold_{n}) & \propto & \max_{L_{\backslash\itime}}f(L_{\itime}|\xBold_{\itime})\label{eq:biVA=00003DFWBW}\\
\fprofile(l_{j}|l_{j+1},\xBold_{j}) & \propto & \max_{L_{j-1}}f(L_{j}|l_{j+1},\xBold_{j})\nonumber \\
\fprofile(l_{j}|l_{j-1},\xBold_{j:n}) & \propto & \max_{L_{j+1:n}}f(L_{j:n}|l_{j-1},\xBold_{j:n})\nonumber 
\end{eqnarray}

In common with the FB algorithm, the backward and forward transition
probabilities in (\ref{eq:biVA=00003DFWBW}) are, respectively: 
\begin{eqnarray*}
\fprofile(l_{j}|l_{j+1},\xBold_{j}) & = & \frac{\fprofile(\xBold_{j},l_{j},l_{j+1})}{\sum_{l_{j}}\fprofile(\xBold_{j},l_{j},l_{j+1})}\\
\fprofile(l_{j}|l_{j-1},\xBold_{j:n}) & = & \frac{f(\xBold_{j:n},l_{j},l_{j-1})}{\sum_{l_{j}}f(\xBold_{j:n},l_{j},l_{j-1})}
\end{eqnarray*}
which can be computed via second-order profile filterings and profile
backward observation: 
\begin{eqnarray*}
\fprofile(\xBold_{i},l_{i},l_{i+1}) & = & f(x_{i},l_{i}|l_{i-1})\fprofile(l_{i-1}|\xBold_{i-1})\\
\fprofile(\xBold_{i:n},l_{i},l_{i-1}) & = & f(x_{i},l_{i}|l_{i-1})\fprofile(\xBold_{i+1:n}|l_{i})
\end{eqnarray*}
These, in turn, can be extracted from forward and backward recursions
in (\ref{eq:TwoVA=00003DFW}-\ref{eq:TwoVA=00003DBW}), respectively.

Note that, because of normalization constants involved in (\ref{eq:biVA=00003DFWBW}),
applying the bi-directional VA algorithm (\ref{eq:TwoVA=00003DHMC})
to (\ref{eq:chap6:bi-VA--HMC}) will not recover the three $\max$
terms on the right hand side of (\ref{eq:biVA=00003DFWBW}). Hence,
the joint MAP of $\fbar_{i}(\LBold_{n}|\xBold_{n})$ is different
from the original joint MAP of $f(\LBold_{n}|\xBold_{n})$ in general.
However, the sequence of modes of the $n$ marginals $\fprofile(l_{i}|\xBold_{n})$
in the $n$ approximations $\fbar_{i}(\LBold_{n}|\xBold_{n})$ is
actually the same as the original joint MAP $\Lhat_{\ndata}$ of $f(\LBold_{n}|\xBold_{n})$.

If we neglect the normalization in (\ref{eq:biVA=00003DFWBW}), we
can find the joint MAP more quickly via VA, as explained below.

\subsection{CE-based approximated HMC \label{subsec:cha6:VA-approximated-HMC}}

VA avoids the normalizations (\ref{eq:biVA=00003DFWBW}) in bi-directional
VA by keeping only their CE values in memory. This scheme yields another
$n$ CE-based approximated HMCs, as follows:

\begin{equation}
\fhat_{i}(\LBold_{n}|\xBold_{n})=\left(\prod_{j=i+1}^{n}\fdelta(l_{j}|l_{j-1},\xBold_{j:n})\right)\fprofile(l_{i}|\xBold_{n})\left(\prod_{j=1}^{i-1}\fdelta(l_{j}|l_{j+1},\xBold_{j})\right)\label{eq:chap6:VA-HMC}
\end{equation}
for $\iinn$ and:

\begin{eqnarray}
\fprofile(l_{i}|\xBold_{n}) & = & Mu_{l_{i}}(\widehat{\gamma_{\itime}}),\ i\in\{1,\ldots,n\}\label{eq:VA=00003DFWBW}\\
\fdelta(l_{j}|l_{j+1},\xBold_{j}) & = & \delta[l_{j}-\lhat_{j}(l_{j+1})]=Mu_{l_{j}}(\widehat{A}_{j}'l_{j})\\
\fdelta(l_{j}|l_{j-1},\xBold_{j:n}) & = & \delta[l_{j}-\lhat_{j}(l_{j-1})]=Mu_{l_{j}}(\widehat{B}_{j}'l_{j})\nonumber 
\end{eqnarray}
in which $\widehat{A}_{j}=[\widehat{l_{j}}(l_{j+1}=\boldsymbol{\epsilon}(1))',\ldots,\widehat{l_{j}}(l_{j+1}=\boldsymbol{\epsilon}(M))']'$
and $\widehat{B}_{j}=[\widehat{l_{j}}(l_{j-1}=\boldsymbol{\epsilon}(1)),\ldots,$
$\widehat{l_{j}}(l_{j-1}=\boldsymbol{\epsilon}(M))]$ are sparse left-stochastic
matrices, with only one element at each column, as illustrated in
Fig. \ref{fig:bi-Viterbi}. The conditional CEs in (\ref{eq:VA=00003DFWBW})
are defined as follows:

\begin{align}
\lhat_{j}(l_{j+1}) & =\arg\max_{l_{j}}\fprofile(l_{j}|l_{j+1},\xBold_{j})=\arg\max_{l_{j}}\fprofile(l_{j},l_{j+1}|\xBold_{n})\nonumber \\
\lhat_{j}(l_{j-1}) & =\arg\max_{l_{j}}\fprofile(l_{j}|l_{j-1},\xBold_{j:n})=\arg\max_{l_{j}}\fprofile(l_{j},l_{j-1}|\xBold_{n})\label{eq:VA=00003DCE=00003DFWBW}
\end{align}
in which the right hand side of (\ref{eq:VA=00003DCE=00003DFWBW})
is implied by CI structure in (\ref{eq:TERNARY}). 

As explained in Section \ref{subsec:chap6:The-Viterbi-Algorithm},
substituting $\lhat_{j+1}$and $\lhat_{j-1}$ of original joint MAP
of $f(L_{n}|\xBold_{n})$ into (\ref{eq:VA=00003DCE=00003DFWBW}),
we can also extract $\lhat_{j}(l_{j+1})$ and $\lhat_{j}(l_{j-1})$
belonging to that joint MAP. Hence, given $\lhat_{i}=\arg\max_{l_{i}}\fprofile(l_{i}|\xBold_{n})$
in $\fhat_{i}(\LBold_{n}|\xBold_{n})$ for any $\iinn$, we can trace
back the original MAP from stored conditional CEs in (\ref{eq:VA=00003DCE=00003DFWBW}).
Although this back-tracking scheme works for any $i\in\{1,\ldots,n\}$,
in practice the traditional VA always assigns $i=n$ in $\fhat_{i}(\LBold_{n}|\xBold_{n})$
and maintains profile filtering distributions $\fprofile(l_{n}|\xBold_{n})$,
in order to achieve an online forward scheme for new data at time
$n+1$, as illustrated in Fig. \ref{fig:Viterbi}. 

It can be proved feasibly that the joint MAP of $\fhat_{i}(\LBold_{n}|\xBold_{n})$,
for any $i\in\{1,\ldots,n\}$, are the same as joint MAP $f(L_{n}|\xBold_{n})$,
since they have the same conditional CEs and the same elements $\lhat_{i}$
of joint MAP. Also, because $\fhat(\LBold_{n}|\xBold_{n})$ only has
$M$ non-zero-probability trajectories, finding the original joint
MAP via $\fhat(\LBold_{n}|\xBold_{n})$ is easier than via the exact
$f(L_{n}|\xBold_{n})$. 

\begin{figure}
\begin{centering}
\includegraphics[width=0.8\columnwidth]{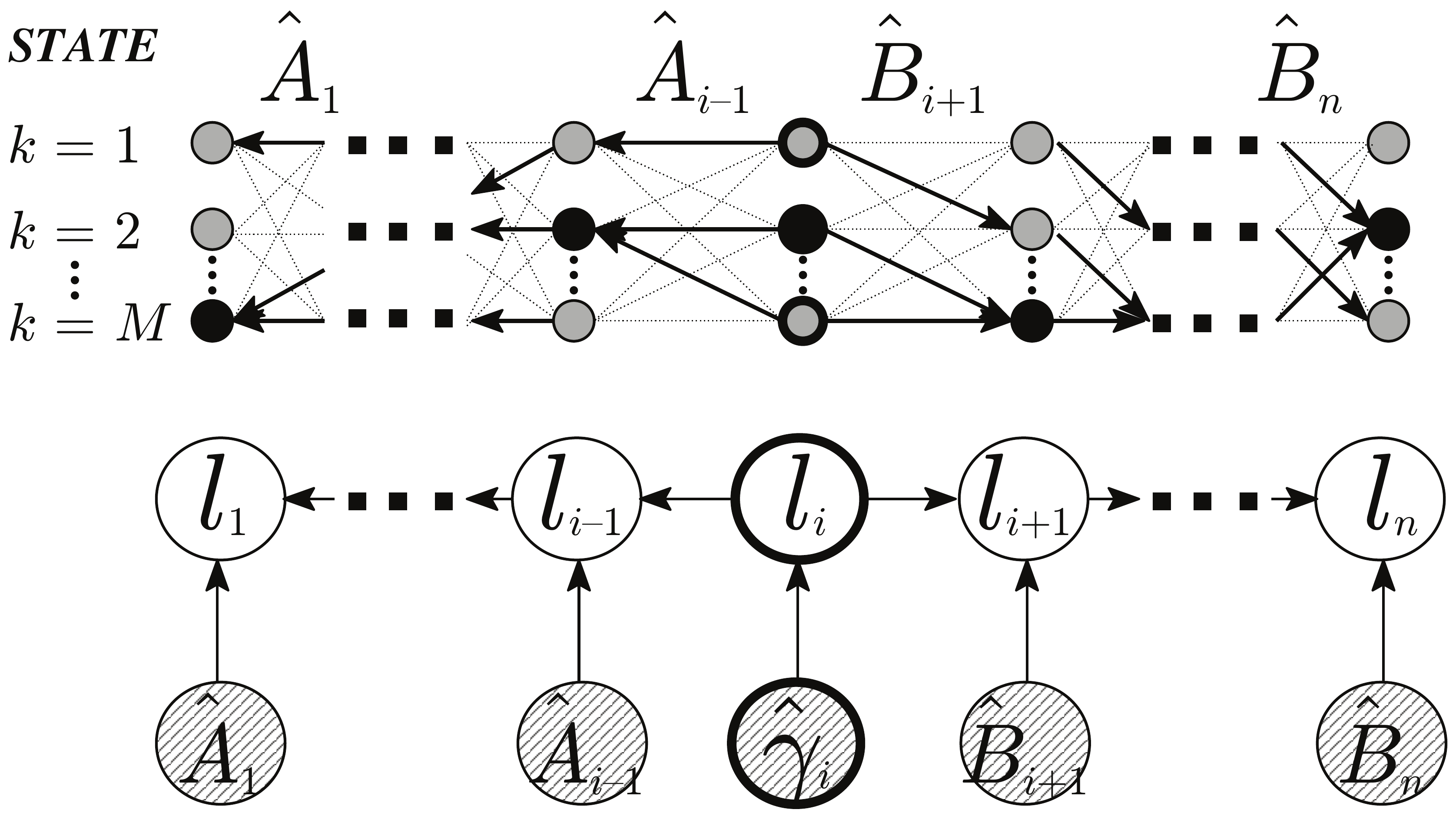}
\par\end{centering}
\caption{\label{fig:bi-Viterbi}Bi-directional Viterbi algorithm (VA) for HMC
posterior in Fig. \ref{fig:FB=00003Db}: trellis diagram (top) and
DAG (bottom). Dotted lines denote zero-probability transitions. Note
that, black circles denote joint MAP of trajectory of $f(\protect\LBold_{n}|\protect\xBold_{n})$
in (\ref{eq:TERNARY}), \textit{but} a sequence of modes of $n$ marginals
$\protect\fprofile(l_{i}|\protect\xBold_{n})$ for $\protect\fbar_{i}(\protect\LBold_{n}|\protect\xBold_{n})$
in (\ref{eq:chap6:bi-VA--HMC}), $\protect\iinn$.}
\end{figure}
\begin{figure}
\begin{centering}
\includegraphics[width=0.8\columnwidth]{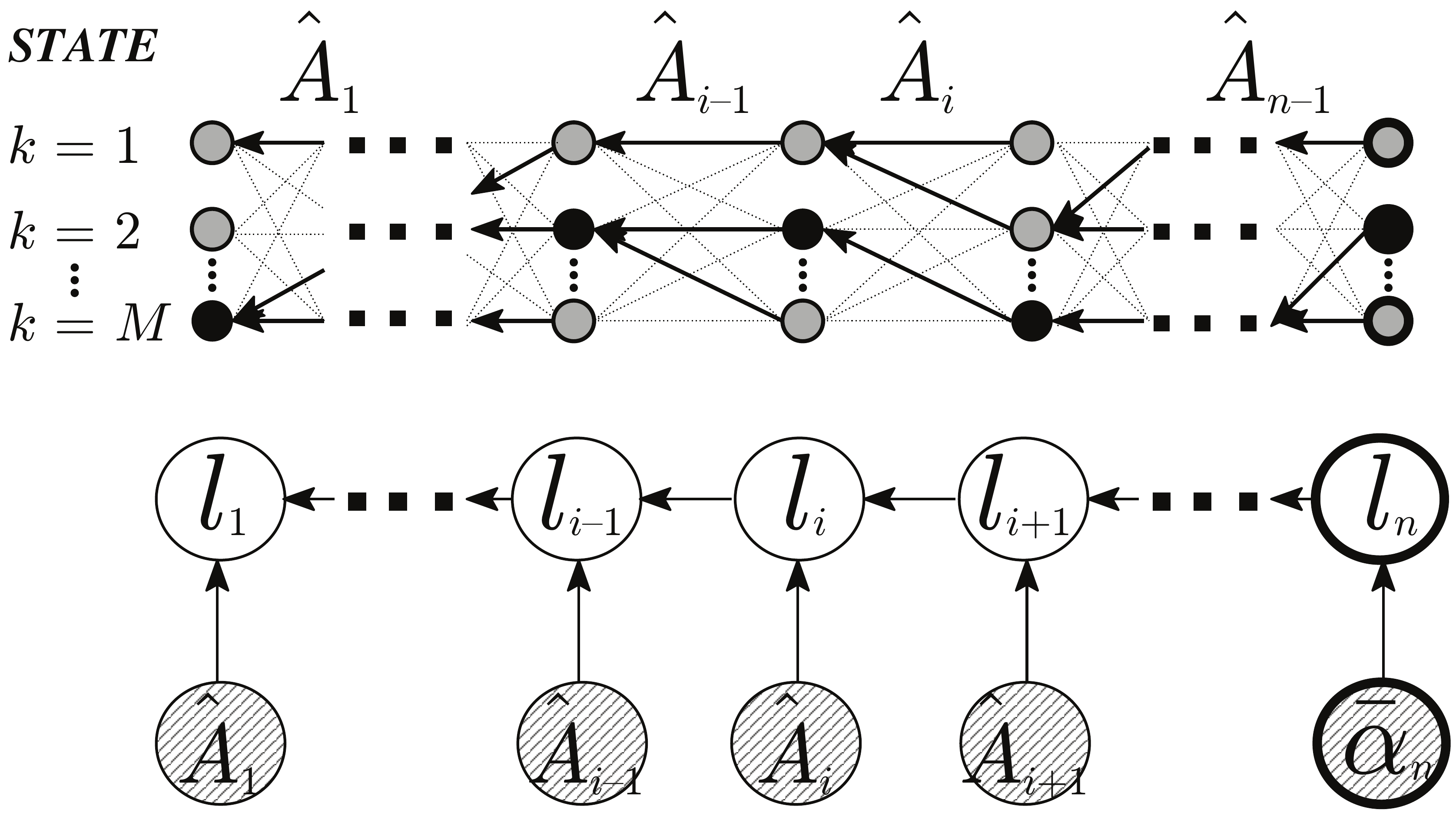}
\par\end{centering}
\caption{\label{fig:Viterbi}Viterbi algorithm (VA) for HMC posterior in Fig.
\ref{fig:FB=00003Db}: trellis diagram (top) and DAG (bottom), with
the same convention as Fig \ref{fig:bi-Viterbi}. Note that, black
circles denote the same joint MAP trajectory of both $f(\protect\LBold_{n}|\protect\xBold_{n})$
in (\ref{eq:TERNARY}) \textit{and} $\protect\fhat_{i}(\protect\LBold_{n}|\protect\xBold_{n})$
in (\ref{eq:chap6:VA-HMC}), for any $\protect\iinn$. These black
circles and $\bar{\alpha}_{\protect\ndata}$ in this figure are exactly
the same as the black circles and $\widehat{\gamma_{\protect\ndata}}$
in Fig \ref{fig:bi-Viterbi}, respectively, when $\protect\itime=\protect\ndata$.}
\end{figure}

\section{Variational Bayes (VB) approximation for CI structure}

In the previous section, we have shown that VA involves a CE-based
approximation of the smoothing marginals $f(l_{i}|\xBold_{n})$ in
the FB algorithm. The VA reduces complexity significantly via two
sequential steps: 
\begin{itemize}
\item The first step is to constrain the inference problem, replacing the
smoothing marginals, $f(l_{i}|\xBold_{n})$, with profile smoothing
distributions, $\fprofile(l_{i}|\xBold_{n})$. 
\item The second step is further to constrain these profile smoothing distributions
to point inferences and directly compute $\widehat{l_{i}}$ of the
joint MAP within this CE-based distribution. 
\end{itemize}
In a similar manner, we seek novel variants of VA via two sequential
steps: 
\begin{itemize}
\item The first step is to derive a distributional approximation of $f(L_{\ndata}|\xBold_{n})$
from the class of independent distributions $\prod_{\itime=1}^{\ndata}\ftilde(l_{i}|\xBold_{n})$
via VB (Section \ref{subsec:chap4:Variational-Bayes-(VB)}). 
\item The second step is further to impose a CE-based constraint upon the
VB marginals, via FCVB (Section \ref{subsec:chap4:Functionally-Constraint-VB}),
in order to reduce complexity significantly.
\end{itemize}
For this purpose, we will apply the VB and FCVB methodology to a general
multivariate posterior distribution in this section, and then to the
HMC in next section.

\subsection{VB approximation via Kolmogorov-Smirnov (KS) distance}

Let us consider a sequence of $\ntheta$ variables $\btheta\TRIANGLEQ\{\theta_{1},\ldots,\theta_{\ntheta}\}$,
provided that there is an arbitrary choice in partition of $\btheta$
into (non-empty) sub-vectors $\vtheta_{i}$. Then, given a joint posterior
distribution $f(\btheta|\xBold)$, the VB method is to seek an approximated
distribution, $\tilde{f}(\btheta|\xBold)$, in independent distribution
class $\mathbb{F}_{c}$, $\breve{f}(\btheta|\xBold)=\prod_{i=1}^{n}\breve{f}(\theta_{i}|\xBold)$,
such that the Kullback-Leibler divergence: 
\[
\KLD_{\ftilde||f}\TRIANGLEQ\KLD(\tilde{f}(\btheta|\xBold)||f(\btheta|\xBold))=E_{\tilde{f}(\btheta|\xBold)}\log\frac{\tilde{f}(\btheta|\xBold)}{f(\btheta|\xBold)}
\]
 is minimized, as explained in Section \ref{subsec:chap4:Variational-Bayes-(VB)}.

\subsubsection{Iterative VB (IVB) algorithm}

Given an arbitrary initial distribution $\fnutilde 0(\btheta|\xBold)=\prod_{i=1}^{\ntheta}\fnutilde 0(\theta_{i}|\xBold)$,
the aim of the IVB algorithm is to update the VB-marginals $\fnutilde{\nu}(\theta_{i}|\xBold)$
iteratively, $\nu=1,2,\ldots$, and cyclically with $\vtheta_{\itime}$,
until a local minimum of $\KLD_{\ftilde||f}$ is reached. This is
achieved as follows (Theorem \ref{thm:chap4:Iterative-VB-(IVB)}):

\begin{eqnarray}
\fnutilde{\nu}(\theta_{i}|\xBold) & \propto & \exp\left(E_{\fnutilde{\nu}(\theta_{\backslash i}|\xBold)}\log f(\btheta|\xBold)\right),\ \iton\label{eq:IVB}
\end{eqnarray}
where we recall that $\theta_{\backslash i}$ denotes the complement
of $\theta_{i}$ in $\btheta$, and: 

\[
\fnutilde{\nu}(\theta_{\backslash i}|\xBold)=\prod_{j=i+1:n}\fnutilde{\nu-1}(\theta_{j}|\xBold)\prod_{j=1:i-1}.\fnutilde{\nu}(\theta_{j}|\xBold)
\]

Note that, in practice, because the normalized form $f(\btheta|\xBold)$
is often unavailable (intractable), we can replace $f(\btheta|\xBold)$
with its unnormalized variant $f(\xBold,\btheta)$ in (\ref{eq:IVB}). 

Also, in the IVB algorithm, we can freely permute the sequence $\{\theta_{\pi(1)},\ldots,\theta_{\pi(n)}\}$
for the update of VB-marginals (\ref{eq:IVB}) in any IVB cycle, where
$\pi$ is a fixed permutation of the index set $\{1,\ldots,n\}$.
Since we have not specified any particular order for the index set
$\{1,\ldots,n\}$, let us denote, for convenience, that $\pi(i)=i$,
$\forall\iton$. 

\subsubsection{Stopping rule for IVB algorithm}

Because $\KLD_{\ftilde||f}$ in IVB algorithm is guaranteed to be
non-increasing with $\nu$ in the IVB algorithm (\ref{eq:IVB}), and
to converge to a local minimum {[}\citet{ch4:BK:AQUINN_06}{]}, we
can propose a stopping rule by choosing a small threshold $\xi_{c}$,
such that $\nu_{c}\TRIANGLEQ\min_{\iIVB}\left(\KLD_{\ftilde||f}^{[\nu]}\leq\xi_{c}\right)$,
as in the definition of the converged cycle number, $\nu_{c}$. However,
the IVB algorithm does not evaluate $\KLD_{\ftilde||f}^{[\nu]}$ itself
and that evaluation is typically prohibitive. 

Therefore, instead of using the joint $\KLD_{\ftilde||f}$, let us
consider individual convergence of each VB-marginal via Kolmogorov-Smirnov
(KS) distance {[}\citet{Kolmogorov_distance}{]}: 

\begin{eqnarray}
KS_{\itime}^{[\nu]} & \TRIANGLEQ & \max_{\theta_{i}}\left|\widetilde{F}^{[\nu]}(\theta_{i}|\xBold)-\widetilde{F}^{[\nu-1]}(\theta_{i}|\xBold)\right|\label{eq:KS}
\end{eqnarray}
where $\widetilde{F}(\theta_{i}|\xBold)$ denotes the cumulative distribution
function (c.d.f) of $\ftilde(\theta_{i}|\xBold)$, $\iinn$. Because
verifying $KS_{\itime}^{[\nu]}$ requires only maximization and addition
operations, this informal scheme for verifying convergence of the
IVB algorithm is much faster than formal evaluation of $\KLD_{\ftilde||f}^{[\nu]}$
itself. Note that, IVB continues to iterate in order to decrease $\KLD_{\ftilde||f}$
and the KS is only proposed for the stopping rule.

Indeed, this stopping rule $KS_{\itime}$ for each VB-marginal is
stricter than that of $\KLD_{\ftilde||f}$ and may lead to higher
value of $\nu_{c}$. This strictness can be compensated by choosing
a greater $\xi_{c}$. In this thesis, the IVB algorithm is considered
to be converged at cycle $\nu_{c}$ if the following condition is
satisfied: 
\begin{equation}
KS_{\itime}^{[\nu_{c}]}\leq\xi=0.01,\ \forall\iinn\label{eq:KS-threshold}
\end{equation}

In simulations, presented in Chapter \ref{=00005BChapter 8=00005D},
no gain in performance was achieved with threshold $\xi$ lower than
$0.01$. Apart from the cost reduction, this KS-based stopping rule
will also lead to a novel speed-up scheme for the IVB algorithm, as
explained in Section \ref{subsec:chap6:Accelerated-IVB-algorithm}.

For later use, let us indicate the status of $KS_{\itime}^{[\nu]}$,
$\iton$, via a sequence of $n$ Boolean indicators $\varrho_{\itime}^{[\nu]}$,
$\iton$, as follows: 
\begin{equation}
\begin{cases}
\varrho_{\itime}^{[\nu]}=0, & \mbox{if}\ KS_{\itime}^{[\nu]}\leq\xi\\
\varrho_{i}^{[\nu]}=1, & \mbox{if}\ KS_{\itime}^{[\nu]}>\xi
\end{cases}\hspace{1em}\iinn,\ \iIVB=1,2,\ldots\label{eq:chap6:KS-indicator}
\end{equation}

Then, the KS-based condition (\ref{eq:KS-threshold}) is equivalent
to the following condition: 
\begin{equation}
\sum_{\itime=1}^{\ntheta}\varrho_{\itime}^{[\nIVB]}=0,\ \mbox{or\ equivalently},\ \varrho_{\itime}^{[\nIVB]}=0,\forall\iinn\label{eq:KS-convergence}
\end{equation}

The IVB algorithm in this case is presented in Algorithm \ref{alg:IVB-algorithm}. 

\textbf{}
\begin{algorithm}
\textbf{Iteration:} 

For $\nu=1,2,\ldots$, do \{

For $i=1,\ldots n$, do \{ 

evaluate either (\ref{eq:IVB}) or (\ref{eq:IVB for CI})

if $KS_{\itime}^{[\nu]}\leq\xi$, \{ set $\varrho_{\itime}=0$ \} 

else: \{set $\varrho_{\itime}=1$\}\}\}\}

\textbf{Termination: }stop if \textbf{$\varrho_{\itime}=0$, $\forall\iinn$}

\textbf{\caption{\label{alg:IVB-algorithm} IVB algorithm using KS-distance stopping
rule\textbf{ }}
}
\end{algorithm}

\subsection{Accelerated IVB approximation \label{subsec:chap6:Accelerated-IVB-algorithm}}

In general, any joint distribution $f(\btheta|\xBold)$ can be binarily
factorized in respect of each choice of $\vtheta_{\itime}$, $\iinn$,
thereby exploiting any CI structure that may be present in the joint
model:

\begin{equation}
f(\btheta|\xBold)=f(\theta_{i}|\theta_{\eta_{i}},\xBold)f(\theta_{\backslash i}|\xBold),\ \iinn\label{eq:C.I}
\end{equation}
where the \foreignlanguage{british}{neighbour} set is $\theta_{\eta_{i}}\subseteq\theta_{\backslash i}$.
Then, the key step in accelerating the IVB algorithm is to exploit
the CI structure of the original distribution $f(\btheta|\xBold)$,
as explained next.

\subsubsection{IVB algorithm for CI structure}

Owing to the CI structure (\ref{eq:C.I}), the IVB algorithm (\ref{eq:IVB})
only involves the \foreignlanguage{british}{neighbour} set $\theta_{\eta_{i}}$,
instead of the whole set $\theta_{\backslash i}$, as follows:

\begin{equation}
\fnutilde{\nu}(\theta_{i}|\xBold)\propto\exp\left(E_{\fnutilde{\nu}(\theta_{\eta_{i}}|\xBold)}\log f(\theta_{i},\theta_{\eta_{i}}|\xBold)\right),\ \iton\label{eq:IVB for CI}
\end{equation}
where: 
\begin{equation}
\fnutilde{\nu}(\theta_{\eta_{i}}|\xBold)=\prod_{j\in\eta_{i}^{BW}}.\fnutilde{\nu-1}(\theta_{j}|\xBold)\prod_{j\in\eta_{i}^{FW}}\fnutilde{\nu}(\theta_{j}|\xBold)\label{eq:ftilde_etaj}
\end{equation}
with $\eta_{i}^{FW}=\eta_{i}\cup\{1:i-1\}$, $\eta_{i}^{BW}=\eta_{i}\cup\{i+1:n\}$
denoting forward and backward \foreignlanguage{british}{neighbour}
index sets, respectively, and $\eta_{i}=\{\eta_{i}^{FW},\eta_{i}^{BW}\}$. 

\subsubsection{Accelerated IVB algorithm \label{subsec:Accelerated-VB-algorithm}}

\begin{figure}
\begin{centering}
\includegraphics[width=0.8\columnwidth]{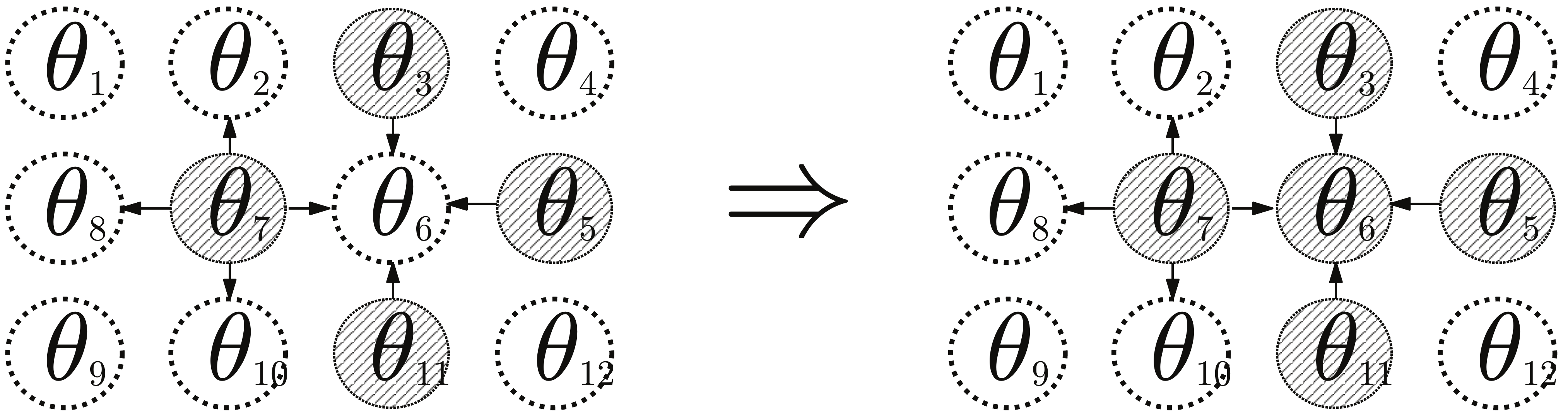}
\par\end{centering}
\caption{\label{fig:many-to-one}Many-to-one scheme. Left: VB-marginals $\protect\ftilde(\protect\vtheta_{3})$,
$\protect\ftilde(\protect\vtheta_{5})$, $\protect\ftilde(\protect\vtheta_{7})$
and $\protect\ftilde(\protect\vtheta_{11})$ in current IVB cycle
$\nu$ are converged (shaded nodes, i.e. $\tau^{[\nu]}=0$). Therefore,
Right: $\protect\ftilde(\protect\vtheta_{6})$ is set as converged
($\tau_{6}^{[\nu+1]}=0$) in next IVB cycle. White nodes denote unknown
converged state of VB-marginals (unknown $\tau_{\protect\itime}$).}
 
\end{figure}
\begin{figure}
\begin{centering}
\includegraphics[width=0.8\columnwidth]{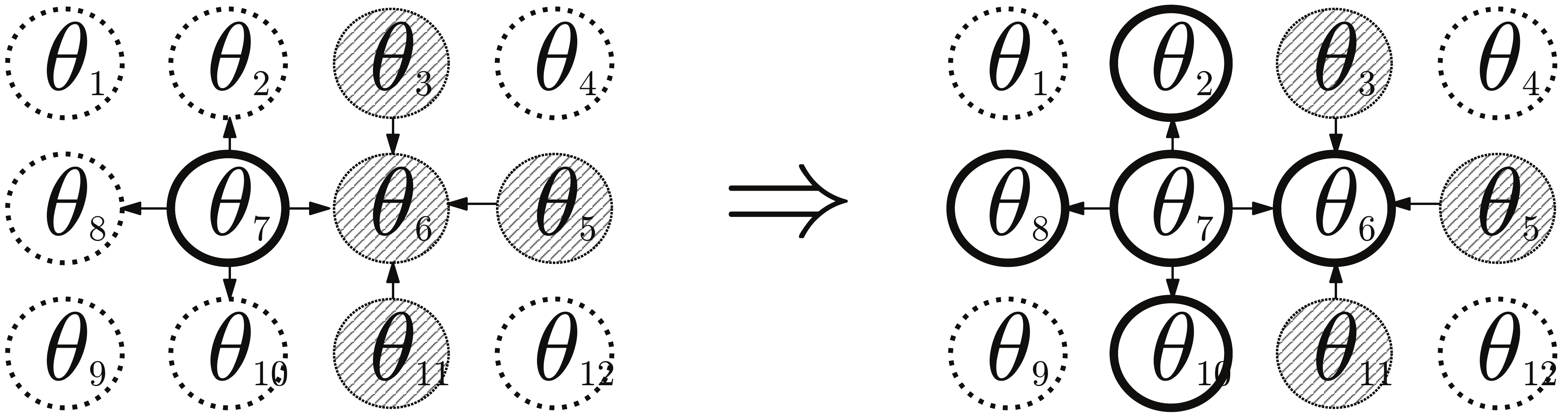}
\par\end{centering}
\caption{\label{fig:one-to-many}One-to-many scheme. Left: $\protect\ftilde(\protect\vtheta_{7})$
is not converged in current IVB cycle $\nu$ (thick, black circle
node $\tau_{7}^{[\nu]}=1$). Therefore, Right: $\protect\ftilde(\protect\vtheta_{2})$,
$\protect\ftilde(\protect\vtheta_{6})$, $\protect\ftilde(\protect\vtheta_{8})$
and $\protect\ftilde(\protect\vtheta_{10})$ are set as not converged
(i.e. $\tau^{[\nu+1]}=1)$ in the next IVB cycle. Shaded nodes denote
converged VB-marginals ($\tau_{i}=0)$. Dotted circles denote unknown
converged state (unknown $\tau_{\protect\itime}$).}
\end{figure}
In IVB (\ref{eq:IVB for CI}), we have to update all $n$ VB-marginals
at any cycle $\nu$, since the KS-based conditions (\ref{eq:chap6:KS-indicator})
needs to be checked all the time. Therefore, if we can quickly identify
the converged VB-marginals at current IVB cycle and exclude them from
next IVB cycle, the computational load will be reduced. 

Hence, the central idea of the accelerated scheme is the concept of
individual convergence. Nevertheless, although the convergence of
$\KLD_{\ftilde||f}$ can be proved to be monotone {[}\citet{ch4:BK:AQUINN_06}{]},
the individual convergence of $KS_{\itime}^{[\nu]}$ for each VB-marginal
is not monotone in general and might vary around the threshold $\xi$.
For resolving this problem, let us formulate two remarks about the
CI structure (\ref{eq:IVB for CI}): 
\begin{itemize}
\item The first remark, called ``many-to-one'', is that the VB-marginal,
$\fnutilde{\nu}(\theta_{i}|\xBold)$, is regarded as converged and
can be ignored at cycle $\nu$, if all of its \foreignlanguage{british}{neighbour}
VB-marginals $\fnutilde{\nu-1}(\theta_{j}|\xBold)$ in (\ref{eq:ftilde_etaj})
already converged at cycle $\nu-1$. 
\item The second remark, called ``one-to-many'', is that the VB-marginal,
$\fnutilde{\nu}(\theta_{i}|\xBold)$, is not yet converged and needs
to be updated at cycle $\nu$, if any of its \foreignlanguage{british}{neighbour}
VB-marginals $\fnutilde{\nu-1}(\theta_{j}|\xBold)$ in (\ref{eq:ftilde_etaj})
is not yet converged at cycle $\nu-1$. 
\end{itemize}
Even though both of those remarks yield, once again, informal convergence
conditions for IVB (\ref{eq:IVB for CI}), they are useful in an accelerated
scheme. For this purpose, let us indicate which VB-marginal $\fnutilde{\nu}(\theta_{i}|\xBold)$
is converged at cycle $\nu$ and which is not, via the toggling values
$\tau_{i}^{[\iIVB]}=0$ and $\tau_{i}^{[\iIVB]}=1$, respectively,
in a sequence of $n$ Boolean indicators $\tau_{i}^{[\nu]}$, $\iton$.
Initially, all the indicators are set to one, i.e. the initialization
$\fnutilde 0(\btheta|\xBold)$ is regarded as not yet converged, as
follows:

\begin{equation}
\tau_{\itime}^{[0]}=1,\ \forall\iinn\label{eq:chap6:accelerate-initial}
\end{equation}

By convention, the first remark yields the ``many-to-one'' scheme
for indicator updates, i.e. each indicator is updated by its \foreignlanguage{british}{neighbour}
indicators, as follows:

\begin{eqnarray}
\sum_{j\in\eta_{i}}\tau_{j}^{[\iIVB-1]}=0 & \Rightarrow & \tau_{\itime}^{[\iIVB]}=0\label{eq:chap6:many-to-one}\\
\sum_{j\in\eta_{i}}\tau_{j}^{[\iIVB-1]}>0 & \Rightarrow & \begin{cases}
\tau_{\itime}^{[\iIVB]}=0, & if\ KS_{\itime}^{[\nu]}\leq\xi\\
\tau_{\itime}^{[\iIVB]}=1, & if\ KS_{\itime}^{[\nu]}>\xi
\end{cases},\ \iinn\nonumber 
\end{eqnarray}
as illustrated in Fig. \ref{fig:many-to-one}. 

Equivalently, the second remark yields the ``one-to-many'' scheme
for indicator updates, i.e. each indicator updates its neighbour indicators,
which involves two steps at each time $\iIVB$. In the first step,
a temporary Boolean sequence is initially set $\varsigma_{\itime}=0$,
$\forall\iinn$, at each time $\iIVB$, then we update those $\varsigma_{\itime}$
by the following relationship:

\begin{eqnarray}
\tau_{\itime}^{[\iIVB-1]}=1 & \Rightarrow & \varsigma_{j}=1,\ \forall j:\ \eta_{j}\ni i,\ \forall\iinn\label{eq:chap6:one-to-many:a}
\end{eqnarray}
in which we do nothing with the case $\tau_{\itime}^{[\iIVB-1]}=0$.
Then, the second step is to update the VB-marginal indicators (Fig.
\ref{fig:one-to-many}):

\begin{eqnarray}
\sum_{j\in\eta_{i}}\tau_{j}^{[\iIVB-1]}=0\Leftrightarrow\varsigma_{\itime}=0 & \Rightarrow & \tau_{i}^{[\iIVB]}=0,\label{eq:chap6:one-to-many:b}\\
\sum_{j\in\eta_{i}}\tau_{j}^{[\iIVB-1]}=1\Leftrightarrow\varsigma_{\itime}=1 & \Rightarrow & \begin{cases}
\tau_{i}^{[\iIVB]}=0, & if\ KS_{\itime}^{[\nu]}\leq\xi\\
\tau_{i}^{[\iIVB]}=1, & if\ KS_{\itime}^{[\nu]}>\xi
\end{cases},\ \iinn\nonumber 
\end{eqnarray}
in which the equivalences ``$\Leftrightarrow$'' in (\ref{eq:chap6:one-to-many:b})
are derived from the relationship (\ref{eq:chap6:one-to-many:a}).
Comparing (\ref{eq:chap6:many-to-one}) with (\ref{eq:chap6:one-to-many:b}),
we can see that the ``one-to-many'' scheme is equivalent to ``many-to-one''
scheme. The common convergence condition for both of ``many-to-one''
and ``one-to-many'' schemes is: 
\begin{equation}
\sum_{\itime=1}^{\ntheta}\tau_{\itime}^{[\nIVB]}=0,\ \mbox{or\ equivalently,}\ \tau_{\itime}^{[\nIVB]}=0,\forall\iinn\label{eq:chap6:accelerate-convergence}
\end{equation}

Of the two, the ``one-to-many'' scheme (\ref{eq:chap6:one-to-many:a}-\ref{eq:chap6:one-to-many:b})
is more useful, because we only need to update one Boolean indicator
$\tau_{i}$ in order to indicate the convergence of neighbour VB-marginals
$\fnutilde{\nu}(\theta_{\eta_{\itime}}|\xBold)$ at cycle $\iIVB$
(hence the name ``one-to-many''), rather than verifying all the
neighbour indicators (hence the name ``many-to-one'') in the first
remark. For simpler programming, in this thesis, the Accelerated IVB
algorithm is designed via ``one-to-many'' scheme (\ref{eq:chap6:one-to-many:b}),
as presented in Algorithm \ref{alg:Accel-IVB-algorithm}. Also, because
the temporary Boolean $\varsigma_{\itime}$ are only introduced in
order to clarify the idea of ``one-to-many'' scheme above, there
is no need to evaluate those $\varsigma_{\itime}$ in Algorithm \ref{alg:Accel-IVB-algorithm}.
To recover the conventional IVB algorithm (Algorithm \ref{alg:IVB-algorithm}),
we only need to omit the indicator condition ``if $\tau_{i}=1$''
in the Accelerated IVB algorithm (Algorithm \ref{alg:Accel-IVB-algorithm}),
i.e. we then update all $n$ VB-marginals in each IVB cycle $\nu$.

\textbf{}
\begin{algorithm}
\textbf{Initialization}: set $\tau_{\itime}=1$, $\forall\iinn$

\textbf{Iteration:} 

For $\nu=1,2,\ldots$, do \{

For $i=1,\ldots n$, do \{ 

if $\tau_{\itime}=1$, \{ 

evaluate (\ref{eq:IVB for CI}) 

if $KS_{\itime}^{[\nu]}\leq\xi$, \{ set $\tau_{\itime}=0$ \} 

else: \{set $\tau_{j}=1$, $\forall j:\ \eta_{j}\ni i$ \}\}\}\}

\textbf{Termination: }stop if\textbf{ $\tau_{\itime}=0$, $\forall\iinn$}

\textbf{\caption{\label{alg:Accel-IVB-algorithm}Accelerated IVB via KS distance stopping
rule }
}
\end{algorithm}
The accelerated scheme can be regarded as another convergence condition
for $\fnutilde{\nu}(\theta_{i}|\xBold)$, $\iinn$, i.e. if we set
$\xi>0$, there might be some VB-marginal updates are excluded in
an IVB cycle in Accelerated IVB, while those VB-marginal updates are
always updated in conventional IVB algorithm. Nevertheless, if we
set $\xi=0$, that excluding step does not yield any difference between
Accelerated IVB's and conventional IVB's outputs, because the accelerated
indicators (\ref{eq:chap6:one-to-many:b}) are equivalent to KS-based
indicators (\ref{eq:chap6:KS-indicator}), as shown in Lemma \ref{lem:(Accelerated-IVB-algorithm)}. 

For $\xi>0$, which is a more reasonable and relaxed case in practice,
such an equivalence is not true in general, as explained above, although
it might be true if $\xi$ is small enough to ensure no change in
the VB-marginal moments in two consecutive iterations.
\begin{lem}
\label{lem:(Accelerated-IVB-algorithm)}\textbf{(Accelerated IVB algorithm)}
If $\xi=0$, the Accelerated IVB algorithm (Algorithm \ref{alg:Accel-IVB-algorithm})
has exactly the same output $\fnutilde{\nu}(\btheta|\xBold)$ as conventional
IVB algorithm (Algorithm \ref{alg:IVB-algorithm}), i.e. we have $\varrho_{\itime}^{[\iIVB]}=\tau_{\itime}^{[\iIVB]}$,
$\forall\iinn$, at any IVB cycle $\iIVB=1,2,\ldots$,$\nIVB$, where
$\nIVB$ is the same number of IVB cycles at convergence in both cases.
\end{lem}

\begin{proof}
If we have $\varrho_{\itime}^{[\iIVB]}=\tau_{\itime}^{[\iIVB]}$,
$\forall\iinn$, at any cycle $\iIVB$, the value $\nIVB$ at convergence
is obviously the same for both stopping rules (\ref{eq:KS-convergence})
and (\ref{eq:chap6:accelerate-convergence}). By comparing (\ref{eq:chap6:KS-indicator})
with (\ref{eq:chap6:one-to-many:b}), it is feasible to recognize
that we only need to prove the following relationship:

\begin{equation}
\sum_{j\in\eta_{i}}\tau_{j}^{[\iIVB-1]}=0\Rightarrow KS_{\itime}^{[\nu]}=0,\ \iinn\label{eq:chap6:proof_rho_tau}
\end{equation}
in order to prove that $\varrho_{\itime}^{[\iIVB]}=\tau_{\itime}^{[\iIVB]}$,
$\forall\iinn$, for the case $\xi=0$. The relationship (\ref{eq:chap6:proof_rho_tau})
can be verified feasibly: Note that $\sum_{j\in\eta_{i}}\tau_{j}^{[\iIVB-1]}=0$
and $\xi=0$ mean $KS_{j\in\eta_{i}}^{[\nu-1]}=0$, i.e.  $\fnutilde{\nu-2}(\theta_{\eta_{i}}|\xBold)=\fnutilde{\nu-1}(\theta_{\eta_{i}}|\xBold)$
(a.s.) in (\ref{eq:ftilde_etaj}), which yields $\fnutilde{\nu-1}(\theta_{i}|\xBold)=\fnutilde{\nu}(\theta_{i}|\xBold)$
(a.s.), owing to (\ref{eq:IVB for CI}), $\iinn$. By definition (\ref{eq:KS}),
we then have $KS_{\itime}^{[\nu]}=0$, $\iinn$.
\end{proof}

\subsubsection{Computational load of Accelerated IVB \label{subsec:chap6:cost-of-Accelerated IVB}}

In the Accelerated IVB algorithm, either via ``many-to-one'' (\ref{eq:chap6:many-to-one})
or ``one-to-many'' schemes (\ref{eq:chap6:one-to-many:b}), we only
need to update $\kappa^{[\nu]}\TRIANGLEQ\sum_{\itime=1}^{\ntheta}\tau_{i}^{[\iIVB]}$
VB-marginals being those that are not yet converged at cycle $\nu$,
and leave out $n-\kappa^{[\nu]}$ converged VB-marginals. Although
the number $0\leq\kappa^{[\nu]}\leq n$ is varying at each cycle $\nu$,
the total number of VB-marginal updates at convergence (i.e. after
$\nu_{c}$ IVB cycles) is: $\sum_{\nu=1}^{\nu_{c}}$$\kappa^{[\nu]}\leq\nu_{c}n$.
If we call $\eIVB\TRIANGLEQ\frac{\sum_{\nu=1}^{\nu_{c}}\kappa^{[\nu]}}{n}$
the effective number of IVB cycles in accelerated scheme, $1\leq\eIVB\leq\nu_{c}$.
Note that, in the first IVB cycle, we have to update all IVB-marginals
from arbitrary initial VB-marginals, i.e. $1\leq\eIVB$, before being
able to identify their convergence. The acceleration rate can be defined
as $\frac{\nu_{c}}{\eIVB}$ and its bound is: 
\begin{equation}
1\leq\frac{\nu_{c}}{\eIVB}\leq\nu_{c}\label{eq:chap6:accelerated rate}
\end{equation}

.

Because the indicator condition (\ref{eq:chap6:one-to-many:b}) is
a relaxed version of KS-based condition (\ref{eq:chap6:KS-indicator}),
the total number of IVB cycles $\nIVB$ in Accelerated IVB and IVB
may not be equal to each other when $\xi>0$, although they are the
same when $\xi=0$, owing to Lemma \ref{lem:(Accelerated-IVB-algorithm)}.
In any case, the computational load of Accelerated IVB depends on
the effective number $\eIVB$, instead of $\nIVB$. Hence, if we denote
$\nIVB$ the true number of IVB cycle for both accelerated and unaccelerated
schemes, the value $\eIVB$ is more important than $\nIVB$ for evaluating
the speed of Accelerated IVB algorithm. In simulations in Chapter
\ref{=00005BChapter 8=00005D}, the value of $\eIVB$ will be shown
to be close to one for the HMC model, on average. 

\subsection{Accelerated FCVB approximation}

Let us recall that the Iterative FCVB approximation in Lemma \ref{lem:(Iterative-FCVB-algorithm)}
involves projecting all VB-marginals $\ftilde(\theta_{i}|\xBold)$
into their MAP values $\widetilde{\fdelta}(\theta_{i}|\xBold)=\delta(\theta_{i}-\widehat{\theta_{i}})$,
$\iinn$, where $\delta(\cdot)$ denotes probability distributions
singular at $\widehat{\theta_{i}}$. The IVB algorithm for CI structure
(\ref{eq:IVB for CI}) now becomes an Iterative FCVB algorithm, as
follows.

\subsubsection{Iterative FCVB algorithm for CI structure}

Owing to the sifting property of $\delta(\cdot)$, the FCVB-marginals
can be updated feasibly, as follows:

\begin{eqnarray}
\widetilde{\fdelta}^{[\nu]}(\theta_{i}|\xBold) & \propto & f(\xBold,\theta_{i},\theta_{\eta_{\itime}}=\widehat{\theta_{\eta_{\itime}}}^{[\nu]}),\ \iton\label{eq:FCVB}
\end{eqnarray}
where $\widehat{\theta_{\eta_{\itime}}}^{[\nu]}=\{\widehat{\theta}_{\eta_{i}^{BW}}^{[\nu-1]},\widehat{\theta}_{\eta_{i}^{FW}}^{[\nu]}\}$,
with the same notation as in (\ref{eq:IVB for CI}). From (\ref{eq:ch4:FCVB:CE})
and (\ref{eq:FCVB}), we only need to update $\fdelta^{[\iIVB]}(\btheta|\xBold)=\prod_{\itime=1}^{\ndata}\delta(\theta_{i}-\widehat{\theta_{i}}^{[\iIVB]})$
via iterative maximization steps, as follows:

\begin{eqnarray}
\widehat{\theta_{i}}^{[\nu]} & = & \arg\max_{\theta_{i}}\widetilde{\fdelta}^{[\nu]}(\theta_{i}|\xBold)\label{eq:FCVB=00003DCE}\\
 & = & \arg\max_{\theta_{i}}f(\xBold,\theta_{i},\theta_{\eta_{\itime}}=\widehat{\theta_{\eta_{\itime}}}^{[\nu]}),\ \iton\nonumber 
\end{eqnarray}
where $\widehat{\theta_{\eta_{\itime}}}^{[\nu]}=\{\widehat{\theta}_{\eta_{i}^{BW}}^{[\nu-1]},\widehat{\theta}_{\eta_{i}^{FW}}^{[\nu]}\}$.

\subsubsection{Stopping rule for Iterative FCVB algorithm}

In general, we can also choose KS distance between FCVB-marginals
$\widetilde{\fdelta}$ in (\ref{eq:FCVB}), as a convergence criterion
for FCVB-marginals, as in IVB. However, in a special case of discrete
$\theta_{i}$, the KS distances $KS_{\theta_{i}}$ will become zero
if $\widehat{\theta_{i}}^{[\nu-1]}=\widehat{\theta_{i}}^{[\nu]},\ \forall\iinn$.
Hence, the latter form will be used as a stopping rule in this thesis.
The convergence condition (\ref{eq:KS-threshold}) in IVB now becomes
the following convergence condition for Iterative FCVB algorithm (Algorithm
\ref{alg:FCVB-discrete}):

\begin{equation}
\widehat{\theta_{i}}^{[\nu-1]}=\widehat{\theta_{i}}^{[\nu]},\ \forall\iinn\label{eq:chap6:FCVB-stopping-rule}
\end{equation}

\textbf{}
\begin{algorithm}
\textbf{Iteration:} 

For $\nu=1,2,\ldots$, do \{

For $i=1,\ldots n$, do \{ 

evaluate (\ref{eq:FCVB=00003DCE})

if $\widehat{\theta_{i}}^{[\nu-1]}=\widehat{\theta_{i}}^{[\nu]}$,
\{ set $\varrho_{\itime}=0$ \} 

else: \{set $\varrho_{\itime}=1$\}\}\}\}

\textbf{Termination: }stop if\textbf{ $\varrho_{\itime}=0$, $\forall\iinn$}

\textbf{\caption{\label{alg:FCVB-discrete}Iterative FCVB for discrete r.v. $\protect\vtheta$}
}
\end{algorithm}

\subsubsection{Accelerated FCVB algorithm \label{subsec:Accelerated-FCVB-algorithm}}

In common with the computational flow of the Accelerated IVB algorithm
(Algorithm \ref{alg:Accel-IVB-algorithm}), we can design the Accelerated
FCVB algorithm (Algorithm \ref{alg:Accel-FCVB-discrete}) by replacing
the traditional IVB step (\ref{eq:IVB for CI}) and the KS-based convergence
condition (\ref{eq:KS-threshold}) with traditional FCVB step (\ref{eq:FCVB=00003DCE})
and the stopping rule (\ref{eq:chap6:FCVB-stopping-rule}), respectively.

Because the Lemma \ref{lem:(Accelerated-IVB-algorithm)} is applicable
to Accelerated VB scheme when converged KS distance is set zero, a
similar Lemma can be proposed for Accelerated FCVB scheme, as follows: 
\begin{lem}
\label{lem:(Accelerated-FCVB-algorithm)}\textbf{(Accelerated FCVB
algorithm)} The Accelerated FCVB algorithm (Algorithm \ref{alg:Accel-FCVB-discrete})
has the same output $\widehat{\theta_{i}}^{[\nu]}$ as Iterative FCVB
algorithm (Algorithm \ref{alg:FCVB-discrete}), i.e. we have $\varrho_{\itime}^{[\iIVB]}=\tau_{\itime}^{[\iIVB]}$,
$\forall\iinn$, at any IVB cycle $\iIVB=1,2,\ldots$
\end{lem}

\begin{proof}
The Lemma is a simple consequence of Lemma \ref{lem:(Accelerated-IVB-algorithm)},
because the stopping rule (\ref{eq:chap6:FCVB-stopping-rule}) is
equivalent to the convergence condition $KS_{\vtheta_{\itime}}=0,$
$\iinn$, where $KS_{\vtheta_{\itime}}$ in this case is the KS distance
between $\widetilde{\fdelta}^{[\nu]}(\theta_{i}|\xBold)$ and $\widetilde{\fdelta}^{[\nu-1]}(\theta_{i}|\xBold)$
in (\ref{eq:FCVB}-\ref{eq:FCVB=00003DCE}).
\end{proof}
\textbf{}
\begin{algorithm}
\textbf{Initialization}: set $\tau_{\itime}=1$, $\forall\iinn$

\textbf{Iteration:} 

For $\nu=1,2,\ldots$, do \{

For $i=1,\ldots n$, do \{ 

if $\tau_{\itime}=1$, \{ 

evaluate (\ref{eq:FCVB=00003DCE})

if $\widehat{\theta_{i}}^{[\nu-1]}=\widehat{\theta_{i}}^{[\nu]}$,
\{ set $\tau_{\itime}=0$ \} 

else: \{set $\tau_{j}=1$, $\forall j:\ \eta_{j}\ni i$ \}\}\}\}

\textbf{Termination: }stop if\textbf{ $\tau_{\itime}=0$, $\forall\iinn$}

\textbf{\caption{\label{alg:Accel-FCVB-discrete}Accelerated FCVB for discrete r.v.
$\protect\vtheta$}
}
\end{algorithm}

\section{VB-based inference for the HMC \label{sec:chap6:VB-infer-for-HMC}}

To the best of our knowledge, the VB methodology for computation in
HMC with known parameters (\ref{eq:ch6:Observation},\ref{eq:HMC})
has not been reported in the literature. In this section, we elaborate
the VB and FCVB approaches for this model, as well as the novel accelerated
schemes, presented in Section \ref{subsec:Accelerated-VB-algorithm}
and \ref{subsec:Accelerated-FCVB-algorithm}. We emphasize the novel
computation flows that result, comparing them to VA and confirming
that the Accelerated FCVB solution is an improved version of ICM,
but now furnished with a Bayesian justification and perspective.

\subsection{IVB algorithms for the HMC \label{subsec:VB-for-HMC}}

In common with the binary-tree approach in the FB and VA methods,
the VB approximation also exploits the Markov property (\ref{eq:MarkovProperty})
and provides two approximations for computation on the forward and
backward trajectories. The computational load of VB for the HMC is
therefore $O(\nu_{c}nM^{2})$ MULs.

Let us define an independent class, $\calF_{c}$, of $n$ variables
for the label field: $\breve{f}(L_{n}|\xBold_{n})=\prod_{i=1}^{n}\breve{f}(l_{i}|\xBold_{n})$.
By substituting the Markov property (\ref{eq:MarkovProperty}) into
the general binary factorization (\ref{eq:FB=00003Dgeneral}) and
then applying the IVB algorithm for CI structure (\ref{eq:IVB for CI}),
we can find the VB-smoothing marginals, $\ftilde(l_{i}|\xBold_{n})$,
feasibly: 
\begin{align}
\ftilde^{[\nu]}(l{}_{i}|\xBold_{n}) & \propto\ftilde^{[\nu-1]}(\xBold_{i+1:n}|l_{i})\ftilde^{[\nu]}(l_{i}|\xBold_{i}),\ \iton\label{eq:VB=00003Dsmoothing}
\end{align}
where the VB-forward filtering distributions $\ftilde(l_{i}|\xBold_{\itime})$
and VB-backward observation models $\ftilde(\xBold_{i+1:n}|l_{i})$
can be defined as follows:

\begin{eqnarray}
\ftilde^{[\nu]}(l_{1}|\xBold_{1})\equiv\ftilde^{[\nu]}(l_{1}|x_{1}) & \propto & f(x_{1}|l_{1})f(l_{1}|p)\label{eq:VB=00003DFW-BW}\\
\ftilde^{[\nu]}(l_{i}|\xBold_{i}) & \propto & \exp(E_{\ftilde^{[\nu]}(l_{i-1}|\xBold_{\ndata})}\log f(x_{i},l_{i}|l_{i-1},\TBold))\nonumber \\
\ftilde^{[\nu-1]}(\xBold_{i+1:n}|l_{i}) & \propto & \exp(E_{\ftilde^{[\nu-1]}(l_{i+1}|\xBold_{n})}\log f(l_{i+1}|l_{i},\TBold))\nonumber 
\end{eqnarray}
for $i\in\{2,\ldots,n\}$. Note that, although the expectations are
taken over $\ftilde^{[\nu]}(l_{i-1}|\xBold_{\ndata})$ and $\ftilde^{[\nu]}(l_{i+1}|\xBold_{\ndata})$
in (\ref{eq:VB=00003DFW-BW}), the notation $\xBold_{i}$ and $\xBold_{i+1:n}$
in distributions $\ftilde^{[\nu]}(l_{i}|\xBold_{i})$ and $\ftilde^{[\nu-1]}(\xBold_{i+1:n}|l_{i})$
are preserved in order to emphasize their forward and backward meaning,
respectively. 

For the discrete label field, the VB-marginals are, of course, always
of multinomial form, i.e. $\fnutilde{\nu}(l{}_{i}|\xBold_{n})=Mu_{l_{i}}(p_{i}^{[\nu]})$,
$i=1,\dots,n$. Moreover, substituting (\ref{eq:ch6:Observation},\ref{eq:HMC})
to (\ref{eq:VB=00003Dsmoothing}), via (\ref{eq:VB=00003DFW-BW}),
the iterative updates for shaping parameters $p_{i}^{[\nu]}$ are
revealed:

\begin{align}
p_{1}^{[\nu]} & \propto\exp(\log\Psi_{1}+(\log\TBold')p_{2}^{[\nu-1]}+p)\label{eq:VB=00003Dshaping}\\
p_{i}^{[\nu]} & \propto\exp(\log\Psi_{i}+(\log\TBold')p_{i+1}^{[\nu-1]}+(\log\TBold)p_{i-1}^{[\nu]})\nonumber 
\end{align}
for $i\in\{2,\ldots,n\}$. At convergence $\iIVB=\nIVB$, the VB inference
for HMC state is, from (\ref{eq:VB=00003Dsmoothing}): 
\begin{equation}
\ftilde^{[\nIVB]}(L_{n}|\xBold_{n})=\prod_{i=1}^{n}\ftilde^{[\nIVB]}(l_{i}|\xBold_{n})\label{eq:VB-HMC}
\end{equation}

The associate VB MAP estimate is, characteristically, the set of VB-marginal
MAP estimates: $\widehat{l_{i}}=\arg\max_{l_{i}}\ftilde(l_{i}|\xBold_{n})$,
$\iinn$, which will be used as state estimates of HMC.

\subsubsection{IVB stopping rule for the HMC via KS distance}

The recursive update of shaping parameters (\ref{eq:VB=00003Dshaping})
per IVB cycle requires $O(2nM^{2})$ MULs, $O(nM)$ EXPs and $O(2nM^{2})$
ADDs. Because the cost of MULs dominates the others, the computational
load of each IVB cycle is almost the same as that of FB (\ref{eq:alpha}-\ref{eq:gamma}).
Nominally, then, for $\nu_{c}$ IVB cycles at convergence, the VB
algorithm (Algorithm \ref{alg:IVB-HMC}) is $\nu_{c}$ times slower
than the FB algorithm. However, with the accelerated scheme of Section
\ref{subsec:Accelerated-VB-algorithm}, the Accelerated IVB for the
HMC (Algorithm \ref{alg:Accel-IVB-HMC}) is only $\eIVB$ times slower
than the FB algorithm, where $1\leq\eIVB\leq\nu_{c}$, as we will
show in simulation (Chapter \ref{=00005BChapter 8=00005D}).

\begin{algorithm}
\textbf{Initialization}: initialize $p_{i}^{[0]}$, $\forall\iinn$

\textbf{Iteration:} 

For $\nu=1,2,\ldots$, do \{

For $i=1,\ldots n$, do \{ 

evaluate (\ref{eq:VB=00003Dshaping}) 

if $KS_{\itime}^{[\nu]}\leq\xi$, \{ set $\varrho_{\itime}=0$ \} 

else: \{set $\varrho_{\itime}=1$\}\}\}\}

\textbf{Termination: }stop if\textbf{ $\varrho_{\itime}=0$, $\forall\iinn$}

\textbf{Return} $\widehat{l}_{i}=\boldsymbol{\epsilon}(k_{i})$, with
$k_{i}=\arg\max_{k}(p_{k,i})$, $i=1,\dots,n$. 

\caption{\label{alg:IVB-HMC} IVB algorithm for the HMC }
\end{algorithm}
$\ $

\begin{algorithm}
\textbf{Initialization:} initialize $p_{i}^{[0]}$ and set $\tau_{i}=1$,
$\iinn$

\textbf{Iteration:} 

For $\nu=1,2,\ldots$, do \{

For $i=1,\ldots n$, do \{ 

if $\tau_{\itime}=1$, \{ 

evaluate (\ref{eq:VB=00003Dshaping}) 

if $KS_{\itime}^{[\nu]}\leq\xi$, \{ set $\tau_{\itime}=0$ \} 

else: \{set $\tau_{j}=1$, $\forall j:\ \eta_{j}\ni i$ \}\}\}\}

\textbf{Termination:} stop if $\tau_{i}=0$, $\forall\iinn$. 

\textbf{Return} $\widehat{l}_{i}=\boldsymbol{\epsilon}(k_{i})$, with
$k_{i}=\arg\max_{k}(p_{k,i})$, $i=1,\dots,n$. 

\caption{\label{alg:Accel-IVB-HMC}Accelerated IVB algorithm for the HMC }
\end{algorithm}

\subsubsection{IVB Stopping rule for the HMC via KLD \label{sec:KLD-for-HMC}}

For comparison in simulations later, let us also consider the traditional
stopping rule for the IVB algorithms (Algorithm \ref{alg:IVB-HMC}-\ref{alg:Accel-IVB-HMC})
involving computation of the KLD. By the definition, the KLD for VB
in the HMC case can be evaluated as:

\[
KLD_{\ftilde||f}=E_{\ftilde(L_{n}|\xBold_{n})}\log\ftilde(L_{n}|\xBold_{n})-E_{\ftilde(L_{n}|\xBold_{n})}\log f(L_{n}|\xBold_{n})
\]
where $\ftilde(L_{n}|\xBold_{n})=\prod_{\itime=1}^{\ndata}\ftilde(l_{\itime}|\xBold_{n})$.
From IVB algorithm (Section \ref{subsec:VB-for-HMC}), the first term
of $KLD_{\ftilde||f}$ is equal to:

\[
\sum_{i=1}^{n}E_{\ftilde(l_{i}|\xBold_{n})}\log\ftilde(l_{i}|\xBold_{n})=\sum_{i=1}^{n}p_{i}'\log p_{i}
\]
From forward factorization of the posterior HMC $f(L_{n}|\xBold_{n})$
(\ref{eq:TERNARY}), the second term of $KLD_{\ftilde||f}$ is:

\begin{eqnarray*}
E_{\ftilde(L_{n}|\xBold_{n})}\log f(L_{n}|\xBold_{n}) & = & \sum_{i=1}^{n-1}E_{\ftilde(l_{i}|\xBold_{n})\ftilde(l_{i+1}|\xBold_{n})}\log f(l_{i}|l_{i+1},\xBold_{i})+E_{\ftilde(l_{n}|\xBold_{n})}\log f(l_{n}|x_{n})\\
 & = & \sum_{i=1}^{n-1}p_{i}'\left(\log A_{i}\right)p_{i+1}+p_{n}'\log\alpha_{n}
\end{eqnarray*}
where $A_{i}$ is computed via Corollary \ref{cor:FB=00003DA,B}. 

\subsection{FCVB algorithms for the HMC \label{subsec:chap6:FCVB-algorithms-HMC}}

By definition, the FCVB-marginals $\widetilde{\fdelta}^{[\nu]}(l_{i}|\xBold_{n})=Mu_{l_{i}}(\widehat{p}_{i}^{[\nu]})$
are of the same multinomial form as the VB-marginals, but with different
shaping parameters, $\widehat{p}_{i}^{[\nu]}$, as specified next. 

\subsubsection{Accelerated FCVB for the homogeneous HMC}

Similarly to the IVB algorithm for HMC (\ref{eq:VB=00003Dsmoothing}-\ref{eq:VB=00003DFW-BW}),
the shaping parameters $\widehat{p}_{i}^{[\nu]}$ of FCVB-marginals
(as illustrated in Fig \ref{fig:VBV}) can be evaluated in a similar
way of (\ref{eq:VB=00003Dshaping}):

\begin{eqnarray}
\log\widehat{p}_{1}^{[\nu]} & = & \log\Psi_{1}+\log\TBold(\widehat{k}_{2}^{[\nu-1]},:)'+\log p+const\label{eq:FCVB=00003Dshaping}\\
\log\widehat{p}_{i}^{[\nu]} & = & \log\Psi_{i}+\log\vartheta(:,\widehat{k}_{i+1}^{[\nu-1]},\widehat{k}_{i-1}^{[\nu]})+const\nonumber \\
\log\widehat{p}_{n}^{[\nu]} & = & \log\Psi_{n}+\log\TBold(:,\widehat{k}_{n-1}^{[\nu]})+const\nonumber 
\end{eqnarray}
where $\TBold(k,:)$ , $\TBold(:,k)$ are $k$th row, column of $\TBold$,
respectively, and: 
\begin{equation}
\widehat{k}_{i}^{[\nu]}=\arg\max_{k}(\log\widehat{p}_{k,i}^{[\nu]}),\ i=\{1,\ldots,n\}\label{eq:FCVB=00003Dk}
\end{equation}

The $M^{3}$ precomputed values of $\vartheta$ are defined as:

\begin{equation}
\log\vartheta(:,\widehat{k}_{i+1}^{[\nu-1]},\widehat{k}_{i-1}^{[\nu]})=\log\TBold(\widehat{k}_{i+1}^{[\nu-1]},:)'+\log\TBold(:,\widehat{k}_{i-1}^{[\nu]})\label{eq:Tdouble}
\end{equation}

For the homogeneous HMC, the purpose of extra precomputed step (\ref{eq:Tdouble})
is to reduce the complexity of FCVB by half of the additions. Although
this method requires extra memory of $O(M^{3})$ for $\vartheta$,
it is safe to assume that $O(M^{3})\leq O(n)$, where $O(n)$ is the
memory requirement of the non-precomputed scheme.

Note that, in the traditional Iterative FCVB (Algorithm \ref{alg:FCVB-HMC}),
the total computational load up for the $\nu_{c}$ cycles is $O(\nu_{c}nM)$.
However, in accelerated scheme (Algorithm \ref{alg:Accel-FCVB-HMC}),
we only need to update a subset, $\kappa^{[\nu]}$, of the $n$ labels
in the $\nu$th cycle, where $0\leq\kappa^{[\nu]}=\sum_{\itime=1}^{\ndata}\tau_{\itime}^{[\iIVB]}\leq\ndata$.
At convergence, the computation cost of accelerated FCVB is therefore
$O(\eIVB nM)$, where $1\leq\eIVB\TRIANGLEQ\frac{\sum_{\nu=1}^{\nu_{c}}\kappa^{[\nu]}}{n}\leq\nIVB$,
as discussed in Section \ref{subsec:chap6:cost-of-Accelerated IVB}.

\textbf{}
\begin{algorithm}
\textbf{Initialization}: initialize $k_{i}^{[0]}\in\{1,\dots,M\}$,
evaluate (\ref{eq:Tdouble}), $\forall\iinn$

\textbf{Iteration:} 

For $\nu=1,2,\ldots$, do \{

For $i=1,\ldots n$, do \{ 

evaluate (\ref{eq:FCVB=00003Dshaping})

if $k_{i}^{[\nu-1]}=k_{i}^{[\nu]}$, \{ set $\varrho_{\itime}=0$
\} 

else: \{set $\varrho_{\itime}=1$\}\}\}\}

\textbf{Termination: }stop if \textbf{$\varrho_{\itime}=0$, $\forall\iinn$}

\textbf{Return} $\widehat{l}_{i}=\boldsymbol{\epsilon}(\widehat{k}_{i}^{[\nIVB]})$,
$i=1,\dots,n$. 

\textbf{\caption{\label{alg:FCVB-HMC}Iterative FCVB for the homogeneous HMC}
}
\end{algorithm}
\textbf{}
\begin{algorithm}
\textbf{Initialization}: initialize $k_{i}^{[0]}\in\{1,\dots,M\}$,
evaluate (\ref{eq:Tdouble}) and set $\tau_{\itime}=1$, $\forall\iinn$

\textbf{Iteration:} 

For $\nu=1,2,\ldots$, do \{

For $i=1,\ldots n$, do \{ 

if $\tau_{\itime}=1$, \{ 

evaluate (\ref{eq:FCVB=00003Dshaping})

if $k_{i}^{[\nu-1]}=k_{i}^{[\nu]}$, \{ set $\tau_{\itime}=0$ \} 

else: \{set $\tau_{j}=1$, $\forall j:\ \eta_{j}\ni i$ \}\}\}\}

\textbf{Termination: }stop if\textbf{ $\tau_{\itime}=0$, $\forall\iinn$}

\textbf{Return} $\widehat{l}_{i}=\boldsymbol{\epsilon}(\widehat{k}_{i}^{[\nIVB]})$,
$i=1,\dots,n$. 

\textbf{\caption{\label{alg:Accel-FCVB-HMC}Accelerated FCVB for the homogeneous HMC}
}
\end{algorithm}
\begin{figure}
\begin{centering}
\includegraphics[width=0.8\columnwidth]{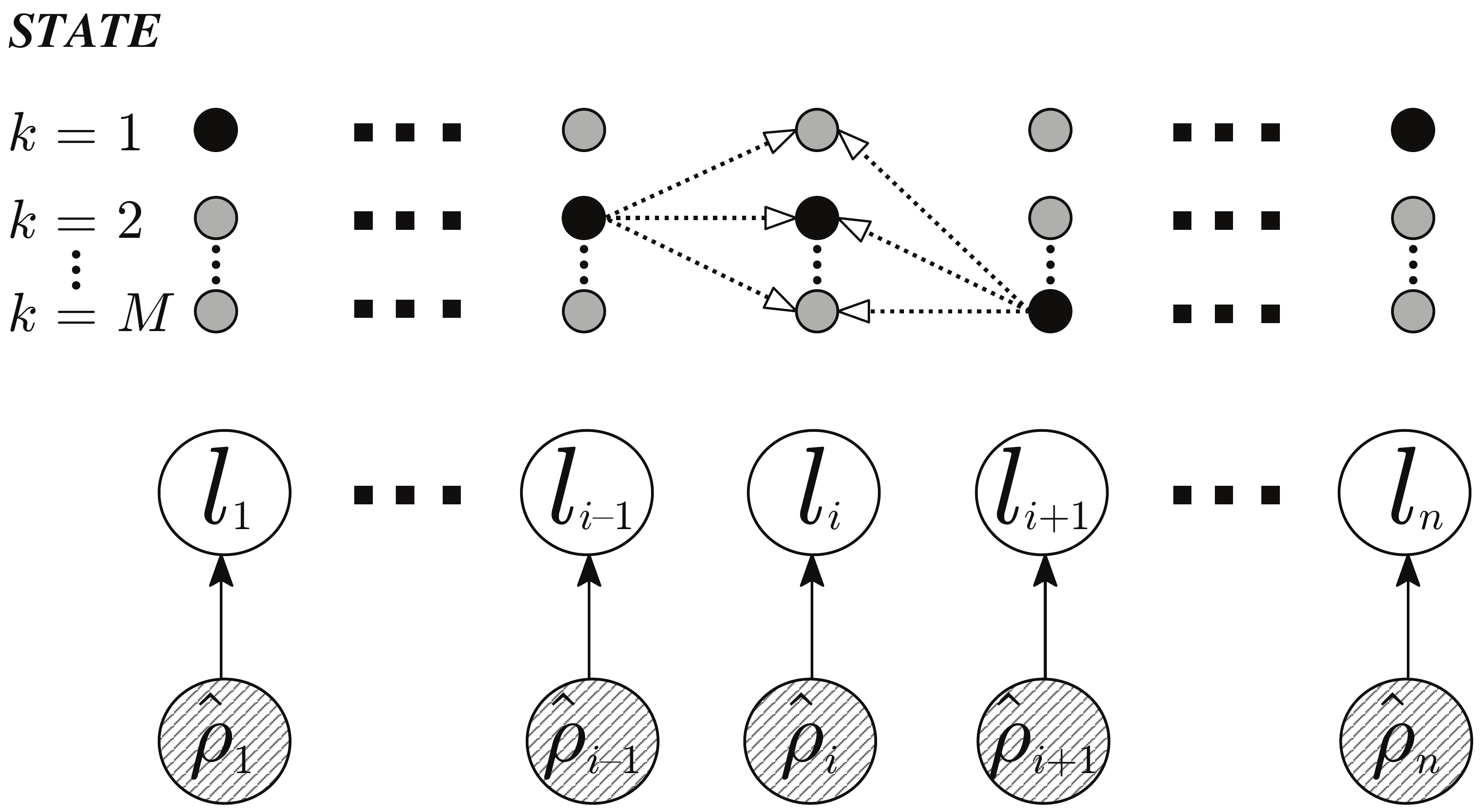}
\par\end{centering}
\caption{FCVB-marginals for HMC posterior (Fig. \ref{fig:FB=00003Db}): trellis
diagram (top) and DAG (bottom). Dotted arrows denote the CE substitutions
in the Iterative FCVB algorithm. Black circles denote the mode of
FCVB-marginal. \label{fig:VBV}}
\end{figure}

\subsubsection{Further acceleration via a bubble-sort-like procedure \label{ch6:sub:bubble-sort-like}}

\begin{table}
\begin{centering}
\begin{tabular}{|c|c|c|c|c|c|}
\hline 
 & ML & FCVB & VA & FB & VB\tabularnewline
\hline 
\hline 
EXP & -- & -- & -- & -- & $O(nM\nu_{c})$\tabularnewline
\hline 
MUL & -- & -- & -- & $O(2nM^{2})$ & $O(2nM^{2}\nu_{c})$\tabularnewline
\hline 
ADD & -- & $O(nM\nu_{c})$ & $O(nM^{2})$ & $O(2nM^{2})$ & $O(2nM^{2}\nu_{c})$\tabularnewline
\hline 
MAX & $O(nM)$ & $O(nM\nu_{c})$ & $O(nM^{2})$ & $O(nM)$ & $O(nM)$\tabularnewline
\hline 
Memory & -- & $O(n+M^{3})$ & $O(nM)$ & $O(2nM)$ & $O(nM)$\tabularnewline
\hline 
\end{tabular}
\par\end{centering}
\caption{\label{tab:ch6:ComputationalComplexity}Computational and memory cost
of algorithms for homogeneous HMC. $\nu_{c}$ is the number of IVB
cycles at convergence (for accelerated scheme, $\nu_{c}$ is replaced
by $\protect\eIVB$). From the lowest to highest (typical) computational
load: ML (Section \ref{subsec:ML}), FCVB (Algorithm \ref{alg:FCVB-HMC}-\ref{alg:Accel-FCVB-HMC}),
VA (Algorithm \ref{alg:chap6:Viterbi-Algorithm}), FB (Algorithm \ref{alg:chap6:FB-algorithm})
and VB (Algorithm \ref{alg:IVB-HMC}-\ref{alg:Accel-IVB-HMC}).}
\end{table}
For comparing the computational complexity of above algorithms in
the homogeneous HMC, a summary of their computational and memory cost
is given in Table \ref{tab:ch6:ComputationalComplexity}. Then, from
Table \ref{tab:ch6:ComputationalComplexity}, it looks like each FCVB
cycle must be slower than ML, since each FCVB cycle seems to always
require more operators than ML. In practice, however, each FCVB cycle
can be implemented more quickly than ML. To achieve this, let us consider
the computational load in hardware level.

The task of finding the maximum value of length-$M$ vector requires
$M$ MAX operations, in which each MAX involves four steps (i-iv),
as illustrated in Algorithm \ref{alg:Traditional-MAX}. Similarly,
the task of finding the sum of a length-$M$ vector requires $M$
ADD operations, in which each ADD requires two steps (a-b), as illustrated
in Algorithm \ref{alg:Traditional-ADD}. Hence, in practice, one MAX
is often considered to be equivalent to $2$ ADDs {[}\citet{ch8:1Max_2Sum}{]}.
In terms of computational load, if both maximum and sum are required,
they can be computed more quickly via a Max-Sum combination, defined
in Algorithm \ref{alg:MAX-ADD}. 

For further speed-up, notice that if the step (ii) in one MAX operation
(Algorithm \ref{alg:Traditional-MAX}) does not detect any higher
value than current maximum value, steps (iii) and (iv) will then not
be implemented in that MAX operation. We can, therefore, design a
pilot-based MAX scheme, which initializes the current maximum value
with the \textit{pilot} element in length-$M$ vector, as illustrated
in Algorithm \ref{alg:pilot-MAX}. The pilot element can be chosen
as any of the $M$ vector elements. Therefore, in the ideal case where
the pilot element is the true maximum, steps (iii)-(iv) can be avoided
completely, i.e. the cost of pilot-based MAX (i.e. each iteration
in Algorithm \ref{alg:pilot-MAX-ADD}) can be as low as half of that
of conventional MAX (i.e. each iteration in Algorithm \ref{alg:Traditional-MAX}).
Likewise, in the ideal case, one pilot-MAX-ADD (i.e. each iteration
in Algorithm \ref{alg:pilot-MAX-ADD}) only needs three steps (i)-(ii)-(b),
which means its cost is in the range $\eta\in\left[\frac{3}{4},\frac{5}{4}\right]$
of the cost of conventional MAX (in Algorithm \ref{alg:Traditional-MAX}).

Now, because each iterative FCVB cycle (\ref{eq:FCVB=00003Dshaping}-\ref{eq:FCVB=00003Dk})
involves Max-Sum scheme and ML requires traditional Max scheme, the
considerations above can be applied to comparing the relative costs
of FCVB and ML. Since the current FCVB cycle relies on the label estimates
in the previous cycle, the latter can be used as pilot elements in
current FCVB cycle. Therefore, the pilot-Max-Sum scheme (Algorithm
\ref{alg:pilot-MAX-ADD}) is applicable in FCVB cycles. In contrast,
there is no scheme for picking pilot elements in ML, and therefore
the ML has to rely on traditional Max scheme. For this reason, the
cost of each pilot-based FCVB cycle is in the range $\eta\in\left[\frac{3}{4},\frac{5}{4}\right]$
of the cost of traditional ML, i.e. it is possible for each FCVB cycle
to run faster than ML.

Notice that, the traditional MAX (Algorithm \ref{alg:Traditional-MAX})
is simply the first step in the \textit{bubble sort} algorithm {[}\citet{ch5:BK:Bible:Algorithms}{]},
in which the maximum value ``floats'' up progressively after each
comparison step (ii). Then, the pilot-based MAX procedure above, which
requires \textit{a priori} knowledge on pilot element, can be loosely
regarded as the first step of a pilot-based bubble sort, i.e. the
maximum value will ``float'' up more quickly, given a good pilot
element.

Obviously, we may have more than one way to make FCVB cycle run faster
than ML. The above procedure is merely to illustrate such a possibility. 

\begin{algorithm}
\textbf{Initialization}: set $k_{max}=1$, $k=1$ and $v_{max}=v_{1}$

\textbf{Iteration: }

(i) increase pointer by 1 (i.e. $k\leftarrow k+1$) and retrieve $v_{k}$ 

(ii) 1 binary comparison between $v_{max}$ and $v_{k}$ 

If $v_{\max}<v_{k}$ in (ii), do: \{

(iii) 1 storage for new maximum value (i.e. $v_{max}\leftarrow v_{k}$)

(iv) 1 storage for position of new maximum value (i.e. $k_{max}\leftarrow k$)
\}

\textbf{Termination: }stop if\textbf{ $k=\nstate$}

\textbf{Return: }$v_{max}$ and $k_{max}$

\caption{\label{alg:Traditional-MAX}Traditional Max for finding maximum  $v_{max}$
of $v=[v_{1},\ldots,v_{\protect\nstate}]'$}
\end{algorithm}
\begin{algorithm}
\textbf{Initialization}: set $v_{sum}=0$ and $k=1$

\textbf{Iteration:}

(a) increase pointer by 1 (i.e. $k\leftarrow k+1$) and retrieve $v_{k}$ 

(b) 1 binary addition $v_{sum}\leftarrow v_{sum}+v_{k}$ 

\textbf{Termination: }stop if\textbf{ $k=\nstate$}

\textbf{Return: }$v_{sum}$ 

\caption{\label{alg:Traditional-ADD}Traditional Sum for finding sum $v_{sum}$
of $v=[v_{1},\ldots,v_{\protect\nstate}]'$}
\end{algorithm}
\begin{algorithm}
\textbf{Initialization}: initialize $k_{pilot},$ set $k=1$, and
$v_{max}=v_{k_{pilot}}$

\textbf{Iteration: }

Implement (i-iv) in Algorithm \ref{alg:Traditional-MAX}

\textbf{Termination: }stop if\textbf{ $k=\nstate$}

\textbf{Return: }$v_{max}$ and $k_{max}$

\caption{\label{alg:pilot-MAX}Pilot-based Max for finding maximum $v_{max}$
of $v=[v_{1},\ldots,v_{\protect\nstate}]'$}
\end{algorithm}
\begin{algorithm}
\textbf{Initialization}: set $k_{max}=1$, $k=1$, $v_{max}=v_{1}$
and $v_{sum}=0$

\textbf{Iteration: }

Implement (i)-(ii)-(b)-(iii)-(iv) in Algorithm \ref{alg:Traditional-MAX}-\ref{alg:Traditional-ADD}.

\textbf{Termination: }stop if\textbf{ $k=\nstate$}

\textbf{Return: }$v_{max}$, $k_{max}$ and $v_{sum}$

\caption{\label{alg:MAX-ADD}Max-Sum for finding maximum $v_{max}$ and sum
$v_{sum}$ of $v=[v_{1},\ldots,v_{\protect\nstate}]'$}
\end{algorithm}
\begin{algorithm}
\textbf{Initialization}: initialize $k_{pilot}$ and set $k=1$, $v_{max}=v_{k_{pilot}}$
and $v_{sum}=0$

\textbf{Iteration: }

Implement (i)-(ii)-(b)-(iii)-(iv) in Algorithm \ref{alg:Traditional-MAX}-\ref{alg:Traditional-ADD}.

\textbf{Termination: }stop if\textbf{ $k=\nstate$}

\textbf{Return: }$v_{max}$, $k_{max}$ and $v_{sum}$

\caption{\label{alg:pilot-MAX-ADD}Pilot-based Max-Sum for finding maximum
$v_{max}$ and sum $v_{sum}$ of $v=[v_{1},\ldots,v_{\protect\nstate}]'$}
\end{algorithm}

\section{Performance versus computational load \label{sec:chap6:Perform-vs-compu}}

In this section, we will examine the trade-off between performance
and computational costs for each of the algorithms above. For comparison
of estimators performance, the bit-error-rate (BER) is widely adopted
in practice {[}\citet{ch2:bk:SEP:Haykin06}{]}. Since minimizing BER
can also be interpreted as minimizing Hamming distance, the simulation
results can be explained intuitively via the Bayesian risk perspective
(Section \ref{subsec:chap4:Bayes-risk}). 

\subsection{Bayesian risk for HMC inference \label{subsec:chap6:Bayesian-risk-for-HMC}}

Let us incorporate Hamming distance into a loss function $Q(\widehat{L_{n}},L_{n})$,
that quantifies the cost of errors in the estimated HMC , $\widehat{l_{i}}\in\widehat{L_{n}}$,
relative to the simulated field, $l_{i}\in L_{n}$, $\iinn$, as follows:
\[
Q(\widehat{L_{n}},L_{n})=1-\frac{1}{n}\sum_{i=1}^{n}\delta[\widehat{l_{i}}-l_{i}]
\]
where $Q:\ M^{n}\times M^{n}\rightarrow[0,1]$ and $Q(\widehat{L_{n}},L_{n})=0\Leftrightarrow\widehat{L_{n}}=L_{n}$.
As shown in Lemma \ref{lem:chap4:MR-Hamming}, the estimate $\widehat{L_{n}*}$
minimizing expected loss - i.e. the minimum risk (MR) estimate - is:
\begin{equation}
\widehat{l_{i}^{*}}=\arg\max_{l_{i}}f(l_{i}|\xBold_{n})\label{eq:chap8:MR-Hamming}
\end{equation}
where $\widehat{l_{i}^{*}}\in\widehat{L_{n}*}$, $\iinn$. This MR
risk provides insight into the observed trade-off in simulations (Chapter
\ref{=00005BChapter 8=00005D}), as follows:
\begin{itemize}
\item FB can provide MR risk estimate $\widehat{L_{n}*}$ , being a sequence
of marginal MAPs (\ref{eq:chap8:MR-Hamming}), where $f(l_{i}|\xBold_{n})$
are the smoothing marginals already computed in FB. 
\item The performance of VA is close to that of FB, with low computational
load, for two reasons: VA replaces the marginal MAP of the smoothing
marginals, $f(l_{i}|\xBold_{n})$, for FB with the MAP of the CE-based
profile inferences, $\fprofile(l_{i}|\xBold_{n})$, with typically
very little difference between these two estimates. However, the computational
load of $\fprofile(l_{i}|\xBold_{n})$ is very low, compared to that
of $f(l_{i}|\xBold_{n})$. 
\item ML yields the worst performance because it returns the estimates based
on the local observation model, i.e. $\widehat{l_{i}}=\arg\max_{l_{i}}f(x_{i}|l_{i})$,
without any prior regularization, i.e. $f(x_{i}|l_{i})$ is a bad
approximation of $f(l_{i}|\xBold_{n})$ since it does not involve
the HMC structure into account. However, the local  observation structure
makes ML work so fast.
\end{itemize}
An important role for the novel VB-approaches to HMC inference is
in furnishing new trade-offs between computational load and accuracy,
other than the extremes confined by ML on one hand, and FB and VA
on the other. This flexibility is achieved via the following two design
steps:
\begin{itemize}
\item In the first step, VB can return the MAP estimates of the respective
VB-marginals, $\widehat{l_{i}}=\arg\max_{l_{i}}\ftilde(l_{i}|\xBold_{n})$
(\ref{eq:VB=00003Dsmoothing}). Since the VB-marginals are approximated
via HMC-regularized model (\ref{eq:Joint}), this VB point estimate
enjoys far better performance that that of ML. 
\item The second step is to reduce the complexity of VB via the CE approach
in FCVB (\ref{eq:FCVB=00003Dshaping}). Since the FCVB performance
is slightly worse than VB, as illustrated in simulation (Chapter \ref{=00005BChapter 8=00005D}),
its performance trade-off can be explained via VB's structure. 
\end{itemize}

\subsection{Accelerated schemes \label{subsec:chap6:Accelerated-schemes}}

The number $\nIVB$ of IVB cycles in the HMC case will be shown, in
simulations in Chapter \ref{=00005BChapter 8=00005D}, as a factor
of logarithmic of either $\ndata$, the number of data, or $\nstate$,
the number of state. We conjecture that this phenomenon is relevant
to the exponential forgetting property of posterior marginals $f(l_{i}|\xBold_{n})$
(Corollary 2.1 in {[}\citet{ch4:art:JuriLember11}{]}): 
\begin{equation}
||f(l_{i}|\xBold_{i-k:i+k})-f(l_{i}|\xBold_{-\infty:\infty})||\leq\frac{C}{\upsilon^{k-1}},\ k\geq0\label{eq:Lember}
\end{equation}
where $\upsilon$ is a constant, $\upsilon>1$, $C$ is a non-negative
finite random variable, $\xBold_{-\infty:\infty}$ denotes an infinite
number $n$ of data and $||\cdot||$ is the total variation distance
(i.e. $\mathcal{L}_{1}$-norm). This property states that $f(l_{i}|\xBold_{n})$,
with high enough $n$, only depends on a factor of $\log(n)$ data.
Because each IVB-cycle updating in HMC (\ref{eq:VB=00003Dsmoothing})
projects one more backward datum into the VB-marginals $\ftilde(l_{i}|\xBold_{n})$,
we conjecture that we only need a factor of $\log(\ndata)$ IVB cycles
in order for $\ftilde(l_{i}|\xBold_{n})$ to converge, i.e. $\nu_{c}\approx O(\log(n))$. 

For the accelerated schemes, $\nu_{c}$ is replaced by $\eIVB$, which
we will find is almost a constant and close to $1$, on average, for
any value of $n$ in simulations (Chapter \ref{=00005BChapter 8=00005D}).
Hence, the accelerated scheme for FCVB and VB reduces significantly
the computational load of traditional FCVB and VB for the HMC. Perhaps,
this is because these schemes only update the non-converged VB-marginals,
whose number decreases in a factor of $\log(n)$, owing to (\ref{eq:Lember}).
This decrease is, therefore, likely to cancel out the $\nu_{c}\approx O(\log(n))$.

\subsection{FCVB algorithm versus VA \label{subsec:chap6:FCVB-versus-VA}}

By combining the CE approach with the independent class approximation,
$\fdelta\in\mathcal{F}_{\delta}$, FCVB is faster than the non-iterative
VA scheme, despite FCVB being an iterative scheme. From Table \ref{tab:ch6:ComputationalComplexity},
we can see that FCVB reduces the computational load from $O(nM^{2})$
for VA down to $O(nM\nu_{c})$. Hence, the computational load of VA
increases quadratically with $M$, while FCVB's computational load
increases only linearly with $\nstate$. Moreover, from simulations
(Chapter \ref{=00005BChapter 8=00005D}) for FCVB, it will be shown
that $\eIVB\ll\nu_{c}\approx O(\log(\nstate))$ when $\nstate$ is
large.

A key advantage of the FCVB scheme is applicability in many practical
applications. Note that, VA is mostly applied to finite state HMC,
because the conditional CE (\ref{eq:condCE}) is very hard to evaluate
for continuous states, other than in the Gaussian context of the Kalman
filter {[}\citet{ch6:art:VA:continuous_2011}{]}. In contrast, since
Iterative FCVB algorithm is equivalent to the ICM algorithm {[}\citet{ch4:art:ICM:Besag86}{]},
with application in the general Hidden Markov Model (HMM) context,
the Accelerated FCVB - a faster version of ICM - can also be applied
feasibly to the continuous state case. 

Another application of VB and FCVB for the HMC is the online context.
Because VB-marginals only depend on their first order neighbour in
(\ref{eq:VB=00003Dshaping}) and (\ref{eq:FCVB=00003Dshaping}), the
VB-marginal updates in IVB cycles can be run consecutively and in
parallel, i.e. $\ftilde^{[\iIVB]}(l_{\itime-1}|\xBold_{\ndata})$
can be updated right after $\ftilde^{[\iIVB-1]}(l_{\itime}|\xBold_{\ndata})$
is updated, without the need to wait for the $(\iIVB-1)$th IVB cycle
to be finished. If we design a lag-window of duration equal to $\nu_{c}$,
the $\nIVB$ FCVB and VB cycles can be evaluated, consecutively and
in parallel, as an online algorithm and return exactly the same results
as the offline case. 

\section{Summary}

In this chapter, fully Bayesian inference and its computation for
the HMC has been developped, revealing a key insight into all of the
state-of-the-art algorithms, namely FB, VA and ICM; i.e. the Markov
property of the HMC has been shown to be a necessary foundation for
their efficient recursive computational flow. In replacing the exact
marginal computations in FB algorithm, the VA has been shown to produce
CE-based approximate inferences, where CE substitution is used to
reduce the computational load significantly. Importantly, the joint
MAP estimate remains invariant under this CE substitution.

Inspired by this insight, FCVB, which is a CE-based variant of the
independent-structure VB approximation, has been proposed as a VA
variant to bridge the trade-off gap between VA and ML. Although FCVB
has previously been reported as the ICM algorithm in the literature,
this novel VB-based derivation and implementation yields insight into
the regimes of operation where the algorithm is expected to perform
well. Indeed, we will see in simulations in Chapter \ref{=00005BChapter 8=00005D}
that when correlation in the HMC is not too high, FCVB is an attractive
algorithm, since its performance is then close to VA, with much lower
computational load. Empirically, these simulations also show that
the accelerated scheme, as designed in this chapter, reduces the number
of IVB cycles to about one, on average. 

Finally, the accelerated ICM/FCVB algorithm---proposed in this chapter---can
work in both online and offline modes, with no difference at the final
output, as noted in Section \ref{subsec:chap6:FCVB-versus-VA}. In
contrast, FB and VA are exclusively the offline schemes.

%auto-ignore
%auto-ignore
%%%% Common

\global\long\def\xbold{\mathbf{x}}%

\global\long\def\btheta{\mathbf{\boldsymbol{\theta}}}%

\global\long\def\xdata{x}%

\global\long\def\vtheta{\theta}%

\global\long\def\htheta{\widehat{\theta}}%

\global\long\def\vxi{\xi}%

\global\long\def\vphi{\phi}%

\global\long\def\vpsi{\psi}%

\global\long\def\veta{\eta}%

\global\long\def\REAL{\mathbb{R}}%

\global\long\def\calTheta{\Theta}%

\global\long\def\spaceO{\Omega}%

\global\long\def\calE{\mathcal{E}}%

\global\long\def\calF{\mathcal{F}}%

\global\long\def\ndata{n}%

\global\long\def\nstate{m}%

\global\long\def\itime{i}%

\global\long\def\istate{k}%

\global\long\def\funh#1{h\left(#1\right)}%

\global\long\def\fung#1{g\left(#1\right)}%

\global\long\def\seti#1#2{#1=1,2,\ldots,#2}%

\global\long\def\setd#1#2{\{#1{}_{1},#1{}_{2},\ldots,#1_{#2}\}}%

\global\long\def\TRIANGLEQ{\triangleq}%

%%%% Chapter 7

\global\long\def\vinterest{\vtheta_{i}}%

\global\long\def\vnuisance{\vtheta_{\backslash i}}%

\global\long\def\zbinary{\vtheta_{\itime},\vtheta_{\backslash\itime}}%

\global\long\def\ftilde{\widetilde{f}}%

\global\long\def\NGauss{\mathcal{N}}%

\global\long\def\CGauss{\mathcal{CN}}%

\global\long\def\xBold{\mathbf{x}}%

\global\long\def\gBold{\mathbf{g}}%

\global\long\def\noise{z}%

\global\long\def\bnoise{\mathbf{z}}%

\global\long\def\Ibold{\mathbf{I}}%

\global\long\def\xmodel{x}%

\global\long\def\AcT{a}%

\global\long\def\ACT{\mathcal{A}}%

\global\long\def\Tdecision{\theta}%

\global\long\def\TDecision{\Theta}%

\global\long\def\Dinference{x}%

\global\long\def\DInference{x}%

\global\long\def\xmodel{x}%

\global\long\def\tmodel{\theta}%

\global\long\def\OSinewave{\Omega}%

\global\long\def\PSymbol{P}%

\global\long\def\qbit{q}%

\global\long\def\asymbol{a}%

\global\long\def\aSymbol{\boldsymbol{a}}%

\global\long\def\SAlphabet{S}%

\global\long\def\MAlphabet{M}%

\global\long\def\ttransform{\theta}%

\global\long\def\ptransform{\phi}%

\global\long\def\ztransformation{\zeta}%

\global\long\def\latransformation{\lambda}%

\global\long\def\ramplitude{r_{a}}%

\global\long\def\muamplitude{\mu_{a}}%

\global\long\def\Ndata{n}%

\global\long\def\rnoise{r_{e}}%

\chapter{The transformed Variational Bayes (TVB) approximation \label{=00005BChapter 7=00005D}}

\section{Motivation \label{sec:chap7:Motivation}}

Although the VB approximation has been proposed as an efficient approximation
for intractable distributions, there is still much room for improvement. 

The VB approximation is, of course, a parametric distribution, $f_{\theta}\TRIANGLEQ f(\theta|s)$,
where $s$ denotes the parameter. For simplicity, let us consider
a binary partition $\theta=\{\zbinary\}$, where, again, $\theta_{\backslash i}$
is the complement set of $\theta_{i}$ in $\vtheta$, with neither
$\theta_{\backslash i}$ nor $\theta_{i}$ is empty. As explained
in Section \ref{subsec:chap4:Variational-Bayes-(VB)}, the VB approximation
reaches a local minimum of Kullback-Leibler divergence $KLD(\tilde{f}_{\theta}||f_{\theta})=E_{\tilde{f}_{\theta}}\log(\tilde{f}_{\theta}/f_{\theta})$,
for approximate distributions $\ftilde_{\vtheta}\in\mathcal{F}_{c}$,
the class of factored distributions (being those for which $\theta_{\backslash i}$
and $\theta_{i}$ are independent, given $s$).

The idea behind the transformed VB (TVB) approximation is illustrated
in Fig. \ref{fig:ch7:TVB}. The transformed distribution $f_{\phi}\TRIANGLEQ f(\phidelay|s')$,
with parameter $s'$ (a function of $s$), should be designed to minimize
the $KLD$ of \textit{its} VB approximation. Since $KLD(\tilde{f}_{\phi}||f_{\phi})<KLD(\tilde{f}_{\theta}||f_{\theta})$,
it follows that $KLD(\tilde{f}_{\theta}^{TVB}||f_{\theta})<KLD(\tilde{f}_{\theta}||f_{\theta})$.

Obviously, VB yields an accurate representation, i.e. $KLD(\tilde{f}_{\theta}||f_{\theta})=0$,
iff $\theta_{i}$, $\theta_{\backslash i}$ are already independent
in $f_{\theta}$. An illustrative context is the multivariate normal
distribution ${\cal N}_{\theta}(0,\Sigma)$. The $KLD$ of its VB
approximation is strictly greater than zero when $\Sigma\neq D$ (diagonal).
Let us consider a rotation operator via eigenvectors of the covariance
matrix, i.e. $\phi=Q^{-1}\theta$, where $\Sigma^{-1}=Q\Lambda Q^{-1}$
and $\Lambda$ is the diagonal eigenvalue matrix. The distribution
of transformed variables $\phi$ is the independent multivariate normal
$N_{\phi}\left(0,\Lambda\right)$, and the VB approximation in this
transformed metric has $KLD(\tilde{f}_{\phi}||f_{\phi})=0$.

\begin{figure}
\begin{centering}
\includegraphics[width=0.8\columnwidth]{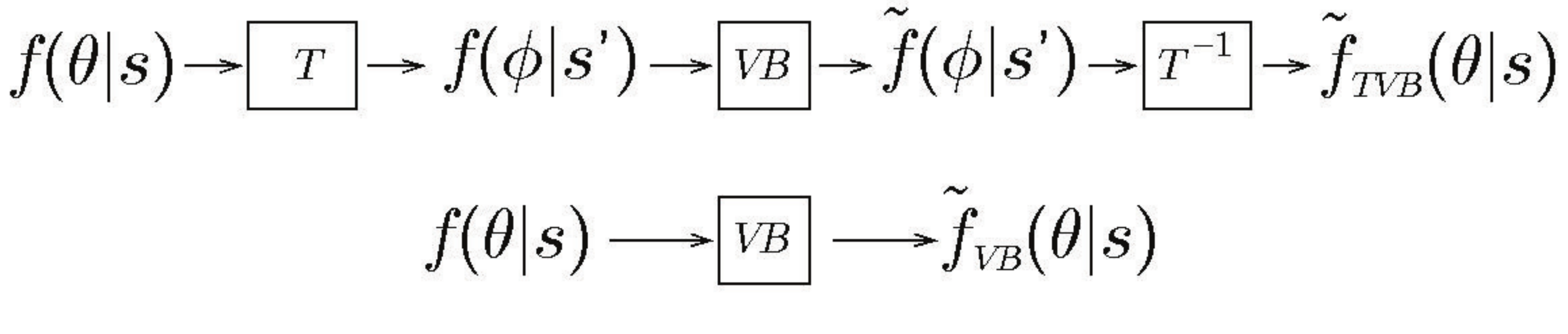} 
\par\end{centering}
\caption{\label{fig:ch7:TVB}Transformed VB (top) and VB (below) approximations.
$T$ denotes a bijective transformation, $T:\protect\vtheta\rightarrow T(\protect\vtheta)=\phi$.}
\end{figure}
\begin{sidewaysfigure}
\begin{centering}
\includegraphics[width=1\textwidth]{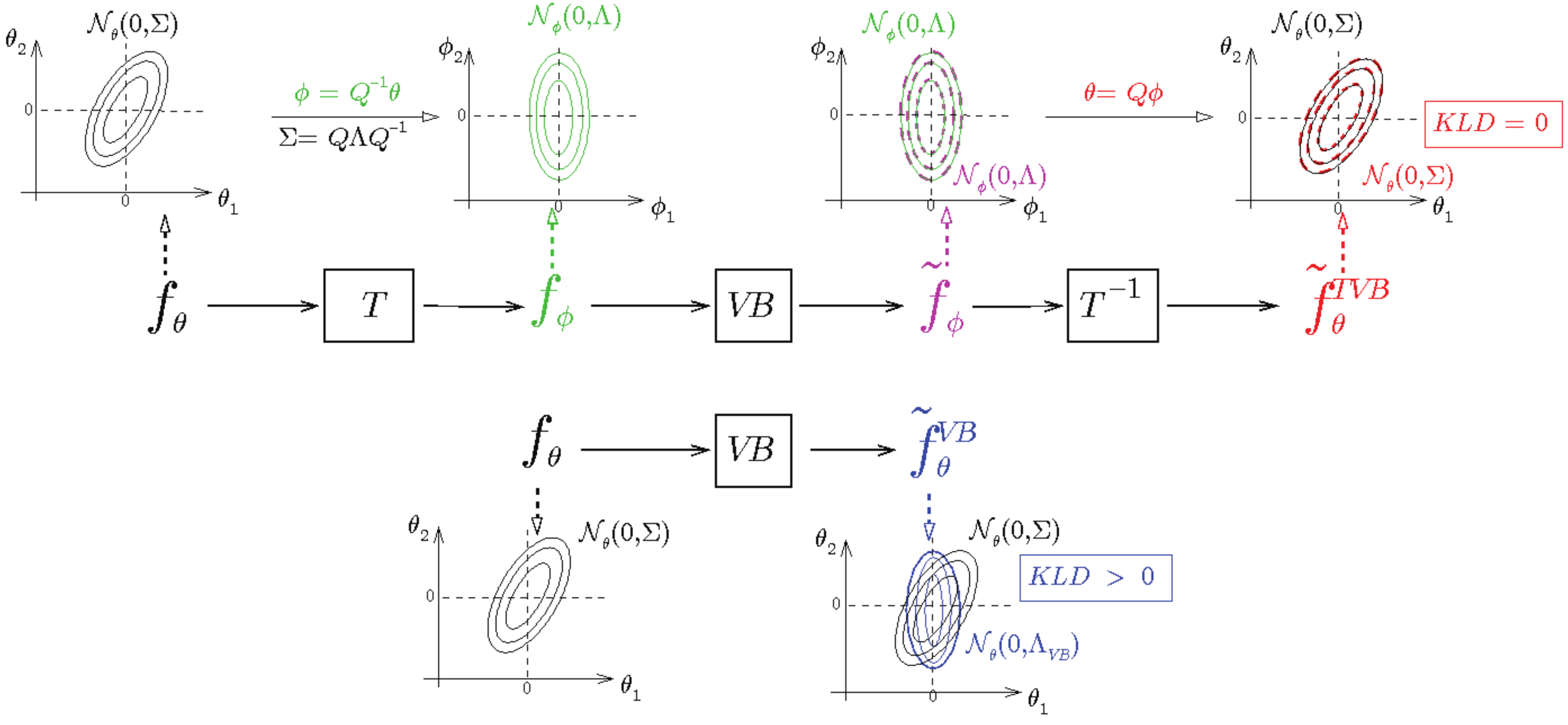}
\par\end{centering}
\caption{\label{fig:ch7:TVB_Normal}Transformed VB (top) with orthogonalization
and VB (below) approximations for the multivariate normal distribution.}

\end{sidewaysfigure}

\section{Transformed VB approximation \label{sec:chap7:Transformed-VB-approximation}}

In the literature, Cox-Reid orthogonalization {[}\citet{ch7:origin:Cox_Reid:87}{]}
has been applied broadly to parameter estimation to achieve robustness,
yet very few paper consider this approach for Bayesian inference.
Based on the Fisher information matrix, its approach is to decouple
a joint distribution up to second order. A slightly more general transformation
will be proposed in this section in order to increase the quality
of the VB approximation. 

\subsection{Distributional transformation}

Let us consider a distribution $f(\vtheta)$ of continuous r.v. $\vtheta$.
Then, given a bijective mapping $\vphi=\fung{\vtheta}$, we can derive
the distribution $f(\vphi)$ of continuous r.v. $\vphi$, as follows
{[}\citet{ch7:bk:Transform:theorem63,ch7:bk:transform:theorem09}{]}:

\[
f(\vphi)=\left.\frac{f(\vtheta)}{J(\vtheta)}\right|_{\vtheta=g^{-1}(\vphi)}
\]
where $J(\vtheta)\TRIANGLEQ\left|\det\left(\frac{d\vphi}{d\vtheta}\right)\right|$
is the Jacobian determinant of the transformation, $\vphi=\fung{\vtheta}$,
and $\left|\cdot\right|$ denotes magnitude. 

\subsection{Locally diagonal Hessian \label{subsec:chap7:Locally-diagonal-Hessian}}

As presented in Section \ref{subsec:chap4:Asymptotic-inference},
let us consider the negative logarithm of the transformed distribution
$L(\phi)\TRIANGLEQ\log f(\vphi)$, expanded up to the second order
of the Taylor approximation at a point $\vphi_{0}$ at which $L(\phi)$
is infinitely differentiable, as follows:
\[
L(\phi)=L(\phi_{0})+(\phi-\phi_{0})'\nabla L(\phi_{0})-\frac{1}{2}(\phi-\phi_{0})'H(\phi_{0})(\phi-\phi_{0})+\ldots
\]
where $\nabla L(\phi_{0})$ and $H(\phi_{0})\TRIANGLEQ-\nabla^{2}L(\phi_{0})$
are gradient vector and Hessian matrix, respectively, evaluated at
$\phi_{0}$. In contrast to the Fisher information matrix approach
in {[}\citet{ch7:origin:Cox_Reid:87}{]}, we propose to design the
transformation, $g(\cdot)$, in order to diagonalize the Hessian matrix.
Its insight is, actually, a quadratic decoupling up to second order,
which yields the asymptotic independence (see Proposition \ref{prop:ch4:posteriorCLT}). 

Such a transformation can be designed via matrix decomposition. The
transformed Hessian, $H(\phi)$, is desired to be diagonal locally
at $\vphi=\phi_{0}=g(\theta_{0})$, in which a specific value $\theta_{0}$
can be chosen freely, typically a certainty equivalent (CE) such as
mean, mode, etc. For this purpose, a linear transformation matrix
$A$ can be defined such that $g(\theta)=A\theta$, where $A$ is
an invertible matrix. The Jacobian of matrix transformation is also
feasible to compute in this case: $J(\vtheta)=\det(A)$, a constant.
We consider two designs for $A$, as follows: 
\begin{itemize}
\item \textbf{Method}\textit{$(I)$ - }\textbf{Eigen decomposition}\textit{:}
Let $A=Q^{-1}$, where $H(\theta_{0})=Q\Lambda Q^{-1}$ in the original
(untransformed) metric. Then, the Hessian in the transformed metric
is $H(\phi_{0})=\Lambda$, $\forall\vphi_{0}\in\Phi$, becomes diagonal
at $\phi_{0}$ and $J(\vtheta)=\det(Q^{-1})=1$.
\item \textbf{Method}\textit{$(II)$ - }\textbf{LDU decomposition}\textit{:}
Let $A=U$, where $H(\theta_{0})=LDU$, and $L=U'$ is a lower triangular
matrix with unit diagonal. Then, transformed Hessian $H(\phi_{0})=D$
is diagonal by design, $\forall\vphi_{0}\in\Phi$, and $J(\vtheta)=\det(U)=1$.
\end{itemize}
While the Eigen decomposition is easier to implement in practice,
the LDU decomposition has one advantage: the variable corresponding
to the last row of $A=U$ is kept unchanged. 

\section{Spherical distribution family \label{sec:chap7:Spherical-family}}

The spherical distribution refers to the family of distributions that
are closed under any diagonalization transformation {[}\citet{ch7:art:spherical:dist70}{]}.
Recently, in Bayesian analysis, it has been shown to be an observation
model, whose conjugate prior is the so-called dispersion elliptical
squared-radial (DESR) distribution - an extension of normal-gamma
family\textbf{ }{[}\citet{ch7:art:spherical:bayes06}{]}. The form
of spherical distribution is defined as 
\begin{equation}
f(\theta|\mu,\Sigma)\propto\left|\Sigma\right|^{-\frac{1}{2}}\psi(u(\theta))\label{eq:ch7:spherical_dist}
\end{equation}
where $\mu$ is the mean vector and $\Sigma$ is the covariance matrix
for $\vtheta$, $\psi$ is a function satisfying the normalizing condition
for $f(\theta|\mu,\Sigma)$ and 
\begin{equation}
u(\theta)=(\theta-\mu)'\Sigma^{-1}(\theta-\mu)\label{eq:chap8:u}
\end{equation}
is the quadratic form implied by $\mu$ and $\Sigma$.

Let us, once again, denote $L(\vtheta)\TRIANGLEQ-\log f(\vtheta)$.
Then, by the chain rule for the composite functions, its Hessian matrix
can be derived, as follows:

\begin{align}
H(\theta)=-\nabla^{2}L(\vtheta) & =-\frac{\partial^{2}L(\vtheta)}{\partial u^{2}}\nabla u(\theta)\nabla u\left(\theta\right)'-\frac{\partial L(\vtheta)}{\partial u}\nabla^{2}u(\theta)\nonumber \\
 & =-\frac{\partial^{2}L(\vtheta)}{\partial u^{2}}\Sigma^{-1}(\theta-\mu)(\theta-\mu)'\Sigma^{-1}-\frac{\partial L(\vtheta)}{\partial u}\Sigma^{-1}\label{eq:Hessian-spherical}
\end{align}

\subsection{Multivariate Normal distribution}

Since the first term in (\ref{eq:Hessian-spherical}) is zero at the
mean $\mu$, and $\Sigma^{-1}$ in the second term is diagonalizable,
$H(\phi_{0})$ can be locally diagonalized at the mean $\theta_{0}=\mu$
for any spherical distribution via the local diagonalization methods
\textit{(I)} and \textit{(II)}. In particular, $H(\vphi_{0})$ is
globally diagonal in the multivariate normal distribution, since,
for this choice of linear transformations, we have $\frac{\partial^{2}L(\vtheta)}{\partial u^{2}}=0$,
$\forall\vtheta\in\Theta$. This corresponds to our setting in Section
\ref{sec:chap7:Motivation}.

\subsection{Bivariate power exponential (PE) distribution}

Let us study another illustrative example for this spherical family
(\ref{eq:ch7:spherical_dist}), namely bivariate power exponential
(PE) distribution, defined as follows {[}\citet{ch7:art:PowerExp:bivariate98}{]}:

\begin{align}
f(\theta|\mu,\Sigma) & =\frac{\sqrt{2}}{\pi\sqrt{\pi}}\left|\Sigma\right|^{-\frac{1}{2}}\exp\left(-\frac{1}{2}u(\theta)^{2}\right)\label{eq:chap8:PE}
\end{align}
where $\vtheta\in\mathbb{R}^{2}$ and $u(\theta)$ is given in (\ref{eq:chap8:u}).
Because it is difficult to express two true marginals in closed form
{[}\citet{ch7:art:PowerExp:marginal:no_closedform}{]}, let us study
the VB approximation for bivariate PE distribution next. 

\subsubsection{VB approximation for bivariate PE}

Without loss of generality, let us assume that $\mu=0$ in the sequel.
Hence, we can write $f(\theta|\Sigma)\TRIANGLEQ f(\theta|\mu=0,\Sigma)$,
where $\Sigma\TRIANGLEQ\left[\begin{array}{cc}
\sigma_{1}^{2} & \rho\sigma_{1}\sigma_{2}\\
\rho\sigma_{1}\sigma_{2} & \sigma_{2}^{2}
\end{array}\right]$ and $\rho$ is the correlation coefficient. At cycle $\nu$, the
VB-marginals (Theorem \ref{thm:chap4:Iterative-VB-(IVB)}) for (\ref{eq:chap8:PE})
can be computed as follows: 

\begin{eqnarray*}
\tilde{f}^{[\nu]}(\theta_{i}) & \propto & \exp\left(E_{\tilde{f}^{[\nu-1]}(\theta_{\backslash i})}\log f(\theta|\Sigma)\right),\hspace{1em}i\in\{1,2\}\\
 & \propto & \exp\left(-\frac{1}{2}E_{\tilde{f}^{[\nu-1]}(\theta_{\backslash i})}\left(\theta'\Sigma^{-1}\theta\right)^{2}\right)\\
 & \propto & \exp\left(-\frac{1}{2\left(1-\rho^{2}\right)^{2}}\sum_{k=1}^{4}\alpha_{i,k}\theta_{i}^{k}\right)
\end{eqnarray*}
where eight VB shaping parameters $\alpha_{i,k}$ are:

\[
\begin{cases}
\alpha_{i,4} & =\frac{1}{\sigma_{i}^{4}}\\
\alpha_{i,3} & =-\frac{4\rho}{\sigma_{i}^{3}\sigma_{\backslash i}}\left[\widehat{\theta_{\backslash i}}\right]\\
\alpha_{i,2} & =\frac{2(1-\rho^{2})}{\sigma_{i}^{2}\sigma_{\backslash i}^{2}}\left[\widehat{\theta_{\backslash i}^{2}}\right]\\
\alpha_{i,1} & =-\frac{4\rho}{\sigma_{i}\sigma_{\backslash i}^{3}}\left[\widehat{\theta_{\backslash i}^{3}}\right]
\end{cases}
\]
with $\widehat{\theta_{\backslash i}}$, $\widehat{\theta_{\backslash i}^{2}}$
and $\widehat{\theta_{\backslash i}^{3}}$ denoting first, second
and third VB moments of $\vtheta_{\backslash i}$, i.e with respect
to $\tilde{f}^{[\nu-1]}(\theta_{\backslash i})$.

\subsubsection{TVB approximation for bivariate PE}

The transformed distribution of $\phi=A\theta$ are designed as $A=Q^{-1}$
, where $\Sigma^{-1}=Q\Lambda Q^{-1}$ for \textit{(I)} and $A=U$,
where $\Sigma^{-1}=LDU$ for \textit{(II).} 

For conciseness, only the case \textit{(I)} is presented below, with
similar findings for the case \textit{(II)}. The transformed distribution
of $\vphi=A\vtheta$ under \textit{(I)} is:

\begin{eqnarray*}
f(\phi|\Lambda) & \propto & \left|\Sigma\right|^{-\frac{1}{2}}\exp\left[-\frac{1}{2}\left(\vphi'\Lambda^{-1}\vphi\right)^{2}\right]
\end{eqnarray*}
The VB approximations are derived as follows:

\begin{eqnarray*}
\tilde{f}^{[\nu]}(\phi_{i}) & \propto & \exp\left(E_{\tilde{f}^{[\nu-1]}(\phi_{\backslash i})}\log f(\phi|\Lambda)\right),\ i\in\{1,2\}\\
 & \propto & \exp\left(-\frac{1}{2}E_{\tilde{f}^{[\nu-1]}(\theta_{\backslash i})}\left(\vphi'\Lambda\vphi\right)^{2}\right)\\
 & \propto & \exp\left(-\frac{1}{2}\left(\lambda_{i}\phi_{i}^{2}+\lambda_{\backslash i}\left[\widehat{\phi_{\backslash i}^{2}}\right]\right)^{2}\right)
\end{eqnarray*}
where $\lambda_{\itime}$,$\lambda_{\backslash i}$ are the two eigenvalues
of $\Sigma^{-1}$ (i.e. the diagonal elements of $\Lambda$) and $\widehat{\phi_{\backslash i}^{2}}$
is the second moment of $\vphi$ with respect to $\tilde{f}^{[\nu-1]}(\vphi_{\backslash i})$.
At convergence, the VB approximation for the transformed $f(\phi|\Lambda)$
distribution are:

\begin{eqnarray*}
\widetilde{f}(\phi|\Lambda) & = & \widetilde{f}(\phi_{1}|\Lambda)\widetilde{f}(\phi_{2}|\Lambda)\\
 & \propto & \exp\left(-\frac{1}{2}\left[\left(\lambda_{1}\phi_{1}^{2}+\lambda_{2}\left[\widehat{\phi_{2}^{2}}\right]\right)^{2}+\left(\lambda_{2}\phi_{2}^{2}+\lambda_{1}\left[\widehat{\phi_{1}^{2}}\right]\right)^{2}\right]\right)
\end{eqnarray*}

Substituting $\phi=Q^{-1}\theta$ into $\widetilde{f}(\phi|\Lambda)$,
we will retrieve the TVB approximation of $f(\vtheta|\Sigma)$, i.e.
$\tilde{f}_{TVB}(\vtheta|\Sigma)$: Also, the normalizing constant
can be evaluated numerically.

\subsubsection{Contour plot}

The results for $\mu=[2.5,1]^{T}$, $\sigma_{1}=0.5$, $\sigma_{2}=1.5$
and various correlation coefficients, $\rho\in(0,1)$, are shown in
Fig. \ref{fig:PE Contour} and Fig. \ref{fig:PE KLD}. In Fig. \ref{fig:PE KLD},
we can see that the $KLD$ of TVB approximation is equal to its minimum
value of original VB and becomes invariant with $\rho$ owing to diagonalization
methods. Although the relationship between \textit{(I)} and \textit{(II)}
is not expressed explicitly, the coincide $KLD$ values show that
they seem to be equivalent to each other. 

\begin{figure}
\begin{centering}
\includegraphics[width=0.8\columnwidth]{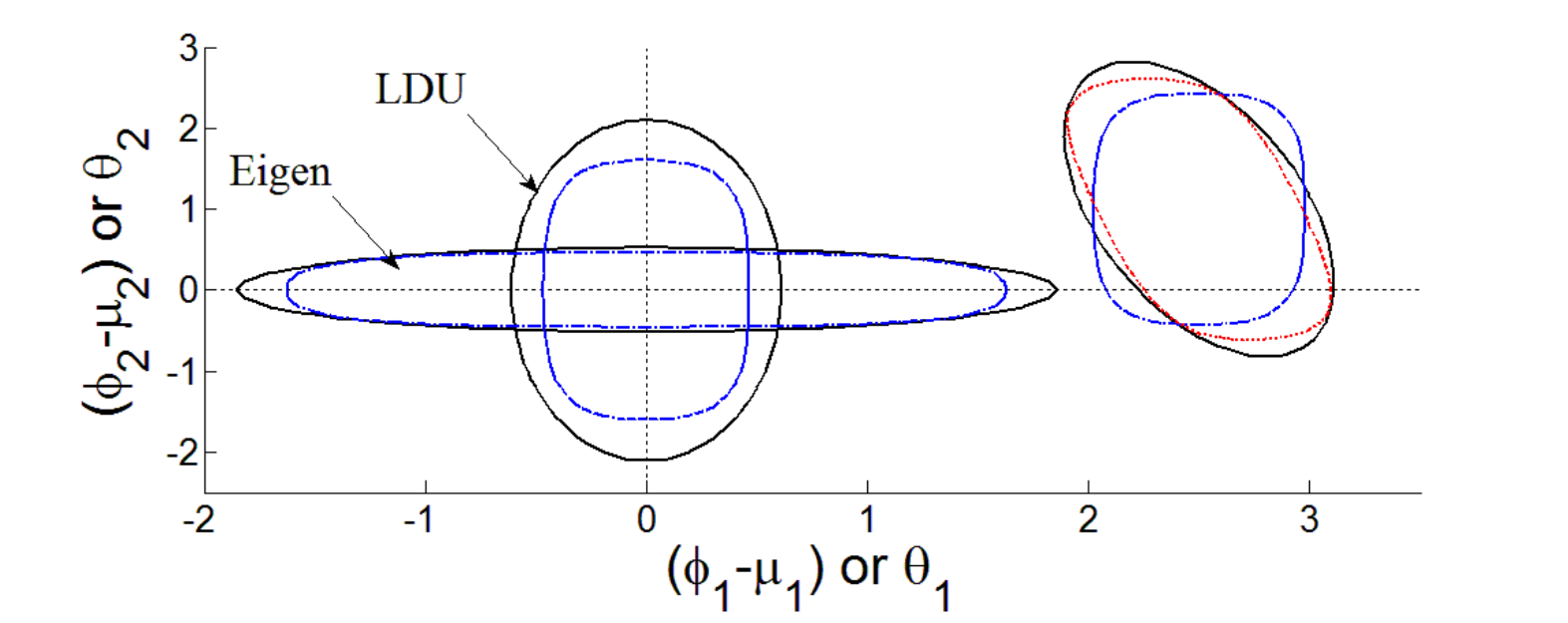} 
\par\end{centering}
\caption{\foreignlanguage{american}{\label{fig:PE Contour}Bivariate PE distribution and its approximations
($\rho=-0.5$). \foreignlanguage{english}{On the right: $f_{\theta},$
$\tilde{f}_{\theta}^{VB}$, $\tilde{f}_{\theta}^{TVB}$ are denoted
by $(-,\ -.\ ,\ :\ )$ respectively; on the left: $f_{\phi}$ and
$\tilde{f}_{\phi}$ are denoted by $(-,\ -.\ )$ respectively. For
clarity, the VB and TVB approximations are shifted to the origin in
that figure.}}}
\end{figure}
\begin{figure}
\begin{centering}
\includegraphics[width=0.6\columnwidth]{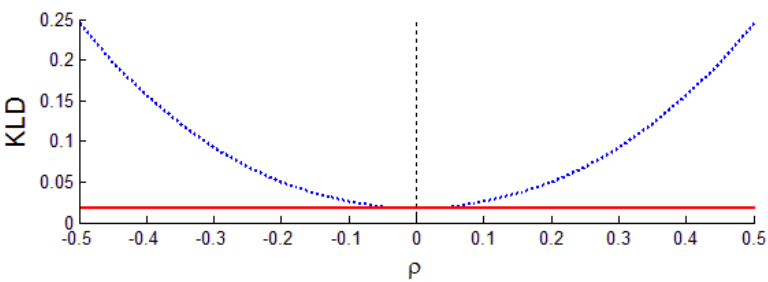} 
\par\end{centering}
\caption{\foreignlanguage{american}{\label{fig:PE KLD}KLD of VB (above, dotted line) and TVB for bivariate
PE distribution}}
\end{figure}

\section{Frequency inference in the single-tone sinusoidal model in AWGN \label{sec:chap7:Frequency-inference}}

In this section, the Bayesian inference for single-tone frequency
will be illustrated. Despite being simple, this canonical model in
digital receivers is non-linear with frequency and, hence, intractable
for frequency's posterior computation. The performance of two posterior's
approximations, VB and TVB, will also be compared and illustrated. 

Based on receiver's model in equation (\ref{eq:ch3:x_k:case2}), let
us consider a received sinusoidal sequence $\xBold_{\ndata}\TRIANGLEQ[x_{1},\ldots,x_{\ndata}]'$
over the AWGN channel:
\begin{equation}
x_{i}=a\sin(\Omega i)+\noise_{i},\ \iinn\label{eq:chap7:x_i}
\end{equation}
 where $\noise_{i}\overset{i.i.d}{\sim}\mathcal{N}(0,\rnoise)$. The
observed sequence (\ref{eq:chap7:x_i}) can be written in vector form,
as follows:

\begin{equation}
\xBold_{\ndata}=a\gBold_{\ndata}+\bnoise_{\ndata},\ \iinn\label{eq:chap7:x_n}
\end{equation}
where $\bnoise_{\ndata}\TRIANGLEQ[\noise_{1},\ldots,\noise_{\ndata}]'$
is the vector of noise samples, and $\gBold_{\ndata}\TRIANGLEQ[\sin(\OMEGA1),\ldots,\sin(\OMEGA\ndata)]'$
is the vector of regression function. The implied observation distribution
is: 
\begin{eqnarray*}
f(\xBold_{\ndata}|a,\Omega) & = & \prod_{i=1}^{\Ndata}\mathcal{N}_{x_{i}}\left(a\sin(\Omega i),\rnoise\right)\\
 & = & \mathcal{N}_{\xBold_{\ndata}}\left(a\gBold_{\ndata},\rnoise\Ibold\right)
\end{eqnarray*}
where $\Ibold$ is the identity matrix. For prior knowledge, $\Omega$
is chosen as uniform over $[0,\pi)$ (rad/sample) \textit{a priori}.
For amplitude $a$, the conjugate prior in this sinusoidal model is
$f(a)=\mathcal{N}_{a}\left(\mu_{a},r_{a}\right)$. This conjugacy
was noted in {[}\citet{ch2:art:sync_freq:MCMC_05}{]} and elsewhere.

\subsection{Joint posterior distribution}

The joint posterior for $\{a,\Omega\}$ can be derived as follows:

\begin{eqnarray}
f(a,\OMEGA|\xBold_{\Ndata}) & \propto & \prod_{i=1}^{\Ndata}\mathcal{N}_{x_{i}}\left(a\sin(\Omega i),\rnoise\right)\mathcal{N}_{a}\left(\mu_{a},r_{a}\right)\nonumber \\
 & \propto & \exp\left(-\frac{a^{2}\left(\frac{1}{r(\Omega)}\right)-2\mu\left(\Omega\right)\frac{1}{r(\Omega)}a}{2}\right)\label{eq:ch8:POSTERIOR:kernel}\\
 & \propto & \mathcal{N}_{a}\left(\mu\left(\Omega\right),r(\Omega)\right)\exp\left(\frac{\mu\left(\Omega\right)^{2}}{2r(\Omega)}\right)\sqrt{2\pi r(\Omega)}\label{eq:ch8:POSTERIOR:factor}
\end{eqnarray}
in which: 

\begin{eqnarray}
\frac{1}{r(\Omega)} & \TRIANGLEQ & \frac{\sum_{i=1}^{N}\sin^{2}(\Omega i)}{r_{e}}+\frac{1}{r_{a}}\label{eq:chap7:r_Omega}\\
 & = & \frac{\left\Vert \gBold_{\ndata}\right\Vert ^{2}}{r_{e}}+\frac{1}{r_{a}}\nonumber \\
\mu\left(\Omega\right) & \TRIANGLEQ & r(\Omega)\left(\frac{\sum_{i=1}^{N}x_{i}\sin(\Omega i)}{r_{e}}+\frac{\mu_{a}}{r_{a}}\right)\nonumber \\
 & = & \frac{\frac{X_{I}(e^{j\OMEGA})}{r_{e}}+\frac{\mu_{a}}{r_{a}}}{\frac{\left\Vert \gBold_{\ndata}\right\Vert ^{2}}{r_{e}}+\frac{1}{r_{a}}}\label{eq:chap7:mu_Omega}
\end{eqnarray}
expressed in term of $X_{I}(e^{j\OMEGA})\TRIANGLEQ Im\{X(e^{j\OMEGA})\}=\xBold_{\ndata}'\gBold_{\ndata}$
and in term of $\left\Vert \gBold_{\ndata}\right\Vert =\gBold_{\ndata}'\gBold_{\ndata}$
.

From (\ref{eq:ch8:POSTERIOR:factor}), we can see that Gaussian form
is preserved for $f(a|\Omega,\xBold_{\ndata})=\mathcal{N}_{a}(\mu\left(\Omega\right),$
$r(\Omega))$, owing to the conjugate prior $f(a)=\mathcal{N}_{a}\left(\mu_{a},r_{a}\right)$
defined above, and the marginal distribution for $\OMEGA$ can be
expressed in closed form as:

\begin{equation}
f(\Omega|\xBold_{\ndata})\propto\exp\left(\frac{\mu\left(\Omega\right)^{2}}{2r(\Omega)}\right)\sqrt{r(\Omega)}\label{eq:chap7:marginal_Omega}
\end{equation}
However, the distribution (\ref{eq:chap7:marginal_Omega}) is non-standard
and intractable in $\OMEGA$, owing to the nonlinear dependence on
$\Omega$, a result that is widely known {[}\citet{ch3:origin:Freq:ML74,ch7:PhD:aquinn92}{]}.
This sinusoidal model is therefore a canonical candidate for distributional
approximation. 

For later use, let us compute joint MAP estimate $\{\widehat{a},\widehat{\Omega}\}$.
From (\ref{eq:ch8:POSTERIOR:factor}), we can see feasibly that $\widehat{a}=\mu\left(\widehat{\Omega}\right)$
and, by substituting that $\widehat{a}$ to $\mathcal{N}_{a}\left(\mu\left(\Omega\right),r(\Omega)\right)$in
(\ref{eq:ch8:POSTERIOR:factor}), we have: 
\begin{eqnarray*}
\widehat{\Omega} & = & \arg\max_{\Omega}\exp\left(\frac{\mu\left(\Omega\right)^{2}}{2r(\Omega)}\right)\\
 & = & \arg\max_{\Omega}\frac{\left(\frac{X_{I}(e^{j\OMEGA})}{r_{e}}+\frac{\mu_{a}}{r_{a}}\right)^{2}}{2\left(\frac{\left\Vert \gBold_{\ndata}\right\Vert ^{2}}{r_{e}}+\frac{1}{r_{a}}\right)}
\end{eqnarray*}

\subsection{VB approximation}

From (\ref{eq:ch8:POSTERIOR:kernel}), the VB approximation (Theorem
\ref{thm:chap4:Iterative-VB-(IVB)}) for $f(a,\OMEGA|\xBold_{\Ndata})$
is (Fig. \ref{fig:ch7:TVB}):

\begin{eqnarray}
\ftilde_{VB}(a,\OMEGA|\xBold_{\Ndata}) & = & \ftilde_{VB}(a|\xBold_{\ndata})\ftilde_{VB}(\OMEGA|\xBold_{\ndata})\label{eq:chap8:VB}\\
 & \propto & \mathcal{N}_{a}\left(\mu_{1},\sigma_{1}^{2}\right)\exp\left(\frac{\alpha_{1}\mu\left(\Omega\right)+\alpha_{2}}{2r(\Omega)}\right)\nonumber 
\end{eqnarray}
in which the iterative VB's shaping parameters are:

\begin{eqnarray}
\mu_{1} & = & E_{\ftilde_{VB}(\OMEGA|\xBold_{\ndata})}\left[\mu\left(\Omega\right)\right]\label{eq:chap7:VB_shaping_para}\\
\sigma_{1}^{2} & = & E_{\ftilde_{VB}(\OMEGA|\xBold_{\ndata})}\left[r(\Omega)\right]\nonumber \\
\alpha_{1} & = & E_{\ftilde_{VB}(a|\xBold_{\ndata})}\left[2a\right]=2\mu_{1}\nonumber \\
\alpha_{2} & = & E_{\ftilde_{VB}(a|\xBold_{\ndata})}\left[-a^{2}\right]=-(\mu_{1}^{2}+\sigma_{1}^{2})\nonumber 
\end{eqnarray}
with $\ftilde_{VB}(\OMEGA|\xBold_{\ndata})\propto\exp\left(\frac{\alpha_{1}\mu\left(\Omega\right)+\alpha_{2}}{2r(\Omega)}\right)$
and $\ftilde_{VB}(a|\xBold_{\ndata})=\mathcal{N}_{a}\left(\mu_{1},\sigma_{1}^{2}\right)$. 

\subsection{TVB approximation}

The TVB approximation via LDU diagonalization (i.e. Method (II) in
Section \ref{subsec:chap7:Locally-diagonal-Hessian}) will be considered
next.

Because the Hessian matrix in this case is a symmetric $2\times2$
matrix, the upper off-diagonal element $u_{12}$ in matrix $U$ of
LDU decomposition for Hessian matrix is $u_{12}=H_{12}/H_{11}=r(\Omega)\left(\sum_{i=1}^{\Ndata}\frac{2i\cos(\Omega i)\left(x_{i}-2a\sin(\Omega i)\right)}{\rnoise}\right)$,
where $H_{11}\TRIANGLEQ-\frac{\partial^{2}\log f(a,\OMEGA|\xBold_{\Ndata})}{\partial a^{2}}$
and $H_{12}\TRIANGLEQ-\frac{\partial^{2}\log f(a,\OMEGA|\xBold_{\Ndata})}{\partial a\partial\Omega}$.
The transformed variable $[\lambda,\Omega]'=U[a,\Omega]'$ in this
case is: 
\begin{equation}
\lambda=a+\widehat{u}_{12}\Omega\label{eq:chap8:LDU_transform}
\end{equation}
where $\widehat{u}_{12}$ denotes $u_{12}$ evaluated at joint MAP
estimate $\{\widehat{a},\widehat{\Omega}\}$. 

Note that, the Jacobian in the LDU transformation is always unity,
i.e. $\det(U)=1$, as explained in Section \ref{subsec:chap7:Locally-diagonal-Hessian}.
Then, by changing $a$ to $\lambda$, the transformed distribution
is $f(\lambda,\Omega|\xBold_{n})=\left.f(a,\Omega|\xBold_{n})\right|_{a=\lambda-\widehat{u}_{12}\Omega}$
and the VB approximation (Theorem \ref{thm:chap4:Iterative-VB-(IVB)})
for $f(\lambda,\Omega|\xBold_{n})$ is (Fig. \ref{fig:ch7:TVB}):

\begin{eqnarray}
\ftilde(\lambda,\OMEGA|\xBold_{\Ndata}) & = & \ftilde(\lambda|\xBold_{\ndata})\ftilde(\OMEGA|\xBold_{\ndata})\label{eq:chap8:VB_for_transform}\\
 & \propto & \mathcal{N}_{\lambda}\left(\mu_{2},\sigma_{2}^{2}\right)\exp\left(\frac{\beta_{1}\mu_{0}\left(\Omega\right)+\beta_{2}}{2r(\Omega)}\right)\exp\left(\frac{\mu^{2}\left(\Omega\right)-\mu_{0}^{2}\left(\Omega\right)}{2r(\Omega)}\right)\nonumber 
\end{eqnarray}
in which: 
\[
\mu_{0}\left(\Omega\right)\TRIANGLEQ\mu\left(\Omega\right)+u_{12}\Omega
\]
and the iterative VB's shaping parameters are:

\begin{eqnarray}
\mu_{2} & = & E_{\ftilde(\OMEGA|\xBold_{\ndata})}\left[\mu\left(\Omega\right)\right]\label{eq:chap7:TVB_shaping_para}\\
\sigma_{2}^{2} & = & E_{\ftilde(\OMEGA|\xBold_{\ndata})}\left[r(\Omega)\right]\nonumber \\
\beta_{1} & = & E_{\ftilde(\lambda|\xBold_{\ndata})}\left[2\lambda\right]=2\mu_{2}\nonumber \\
\beta_{2} & = & E_{\ftilde(\lambda|\xBold_{\ndata})}\left[-\lambda^{2}\right]=-(\mu_{2}^{2}+\sigma_{2}^{2})\nonumber 
\end{eqnarray}
with $\ftilde(\OMEGA|\xBold_{\ndata})\propto\exp\left(\frac{\beta_{1}+\beta_{2}\mu_{0}\left(\Omega\right)}{2r(\Omega)}+\frac{\mu^{2}\left(\Omega\right)-\mu_{0}^{2}\left(\Omega\right)}{2r(\Omega)}\right)$
and $\ftilde(\lambda|\xBold_{\ndata})=\mathcal{N}_{\lambda}\left(\mu_{2},\sigma_{2}^{2}\right)$.

Note that, the Jacobian in inverse LDU transformation is, once again,
unity. By applying the inverse transformation, i.e. $a=\lambda-\widehat{u}_{12}\Omega$
(from (\ref{eq:chap8:LDU_transform})), the TVB approximation for
$f(a,\OMEGA|\xBold_{\Ndata})$ is $\left.\ftilde(\lambda,\OMEGA|\xBold_{\Ndata})\right|_{\lambda=a+\widehat{u}_{12}\Omega}$,
i.e. (Fig. \ref{fig:ch7:TVB}):

\begin{eqnarray}
\ftilde_{TVB}(a,\OMEGA|\xBold_{\Ndata}) & = & \ftilde_{TVB}(a|\Omega,\xBold_{\ndata})\ftilde_{TVB}(\OMEGA|\xBold_{\ndata})\label{eq:chap8:TVB}\\
 & = & \mathcal{N}_{a}\left(\mu_{2}-\widehat{u}_{12}\Omega,\sigma_{2}^{2}\right)\ftilde(\OMEGA|\xBold_{\ndata})\nonumber 
\end{eqnarray}
where $\ftilde_{TVB}(\OMEGA|\xBold_{\ndata})=\ftilde(\OMEGA|\xBold_{\ndata})$,
defined in (\ref{eq:chap8:VB_for_transform}). 

Comparing TVB approximations (\ref{eq:chap8:TVB}) with VB (\ref{eq:chap8:VB})
, we can see that the two factors in TVB are not independent anymore
like those in VB, owing to linearization (\ref{eq:chap8:LDU_transform})
at specific point $\{a,\Omega\}$ in coefficient $u_{12}$. This result
shows that TVB is, in this case, a non-naive mean field approximation. 
\begin{rem}
\label{Remark:chap7:As-a-remark}Note the similarity between $\ftilde_{VB}(\OMEGA|\xBold_{\ndata})$
in (\ref{eq:chap8:VB}) and $\ftilde_{TVB}(\OMEGA|\xBold_{\ndata})=\ftilde(\OMEGA|\xBold_{\ndata})$
in (\ref{eq:chap8:VB_for_transform},\ref{eq:chap8:TVB}), whose key
difference is the extra factor $\exp\left(\frac{\mu^{2}\left(\Omega\right)-\mu_{0}^{2}\left(\Omega\right)}{2r(\Omega)}\right)$.
Because this extra factor involves the second order $\mu^{2}\left(\Omega\right)$
like the true marginal $f(\OMEGA|\xBold_{\ndata})$ in (\ref{eq:ch8:POSTERIOR:factor}),
the frequency estimation of TVB is expected to be better than that
of VB, which only takes into account the first order $\mu\left(\Omega\right)$.
\end{rem}

\subsection{Simulation}

All iterative shaping parameters, $\mu$, $\sigma$, $\alpha$, $\beta$
in (\ref{eq:chap7:VB_shaping_para}, \ref{eq:chap7:TVB_shaping_para})
are evaluated numerically at DFT-bins of $\Omega$. The performance
of $\Omega$ estimators for various schemes is shown in Fig. \ref{fig:sinewave}.
Here, $\ndata=1024$ and $\Omega=1.1$ DFT-bins, i.e. $\Omega=1.1\frac{2\pi}{\ndata}$
rad/sample. The value 1.1 was chosen so that the true digital frequency
$\Omega$ is always off-bin, no matter how high the resolution of
the DFT-bin quantization is. Also, the fact that $\Omega$ is close
to one represents a stressful regime for the $\OMEGA$ estimator {[}\citet{ch3:origin:Freq:ML74,ch7:PhD:aquinn92}{]}.
The SNR is $SNR=\frac{\muamplitude^{2}+\ramplitude}{2\rnoise}$, where
the prior parameters were chosen as $\muamplitude=1$ and $\ramplitude=0.1$,
representing small variance of normalized attenuation. The number
of Monte Carlo runs is $10^{6}$. 

\subsubsection{Performance of frequency estimates \label{subsec:chap7:Performance-of-frequency}}

Because the loss function is the root mean square (RMS) error, the
posterior mean $\widehat{\Omega}_{MEAN}$ is the minimum-risk estimator
(\ref{eq:ch4:Bayesian_Estimator}), as verified in Fig. \ref{fig:sinewave}.
Owing to prior information, the joint MAP $\widehat{\Omega}_{MAP}$
estimator is slightly better than joint ML $\widehat{\OMEGA}_{ML}$
. However, because both of them can only detect the frequency at DFT-bins,
they are much worse than $\widehat{\Omega}_{MEAN}$ in this off-bin
case. 

For illustration of the VB schemes. both VB $\widehat{\OMEGA}_{VB}$
and TVB $\widehat{\OMEGA}_{TVB}$ estimators are chosen as the mean
of $\ftilde_{VB}(\OMEGA|\xBold_{\ndata})$ and $\ftilde_{TVB}(\OMEGA|\xBold_{\ndata})$,
respectively. Because the joint distribution of $\OMEGA$ concentrates
around single point at high SNR, the performance of a naive mean field
approximation like VB does increase but is still not good in this
case. In contrast, the performance of the TVB estimator is much higher
and close to joint MAP performance, owing to its linearization around
$\widehat{\Omega}_{MAP}$. 

\subsubsection{Evaluation of computational load via FFT \label{subsec:chap7:FFT}}

Since FFT algorithm was applied to all schemes, their computational
load was of the same efficient order $O(\ndata\log\ndata)$. Note
that, owing to the asymptotic nature of $O(\cdot)$, we do not consider
the difference of a factor $O(1)$ in the computational load of the
schemes above. In practice, iterative schemes such as VB and TVB may
be $\nu$-times slower than a standard FFT-based scheme such as ML,
where $\nu$ is the number of iterations at convergence. The number
of iteration $\nu$ for VB and TVB was fixed at $5$, since no increase
in performance was visible for higher $\nu$. 

Owing to FFT, the Bayesian posterior inference $\widehat{\Omega}_{MEAN}$
yields a good trade-off scheme, i.e. small loss in speed but high
gain in performance, compared to ML estimator via periodogram. However,
the iterative VB and TVB methods, although being $\nu$ times slower
than Bayesian scheme, do not yield better performance than $\widehat{\Omega}_{ML}.$
or $\widehat{\Omega}_{MAP}.$ Hence, although they cannot be recommended
for this single-tone problem, the simulation has verified the superiority
of TVB to VB method, which is the main motivation of designing TVB.

\begin{figure}
\begin{centering}
\includegraphics[width=0.8\columnwidth]{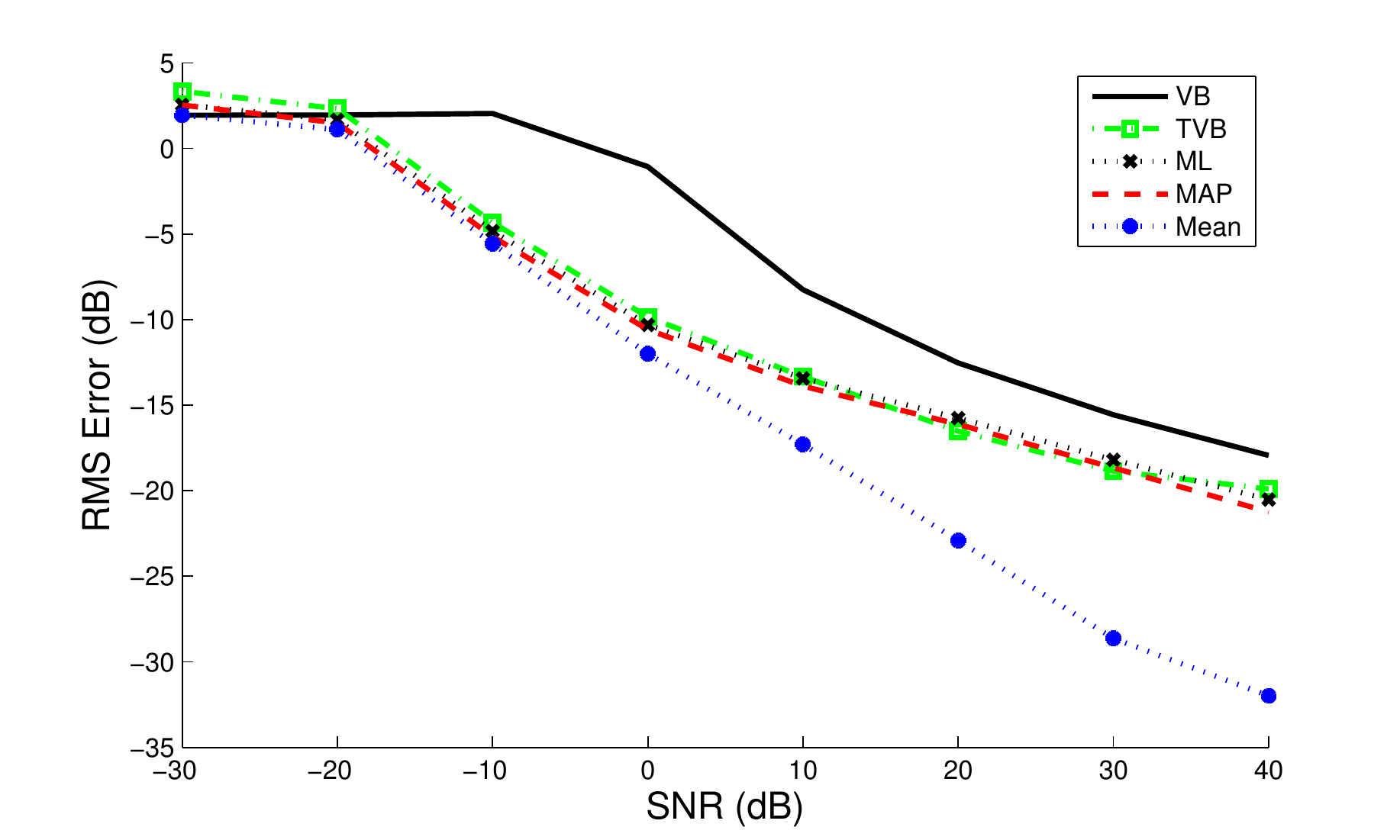}
\par\end{centering}
\caption{\foreignlanguage{american}{\label{fig:sinewave}\foreignlanguage{english}{ Root mean square (RMS)
error in estimation of $\Omega=1.1$ DFT-bins for the single-tone
sinusoidal AWGN, via conventional methods and via VB variants}}}
\end{figure}

\section{Summary}

This chapter began with an illustrative example, which showed that
rotation can render independent (i.e.$\ $decouple) bivariate Gaussian
random variables and, hence, reduce the Kullback-Leibler divergence
(KLD) in the VB approximation to zero. This idea encouraged us to
design the rotation via the transformed Hessian matrix in any posterior
distribution at a desired CE point, with the VB approximation then
being performed in the transformed metric. This novel scheme was defined
as transformed VB (TVB) approximation. Thanks to the asymptotic independence
property of the transformed posterior distribution, as previously
reviewed in Section \ref{subsec:chap4:Asymptotic-inference}, this
rotation scheme achieved a local decoupling (independence), and, hence,
the anticipated reduction in KLD via TVB. The spherical distribution
family, which is a generalized form of Gaussian distribution, is closed
under rotation and, hence, was used as an illustrative example for
TVB scheme. 

Finally, TVB was applied to the frequency carrier offset estimation
problem, a simplified context for frequency synchronization in the
AWGN channel, as explained in Chapter \ref{=00005BChapter 3=00005D}.
As expected, the accuracy of TVB was shown to be much better than
that of VB in simulation. This improvement will motivate further research,
which will be discussed in Chapter \ref{=00005BChapter 9=00005D}.

%auto-ignore
%auto-ignore
%%%% Common

\global\long\def\REAL{\mathbb{R}}%

\global\long\def\DEAL{\mathbb{D}}%

\global\long\def\COMPLEX{\mathbb{C}}%

\global\long\def\RING{\mathcal{R}}%

\global\long\def\MSET{\mathcal{M}}%

\global\long\def\ASET{\mathcal{A}}%

\global\long\def\OMEGA{\Omega}%

\global\long\def\calO{{\cal O}}%

\global\long\def\calX{\mathcal{X}}%

\global\long\def\ndata{n}%

\global\long\def\nstate{M}%

\global\long\def\itime{i}%

\global\long\def\istate{k}%

\global\long\def\funh#1{h\left(#1\right)}%

\global\long\def\fung#1{g\left(#1\right)}%

\global\long\def\seti#1#2{#1=1,2,\ldots,#2}%

\global\long\def\setd#1#2{\{#1{}_{1},#1{}_{2},\ldots,#1_{#2}\}}%

\global\long\def\TRIANGLEQ{\triangleq}%

%%%% Chapter 8

\global\long\def\fBold{\boldsymbol{f}}%

\global\long\def\TBold{\mathbf{T}}%

\global\long\def\LBold{L}%

\global\long\def\fprofile{f_{p}}%

\global\long\def\fdelta{f_{\delta}}%

\global\long\def\ftdelta{\widetilde{f}_{\delta}}%

\global\long\def\fdbar{\overline{f_{\delta}}}%

\global\long\def\WBold{\mathbf{W}}%

\global\long\def\XBold{\mathbf{X}}%

\global\long\def\ftilde{\widetilde{f}}%

\global\long\def\fnutilde#1{\widetilde{f}^{[#1]}}%

\global\long\def\fctilde{\widetilde{f}_{c}}%

\global\long\def\fbar{\overline{f}}%

\global\long\def\fhat{\widehat{f}}%

\global\long\def\fpbar{\overline{\fprofile}}%

\global\long\def\xBold{\mathbf{x}}%

\global\long\def\sBold{\mathbf{s}}%

\global\long\def\oneBold{\boldsymbol{1}}%

\global\long\def\KLDVB{KLD_{VB}}%

\global\long\def\KLDFCVB{KLD_{FCVB}}%

\global\long\def\Fc{\mathbb{F}_{c}}%

\global\long\def\Ffc{\mathbb{F}_{f.c}}%

\global\long\def\Normal{\mathcal{N}}%

\global\long\def\Oc{\mathcal{O}}%

\global\long\def\Rc{\mathbb{R}}%

\global\long\def\Psi{\psi}%

\global\long\def\LAMBDA{\boldsymbol{\Lambda}}%

\global\long\def\lAMBDA{\boldsymbol{\lambda}}%

\global\long\def\khat{\widehat{k}}%

\global\long\def\lhat{\widehat{l}}%

\global\long\def\lVA{\widehat{l}^{(VA)}}%

\global\long\def\Lhat{\widehat{L}}%

\global\long\def\LVA{\widehat{L}^{(VA)}}%

\global\long\def\tradVBML{tradVB_{(ML)}}%

\global\long\def\tradFCVBML{tradFCVB_{(ML)}}%

\global\long\def\VBML{VB_{(ML)}}%

\global\long\def\FCVBML{FCVB_{(ML)}}%

\global\long\def\Vst{VB_{(ML)}^{1st}}%

\global\long\def\Fst{FCVB_{(ML)}^{1st}}%

\global\long\def\Mu{Mu}%

\global\long\def\iton{i=1,\ldots,n}%

\global\long\def\iinn{i=\{1,\ldots,n\}}%

\global\long\def\gbar{\bar{g}}%

\global\long\def\fDoppler{f_{D}}%

\global\long\def\Tsample{T_{s}}%

\chapter{Performance evaluation of VB variants for digital detection \label{=00005BChapter 8=00005D}}

In Section \ref{sec:chap7:Frequency-inference}, we have considered
one of three basic digital receivers of Chapter \ref{=00005BChapter 3=00005D},
namely pilot-based unsynchronized frequency receiver. In this chapter,
we will consider the other two receivers, namely digital detections
in AWGN and quantized Rayleigh fading channels.

Firstly, a toy problem will be investigated. In the communication
context, we consider a Markov source transmitted over an AWGN channel,
with known parameters. The performance will be studied in two scenarios:
a fixed number, $M$, of state and fixed number, $\ndata$, of computational
time-resource.

Secondly, an appropriate model for practical Rayleigh fading channel
will be studied. Since the bivariate Rayleigh distribution is a complicated
function, it is often quantized to yield a closed-form Markov channel
{[}\citet{ch3:ART:FadingMarkov:tutorial08}{]}. By way of correlation
coefficient in the bivariate Rayleigh distribution, we can investigate
the influence of correlation via transition matrix of Markov channel
over the performance of inference methods.

In this chapter, we will apply the FB, VA, VB and FCVB algorithms
in Chapter \ref{=00005BChapter 6=00005D} to those two Markovian digital
detection scenarios. The simulation evidence will be provided and
illustrate the trade-offs between performance and computational load.
Also, in order to avoid the ambiguity, $\log(0)=-\infty$, and to
protect the convention $0\log0=0$, we assign $\log(0)=-10^{10}$
in the simulations. 

\section{Markov source transmitted over AWGN channel \label{sec:chap8:Markov-source-AWGN}}

Using the receiver model in equation (\ref{eq:ch3:x_k:case1}), let
us consider an AWGN channel with the classical Wold decomposition:

\[
x_{i}=s_{i}+e_{i},\ i=1,\ldots n
\]
where $x_{i}$ are the complex, observed noisy signal (data) samples;
$e_{i}$ is a realization of complex AWGN with variance per dimension
$N_{0}/2$ (Watts), and $N_{0}$ (Watts per radian/samples) is the
power spectral density (PSD). The source symbol $s_{i}$ is $i$th
realization of an $M$ states homogeneous Markov chain. Each state
is then mapped to a constellation point in a rectangular Gray-code
$M$-QAM. 

Because the observation model (\ref{eq:ch3:f(x|s):case1}) in this
case corresponds to the i.d. observation model in (\ref{eq:ch6:Observation})
and the distribution for HMC sequence $s_{i}$ corresponds to HMC
prior model in (\ref{eq:ch6:PRIOR}), this toy model is recognized
as a time-homogeneous HMC (\ref{eq:Joint}) with $M$ known Gaussian
components. 

In simulation, each element of $M\times M$ transition matrix $\TBold$
for the source was generated as an iid realization from uniform distributions
$U(0,1)$, and then, the columns are normalized to satisfy the stochastic
matrix constraint. The length-$\nstate$ vector with uniform elements,
$1/M$, was chosen as the initial probability vector $p$ of the HMC.
For the channel, $M$-QAM constellation point represents a block of
$\log_{2}M$ source bits. Therefore, although our label estimates
return a symbol-error-rate (SER), which is a Hamming distance on the
sequence of symbols, SER is often converted to bit-error-rate (BER)
in the literature {[}\citet{ch2:bk:Richardson}{]}. In this section,
BER is confined to a range much lower than $10^{-1}$, reflecting
the requirement in practice {[}\citet{ch2:bk:SEP:Haykin06}{]}. The
amplitudes, $a_{i}$, of the constellation points are also normalized,
such that the average Energy per bit ($E_{b}$) is unity, i.e. $E_{b}=\frac{\sum_{i=1}^{M}a_{i}^{2}/M}{\log_{2}M}=1$,
c.f. {[}\citet{ch2:bk:ToddMoon}{]}. By this way, we can regard the
average SNR per bit, $SNR_{b}\TRIANGLEQ E_{b}/N_{0}$, as an interpretation
of signal-to-noise ($SNR$) ratio.

\subsection{Initialization for VB and FCVB}

For initialization purposes, the initial shaping parameters $p_{k,i}^{[0]}$
(\ref{eq:VB=00003Dshaping}) and $\widehat{p}_{k,i}^{[0]}$ (\ref{eq:FCVB=00003Dshaping})
of multinomial distribution for VB and FCVB can be chosen either uniform,
with $p_{k,i}^{[0]}=\widehat{p}_{k,i}^{[0]}=1/M$, $k=1,\ldots,M$,
$\forall\iinn$ or ML-based scheme with $p_{i}^{[0]}=\widehat{p}_{i}^{[0]}\propto\Psi_{i}$,
$\iinn$, where $\psi_{i}$ is defined in (\ref{eq:ch6:DEF:psi}).
Since the ML estimate is fast, this initialization scheme does not
greatly affect the overall complexity of VB and FCVB, as summarized
in Table \ref{tab:ch6:ComputationalComplexity}. 

In all simulations, the converged performance of these two initializations
are similar. However, the ML-based initialization often reduces the
total number of IVB cycles, $\nu_{c}$, by $1$, i.e. one cycle of
IVB (for VB and FCVB) with uniform initialization has the effect of
ML-based initialization (for VB and FCVB). This means that the ML-based
initialization scheme is slightly faster in the simulations. Hence,
for clarity, only the curves of ML-based initialization are shown
in the figures.

For convention, the terms $\tradVBML$, $\tradFCVBML$ and $\VBML$,
$\FCVBML$ denote traditional and accelerated VB and FCVB, respectively,
with ML-based initialization. Since, in the first IVB cycle, the traditional
and accelerated methods are identical, let us call them $\Vst$ and
$\Fst$, respectively.

\subsection{Performance of HMC source estimates}

\begin{figure}
\begin{centering}
\includegraphics[width=1\columnwidth]{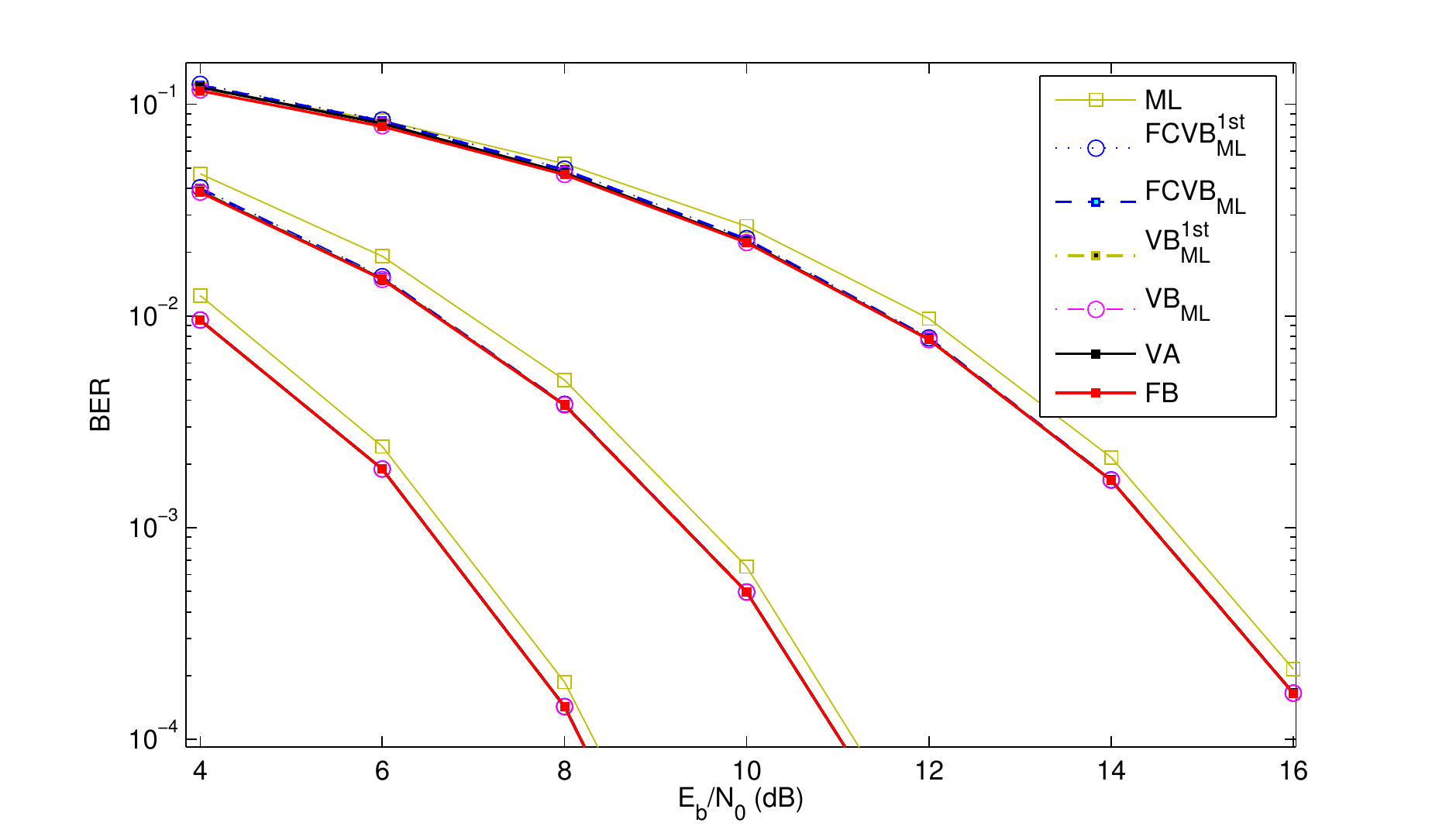}
\par\end{centering}
\centering{}\caption{\label{fig:M=00003D2-8-64}BER versus SNR per bit, $E_{b}/N_{0}$
(dB), for $2,8,64$-QAM (left,middle,right), with $10^{5}$ Monte
Carlo runs.}
\end{figure}
In Fig. \ref{fig:M=00003D2-8-64}, the BER performance is plotted
versus $E_{b}/N_{0}$, for the competing schemes in Chapter \ref{=00005BChapter 6=00005D}.
In all cases of $M\in\{2,8,64\}$, ML is the worst estimator, as anticipated
in Section \ref{subsec:chap6:Bayesian-risk-for-HMC}, while the performances
of all other algorithms are identical. Intuitively, this is an expected
result for VB approximation and its variant, FCVB, all of which seek
an approximating distribution in the class $\mathcal{F}_{c}$ and
$\mathcal{F}_{\delta}$, respectively, of independent distributions.
The reason is that the correlation coefficient, $\rho$, implied by
a transition matrix with uniform elements, is low.

\begin{figure}
\begin{centering}
\includegraphics[width=1\columnwidth]{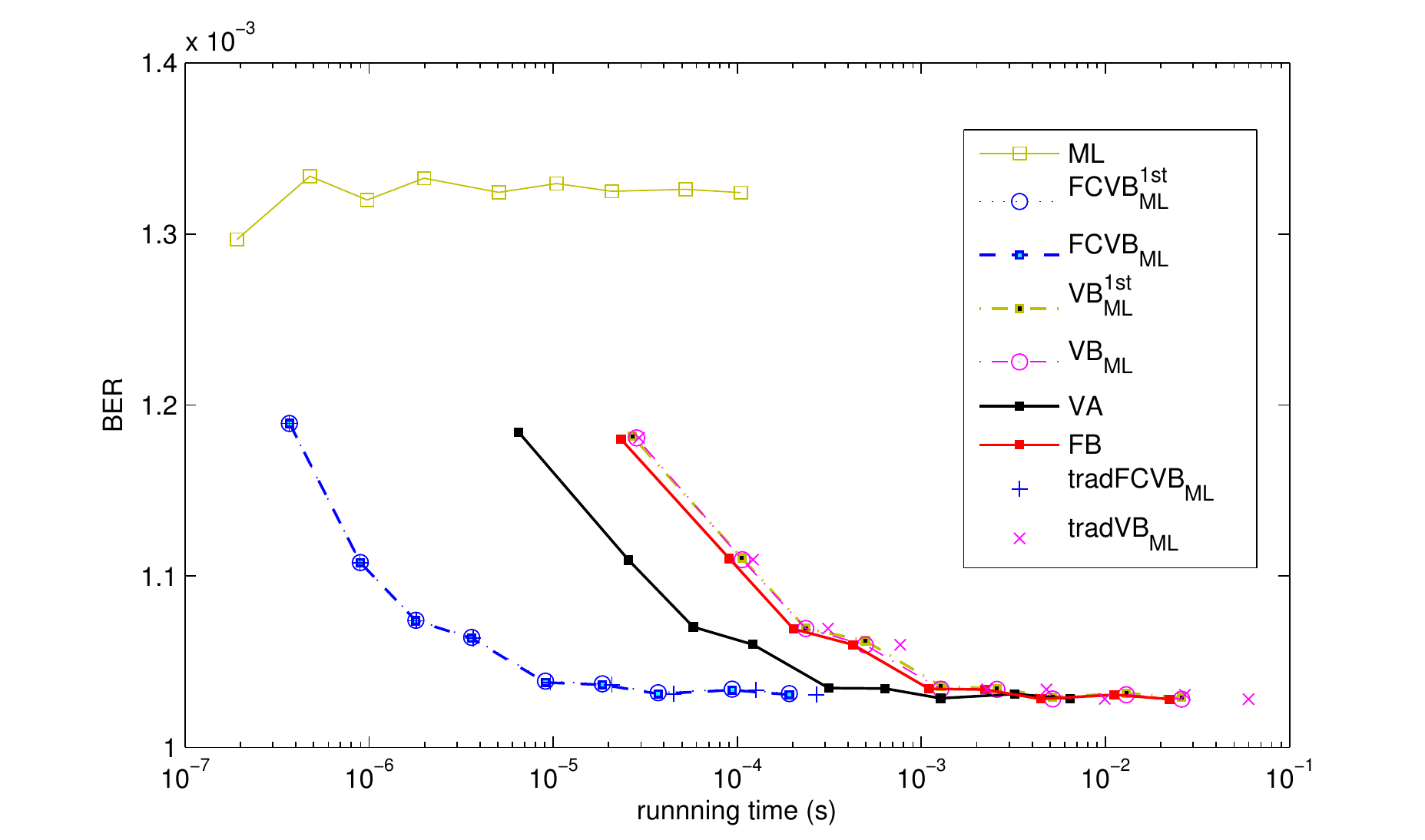}
\par\end{centering}
\centering{}\caption{\label{fig:TimeResource}BER versus running-time for 64-QAM, with
$10^{5}$ Monte Carlo runs. From left to right: $n=\{2,5,10,20,50,100,200,500,1000\}$.
The running-time is measured by C++ implementation and 3 GHz Core2Duo
Intel processor.}
\end{figure}
The effect on performance of a constrained running-time is illustrated
in Fig. \ref{fig:TimeResource}, in which we set $M=64$ , $E_{b}/N_{0}=14.5$
(dB), and we varied the number of data $n$ such that BER performance
of all methods are convergent to $10^{-3}$. In this scenario of high
SNR, the algorithms have almost identical performance, but with different
running-times. Hence, all curves in Fig. \ref{fig:TimeResource} appear
as x-shifted variants of each other. 

The results show that, given a fixed-time resource, we can run the
low complexity FCVB with more data than is possible for other algorithms.
The maximum gain in FCVB's performance over VA's is about $12$\%,
with a fixed time resource around $10$ microseconds. The simulation
results in Fig. \ref{fig:TimeResource}, for the case $\ndata=50$,
are also extracted in Fig. \ref{fig:ch8:Layman} in order to illustrate
the superiority of the Accelerated $\FCVBML$ to VA and FB methods. 

As explained in Section \ref{subsec:chap6:Bayesian-risk-for-HMC},
the FB algorithm is the most accurate method, since it returns the
sequence of marginal MAP estimates of the HMC labels, i.e. the exact
minimum risk (MR) estimate in this case (\ref{eq:ch4:MR:marginalMAP}).
The VA and $\FCVBML$, despite of not being MR risk estimators, return
the exact global and local MAP trajectory estimate for the HMC, respectively.
For the weakly correlated HMC model in Fig. \ref{fig:ch8:Layman},
we can see that these two methods are very close, to the extent that
they are visually identical to FB in performance.

\begin{figure}
\begin{centering}
\includegraphics[width=1\columnwidth]{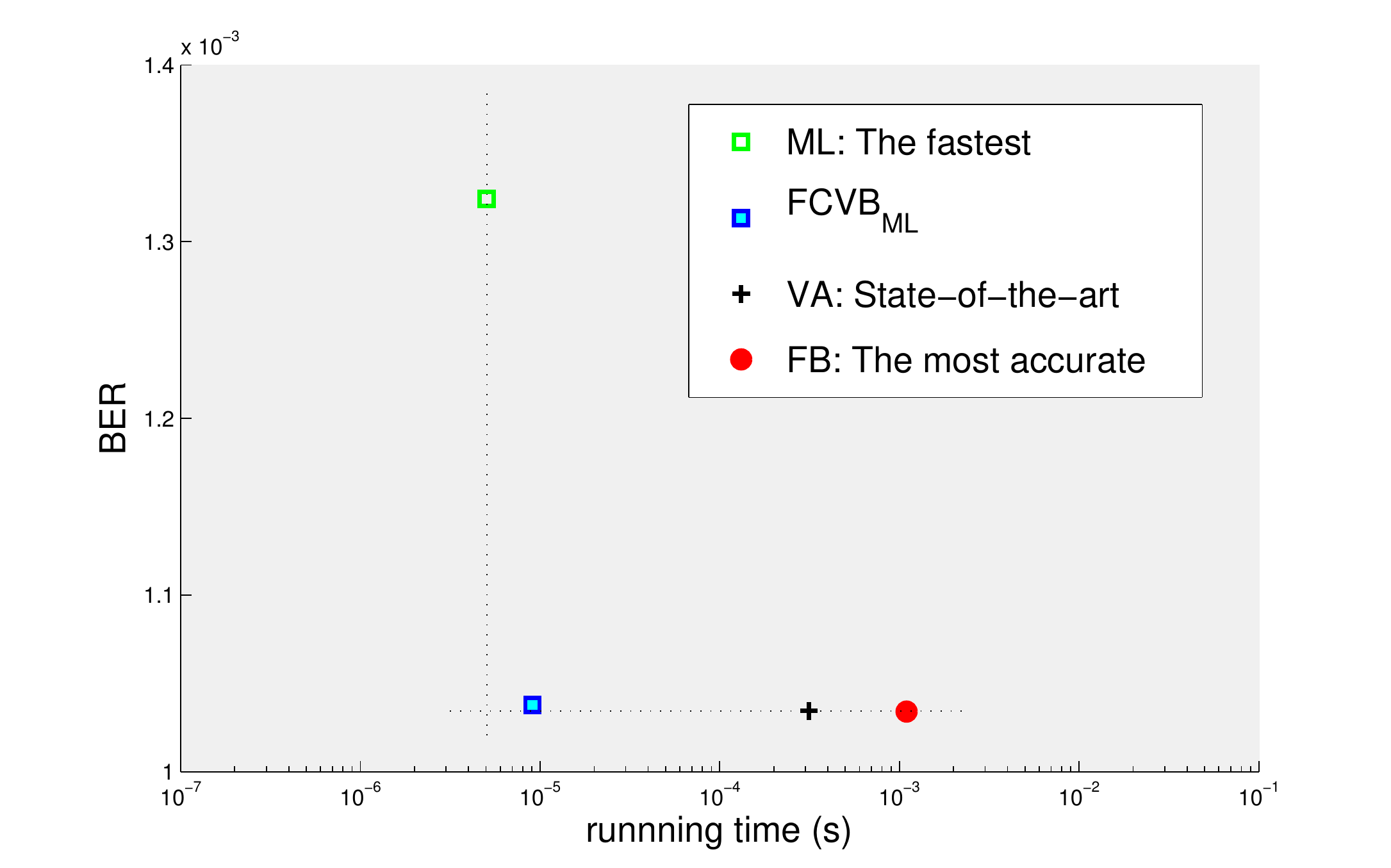}
\par\end{centering}
\centering{}\caption{\label{fig:ch8:Layman} BER versus running-time for 64-QAM (extracted
from Fig. \ref{fig:TimeResource}, with the case $n=50$).}
\end{figure}

\subsection{Computational load of HMC source estimates \label{subsec:chap8:Computational-cost-AWGN}}

For brevity, we only consider the accelerated scheme for VB-based
inference here. The comparison between traditional and accelerated
schemes will be studied in Section \ref{subsec:ch8:Empirical-speed-up}.

From Table \ref{tab:ch6:ComputationalComplexity}, we can roughly
estimate the cost of each algorithm via the number of equivalent operators.
by normalizing the cost of $ML$ as $O(1)$, as follows (from the
lowest to highest anticipated cost):
\begin{itemize}
\item For FCVB: By normalizing the cost of $ML$ as $O(1)$, the cost of
each FCVB cycle and VA should be $O(1)$ and $O(M)$, respectively.
Because $\FCVBML$ requires at least one IVB cycle $O(\eta\eIVB)$
and $O(1)$ $ML$-based initialization, in respect to ML's cost, the
total cost of $\FCVBML$ should be at least $O(\eta\eIVB+1)\geq O(1.75)$,
where $\eIVB\geq1$ and $\eta\in\left[0.75,1.25\right]$, as explained
in Section \ref{ch6:sub:bubble-sort-like}. 
\item For VA: the cost should be $O(M)$. 
\item For FB: because the cost of MUL is not deterministic, we consider
FB to be at least three times slower than VA, as is often noted in
the literature {[}\citet{ch8:BCJR_3x_VA}{]}. 
\item For VB: Each VB cycle is slightly slower than FB, hence the total
cost of $\VBML$ is at least $\eIVB$-fold slower than FB. 
\end{itemize}
In summary, we can predict the running time of $\FCVBML$, $VA$,
$FB$ and $\VBML$ versus $ML$'s to be $O(\eta\eIVB+1)\geq O(1.75)$,
$O(\nstate)$, $O(3\nstate)$ and $O(3\nstate\eIVB)$, respectively. 

Let us verify above computational prediction via the simulation results
in Fig. \ref{fig:TimeResource}, $M=64$. the average ratios of running
time of $\FCVBML$, $VA$, $FB$ and $\VBML$ versus $ML$'s were
found to be $1.83$, $57.1$, $200.3$ and $233.9$, respectively.
These results are consistent with our estimates of the ratios in the
last paragraph. 

Also, from Fig. \ref{fig:TimeResource}, we can see that FCVB is completely
superior to VA in this case, since they achieve similar performance,
given the same number $\ndata$ of data, but FCVB runs much faster.
The average gain in FCVB's speed over VA's is about $57.1/1.83=31.2$
times, i.e. around half of $M=64$. This gain in simulation is also
consitent with the theoretical gain $\nstate/(\eta\eIVB+1)\leq36.6$
in computational load, with $\nstate=64$, $\eIVB\geq1$ and $\eta\in\left[0.75,1.25\right]$,
as explained above.

In the same simulation of Fig. \ref{fig:TimeResource}, we also found
that the average $\bar{\eIVB}$ for $\FCVBML$ is $\bar{\eIVB}=1.01$,
i.e. $\FCVBML$ almost converged right after the first IVB cycle.
Then, we can deduct the average value $\bar{\eta}$ to be $\bar{\eta}=\frac{(1.83-1)}{\bar{\eIVB}}=0.82$.
This value $\bar{\eta}=0.82$ belongs to the theoretical range $\left[0.75,1.25\right]$. 
\begin{rem}
Note that, if we exclude the relative cost $O(1)$ of ML's initialization
step in the $\FCVBML$'s relative cost $O(\eta\eIVB+1)$, the average
computational load $\overline{\eta\eIVB}$ of Accelerated FCVB algorithm
in this case is only $\overline{\eta\eIVB}=1.83-1=0.83$ times of
$ML$'s cost, i.e., in this case, the Accelerated FCVB is faster than
the currently-supposed fastest algorithm, ML, for HMC estimate.
\end{rem}

\subsection{Evaluation of VB-based acceleration rate \label{subsec:ch8:Empirical-speed-up}}

From two Lemmas \ref{lem:(Accelerated-IVB-algorithm)}-\ref{lem:(Accelerated-FCVB-algorithm)},
the BER performance for accelerated scheme is expected to be the same
as traditional scheme for VB-based (VB and FCVB) inference, which
is indeed the case in simulations in this chapter. For comparison
purpose, the gain factor in this case is the acceleration rate $\frac{\nu_{c}}{\eIVB}$
in speed, as defined in (\ref{eq:chap6:accelerated rate}), between
the effective number $\eIVB$ and total number $\nIVB$ of IVB cycles.

\begin{figure}
\begin{centering}
\includegraphics[width=1\columnwidth]{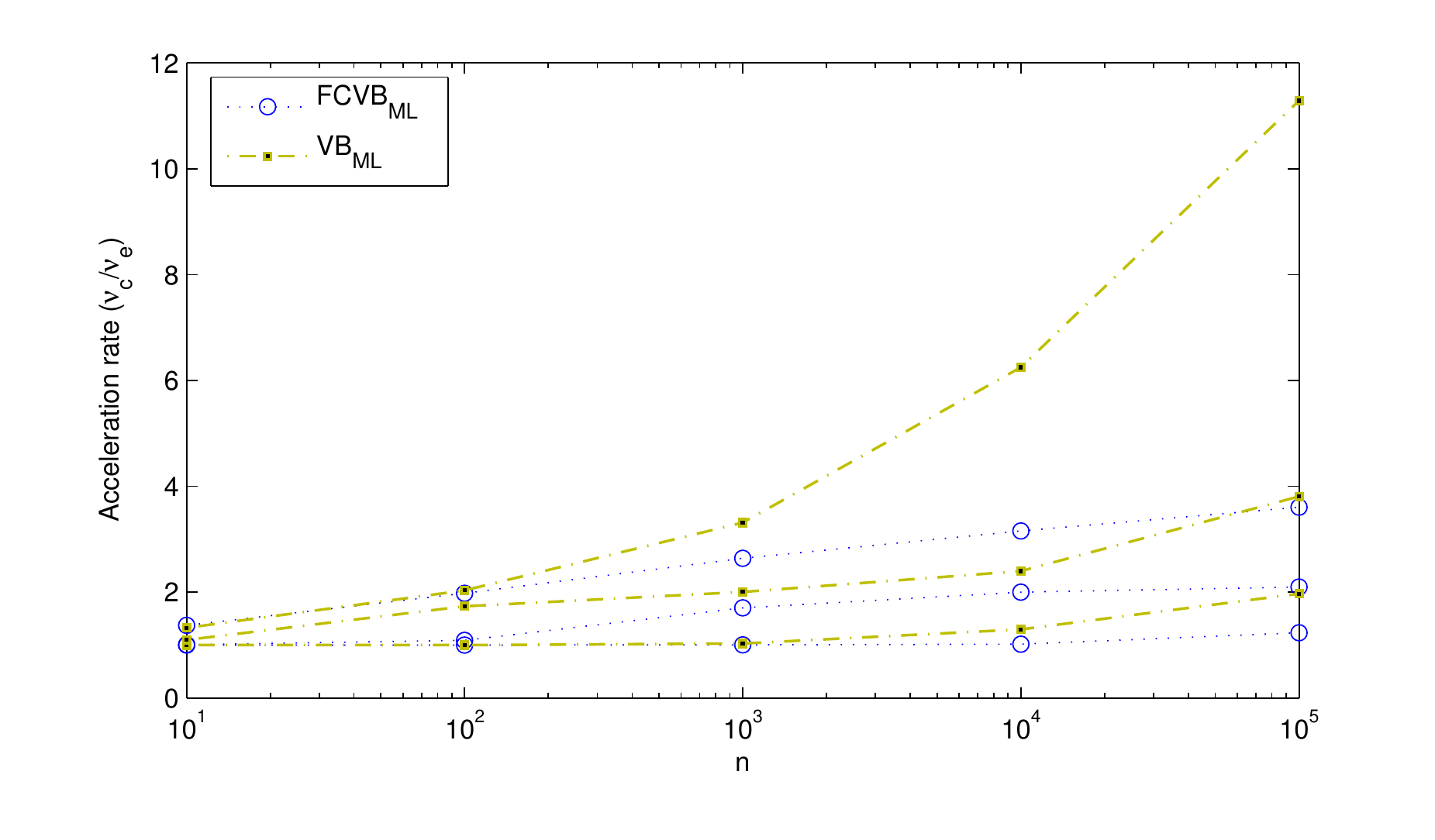}
\par\end{centering}
\begin{centering}
\includegraphics[width=1\columnwidth]{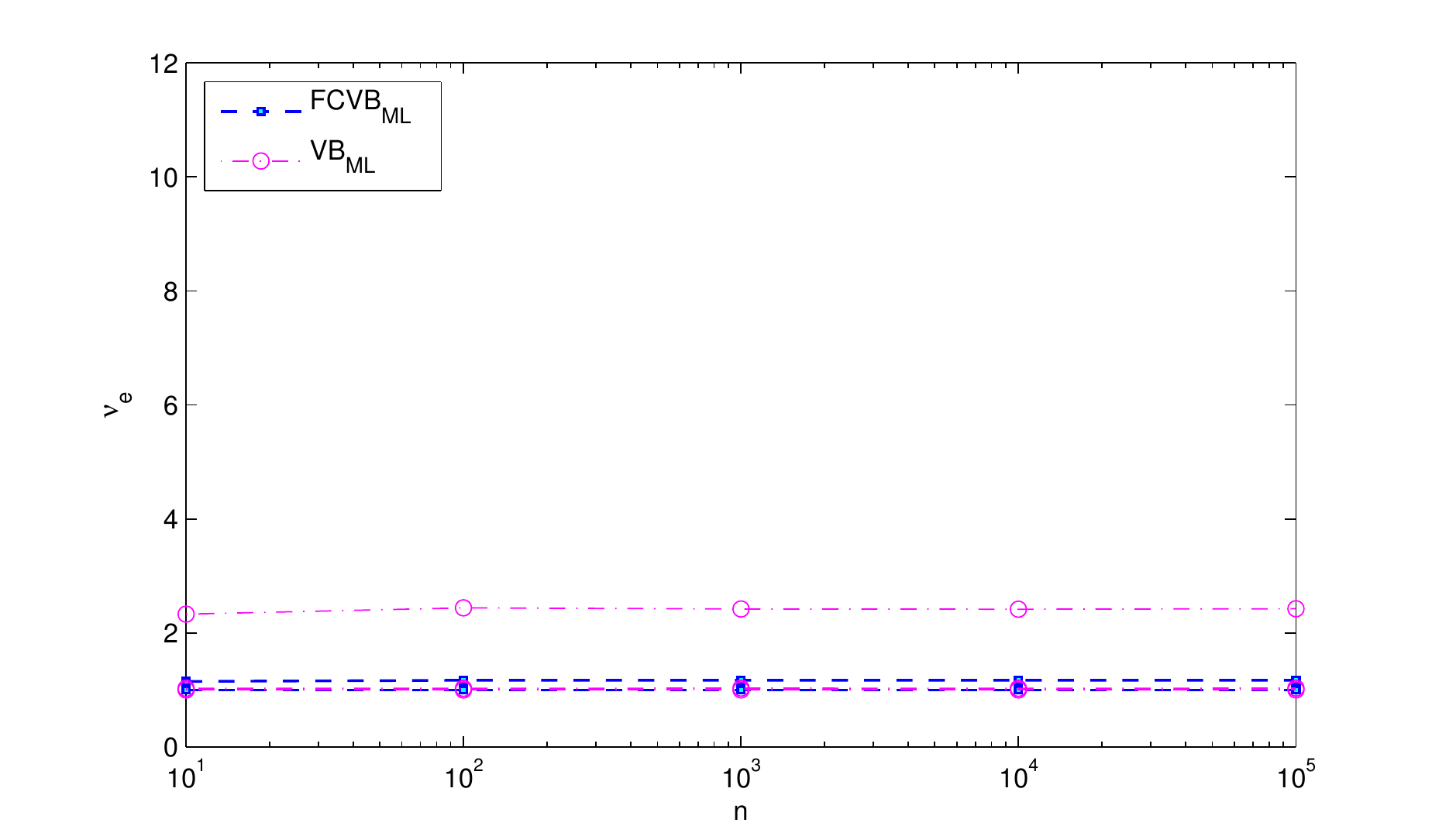}
\par\end{centering}
\centering{}\caption{\label{fig:N}Acceleration rate $\protect\nIVB/\protect\eIVB$ (above)
and effective number $\protect\eIVB$ of IVB cycles (below) versus
number of data $n$, with $10^{3}$ Monte Carlo runs.. Three curves
for each algorithm, from bottom to top, correspond to number of states:
$M=\{2,8,64\}$. Note that, some curves are almost identical to each
other.}
\end{figure}
In Fig. \ref{fig:M=00003D2-8-64}, the overall average values of $\nu_{c}$
for $\tradVBML$ and $\tradFCVBML$ and those of $\eIVB$ for $\VBML$,
$\FCVBML$ are $2.6\pm0.7$, $1.5\pm0.3$ and $1.6\pm0.6$, $1.1\pm0.1$
over $10^{5}$ Monte Carlo runs, respectively. Hence the acceleration
rate $\nIVB/\eIVB$ in this context is about $1.5$ for these VB and
FCVB schemes.

In Fig. \ref{fig:N}-\ref{fig:M} (lower panels), we can see that
$\eIVB$ for both $\VBML$ and $\FCVBML$ are very small, $O(1)$,
and independent of $n$, even at $n=10^{5}$. At high values of $M$
($32$ and $64$), $\eIVB$ for $\VBML$ increases considerably, while
$\eIVB$ for $\FCVBML$ increases only slightly. 

Compared with $\tradVBML$ and $\tradFCVBML$, we can see that the
acceleration rate of $\FCVBML$ is approximately linear in the $\log$
of $\ndata$ or $\nstate$, i.e. $O\left(\frac{\nu_{c}}{\eIVB}\right)=O(\log\nstate)$
and $O\left(\frac{\nu_{c}}{\eIVB}\right)=O(\log n)$, with fixed $n$
and fixed $M$, respectively. The acceleration rate of $\VBML$ is
super-linear against a $\log$ scale, when $n$ and $M$ are high.
Hence, for $\VBML$, we have $O\left(\frac{\nu_{c}}{\eIVB}\right)\geq O(\log\nstate)$
and $O\left(\frac{\nu_{c}}{\eIVB}\right)\geq O(\log n)$, with fixed
$n$ and fixed $M$, respectively. As a consequence, from simulation
results of $\eIVB$ and acceleration rate $\frac{\nIVB}{\eIVB}$,
we can deduct that the converged IVB cycles, $\nu_{c}$, of $\tradVBML$
and $\tradFCVBML$ are also logarithmically scale against  both $n$
and $M$. 

Overall, this $\log$scale phenomenon may be relevant to exponential
forgetting property of HMC, as explained in Section \ref{subsec:chap6:Accelerated-schemes}.
The simulations also show that VB requires slightly more number of
IVB cycles than FCVB, possibly because FCVB circulates hard-information,
which is likely to converge faster than soft-information used in VB.

\begin{figure}
\begin{centering}
\includegraphics[width=1\columnwidth]{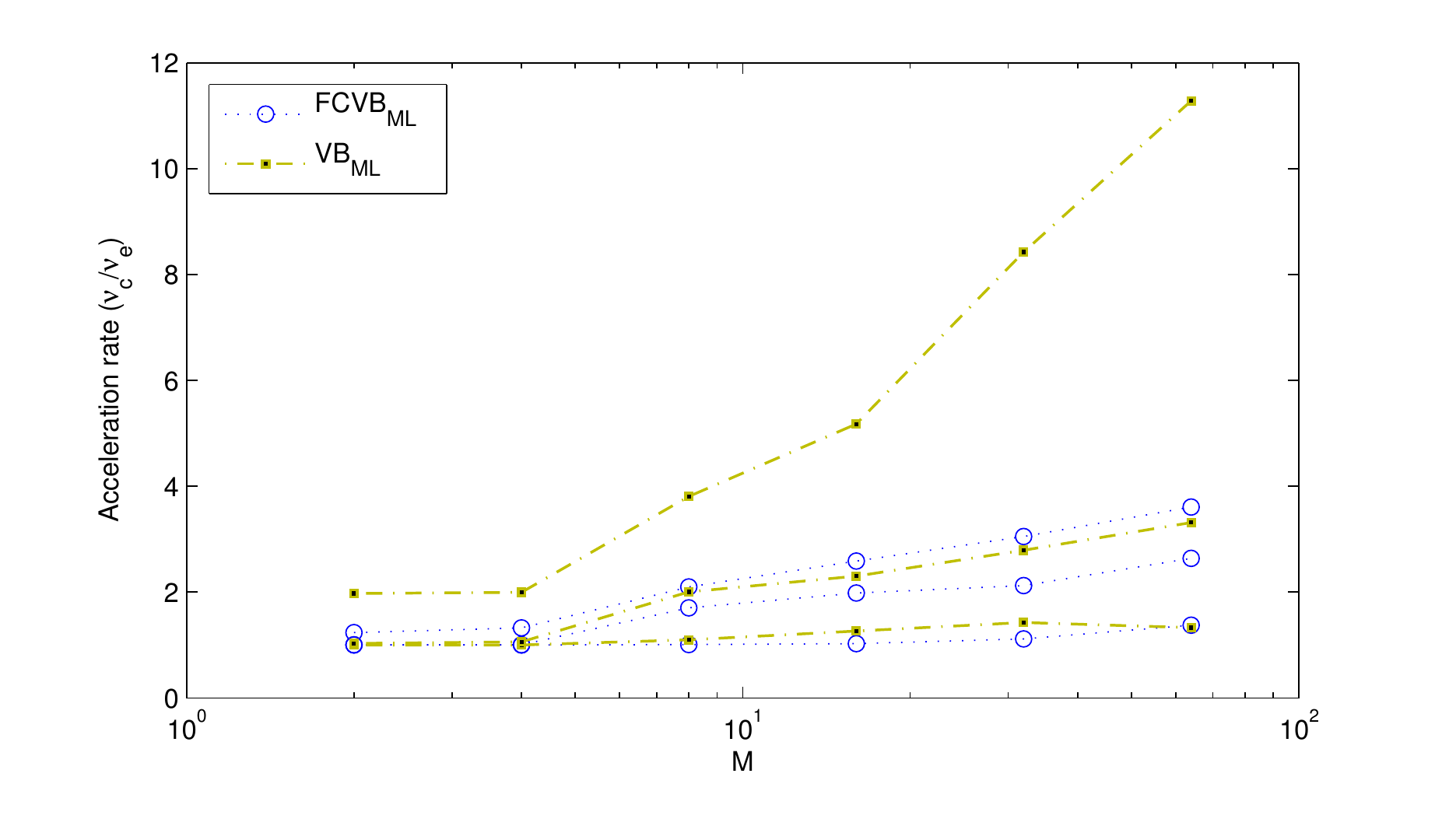}
\par\end{centering}
\begin{centering}
\includegraphics[width=1\columnwidth]{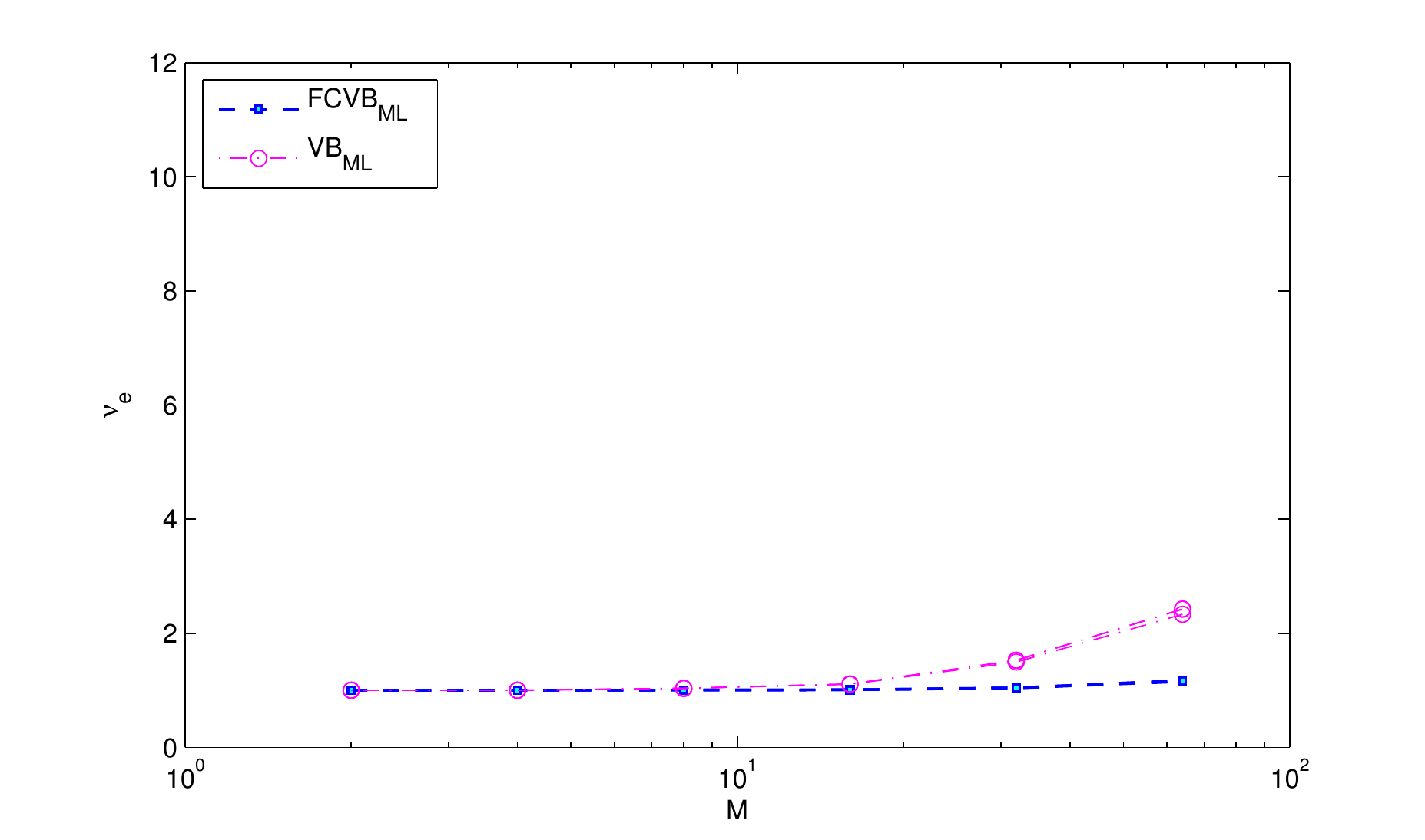}
\par\end{centering}
\centering{}\caption{\label{fig:M}Acceleration rate $\protect\nIVB/\protect\eIVB$ (above)
and effective number $\protect\eIVB$ of IVB cycles (below) versus
number of states $\protect\nstate$, with $10^{3}$ Monte Carlo runs..
Three curves for each algorithm, from bottom to top, correspond to
number of data: $n=\{10^{1},10^{3},10^{5}\}$. Note that, some curves
are almost identical to each other.}
\end{figure}

\section{Rayleigh fading channel \label{subsec:chap8:Rayleigh-fading-channel}}

Although the Rayleigh fading channel was thoroughly reviewed in Section
\ref{sec:ch3:Fading-channel}, some key points for simulation will
be summarized here for clarity. 

Let us recall that, in practice, the receiver may be moving with velocity
$v$. Because of the Doppler effect, this movement causes a fading
rate $g_{i}\TRIANGLEQ\left\Vert \hdelay_{\itime}\right\Vert $ (\ref{eq:ch3:g=00005Bk=00005D:case3})
for the amplitude's average of scattering received signals, as illustrated
in Fig. \ref{fig:chap3:fading_path}. A statistical model for $|\hdelay_{\itime}|$,
firstly proposed in {[}\citet{ch2:origin:Fading:Clark68}{]}, is an
envelope of a stationary complex Gaussian process, whose autocorrelation
function (ACF) is given by (\ref{eq:ch3:ACF(t)}): 
\begin{equation}
\rho(\Tsample)=\sigma^{2}J_{0}(2\pi\fDoppler\Tsample)\label{eq:chap8:rho-ACF}
\end{equation}
where $\sigma^{2}$ is the variance of the complex Gaussian process
per dimension, $\fDoppler=v/\lambda$ (Hz) is the maximum Doppler
frequency, $\lambda$ (m) is the transmitted carrier wavelength, $J_{0}(\cdot)$
is the zero-order Bessel function of the first kind and $\Tsample$
is the sampling period. Note that, at any sampling time $\iinn$,
the marginal distributions of $\gbar_{i}$ and $\gbar_{i}^{2}$ for
this Gaussian process are Rayleigh {[}\citet{Rayleigh_process,Vehicular_Fading}{]}
and $\chi^{2}(2)$ {[}\citet{ch3:bk:Fading:cavers00}{]} distribution,
respectively, as shown in equations (\ref{eq:ch3:Chi-squared}-\ref{eq:ch3:Rayleigh}).
Hence, this model is called a Rayleigh fading channel. Because it
is prohibitive to evaluate $\gbar_{i}$ via ARMA process, a quantized
HMC model {[}\citet{ch3:ART:FadingMarkov:tutorial08}{]} for $\gbar_{i}$
is currently a popular choice for the decoder over fading channel.
In our simulation, all values of $\gbar_{i}$ are generated from the
quantized HMC, as defined below.

\subsection{Markov source with HMC fading channel}

From receiver's model in equation (\ref{eq:ch3:receiver_case3}),
let us consider a fading channel model, with the same Markov source
$s_{i}$ and notations in previous section:

\begin{equation}
x_{i}=\gbar_{i}s_{i}+e_{i},\ \iinn\label{eq:Rayleigh_Markov}
\end{equation}

where $\gbar_{i}$ is one of quantized $K$-levels of Rayleigh distribution
$f(g_{i})$. The transition matrix, $\TBold_{c}$, of fading HMC is
a quantized version of the conditional distribution: $f(g_{i}|g_{i-1})=f(g_{i},g_{i-1})/f(g_{i})$,
where $f(g_{i},g_{i-1})$ is a time-invariant bi-variate Rayleigh
distribution (see Appendix. \ref{App:chap:Quantization} for details).
Then, the model (\ref{eq:Rayleigh_Markov}) can be augmented to be
an HMC with $MK$ states. Let us then define $MK\times MK$ transition
matrix as $\TBold_{cs}=\TBold_{c}\bigotimes\TBold_{s}$, where $\bigotimes$
denotes Kronecker product for matrix. Because $K=8$ channel quantization
levels are found to be accurate enough for $\fDoppler\Tsample\leq0.01$
{[}\citet{ch2:art:Fading:Capacity_05}{]}. Then, in the sequel, let
us consider an $\nstate=16$-QAM signal transmitted over a Rayleigh
channel, which yields an augmented source-channel HMC of $MK=128$
states. The algorithms for HMC in Chapter \ref{=00005BChapter 6=00005D}
will infer the $MK$-state label of the augmented HMC, $\gbar_{i}s_{i}$,
in (\ref{eq:Rayleigh_Markov}). We can then marginalize out the $K$
channel levels to compute the $M$-state Markovian source. Hence,
the BER in our simulations took only the source state estimates into
account. 

For parameter settings, since $g_{i}$ and $s_{i}$ are assumed independent,
the fading power $E(g_{i}^{2})=2\sigma^{2}$ is normalized to unity
in this section, i.e. $\sigma^{2}=0.5$, so that the average SNR per
bit $SNR_{b}$ is still the same as average energy per bit of the
source, i.e. $SNR_{b}=E_{b}/N_{0}$ (Section \ref{sec:chap8:Markov-source-AWGN}).
Also, as shown in (\ref{eq:chap8:rho-ACF}), the correlation coefficient,
$\rho\TRIANGLEQ\rho(\Tsample)$, of the time-invariant $f(g_{i},g_{i-1})$
is a function of the normalized Doppler frequency $\fDoppler\Tsample$,
whose meaning is explained in Section \ref{subsec:chap2:Fading-channel}.
This relationship is illustrated in in Fig. \ref{fig:rho_vs_fdT}.
Then, we can vary $\rho$ via three practical regimes of fading channel,
i.e. slow, intermediate and fast fading regimes, corresponding to
$\fDoppler\Tsample\lesssim0.01$, $0.01\lesssim\fDoppler\Tsample\lesssim0.4$,
and $\fDoppler\Tsample\gtrsim0.4$, respectively {[}\citet{ch3:ART:FadingMarkov:tutorial08}{]}.
Because the exact thresholds for those three regimes are not clearly
defined in the literature, let us re-define the range $\fDoppler\Tsample\leq0.01$,
$0.01<\fDoppler\Tsample\leq0.1$, and $\fDoppler\Tsample>0.1$ for
those three regimes, respectively, in this thesis. 

Note that, the correlation in $\TBold_{c}$ is implied by value of
$\rho$ in (\ref{eq:chap8:rho-ACF}). Because the correlation in our
uniform samples-based $\TBold_{s}$ is low, the correlation in $\TBold_{cs}$
mostly depends on $\rho$. Hence, by varying $\rho$, we are actually
varying the correlation in the augmented HMC model (\ref{eq:Rayleigh_Markov}).

\begin{figure}
\begin{centering}
\includegraphics[width=0.8\columnwidth]{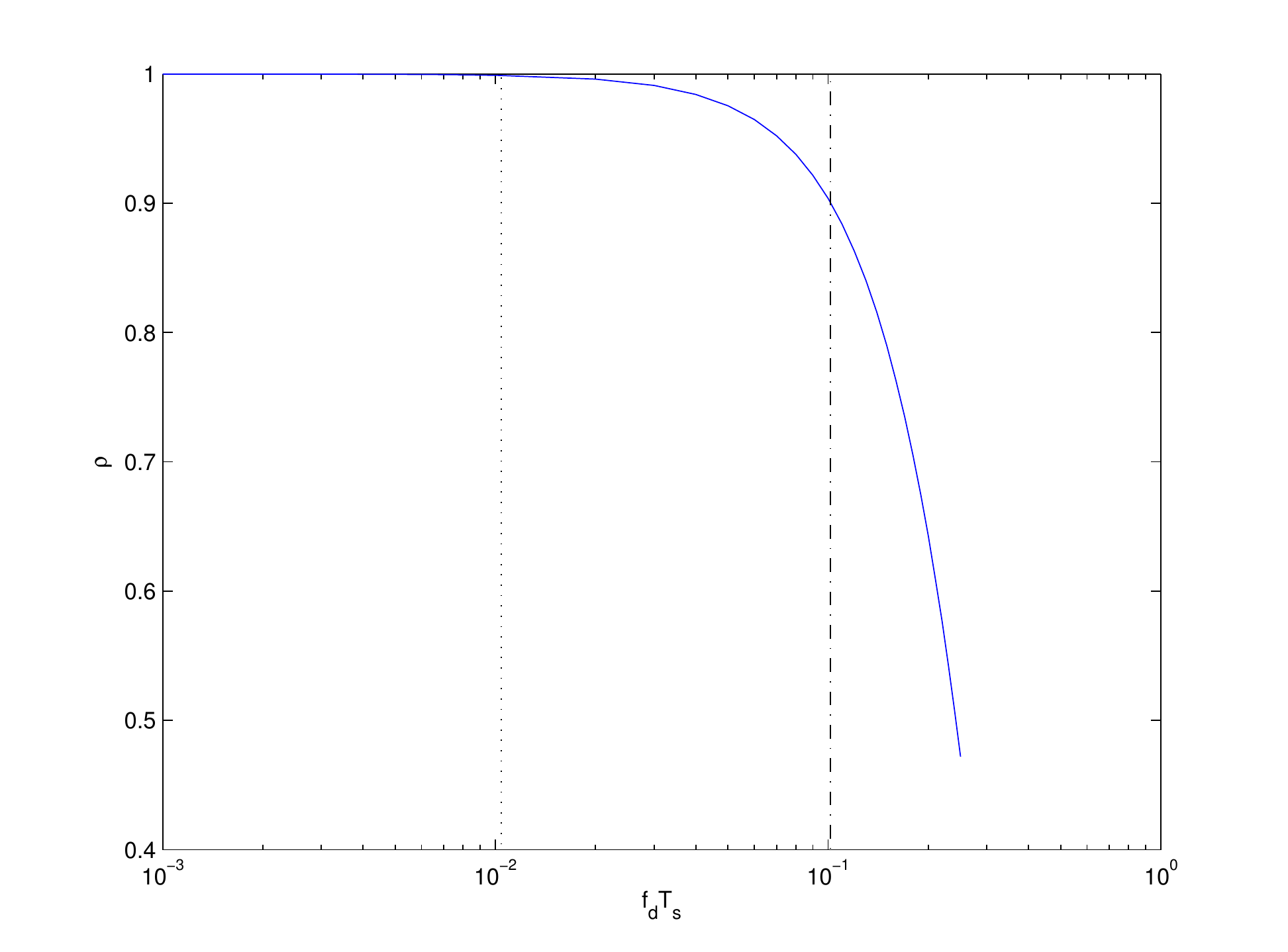}
\par\end{centering}
\centering{}\caption{\label{fig:rho_vs_fdT}Correlation coefficient $\rho=J_{0}(2\pi\protect\fDoppler\protect\Tsample)$
versus normalized Doppler frequency $\protect\fDoppler\protect\Tsample$.
From left to right: three fading regimes are slow, intermediate, and
fast fading regimes, corresponding to $\protect\fDoppler\protect\Tsample\protect\leq0.01$,
$0.01<\protect\fDoppler\protect\Tsample\protect\leq0.1$, and $\protect\fDoppler\protect\Tsample>0.1$,
respectively.}
\end{figure}

\subsection{Performance of source estimates in HMC channel}

\begin{figure}
\begin{centering}
\includegraphics[width=1\columnwidth]{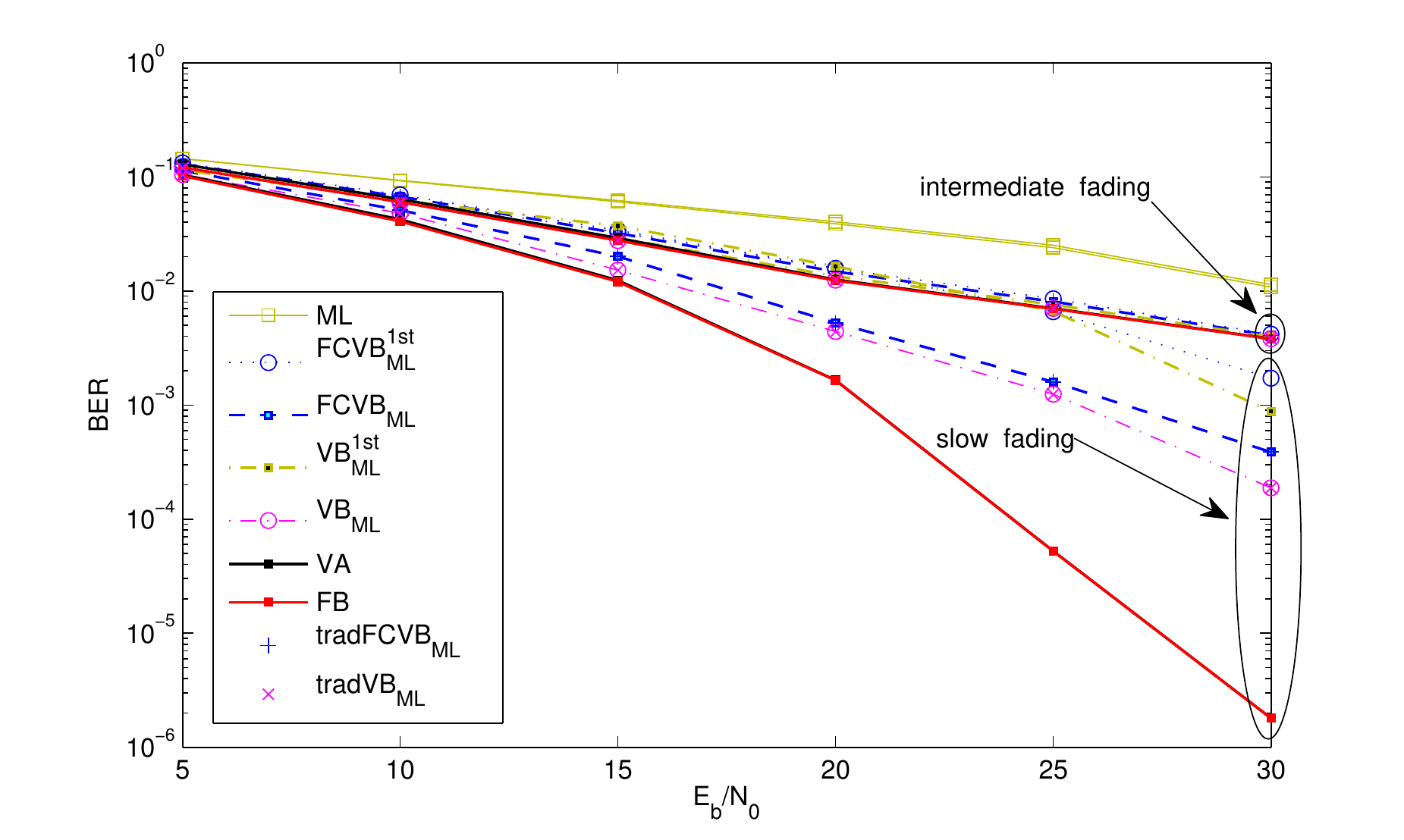}
\par\end{centering}
\centering{}\caption{\label{fig:Rayleigh_BER}BER versus SNR per bit, $E_{b}/N_{0}$(dB),
for HMC source ($\protect\nstate=16$-QAM)) over HMC fading channel
($K=8$ levels), with $10^{5}$ Monte Carlo runs.}
\end{figure}
For evaluating performance, the simulation against variable SNR per
bit, $E_{b}/N_{0}$, is displayed in Fig. \ref{fig:Rayleigh_BER}.
Two values $\fDoppler\Tsample=0.01$ and $\fDoppler\Tsample=0.1$
are considered as slow and intermediate fading regimes, respectively,
in this figure. 

We can see that, in the fast fading regime, the correlation coefficient
$\rho$ is not too high (less than $0.9$), all the algorithms (except
ML) have the same performance and similar to those for the toy HMC
example in Fig. \ref{fig:M=00003D2-8-64}. This result is expected,
because when the fast fading channel becomes dominant, the samples
between two time point becomes more independent. In the slow fading
regime, which is more popular in practice {[}\citet{ch2:art:Fading:Capacity_05}{]},
the performances of both FB and VA are almost coincide with each other
and better than those in the fast fading regime. However, VB's and
FCVB's performance become closer to ML than to FB or VA in high SNR
per bit. This fact implies that VB and FCVB approximations become
less and less accurate. In order to corroborate this finding, two
plots of BER and $KLD_{\ftilde||f}$ versus $\rho$ are displayed
in Fig. \ref{fig:BER-vs-rho} and Fig. \ref{fig:KLD-vs-rho}, respectively
(for the computation of $KLD_{\ftilde||f}$, see Section \ref{sec:KLD-for-HMC}).

\begin{figure}
\begin{centering}
\includegraphics[width=1\columnwidth]{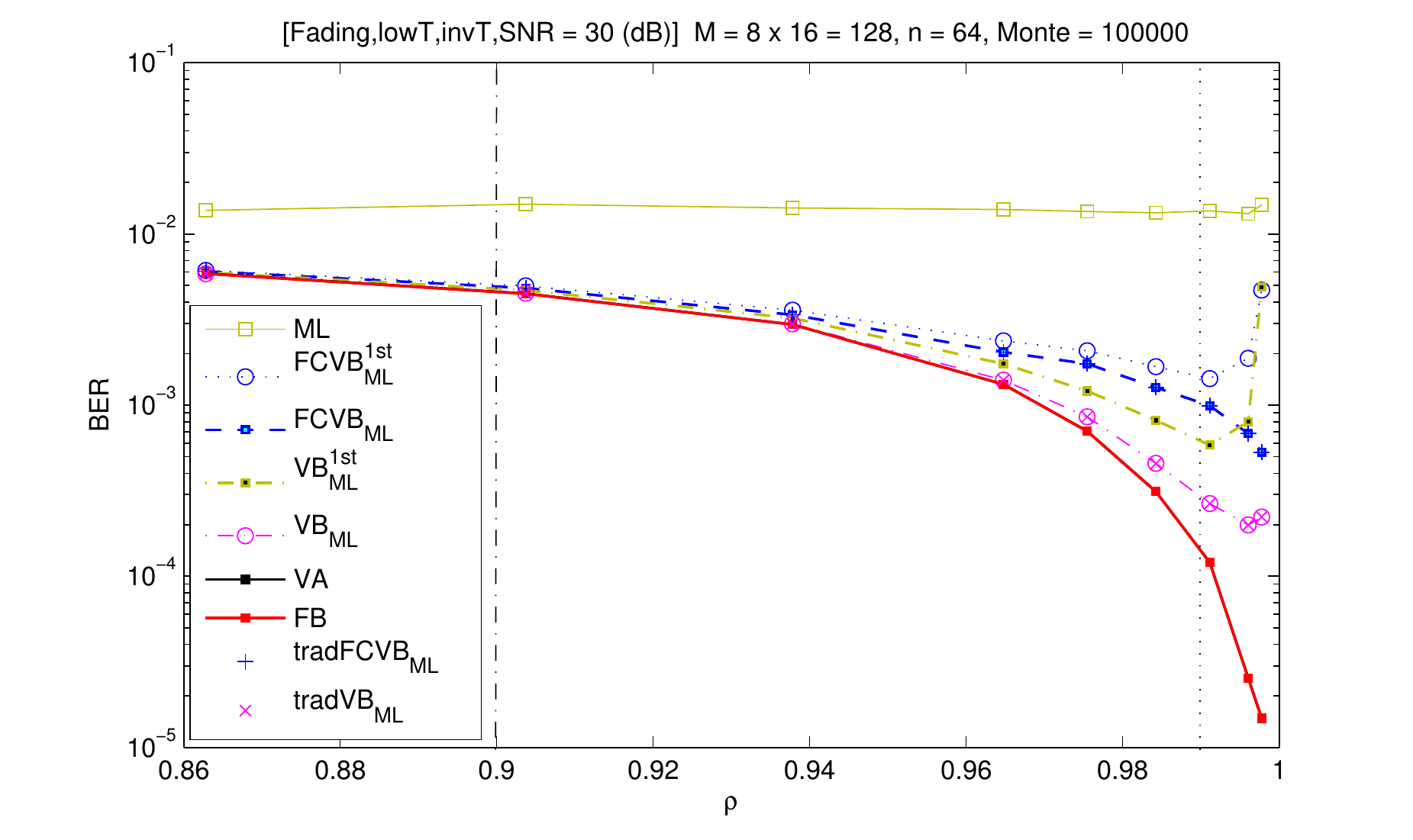}
\par\end{centering}
\centering{}\caption{\label{fig:BER-vs-rho}BER versus correlation coefficient $\rho$
of Rayleigh channel, at $E_{b}/N_{0}=30$ dB, for three fading regimes
in Fig. \ref{fig:rho_vs_fdT}, with $10^{5}$ Monte Carlo runs.}
\end{figure}
In Fig. \ref{fig:BER-vs-rho}, we focus on the case $E_{b}/N_{0}=30$
dB in Fig. \ref{fig:Rayleigh_BER}. By varying the fading channel
from slow to fast regimes, i.e. from $\fDoppler\Tsample\leq0.01$
up to $\fDoppler\Tsample\geq0.1$, we can investigate many cases,
$\rho\geq0.99$ down to $\rho\leq0.9$, respectively. In all cases,
ML's performance does not change and remains with the worst performance.
For the fast regime ($\rho\leq0.9$), all algorithms (except ML) have
virtually the same performance. For the intermediate regime, $0.9\leq\rho\leq0.99$,
there is a trade-off in performance between two groups of exact and
approximate estimates, i.e. between FB (and VA) and VB (and FCVB).
For the slow regime ($\rho\geq0.99$), although VB and FCVB's estimates
are still better than ML's, their performance deteriorates compared
to FB and VA. 

\begin{figure}
\begin{centering}
\includegraphics[width=1\columnwidth]{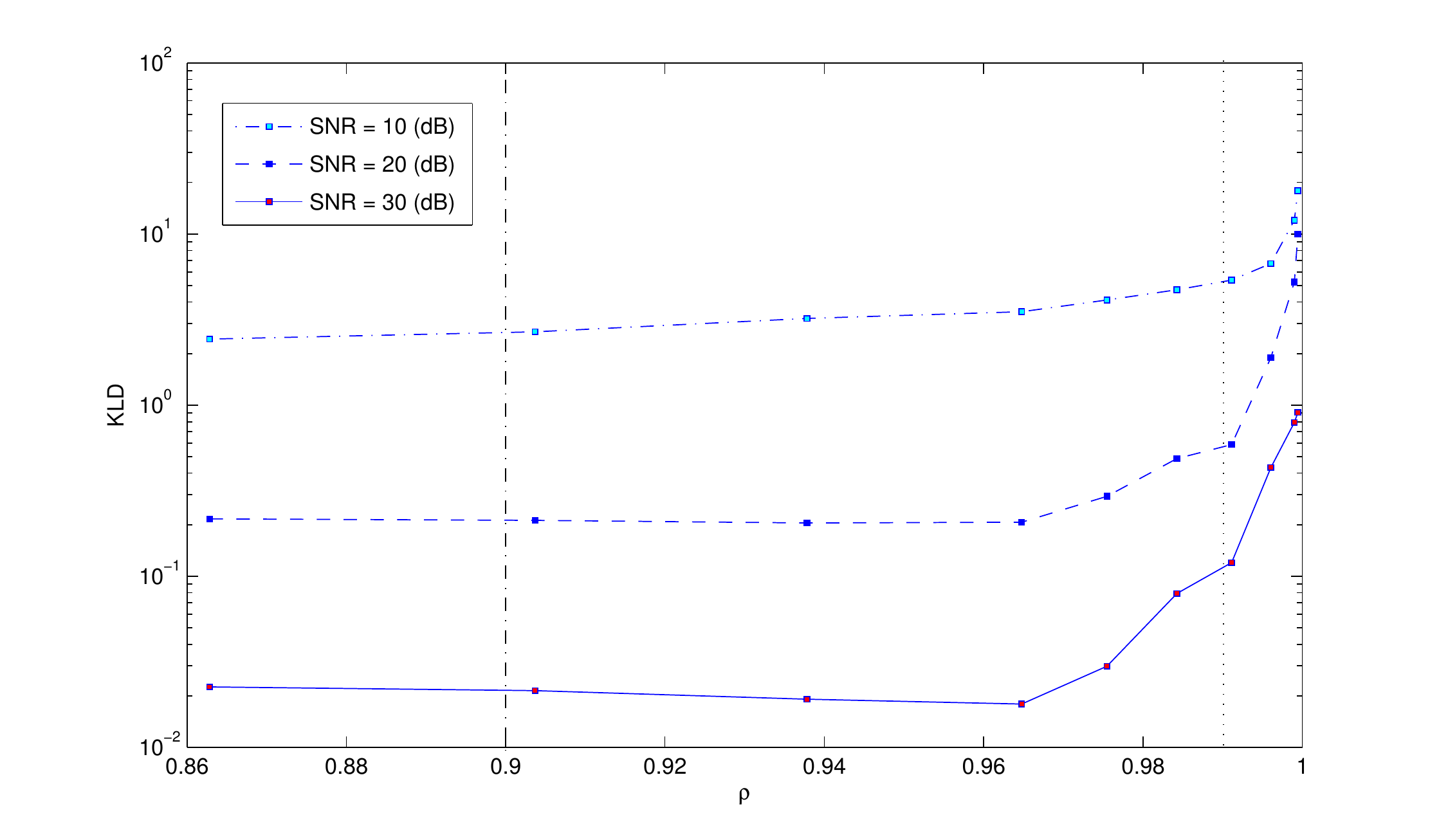}
\par\end{centering}
\centering{}\caption{\label{fig:KLD-vs-rho}$KLD_{\protect\ftilde||f}$, for the VB approximation
of source-channel HMC, versus $\rho$, for the slow fading regime
in Fig. \ref{fig:Rayleigh_BER}, with $10^{4}$ Monte Carlo runs.}
\end{figure}
The dependence of VB's accuracy of approximation on $\rho$ is shown
clearly in Fig. \ref{fig:KLD-vs-rho}. For fast and intermediate regimes
($\rho\leq0.99$), $KLD_{\ftilde||f}$ is small, implying that VB
yields a good approximation. For the slow regime ($\rho\geq0.99$),
the $KLD_{\ftilde||f}$ increases sharply with $\rho$, and, hence,
VB for the HMC is not a good approximation under this fading conditions.
Since FCVB is a CE-based version of VB, the trend of $KLD_{\ftilde||f}$
in Fig. \ref{fig:KLD-vs-rho} explains the diminished performance
of VB and FCVB compared to FB and VA, observed in Fig. \ref{fig:BER-vs-rho}.
We also see that this phenomenon is repeated for many values of SNR
per bit $E_{b}/N_{0}$, although $KLD_{\ftilde||f}$ becomes smaller
(i.e. VB yields a better approximation) in higher SNR regimes, as
expected.

The empirical results on relationship between digital detection accuracy
and correlation coefficient $\rho$ also proposes a trade-off situation
in practice: 
\begin{itemize}
\item By increasing $\rho$, the performance of Markov-based algorithms
(i.e. FB and VA), is likely to be increased, but the approximations
in class of independent distributions (i.e. VB and FCVB) is decreased.
The higher $\rho$ is, the more significant this phenomenon becomes.
This fact is actually reasonable, since the original model become
more correlated in this case. 
\item In simulations, it is shown that there are three working regimes for
FCVB algorithm. If correlation in transition matrix of HMC is not
high ($\rho\leq0.9$), FCVB can achieve the same performance as VA
and FB. When the correlation is high ($0.9\leq\rho\leq0.99$), FCVB
yields a trade-off between performance and computational load. And
finally, if the correlation is too high ($\rho\geq0.99$), FCVB is
not an attractive algorithm, since approximations in independent class
for HMC are not suitable. 
\end{itemize}

\subsection{Computational load of source estimates in HMC channel \label{subsec:chap8:cost-HMC-Rayleigh}}

The average running time (over all tested $\rho$) of all algorithms
in Fig. \ref{fig:BER-vs-rho} are displayed in Fig. \ref{fig:Fading-Time}.
This result shows that FCVB is an attractive algorithm, with much
lower complexity than VA. The ratios of averaged running-time of $\FCVBML$,
$VA$, $FB$ and $\VBML$ versus $ML$'s are $2.03$, $139.7$, $399.7$
and $575.4$, respectively. For the number of IVB cycles, we have
$\eIVB=1.04\pm0.04$, $\nu_{c}=1.79\pm0.54$ and $\eIVB=1.24\pm0.10$,
$\nu_{c}=2.28\pm0.64$ for $\FCVBML$, $\tradFCVBML$ and $\VBML$,
$\tradVBML$, respectively. Hence, the acceleration rate is about
$1.25$ for both FCVB and VB schemes in this augmented HMC context.
These results are all consistent with Table \ref{tab:ch6:ComputationalComplexity},
and with the explanation in sections \ref{subsec:chap8:Computational-cost-AWGN}-\ref{subsec:ch8:Empirical-speed-up}.

\begin{figure}
\begin{centering}
\includegraphics[width=1\columnwidth]{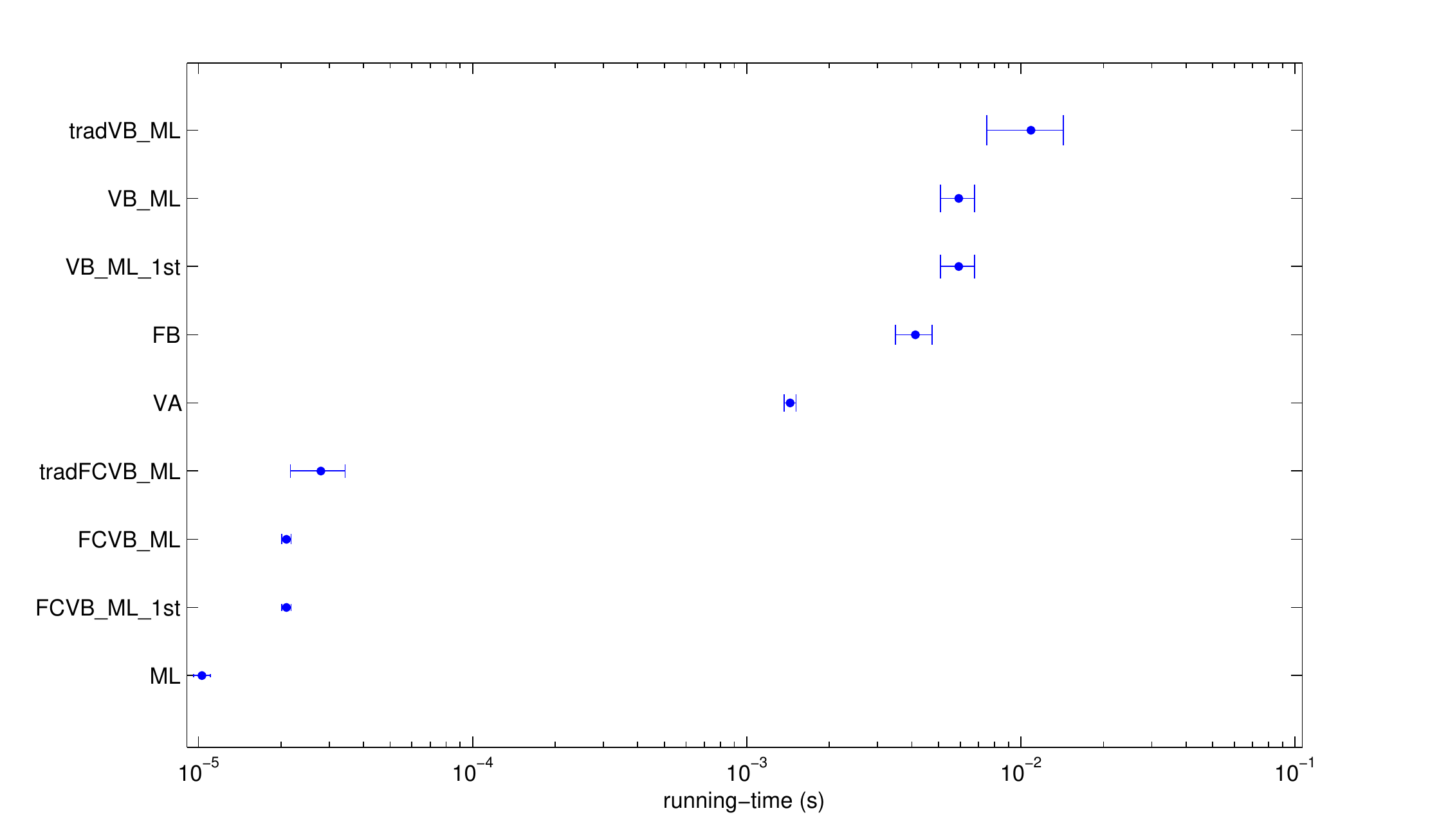}
\par\end{centering}
\centering{}\caption{\label{fig:Fading-Time} Running-time averaged over all tested $\rho$
in Fig. \ref{fig:BER-vs-rho}, with $10^{5}$ Monte Carlo runs. The
running-time is measured by C++ implementation and 3 GHz Core2Duo
Intel processor.}
\end{figure}

\section{Summary}

This chapter presents simulation results in the context of two digital
receiver models, developed in Chapter \ref{=00005BChapter 3=00005D}.
The first assumes a Markov source, whose transition matrix was generated
uniformly and simulated as input to a synchronized AWGN channel. In
the simulations, the accuracy of all state-of-the-art algorithms FB,
VA and ICM (i.e. FCVB) was shown to be the same, but the accelerated
ICM/FCVB algorithm reduced the effective number of iteration cycles
to about one, on average. 

Secondly, the same Markov source was used as an input to a synchronized
finite state Markov fading channel, with number of data, $n$, and
number of states, $M$. In this case, although the computational load
$O(nM)$ of the accelerated FCVB scheme was still much smaller than
$O(nM^{2})$ of VA, its accuracy was only comparable to VA's when
the correlation coefficient of Rayleigh fading process was not too
high. This trade-off facility can be applied in other contexts, as
discussed in the next chapter.

%auto-ignore
\chapter{Contributions of the thesis and future work\label{=00005BChapter 9=00005D}}

The pathway established in Chapter \ref{=00005BChapter 2=00005D}
has ended. It started by reviewing the telecommunications literature
for the purpose of identifying the main challenges that we wanted
to address using Bayesian methodology. The interior chapters have
provided progress with this aim. In this closing chapter, we will
summarize the contributions of the thesis and offer suggestions for
future work, before closing the thesis with concluding remarks. 

\section{Progress achieved by the thesis \label{sec:chap9:Progress-achieved}}

In this section, the strengths and weaknesses of the thesis will be
reflected by considering the major contributions and suggesting future
work. In turn, these can be divided into four principal themes, corresponding
to four key tasks in Chapter \ref{=00005BChapter 1=00005D}. For each
theme, the discussion will be presented in three steps: key achievements
in this thesis, the generalization, and proposals for future work.

\subsection{Optimizing computational flow via ring theory}

For this theme, the thesis' purpose is to reduce the computational
load involving the evaluation of objective functions arising in inference
problems relevant to telecommunication. 
\begin{itemize}
\item \textit{Main contributions:}
\end{itemize}
The thesis' original idea is to separate computational flow into two
domains: operators and variables (Section \ref{sec:Conditionally-separable}).
Respectively, two key contributions are: in the former domain, a novel
theorem (Theorem \ref{thm:GDL}) on the computational perspective
of the generalized distributive law (GDL) in ring theory, and in the
latter domain, no-longer-needed (NLN) algorithm (Algorithm \ref{alg:chap7:(NLN)-algorithm}).
The theorem guarantees computational reduction for any valid application
of GDL upon operators, while the NLN algorithm exploits the conditionally
independent (CI) topology of variables. Together, they open up two
other contributions. The first one is the efficient Forward-Backward
(FB) computational flow (Section \ref{sec:chap5:GDL-for-objective})
for distributing the operators over variables when needed, while the
second one is the discovery of explicit formulae for counting the
number of operators, ring-sum and ring-products (Section \ref{subsec:ch5:Computational-reduction-via-GDL}),
in that FB flow.

In general probability context, there is an increase exponentially
of number of operators with the number of data in computation of the
joint distribution, because the number of state increases exponentially
(the curse of dimensionality). However, the number of operators falls
exponentially with the number of NLN variables. In hidden Markov chain
context, the FB algorithm and VA, two special cases of FB recursion,
exploit the latter in combating the former. Hence, the FB recursion
and GDL helps interpret the linear dependence on data in the complexity
of FB and VA, as a consequence of numerical cancellation over the
number of operators.
\begin{itemize}
\item \textit{Generalization:}
\end{itemize}
In principle, the recursive FB flow is applicable to objective functions
to which GDL is valid. For example, it is applicable to computation
of the message-passing algorithm in graphical learning {[}\citet{ch2:origin:GDL:McEliece,ch2:bk:ToddMoon}{]}
and Markov random field {[}\citet{ch2:origin:MCMC:Gibbs}{]}, computation
of marginalization, maximization, entropy in Bayesian learning, evaluation
of Iterative VB algorithm {[}\citet{ch4:BK:AQUINN_06,ch4:art:Jordan:VB08}{]}. 
\begin{itemize}
\item \textit{Future work:}
\end{itemize}
Based on the above generalization, future work can be designed in
two directions: optimizing the FB flow and finding more applications. 

- For the former: the objective is to find the CI topology that minimizes
the computational reduction achieved via FB recursion. Although the
global minimization is an NP-complete problem, as discussed in Section
\ref{sec:chap5:optimizations}, a local solution can be found by extending
the computational flow from two-directions in FB recursion to multi-directions
in a topological graph. Another potential solution is to apply topological
sorting algorithms to the CI topology before implementing the FB recursion,
owing to the explicit formulae for determining computational complexity
in this case.

- For the latter: the objective is to verify whether GDL is valid
for a particular objective function. One key property is Markovianity,
owing to its natural CI topology. Because Markovianity is assumed
in many efficient algorithms in telecommunications system, as reviewed
in Section \ref{sec:chap2:Review-of-digital}, the FB recursion can
be applied to studying the computational reduction in these algorithms,
e.g. in forward-backward lattice filters and in other scenarios in
telecommunications. Another interesting issue is to explore the relationship
between recursion (Section \ref{sec:chap6:FB and VA via CI}) and
iteration (Section \ref{sec:chap6:VB-infer-for-HMC}) in Markovian
objective functions. If the recursion is carried out via GDL, it is
likely that the computational load can be further reduced in two cases:
recursion embedded in iteration, and iteration embedded in recursion.
The first case was considered in Accelerated FCVB algorithm in this
thesis (Section \ref{subsec:chap6:FCVB-algorithms-HMC}), while the
second case can be explored in future online variance of VB scheme.

\subsection{Variational Bayes (VB) inference }

For this theme, the thesis' purpose is to extend the VB methodology
in order to achieve more accurate deterministic distributional approximation.
Our aim was not to reduce computational load, \textit{per se}.
\begin{itemize}
\item \textit{Main contributions}
\end{itemize}
The thesis' original idea is to weaken the coupling between parameters
in the transformed metric by diagonalizing the Hessian matrix at a
specific point. The VB approximation is then applied to the transformed
posterior distribution. The technique is referred to as the transformed
VB (TVB) approximation (Section \ref{sec:chap7:Transformed-VB-approximation}).
Compared with VB, the TVB scheme was shown to yield significant improvement
in accuracy when the Hessian of the transformed distribution is designed
to be diagonal at its MAP point. Intuitively, this improvement is
achieved because the transformed variables are asymptotically independent,
in which case the VB approximation is exact. Note that, the TVB approximation
has a fundamental output as an approximate distribution in the \textit{original}
metric (Fig. \ref{fig:ch7:TVB}), a novel contribution, when compared
to classical orthogonalization approaches, whose purpose is to produce
estimates of transformed variables.
\begin{itemize}
\item \textit{Generalization:}
\end{itemize}
In principle, the TVB approximation is applicable to any multivariate
posterior distribution, whose desired marginalization is intractable.
In practice, for tractability of Iterative VB algorithm, the transformed
distribution should be separable-in-parameters (Definition \ref{DEF:(Separable-in-parameter-(SEP)}),
i.e its logarithm can be factorized into product of functions for
each parameter separately. 
\begin{itemize}
\item \textit{Future work:}
\end{itemize}
Based on the above generalization, future work can be designed in
two directions: optimizing computational load and designing new transformations,
such that the transformed distribution is separable-in-parameter (Definition
\ref{DEF:(Separable-in-parameter-(SEP)}). 

- For the former: the current TVB approximation may involve computational
intensive IVB cycles. A potential solution is to replace the involved
expectation with maximization via the FCVB scheme (Lemma \ref{lem:(Iterative-FCVB-algorithm)}).
However, this scheme reduces to a point estimation and neglects all
the moments, which, in turn, may significantly reduce the quality
of distributional approximation (Fig. \ref{fig:BER-vs-rho}).

- For the latter: Two potential transformations are global diagonalization
of transformed Hessian matrix and frequentist's transformation techniques
{[}\citet{ch7:origin:Box_Cox:64,ch2:art:BoxCox:sakia92,ch2:art:transform:semi_para2008}{]}.
The task is, however, not trivial, because of difficulty with each
of transformation design. Global diagonalization of the Hessian matrix
is only feasible for bivariate distribution {[}\citet{ch7:origin:Cox_Reid:87}{]}.
Furthermore, in frequentist's transformation, the inverse transformation
can be applied in the point estimate. In contrast, the inverted distribution,
i.e. the TVB approximation, may be highly complicated and, in particular,
its marginalization is not available. 

These difficulties show that further work are required for TVB. Even
in the current form, TVB needs to be applied on a case-by-case basis,
since the certainty equivalent (CE) points, like the MAP point, may
not be available at the beginning. Nevertheless, the TVB shows the
potential for relaxing VB methodology to achieve more accurate distributional
approximation.

\subsection{Inference for the Hidden Markov Chain}

The previous two themes focus exclusively on computational reduction
and enhancing accuracy, respectively, but not on both together. The
third main theme of this thesis is to provide new algorithms which
achieves better trade-offs between performance and speed for label's
inference in the HMC.
\begin{itemize}
\item \textit{Main contributions:}
\end{itemize}
The thesis' idea is to replace the VA with the ICM algorithm for better
trade-off. Two contributions, one for performance and one for speed,
were achieved using this approach. 

- For performance, a Bayesian interpretation was given for both the
ICM (Section \ref{subsec:chap4:Functionally-Constraint-VB}) and the
VA algorithm (Section \ref{subsec:cha6:VA-approximated-HMC}). We
show that the criteria are to preserve the global MAP trajectory and
the local MAP trajectory at any recursive and iterative step, respectively.
This interpretation also explained why the accuracy of VA and ICM
are comparable when correlation in the HMC is not too high. 

- For speed, an accelerated scheme was designed for the VB scheme,
in which any VB marginals that have converged are flagged and are
not updated in the next IVB iteration. We shows that this accelerate
scheme provides the same output as original scheme (Lemma \ref{lem:(Accelerated-IVB-algorithm)},\ref{lem:(Accelerated-FCVB-algorithm)}).
Since ICM can be re-interpreted as the functionally constrained VB
(FCVB) approximation, the computational load of Accelerated ICM/FCVB
was reduced, in simulation, from $O(\nu nM)$ down to nearly $O(nM)$,
where $\nu$, $n$ and $M$ are number of ICM iterations, number of
time points and number of states in HMC, respectively. 
\begin{itemize}
\item \textit{Generalization:}
\end{itemize}
In principle, the accelerated scheme for Iterative VB and ICM algorithm
can be applied to any inference problem involving hidden field of
CI variables, notably the Markov random field, when correlation is
not too strong.\\
\\

\begin{itemize}
\item \textit{Future work:}
\end{itemize}
Based on the above principle, we may investigate further the computational
reduction achieved by the accelerated scheme. The simulation in the
thesis showed that, in the HMC, the number of iteration for traditional
VB and ICM/FCVB is almost linear to the logarithm of both $n$ and
$M$ (Section \ref{subsec:ch8:Empirical-speed-up}). Also, the effective
number of IVB cycle for accelerated VB and ICM/FCVB scheme was close
to one and stayed nearly constant with $n$ and $M$ (Section \ref{subsec:ch8:Empirical-speed-up}).
This reduction in log-scale suggests HMC's exponentially-forgetting
property, whose influence on the number of IVB cycle should be investigated. 

\subsection{Inference in digital receivers}

For this theme, the thesis' purpose is to apply the three themes above
to practical concern in the digital demodulation in digital receivers.
\begin{itemize}
\item \textit{Main contributions:}
\end{itemize}
The thesis' idea was to apply the Accelerated ICM/FCVB algorithm and
TVB approximation to demodulation in digital receivers. Two main contributions,
one for Markovian digital detector and one for frequency synchronization,
were given in the thesis.

For Markovian digital detector (Chapter \ref{=00005BChapter 8=00005D}),
the Accelerated ICM/FCVB algorithm was applied to detecting modulated
bit stream transmitted over a quantized Rayleigh fading channel. When
the fading is not too slow, i.e. correlation between samples is not
too high, the performance of Accelerated ICM/FCVB is comparable to
the state-of-the-art VA, but with a greatly reduction of computational
load (Section \ref{subsec:chap8:Computational-cost-AWGN},\ref{subsec:chap8:cost-HMC-Rayleigh}). 

For frequency synchronization, the full Bayesian inference was studied
for a toy problem, namely frequency inference for the single-tone
sinusoidal model in AWGN channel (Section \ref{sec:chap7:Frequency-inference}).
Note that, when the frequency is off-bin, the posterior mean yields
far more accurate (Fig. \ref{fig:sinewave}) than periodogram-based
ML estimate, since posterior mean is continuous value while the DFT-based
periodogram is not (Section \ref{subsec:chap7:Performance-of-frequency}).
The accuracy of the TVB approximation was also found significantly
better than that of the VB approximation (Fig. \ref{fig:sinewave})
from the point of view of posterior mean (Remark \ref{Remark:chap7:As-a-remark}).
It is important to remember that all of these techniques - VB, TVB,
and exact posterior mean, as well as ML - are all computed via the
DFT (and implemented via FFT), and therefore have similar computational
load.
\begin{itemize}
\item \textit{Generalization:}
\end{itemize}
In principle, the Accelerated ICM/FCVB can be successfully applied
to the finite-state Markov channel (FSMC) {[}\citet{ch3:ART:FadingMarkov:tutorial08}{]}
when correlation is not too high. Also, the TVB method is attractive
for maintaining accuracy in nonlinear synchronization problem {[}\citet{ch2:art:sync_phase:SEP_ICASSP}{]}. 
\begin{itemize}
\item \textit{Future work:}
\end{itemize}
Based on the above principle, future work can be proposed in two directions:
Markovian digital decoder and carrier synchronization.

- For Markovian digital decoder (Chapter \ref{=00005BChapter 8=00005D}),
perhaps the most obvious proposal is to replace the VA with the Accelerated
ICM/FCVB. The evidence supporting proposal was provided in (Fig. \ref{fig:TimeResource},\ref{fig:ch8:Layman},\ref{fig:Fading-Time}),
showing great increase in speed without much loss of accuracy. Note
that, the Accelerated ICM/FCVB is more broadly applicable than the
Markovian context of VA. Furthermore, the Accelerated ICM/FCVB can
be implemented in both online and offline scenarios, yield the same
output in these cases, while VA is the offline algorithm.

- For carrier synchronization, the Bayesian inference is mostly preferred
when the accuracy is a premium. For example, the accuracy in frequency
and phase synchronization is critical in OFDM scheme for 4G system
(Section \ref{subsec:chap2:Memoryless-modulation}), and in joint
decoding and synchronization {[}\citet{ch2:art:sync_Turbo:Herzet07}{]}.
In the future, the challenge will to elaborate VB and TVB solution
for these problems.

\section{Conclusion}

The thesis has considered both the computational side of VB-based
inference methodology and its application in digital receivers.

For the inference tasks we considered, the mathematical tools were
Bayesian methodology and ring theory, whose purpose is to update the
belief on unknown quantities and to generalize the operators for computing
these beliefs, respectively. The required computations were efficiently
implemented via two approaches, namely recursive flow via the generalized
distributive law (GDL) from ring theory, and iterative deterministic
approximation via the Variational Bayes (VB) approximation in mean
field theory. Two key contributions were given for each of the two
approaches. For GDL, the first contribution was a novel theorem on
GDL, guaranteeing the reduction in the number of operators and providing
the formula for quantifying this reduction. Secondly, a novel Forward-Backward
(FB) recursion for achieving this reduction was derived. Meanwhile,
for VB, the first contribution was the Transformed VB (TVB) scheme
for asymptotically decoupling the transformed distribution to which
VB is applied. Secondly, we develop a novel accelerated scheme for
VB, reducing the effective number of iterative VB cycles to about
one in the case of hidden Markov chain (HMC) inference. 

For digital receivers, the four achievements in inference methodology
above were then applied to digital demodulation, which consists of
synchronization and digital detection. Respectively, a TVB-based frequency
synchronizer and a fast digital detector for the quantized Rayleigh
fading channel were derived in the thesis. Each performs well in specific
operating conditions, specified in Section \ref{sec:chap7:Frequency-inference}
and Section \ref{subsec:chap8:Rayleigh-fading-channel}, respectively.
However, further work is needed to formalize these operating conditions
and to achieve a robust extension of the algorithms. Nevertheless,
these two applications illustrate the applicability to telecommunications
systems of the novel inference methodologies, described in the previous
paragraph. Undoubtedly, these approaches can address the technical
demands of digital decoders in 4G mobile systems, as reviewed in Section
\ref{subsec:ch2:Challenges-in-mobile}.

As an outcome of this thesis, two related journal papers, based on
Chapter \ref{=00005BChapter 5=00005D} and Chapter \ref{=00005BChapter 6=00005D}
respectively, are about to be submitted to the IEEE Transactions on
Information Theory. The novel algorithms derived from the generalized
distributive law (GDL) in Chapter \ref{=00005BChapter 5=00005D},
which will be reported in the first of these papers, should have impact
in the future design of optimal computational flows for arbitrary
networks, particularly Bayesian networks. The novel Variational Bayes
(VB) variants of the Viterbi algorithm, developed in Chapter \ref{=00005BChapter 6=00005D}
of this thesis, will be published in the second of these forthcoming
journal papers, and were partly published in {[}\citet{ch9:VH:VBV_ISSC}{]}.
As explained in Chapter \ref{=00005BChapter 6=00005D}, these methods
lead to better trade-offs between computational load and accuracy
than the state-of-the-art Viterbi algorithm, and should yield more
efficient decoders for hidden Markov chains. Note that preliminary
work on Bayesian inference of hidden discrete fields was published
in {[}\citet{ch9:VH:onlineVB-ISSC}{]}. Finally, Chapter \ref{=00005BChapter 7=00005D}
of this thesis, which proposes a novel inference scheme for frequency
inference, was partly published in {[}\citet{ch9:VH:TVB_ICASSP}{]}.
A fuller account of the TVB methodology in signal processing will
be submitted to the IEEE Transactions on Signal Processing at the
end of this year.

\appendix 

%auto-ignore
\chapter{Dual number \label{App:chap:Dual-Number}}

A dual number $d\in\DEAL\subset\mathbb{R}^{2\times2}$ {[}\citet{App:Dual:book68,ch5:art:DualNumber:complex75}{]}
may be defined in two ways, as either (i) $d=a\mathbf{I}_{2}+b\boldsymbol{\epsilon}=\left[\begin{array}{cc}
a & b\\
0 & a
\end{array}\right]$, $a\in\mathbb{R}$ is called the real part and $b\in\REAL$ is called
the dual part, or (ii) $d=\{Re(d),Im(d)\}$, i.e. :

\begin{eqnarray*}
Re(d) & \TRIANGLEQ & a\mathbf{I}_{2}=a\left[\begin{array}{cc}
1 & 0\\
0 & 1
\end{array}\right]\\
Im(d) & \TRIANGLEQ & b\boldsymbol{\epsilon}=b\left[\begin{array}{cc}
0 & 1\\
0 & 0
\end{array}\right]
\end{eqnarray*}
where the latter is a Catersian form representation of $d$ alternatively,
writing $d=a(\mathbf{I}_{2}+\frac{b}{a}\boldsymbol{\epsilon})$, and
$\mathbf{I}_{2}\TRIANGLEQ\left[\begin{array}{cc}
1 & 0\\
0 & 1
\end{array}\right]$ , $\boldsymbol{\epsilon}\TRIANGLEQ\left[\begin{array}{cc}
0 & 1\\
0 & 0
\end{array}\right]$. We can propose a polar-form, representation of $d$, as follows: 

\[
d=a\angle\tan\theta
\]
with the argument $\tan\theta\TRIANGLEQ b/a$ and $a\neq0$. Then,
with the sum and product in $\DEAL$ defined as usual matrix sum and
product, it is easy to verify that:

\begin{eqnarray*}
d_{1}+d_{2} & = & (a_{1}+a_{2})+\epsilon(b_{1}+b_{2})\\
d_{1}d_{2} & = & (a_{1}a_{2})\angle(\tan\theta_{1}+\tan\theta_{2})
\end{eqnarray*}
where $\epsilon$ is called the dual unit, and $\epsilon^{2}=0$,
corresponding to the matrix form:

\[
\boldsymbol{\epsilon}\boldsymbol{\epsilon}=\mathbf{0}\TRIANGLEQ\left[\begin{array}{cc}
0 & 0\\
0 & 0
\end{array}\right]
\]

\chapter{Quantization for the fading channel \label{App:chap:Quantization}}

Our aim is to derive the probability mass function (pmf) induced by
quantization of the amplitude paramater of the Rayleigh fading channel.
In common with the literature {[}\citet{Rayleigh_process,Fading_vehicular,ch2:art:Fading:Capacity_05}{]},
we design the quantization thresholds such that each quantized state,
$\overline{g}_{k,i}$ of $g_{i},$ is equi-probable.

At each time $i$, the first-order Rayleigh distribution is quantized
to $K$-levels, as follows:

\begin{equation}
f(g_{i})=\begin{cases}
\frac{g_{i}}{\sigma^{2}}\exp\left(-\frac{g_{i}^{2}}{2\sigma^{2}}\right) & ,\ g_{i}\geq0\\
0 & ,\ \mbox{otherwise}
\end{cases}\label{eq:Rayleigh}
\end{equation}
where $E(g_{i}^{2})=2\sigma^{2}$ is called the fading energy and
$\sigma^{2}$ is called the variance of underlying complex Gaussian
process per dimension {[}\citet{ch3:ART:FadingMarkov:tutorial08}{]}
(see Section \ref{subsec:chap3:Rayleigh-process} for details). Note
that, the fading energy $E(g_{i}^{2})=2\sigma^{2}$ can also be found
via distribution  $f(g_{i}^{2})$, which is the $\chi^{2}$ distribution
with two degree of freedoms in this case {[}\citet{ch3:bk:Fading:cavers00}{]}:

\[
f(g_{i}^{2})=\chi_{g_{i}^{2}}^{2}(2)=\frac{1}{2\sigma^{2}}\exp\left(-\frac{g_{i}^{2}}{2\sigma^{2}}\right)
\]
whose the mean is $E(g_{i}^{2})=2\sigma^{2}$.

For quantization, an equiprobable partitioning approach similar to
{[}\citet{Rayleigh_process,Fading_vehicular}{]} will be applied.
Let us consider the continuous distribution function (c.d.f) of Rayleigh
distribution (\ref{eq:Rayleigh}), as follows: 

\begin{equation}
F(g_{i})=\begin{cases}
1-\exp\left(-\frac{g_{i}^{2}}{2\sigma^{2}}\right) & ,\ g_{i}\geq0\\
0 & ,\ \mbox{otherwise}
\end{cases}\label{eq:Rayleigh-cdf}
\end{equation}

Now we can find $K$ thresholds $\zeta_{1},\ldots,\zeta_{K}$ of $K$
equiprobable intervals, i.e. $F(\zeta_{k})-F(\zeta_{k-1})=\frac{1}{K}$,
with $\zeta_{0}=0$, which yields: 
\begin{equation}
F(\zeta_{k})=k/K,\ k=1,\ldots,K\label{eq:F(K)}
\end{equation}
From (\ref{eq:Rayleigh-cdf}), these $\zeta_{k}$ can be expressed
in closed form, as follows:

\[
\zeta_{k}=\sqrt{-2\sigma^{2}\log(1-\frac{k}{K})},\ k=1,\ldots K-1
\]
where, for truncation at $k=K$, we set $\zeta_{K}=5$, since $F(\zeta_{K}=5)\approx1-10^{-11}$,
if $E(g_{i}^{2})=2\sigma^{2}=1$. Then, the state of quantized fading
channel are defined as the continuous mean $\overline{g}_{k,i}$ of
each interval:

\[
\overline{g}_{k,i}=K\int_{\zeta_{k-1}}^{\zeta_{k}}g_{i}f(g_{i})dg_{i},\ k=1,\ldots,K
\]
which can be computed numerically {[}\citet{Fading_vehicular}{]}.
Under this $K$-state quantization procedure, we can define the bivariate
(second-order) pmf of the $K^{2}$-state samples. The bivariate Rayleigh
probability for a pair $g_{i}$, $g_{j}$ is also quantized into $K\times K$
intervals, as follows:

\begin{equation}
\Pr[\overline{g}_{m,i},\overline{g}_{k,j}]=\int_{g_{i}=\zeta_{m-1}}^{g_{i}=\zeta_{m}}\int_{g_{j}=\zeta_{k-1}}^{g_{j}=\zeta_{k}}f(g_{i},g_{j})dg_{i}dg_{j}\label{eq:Rayleigh joint}
\end{equation}
where the integral of the bivariate distribution, $f(g_{i},g_{j})$,
can be computed numerically via the following form of bivariate Rayleigh
distribution $f(g_{i},g_{j})$ {[}\citet{ch2:art:Fading:Capacity_05}{]}:
\begin{eqnarray}
f(g_{i},g_{j}) & = & \frac{g_{i}g_{j}}{\sigma^{4}(1-\rho^{2})}\exp\left(-\frac{(g_{i}^{2}+g_{j}^{2})}{2\sigma^{2}(1-\rho^{2})}\right)I_{0}\left(\frac{g_{i}g_{j}}{\sigma^{2}}\frac{\rho}{(1-\rho^{2})}\right)\label{eq:bivariate_Rayleigh}
\end{eqnarray}
in which $I_{0}$ denotes zero-order modified Bessel function of the
first kind and $\rho$ is the correlation coefficient between $g_{i}$
and $g_{j}$. Note that, when $\rho$ is very close to $1$, then
the argument of $I_{0}(\cdot)$ in (\ref{eq:bivariate_Rayleigh})
is large. For computation in that case, we can replace the above $I_{0}(\cdot)$
with its approximation $I_{0}(x)\approx\exp(x)/\sqrt{2\pi x}$, for
large $x$ {[}\citet{Bessel_approx}{]}. 

From (\ref{eq:Rayleigh}) and (\ref{eq:Rayleigh joint}), the conditional
pmf of the quantized Rayleigh fading amplitude can be defined as $f(l_{i}|l_{i-1})=Mu_{l_{i}}(\TBold_{K}l_{i-1})$,
where $l_{i}\in\{\boldsymbol{\epsilon}(1),\ldots,\boldsymbol{\epsilon}(K)\}$
is label variable pointing to $K$ quantized levels $\overline{g}_{i}=[\overline{g}_{1,i},\ldots,\overline{g}_{K,i}]'$
at time $i$ and $\TBold_{K}$ is positive $K\times K$ transition
probability matrix of an homogeneous Markov chain, with elements:
$t_{k,m}=\frac{\Pr[\overline{g}_{k,i},\overline{g}_{m,i-1}]}{\Pr[\overline{g}_{m,i-1}]}=\Pr[\overline{g}_{k,i},\overline{g}_{m,i-1}]\times K$,
$1\leq k,m\leq K$. The columns of $\TBold_{K}$ are then normalized
to 1, by definition, in order to avoid any numerical computation's
error. Finally, note that all initial probabilities of this Markov
chain in this equi-probable scheme are equal to $1/K$, from (\ref{eq:F(K)}).

\bibliographystyle{klunamed}
\bibliography{chap1,chap2,chap3,chap4,chap5,chap6,chap7,chap8,chap9,Appendix}

\end{document}